\documentclass[11pt]{article}
\usepackage[left=1in,top=1in,right=1in,bottom=1in, nohead]{geometry}
\usepackage[colorlinks=true,backref=true,hyperindex,citecolor=blue,linkcolor=blue]{hyperref}
\usepackage{booktabs}
\usepackage{subcaption}
\usepackage{graphicx}

\usepackage{amssymb, bbm}

\usepackage{multirow}
\usepackage{multicol}
\usepackage{amsmath}
\usepackage{float}
\usepackage[title]{appendix}
\usepackage{color}
\usepackage{tikz}
\usepackage[justification=centering]{caption}
\usepackage{setspace}
\usepackage{colortbl}
\usepackage{varwidth}
\usepackage{authblk}
\usepackage[numbers]{natbib}
\usepackage{mathrsfs}
\usepackage{amsthm}
\usepackage{commath}
\usepackage{amsfonts}
\usepackage{bm}
\usepackage{enumerate}
\usepackage{enumitem}
\usepackage{upgreek}
\usepackage{nicefrac}
\usepackage[normalem]{ulem}
\allowdisplaybreaks
\usepackage{extpfeil}
\usepackage{threeparttable}
\usepackage{mathtools}
\usepackage[linesnumbered,ruled]{algorithm2e}
\usepackage{algpseudocode}

\bibpunct{(}{)}{;}{a}{}{,}
\newtheorem{thy}{Theorem}
\newtheorem{lemma}{Lemma}
\newtheorem{corollary}{Corollary}

\newtheorem{assumption}{Assumption}
\newtheorem{defin}{Definition}

\theoremstyle{remark}
\newtheorem{remark}{Remark}
\newtheorem*{example}{Example}

\makeatletter
\newcommand*{\rom}[1]{\expandafter\@slowromancap\romannumeral #1@}
\makeatother

\title{Community Detection on a Randomly Growing Network}

\author{Jianxiang Wang and Min Xu\footnote{Corresponding author: \texttt{mx76@stat.rutgers.edu}} \\ 
Statistics Department \\
Rutgers University \\
New Brunswick, NJ, USA}
\date{}

\begin{document}
\maketitle

\begin{abstract}
We study community detection on Markovian random networks outside of the Stochastic Block Model (SBM) framework. Specifically, we consider a random network growth process which generates $K$ separate preferential attachment trees and connects them with Erd\H{o}s--R\'enyi edges, so that each tree represents a community and each node inherits the label of the tree to which it belongs. This model is able to produce many features of real-world networks that are improbable under SBM, such as power law degree distribution and the existence of chains and hubs. Given only the final graph, without any knowledge of the growth process, we seek to recover the unobserved community membership of the nodes. We first prove that it is impossible for any algorithm to consistently recover the community label of all the nodes. However, we design algorithms that are provably able to recover the community labels of subsets of central nodes, for several different notions of node centrality, such as arrival time or degree. Our procedure consists of two stages where, in the first stage, we classify high degree nodes and then, in the second stage, extend the community assignments to the remaining vertices. Numerical experiments and a real data application on a coauthorship network demonstrate the effectiveness of our proposed approach.

% Community detection is a fundamental task in network analysis.
% Classical approaches based on the stochastic block model (SBM) assume a static network structure and often fail to capture growth mechanisms, such as preferential attachment, commonly observed in real-world networks.
% To address this limitation, we model the observed network as a superposition of a planted preferential attachment forest and an Erd\H{o}s--R\'enyi random graph.
% We propose a two-step algorithm that identifies core components via graph pruning, followed by a label propagation step that extends community assignments to all vertices.
% We show that exact recovery is impossible under constant order noise and instead establish guarantees for consistent recovery on informative vertex subsets.
% Our work introduces a growing network framework for community detection and characterizes regimes where recovery is possible.
% Numerical experiments and a real data application to a coauthorship network support the effectiveness of the proposed approach.
\end{abstract}

\noindent%

% {\color{red}
% \begin{enumerate}
% \item Theory Add two remarks explaining: why the layer-2 result is limited to the $\alpha=0$ (LPA) case, and why the current theory is limited to the first 1–2 layers.
% \item Simulation Go through the numerical analysis section and rewrite parts of the expressions to make the interpretation clearer.
% \item Consider revising the tables/figures to improve readability, as there are currently too many of them. 
% \item Algorithm Revise the first section, since what we currently describe as the simple algorithm is mainly the graph pruning part, and we are not sure whether it should remain in its current section.
% \item Rewrite the Monte Carlo method section from a model-generation perspective. 
% \item Go through the paper and consistently update the algorithm name to SPAR, and check where “core node” should be replaced with “anchor node”. 
% \item Method Clarify the relationship between the model output and the root-node confidence interval output.
% \item In Section~\ref{secsec: The simple Two-Step Algorithm}, add a Proposition stating that $\hat{K} = K$ with high probability. 
% \end{enumerate}

% }

\section{Introduction}

Community detection is a central problem in network data analysis. It seeks to cluster network nodes based on connectivity patterns, with applications in various domains such as biology \citep{krzakala2013spectral,traag2019louvain,yu2007graph}, social science \citep{ji2022co,akbaritabar2020italian}, computer science \citep{velickovic2018deep,wang2016botnet}, and more \citep{de2020unveiling}. The predominant approach to community detection among statistical researchers has focused on the so-called Stochastic Block Model, which is a random graph model where edges are added independently, with probability dependent on whether the edge is between or within community. The order in which the edges are added is of no importance to the model. 

%Network data analysis has attracted large amount of attention across diverse fields ranging from biology \citep{zeng2018prediction}, social science \citep{hunter2008goodness,van2018social} to computer science \citep{grover2016node2vec}.
%Network data analysis has received significant attention from many different disciplines including biology~\citep{zeng2018prediction}, social science \citep{hunter2008goodness,van2018social}, computer science \citep{grover2016node2vec}, and statistics \citep{ji2022co,lei2021network}. 
%Within network data analysis, one of the most important tasks is that of community detection, where the aim is to cluster network nodes based on connectivity patterns.
%A wide range of methodologies for community detection, predominantly developed within the framework of the stochastic block model (SBM) \citep{holland1983stochastic} and its variants \citep{jin2015fast,aicher2015learning,matias2017statistical}, have been extensively reviewed in \cite{abbe2018community,jin2021survey,de2017community}.

%Community detection in networks is largely built upon the framework of the stochastic block model (SBM) \citep{holland1983stochastic} and its variants \citep{jin2015fast, aicher2015learning, matias2017statistical}. Within this framework, standard classification approaches include spectral clustering \citep{amini2024hierarchical,li2022hierarchical}, likelihood-based inference \citep{zhen2023community, cerqueira2023pseudo}, or a hybrid approach that combines spectral initialization with likelihood refinement \citep{xu2020optimal,deng2024distributed}.

In contrast, real-world networks are often formed from a growth process where vertices and edges are
added sequentially. This explains the prevalence of certain network features, such as a power-law degree distribution \citep{barabasi1999emergence} and the existence of long chains and pendants (a node that connects to many degree 1 nodes). Examples of graphs that exhibit such features include co-authorship networks \citep{ji2016coauthorship} and gene co-expression networks \citep{cline2007integration}; see Figure~\ref{fig: co-authorship graph info}.

\begin{figure}[h]
    \centering
    % First subfigure: Hard classification
    \begin{subfigure}[b]{0.26\textwidth}
        \centering
        \vspace{-0.4cm}
        \includegraphics[width=\linewidth]{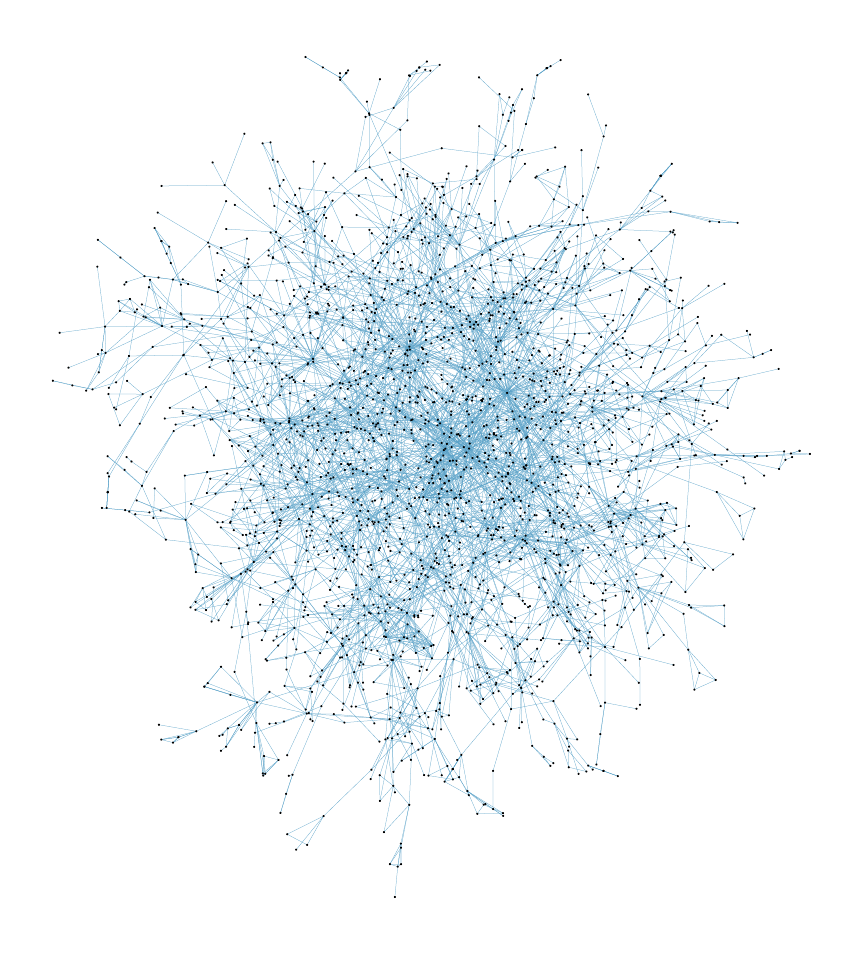}
        \caption{Observed graph}
        \label{fig: co-authorship graph}
    \end{subfigure}
    \hfill
    \begin{subfigure}[b]{0.45\textwidth}
        \centering
        \vspace{-0.4cm}
        \includegraphics[width=\linewidth]{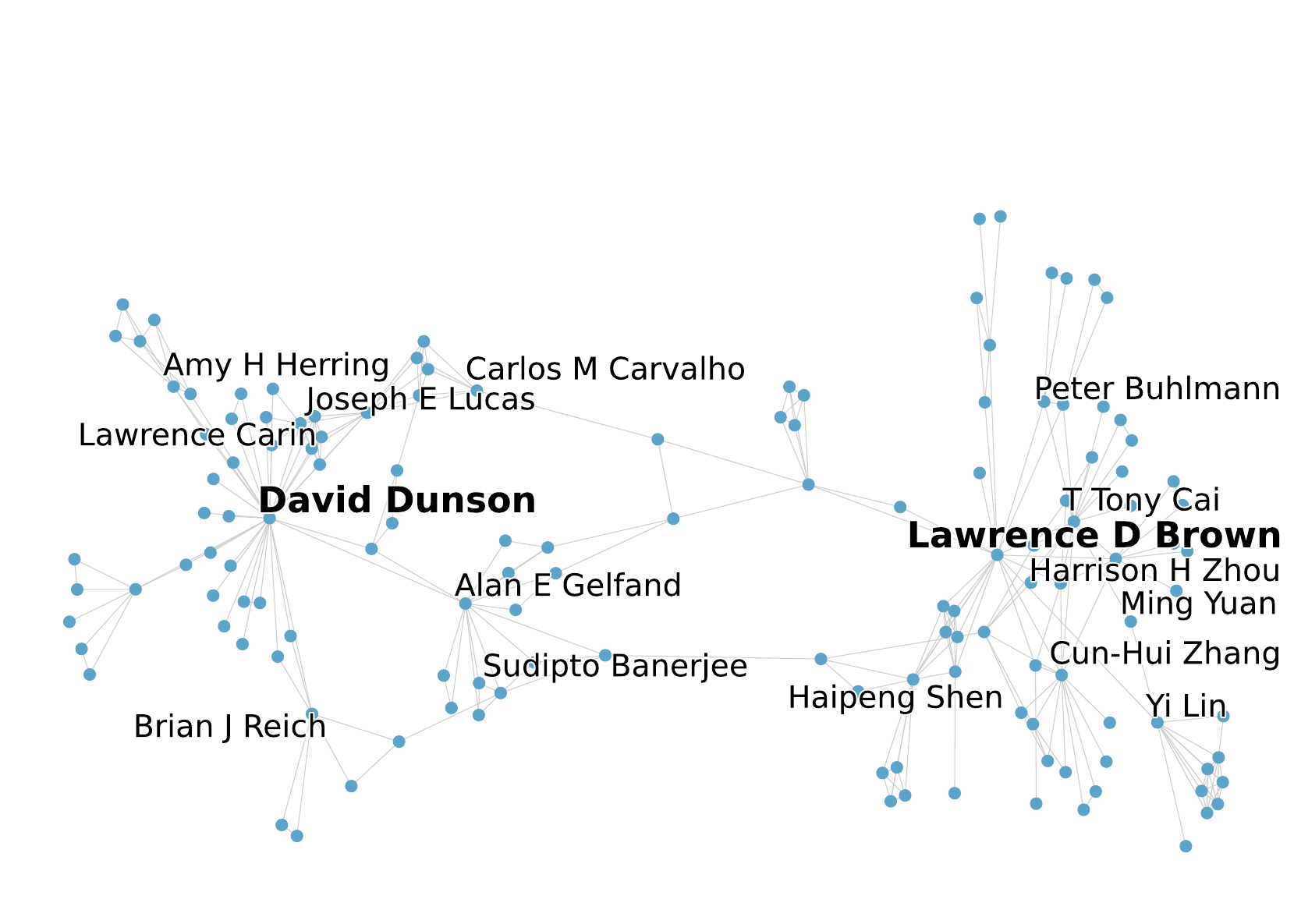}
        \caption{Sub network}
        \label{fig: co-authorship graph local}
    \end{subfigure}
    \hfill
    % Second subfigure: Soft classification
    \begin{subfigure}[b]{0.26\textwidth}
        \centering
        \vspace{-0.4cm}
        \includegraphics[width=\linewidth]{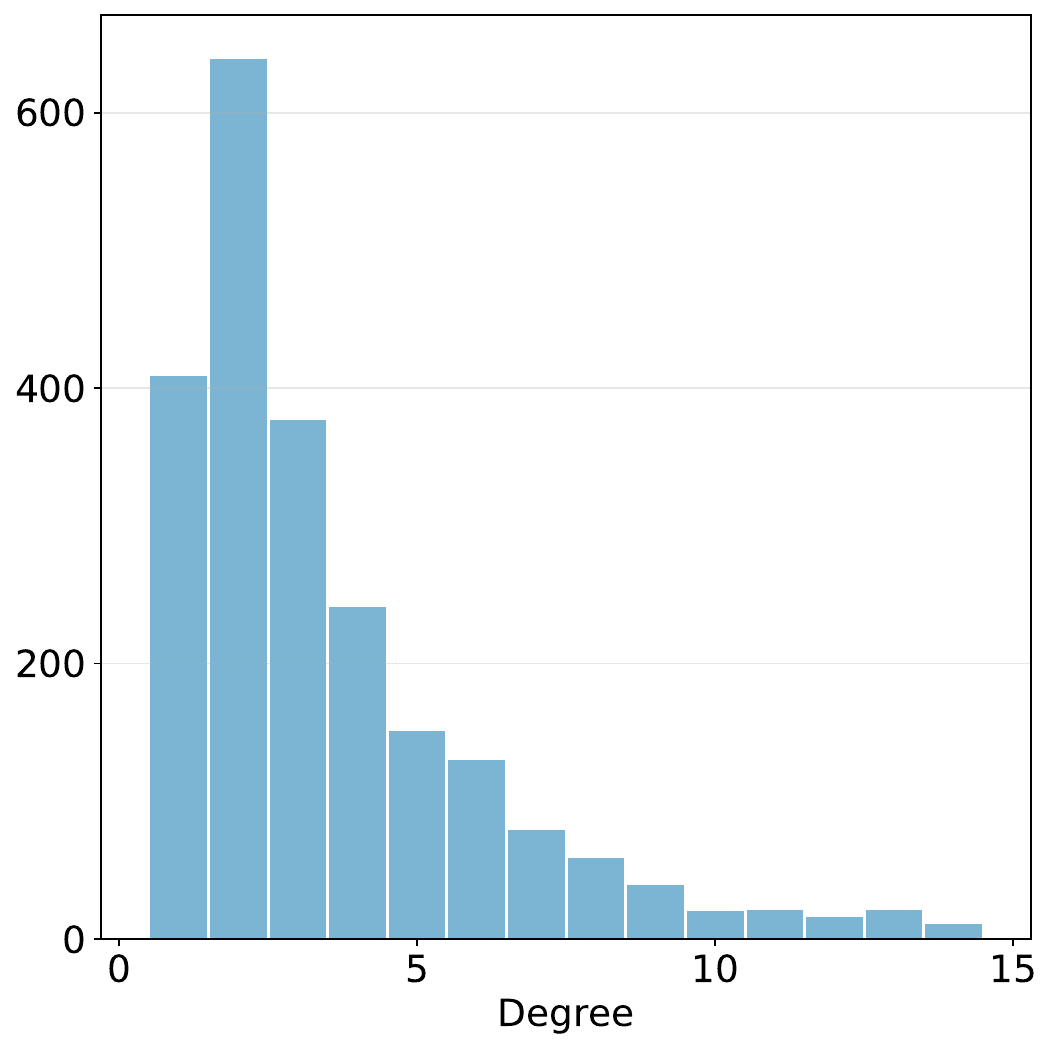}
        \caption{Degree distribution}
        \label{fig: co-authorship graph degree}
    \end{subfigure}
    
    \caption{
     Co-authorship network from \cite{ji2016coauthorship}. 
    \textbf{(a)} Observed graph topology, where vertices correspond to authors and edges represent co-authorship relations. \textbf{(b)} Sub-network consisting of all nodes within graph distance two of Lawrence D.~Brown or David Dunson.
    \textbf{(c)} Degree distribution of the network. For readability, authors with degrees exceeding $15$ are omitted.
    }
    \label{fig: co-authorship graph info}
\end{figure}

The growth process of a network leaves distinctive topological structures which contain information useful for estimating the communities. In the statistician co-authorship network from \cite{ji2016coauthorship} for example, a research community tends to coalesce around a small number of influential researchers, creating a central core of high-degree nodes in that community. This is illustrated in Figure~\ref{fig: co-authorship graph local}, where we show the subnetwork of all the nodes within distance two of either Larry Brown, the renowned theoretical statistician, or David Dunson, the famous Bayesian statistician. A clear community of theoretical researchers forms around the former, and a separate community of Bayesian modelers around the latter. Identifying hub nodes could substantially aid community detection, but existing SBM-based approaches may overlook such growth-induced features due to the static nature of the block model. Our work addresses this gap. 

%In contrast, our work aims to develop community detection methods that can exploit the information encoded in the underlying network formation process. 

To formalize the definition of our problem and to guide the development of our methodology, we consider a random graph model that captures both sequential growth and community structure in a simple way. Our random graph model, which we call the planted forest model, first generates $K$ independent tree graphs using the preferential attachment process, where each node arrives sequentially and attaches to an existing node with probability dependent on the degree of the existing node. We let each tree represent a separate community so that the nodes in that tree take on the community label of that tree. We refer to this collection of disjoint trees as a forest graph. We then take the superposition of this forest graph with an independent Erd\H{o}s--R\'enyi (ER) noise graph collapsing any multi-edges. Given an observation of the final graph, without any knowledge of the growth process or whether the edges belong to the forest graph or the noise graph, we seek to recover the unobserved community labels of the nodes.

Our first main contribution is an information-theoretic lower bound showing that no estimator can consistently recover the community labels of all the nodes even when the ER probability is at a relatively low level of $O(1/n)$. This is because the misclustering error is driven by the later-arriving nodes, which cannot be classified correctly due to their peripheral location in the graph. 

Despite the impossibility of global recovery, consistent recovery may still be possible \emph{locally}, on a specific \emph{subset} of the nodes. In our second main contribution, we propose a community recovery algorithm and prove that, asymptotically, it has either zero or vanishing misclustering errors on various subsets of central nodes, defined through three different but related notions of centrality: the arrival time of the node, the degree, and the distance of the node to the initial root node of the community. Our guarantees require only that community sizes be of comparable order and that the ER noise probability is not so high that the degrees of the nodes become uncorrelated with their arrival time. 

Our estimation algorithm consists of two stages. We first iteratively remove low-degree nodes to reveal $K$ disjoint connected components, which we refer to as anchor components. Intuitively, the anchor components capture the high-degree core of each of the communities and supply the initial seeds for our clustering. In the second stage, we propagate the community labels from the anchor components to the entire network using a range of recovery methods, the simplest of which is to label the nodes based on their shortest-path distance to the anchors. 

Although our algorithm is motivated by the planted forest model, it depends on the model only loosely, mainly through the assumption that each community has a high-degree core that can be isolated by iterative degree-based pruning. In particular, it does not rely on the assumption that the connectivity structure within a community is tree-shaped. We demonstrate the practical effectiveness of our approach on a co-authorship network of statisticians, where we recursively apply our community detection procedure and recover a hierarchical clustering of research communities that largely aligns with recognized sub-fields of statistics.

%Our theoretical analysis assumes only that the community sizes are somewhat balanced and that the sparsity level of the Erd\H{o}s--R\'enyi (ER) noise graph is not too high.

The remainder of this paper is organized as follows. In Section~\ref{sec: model}, we introduce the planted forest model and present an impossibility result that frames the objective of our work. 
Section~\ref{sec: algorithm} details our SPAR (Selective Pruning with Anchor-based Recovery) algorithm. The theoretical guarantees of the algorithm are established in Section~\ref{sec: theory}. Section~\ref{sec: simulation} validates the method through simulation experiments under varied preferential attachment models and noise levels. 
Finally, Section~\ref{sec: case study} applies our approach to a co-authorship network, illustrating its practical performance.

\subsection{Literature Review}
Community detection in networks has been predominantly studied within the framework of SBM. 
Within this paradigm, widely used methods include spectral clustering \citep{amini2024hierarchical,li2022hierarchical}, likelihood-based inference \citep{zhen2023community,cerqueira2023pseudo}, and hybrid approaches that combine spectral initialization with likelihood refinement \citep{xu2020optimal,deng2024distributed}.
By contrast, relatively little work has addressed community detection in networks formed through sequential growth processes. \cite{ben2025inference} and \cite{hajek2019community} extend classical preferential attachment models to multi-class settings in which incoming nodes connect to existing vertices according to both degree and latent class labels. Our setting differs in two respects: first, we consider an entirely different notion of community structure--they model communities through class-dependent attachment probabilities, whereas our communities arise from disjoint growth processes--and second, most of their method and theory assume that the node arrival times are known, whereas we take the arrival ordering to be unobserved. \cite{crane2024root} is most closely related to our work; they propose a Bayesian method
for community detection on a similar model without any theoretical analysis. Our work differs in three crucial ways: first, our method comes with rigorous theoretical guarantees; second, our estimator is provably computationally efficient (see Remark~\ref{rem:runtime}) whereas \cite{crane2024root} uses MCMC with no mixing time bounds; lastly, our approach depends only on the high-degree core structure within each community and may be less sensitive to deviation from the model. 

%While these models capture preferential-attachment features observed in real networks, their theoretical analyses typically assume that node arrival times are observed. This requirement substantially limits their applicability in practice, where inference is often based on a single, time-aggregated snapshot of the network.

Another related line of work in network data analysis considers inference on special vertices, which includes ``root finding" \citep{bubeck2017finding, banerjee2022root,banerjee2023degree,crane2021inference,crane2024root,addario2025leaf} and ``diffusion source identification" \citep{dawkins2021diffusion,li2021propagation,dong2022wavefront,shah2011rumors}.
Most work in the root finding literature assumes that the observed network is generated entirely by a sequential growth mechanism, whereas most work in diffusion source identification literature assumes a fixed underlying network over which a stochastic diffusion process evolves.
Despite these differing modeling assumptions, both lines of work focus on identifying or constructing confidence sets for a single distinguished vertex (the root node or the diffusion source) of a sequential growth process based on a single observed network snapshot.
Among these works, \cite{addario2025leaf} employs a leaf-stripping procedure to construct confidence sets for root nodes, which closely parallels the graph pruning step of our algorithm. 
Our objective, however, is fundamentally different: rather than identifying a single special vertex, we aim to recover and cluster multiple planted components generated by distinct growth processes.

% A key distinction between our work and the aforementioned studies is that we consider networks in which multiple sequential growth processes are intermixed, and aim to classify vertices according to the growth process that generated them using only a single observed noisy network snapshot.
% Our analysis provides both algorithmic procedures and theoretical guarantees that characterize how misclassification behavior depends on vertex-specific features, including arrival order and structural position, under perturbations induced by preferential attachment dynamics and additional random noise.

\section{Model and Problem Formulation}
\label{sec: model}
In this section, we first introduce the notation and definition for the Planted Forest model.
Then, we formalize the problem setup for community detection under the proposed model.

\subsection{Notation}
Throughout the paper, we will adopt the convention $ 0/0: = 0$. We define the natural numbers $\mathbb{N} = \{1,2,3,\ldots\}$.
For an integer $n$, we write $[n] := \{1, 2, \ldots, n\}$ and let $S_n$ denote the set of all permutations on $[n]$. 
For a discrete set $A$, we write $|A|$ as the cardinality of $A$.
For two sequences of positive scalars $\{a_n\}_{n=1}^\infty$ and $\{b_n\}_{n=1}^\infty$, we write
$a_n=o(b_n)$ and $a_n=O(b_n)$ if $a_n/b_n$ converges to zero and $a_n/b_n$ is bounded, respectively. 
Furthermore, we write $a_n \asymp b_n$ if $a_n/b_n=O(1)$ and $b_n/a_n=O(1)$ hold. 
For a sequence of random variables $\{X_n\}_{n=1}^\infty$,
we write $X_n=o_p(b_n)$ and $X_n=O_p(b_n)
$ if $X_n/b_n$ converges to zero and is bounded in probability, respectively. We write $X_n \asymp_p b_n$ if $X_n/b_n = O_p(1)$ and $b_n/X_n = O_p(1)$. 

We denote a graph by $\boldsymbol{g}= \left(V, E\right)$, where $V$ and $E \subset V \times V$ represent the set of nodes and the set of undirected edges, respectively.
For any two vertex sets $V_1, V_2 \subset V\left(\boldsymbol{g}\right)$, we define the edge between them as  $E_{\boldsymbol{g}}\left(V_1,V_2\right):=\left\{\left(u,v\right) \mid \left(u,v\right)\in E\left(\boldsymbol{g}\right), u\in V_1, v\in V_2\right\}$. 
When $V_1=V_2$, we abbreviate
$E_{\boldsymbol{g}}\left(V_1,V_1\right)$ as $E_{\boldsymbol{g}}\left(V_1\right)$.
We write $ \operatorname{deg}_{\boldsymbol{g}}\left(u\right)$ to denote $u$'s degree in graph $\boldsymbol{g}$, i.e., the number of nodes directly connected to $u$ via an edge.
For a vertex set $V_1 \subset V$, we write $\operatorname{deg}_{\boldsymbol{g}}(V_1):= \max_{u \in V_1} \operatorname{deg}_{\boldsymbol{g}}\left(u\right)$.

For two labeled graphs $\boldsymbol{g}$ and $\boldsymbol{g}^{\prime}$ defined on the same vertex set $V$, we write $\boldsymbol{g}+\boldsymbol{g}^{\prime}$ to denote the graph obtained by taking the union of their edge sets, i.e., $E\left(\boldsymbol{g}+\boldsymbol{g}^{\prime}\right)=E\left(\boldsymbol{g}\right) \cup E\left(\boldsymbol{g}^{\prime}\right)$, with multi-edges collapsed. We write $\boldsymbol{g} \subset \boldsymbol{g}^{\prime}$ if $\boldsymbol{g}$ is a subgraph of $\boldsymbol{g}^{\prime}$, meaning $E\left(\boldsymbol{g}\right) \subset E\left(\boldsymbol{g}^{\prime}\right) $.
When $\boldsymbol{g}$ and $\boldsymbol{g}^{\prime}$ are two disjoint graphs (no shared vertices or edges), we use $\boldsymbol{g} \oplus \boldsymbol{g}^{\prime}$ to denote their disjoint union. 
In addition, for $v \in V(\boldsymbol{g})$, we write $\boldsymbol{g} \setminus v$ for the subgraph obtained by removing the vertex $v$ together with all edges incident to $v$. Given a subset $\tilde{V} \subset V(\bm{g})$, we define $\bm{g} \cap \tilde{V}$ to be the subgraph induced by the nodes in $\tilde{V}$. 
Given a graph $\bm{g}$, we define $\mathcal{C}(\bm{g})$ to be the set of connected components of $\bm{g}$; if $\bm{g}$ is connected then $\mathcal{C}(\bm{g}) = \{ \bm{g}\}$.

%In contrast, we use the notation $\boldsymbol{g} \oplus \boldsymbol{g}^{\prime}$ to denote the disjoint union of two graphs $\boldsymbol{g}$ and $\boldsymbol{g}^{\prime}$. This operation assumes that the vertex sets $V(\boldsymbol{g})$ and $V(\boldsymbol{g}^{\prime})$ are disjoint; the resulting graph has vertex set $ V(\boldsymbol{g}) \cup V(\boldsymbol{g}^{\prime})$ and edge set $E\left(\boldsymbol{g}\right) \cup E\left(\boldsymbol{g}^{\prime}\right)$, with no edges between $\boldsymbol{g}$ and $\boldsymbol{g}^{\prime}$. 

Throughout the paper, we use capital font (e.g., $\boldsymbol{G}$) to denote random objects and lowercase font to denote fixed objects. Graphs are represented via bold font.

\subsection{Planted Forest Model}

We first define the affine preferential attachment process for a randomly growing tree~\citep{crane2021inference}, which models the growth process of a single community.

\begin{defin}
\label{defin: APA}
Let $V_n$ be a set of $n$ node labels and let $\pi \,:\, [n]  \rightarrow V_n$ be the arrival ordering so that $\pi_t$ denotes the $t$-th node to be added to the tree. Let $\alpha > -1$; we say that a random tree $\boldsymbol{T}$ with $n$ nodes has the affine preferential attachment distribution $\mathrm{APA}(\alpha, \pi)$ if it is generated according to the following process: start with a single node $\pi_1$ as the root node and at every time step $t = 2, \ldots, n$, add a new node $\pi_t$ and connect it to an existing node $w \in \{\pi_1, \ldots, \pi_{t-1}\}$ chosen with probability
\begin{equation}
\label{defin: attachment rule 1}  
w \mapsto \frac{ \mathrm{deg}_{\bm{T}_{t-1}}(w) + \alpha}{2(t-2)+ (t-1)\alpha},
\end{equation}
where $\mathrm{deg}_{\bm{T}_{t-1}}(w)$ is the degree of node $w$ at time $t-1$. 

\end{defin}

In affine preferential attachment, higher-degree nodes are more likely to attract new neighbors so that we get a rich-get-richer phenomenon. One important consequence of this is that the APA random graph tends to exhibit strong degree heterogeneity: asymptotically, the percentage of nodes with degree $k$ is approximately proportional to $\frac{1}{k^{3 + \alpha}}$ \citep[][Section 8.4]{van2024random}. The parameter $\alpha$ controls the strength of the degree preference: larger values of $\alpha$ diminish the degree preference in the attachment probabilities, with the limiting case $\alpha = \infty$ corresponding to uniform attachment, in which each new node attaches to an existing node chosen uniformly at random.

To model multiple communities, we associate each community with an independent tree generated by the APA model, producing a forest $\bm{F}_n$ of disjoint community-trees. To model between-community edges and to ensure that the overall graph is connected, we superpose $\bm{F}_n$ with an independent Erd\H{o}s--R\'enyi random graph to obtain the final observed graph $\bm{G}_n$. The tree structure within each community is a simplification made for model parsimony. Definition~\ref{defin: APA} can be extended to allow for multiple edges per arrival, but doing so requires additional modeling choices, such as deciding how many edges to add at each time step, that distract from our main focus on the interplay between the growth dynamics and community detection.

% \begin{defin}
% \label{defin: APA}
% The affine preferential attachment tree model, which we denote by $\mathrm{APA}\left(\alpha, n\right)$ for parameters $\alpha \in \mathbb{R}$, generates an increasing sequence $\boldsymbol{T}_1 \subset \boldsymbol{T}_2 \subset \cdots \subset \boldsymbol{T}_n$ of random trees where $\boldsymbol{T}_t$ is a tree with $t$ nodes and where nodes are labeled by their arrival time so that $V(\boldsymbol{T}_t) = [t]$. The first tree $\boldsymbol{T}_1 = \{1\}$ is a singleton and for $t > 2$, we define the transition kernel $\mathbb{P}(\boldsymbol{T}_t \mid \boldsymbol{T}_{t-1})$ in the following way: given $\boldsymbol{T}_{t-1}$, we add a node labeled $t$ and a random edge $(t, w_t)$ to obtain $\boldsymbol{T}_t$, where the existing node $w_t \in [t-1]$ is chosen with probability
% \begin{equation}
%     \label{eq: APA}
%     \frac{\beta \operatorname{deg}_{\boldsymbol{T}_{t-1}}(w_t) + \alpha}{\beta 2(t-2) + \alpha(t-1)}.
% \end{equation}
% For tree $\boldsymbol{T}_n$, we define a latent arriving order function ${\pi}$, where ${\pi}_t\left(\boldsymbol{T}_n\right)$ denotes the $t$th arriving node in graph $\boldsymbol{T}_n$ for all $1\le t\le n, t\in \mathbb{N}$.
% Furthermore, we denote by $\pi_{i:j}\left(\boldsymbol{T}_n\right)$ the set of vertices arriving from time $i$ to time $j$ in the tree $\boldsymbol{T}_n$. Formally,
% %
% \begin{equation}
% \label{defin: pi}
%     \pi_{i:j}\left(\boldsymbol{T}_n\right):=\left\{{\pi}_t\left(\boldsymbol{T}_n\right)\mid i\le t\le j, t\in \mathbb{N}\right\}.
% \end{equation}
% %
% \end{defin}

\begin{defin}
\label{defin: PF}
Let  $\alpha>-1$, $V_n$ be a set of $n$ node labels and let $\ell : V_n \rightarrow [K]$ be the community label function so that $\ell(v)$ is the community membership of node $v$. Let $\pi \,:\, V_n \rightarrow [n]$ be the arrival ordering of all $n$ nodes. For each community $k \in [K]$, define $V^k := \{ u \in V_n \,:\, \ell(u) = k \}$, $n_k := |V^k|$, and $\pi^k \,:\, [n_k] \rightarrow V^k$ as an ordering of $V^k$ induced by $\pi$. For each $k \in [K]$, let $\bm{T}^k$ be an independent random tree generated according to the $\mathrm{APA}(\alpha, \pi^k)$ distribution. Let $\boldsymbol{F}_n=\oplus_{i=1}^K \boldsymbol{T}^i$ be the disjoint union of the $K$ trees. 

Let $\theta \in [0, 1]$ and let $\boldsymbol{R}_n$ be an Erd\H{o}s--R\'enyi random graph with edge probability $\theta$ defined on $V_n$. Define
\begin{equation}
\label{decomposition}
  \boldsymbol{G}_n := \boldsymbol{F}_n + \bm{R}_n,\quad \text{ collapsing multi-edges.}
\end{equation}

We then say that $\boldsymbol{G}_n$ is distributed according to the planted forest distribution $\mathrm{PF}(\alpha, \theta, \ell, \pi)$. As a short hand, we also write $\pi(\bm{T}^k) \equiv \pi^k$. 
\end{defin}

The random tree $\bm{T}^k$ and the permutation $\pi^k$ capture the growth process of community $k$. Since the model depends on the global arrival ordering $\pi$ only through the community-specific orderings $\pi^1, \ldots, \pi^K$, the full $\pi$ is not identifiable. It is however still convenient to have the notion of a global arrival ordering. We see that the forest graph $\bm{F}_n$ contains all the information about the community membership and that $\bm{R}_n$ may be regarded as noise. 

\begin{remark}
The planted forest model is basically equivalent to the multiple-roots PAPER (Preferential Attachment Plus Er\"{o}s--R\'{e}nyi) model proposed by \cite{crane2024root}. The main difference is that we take the community membership $\ell(\cdot)$ to be a fixed parameter. We use the name ``planted forest" to emphasize the community structure induced by the disjoint collection of tree graphs. 
\end{remark}

\begin{remark}
Our formulation of the planted forest model as having disjoint subgraphs (defined through the forest $\bm{F}_n$) embedded in an Erd\H{o}s--R\'enyi graph $\bm{R}_n$ has strong parallels with the definition of SBM. To see this, let $\bm{G}_n$ be a SBM random graph with two blocks, with $p, q \in (0, 1)$ as the within-block and between-block edge probabilities, and with assortativity so that $p > q$. Then, we have that $\bm{G}_n = \bm{B}^1 + \bm{B}^2 + \bm{R}_n$ where $\bm{B}^1$ is an Erd\H{o}s--R\'enyi random graph on the first block with edge probability $p-q$, $\bm{B}^2$ is an ER random graph on the second block also with edge probability $p-q$, and $\bm{R}_n$ is an Erd\H{o}s--R\'enyi random graph on all $n$ nodes with edge probability $q$. We may analogously view $\bm{B}^1 + \bm{B}^2$ as signal and $\bm{R}_n$ as noise. 
\end{remark}

In our definition of the $\mathrm{PF}(\alpha, \theta, \ell, \pi)$ model, the community membership function $\ell(\cdot)$ is viewed as a fixed parameter; $\ell(u) = k$ means that the node $u$ is in community-tree $\bm{T}^k$. However, for our lower bound analysis in Theorem~\ref{thm: impossible} and algorithmic development in Section~\ref{alg: Model recovery}, it is useful to define an alternative model where $\ell(\cdot)$ is random. In this model, the nodes enter the network according the global ordering $\pi$ and the $t$-th node $\pi_t$ arrives, it chooses a parent-node at random from the set of all existing nodes $\{\pi_1, \ldots, \pi_{t-1}\}$. The new node $\pi_t$ then joins the same community-tree as its parent. We refer to this model as the random community planted forest model (RC-PF) and define it formally below.

\begin{defin}
\label{defin: RC-PF}
Let  $\alpha>-1$, $\theta \in [0, 1]$, and $K \in \mathbb{N}$. Let $\pi \,:\, V_n \rightarrow [n]$ be an arrival ordering. For each $k \in [K]$, we initialize tree $\bm{T}^k$ as a graph with two nodes $\pi_{2k-1}, \pi_{2k}$ and a single edge between them and assign $\ell(\pi_{2k-1})=\ell(\pi_{2k}) = k$. We initialize $\bm{F}$ as the disjoint collection of $\bm{T}^1, \ldots, \bm{T}^K$. 

Then, for each $t = 2K+1, \ldots, n$, we connect the new node $\pi_t$ to an existing node $w \in \{\pi_1,\ldots,\pi_{t-1}\}$ chosen with probability
\begin{equation}
\label{defin: attachment rule 2}
w \mapsto \frac{ \mathrm{deg}_{\bm{F}_{t-1}}(w) + \alpha}{2(t-K-1) + (t-1)\alpha},
\end{equation}
where $\mathrm{deg}_{\bm{F}_{t-1}}(w)$ is the degree of node $w$ at time $t-1$. We set $\ell(\pi_t) = \ell(w)$. 

Let $\bm{F}_n$ be the resulting random forest graph after all $n$ nodes are added. Let $\bm{R}_n \sim \text{Erd\H{o}s--R\'
enyi}(\theta)$ be an independent ER graph defined on the same set of vertices $V_n$ and, as before, define
\[
\bm{G}_n = \bm{F}_n + \bm{R}_n, \quad \text{collapsing multi-edges}
\]
We then say that $\bm{G}_n$ is distributed according to the $\mathrm{PF}(\alpha, \theta, \pi, K)$ model. 
\end{defin}

We note that the RC-PF model produces a random community membership function $\ell$ so that the size of each community is random. A P\'olya-urn argument (see e.g. Section 4.3.2 in~\cite{durrett2019probability}), where we view each ball as one node and two edge-endpoints, shows that the community size proportions $\bigl( \frac{n_1}{n}, \frac{n_2}{n}, \ldots, \frac{n_K}{n}\bigr)$ converge in distribution to $\operatorname{Dirichlet}(\frac{2+2\alpha}{2+\alpha}, \frac{2+2\alpha}{2+\alpha}, \ldots, \frac{2+2\alpha}{2+\alpha})$ supported on the $K-1$ dimensional simplex. 

We also note that the RC-PF model is almost identical to the fixed $K$ multiple roots PAPER model in~\cite{crane2024root}. The only difference is that, in the latter, each tree is initialized with a single root node with a self-loop to guarantee positive degree.

\subsection{Community Detection Problem}

We observe the final snapshot $\boldsymbol{G}_n$ generated by the $\mathrm{PF}(\alpha,\theta,\ell,\pi)$ model. The community detection task is to recover the community membership function $\ell$. 

To assess the performance of estimated community labels $\hat{\ell}$, we first define a discrepancy measure between two membership functions on the same graph. Since the community labels are arbitrary, the true membership function $\ell$ is identifiable only up to a permutation of the $K$ community labels. Therefore, we do not compare $\hat{\ell}$ with the true membership $\ell$ directly but rather with $\sigma \circ \ell$ where $\sigma \in {S}_K$ is a permutation of the $K$ labels. We choose $\sigma$ to minimize the Hamming distance between the two membership functions. 

\begin{defin}
\label{defin mismatch}
Let $K \in \mathbb{N}$, let $\bm{g}$ be a graph, and let $\ell_1, {\ell}_2 :V(\bm{g}) \rightarrow [K]$ be two labelings. We define the Hamming distance and the global misclustering error (viewing $\ell_1$ as the true clustering) respectively as
\begin{equation*}    
d^{\mathrm{Ham}}({\ell}_1,\ell_2):=\sum_{u\in V\left(g\right)}\mathbbm{1}\{{\ell}_1(u)\neq \ell_2(u)\},\quad d({\ell}_1,\ell_2) := \min_{\sigma \in S_K} d^{\mathrm{Ham}}({\ell}_1, \sigma \circ \ell_2). 
\end{equation*}

For a given subset of the nodes $V_0 \subseteq V(\bm{g})$, we also define the \emph{local} misclustering error with respect to $V_0$:
\begin{equation*}    d^{\mathrm{Ham}}_{V_0}({\ell}_1,\ell_2):=\sum_{u\in V_0}\mathbbm{1}\{{\ell}_1(u)\neq \ell_2(u)\}, \quad d_{V_0}({\ell}_1,\ell_2) := \min_{\sigma \in S_K} d^{\mathrm{Ham}}_{V_0}({\ell}_1, \sigma \circ \ell_2). 
\end{equation*}
\end{defin}

Let us first consider the goal of designing an estimator $\hat{\ell}$ of the community membership labels so that the expected misclustering error rate $\frac{1}{n} \mathbb{E}[ d(\hat{\ell}, \ell)]$ tends to 0 as $n$ increases, which is known as weak recovery or consistent recovery. The difficulty of this naturally depends on the noise level $\theta$ since weak recovery is clearly impossible if $\theta$ is close to 1. Somewhat surprisingly, we prove that even when $\theta$ is at a relatively small level of $\Omega(1/n)$, weak recovery is still unattainable for any estimator. 

\begin{thy}
\label{thm: impossible}
Let $\boldsymbol{G}_{n} \sim \mathrm{PF}(\alpha, \theta, \pi, K=2)$. Let $c > 0$ and suppose $\theta \geq \frac{c}{n}$. Then, for all sufficiently large $n$, there exists $\lambda := \lambda(c, \alpha) > 0$ such that 
\begin{equation*}
    \inf_{\hat{\ell}}  \, \mathbb{E}\Biggl[ \frac{d\bigl(\hat{\ell}(\mathbf{G}_n), \ell\bigr)}{n}\Biggr]\geq \lambda,
\end{equation*}
where the infimum is taken over all estimators of the community label, that is, over all function that on input $\boldsymbol{G}_{n}$ outputs a labeling function $\hat{\ell}(\mathbf{G}_n): V(\mathbf{G}_n) \rightarrow \{1, 2\}$, and where $d$ is the global misclustering error introduced in Definition~\ref{defin mismatch}.
\end{thy}

We defer the proof of Theorem~\ref{thm: impossible} to Section~\ref{sec: appendix model} of the appendix, but outline the high-level idea in Remark~\ref{rem: impossible proof} below. 

We see from Theorem~\ref{thm: impossible} that even when $\theta = \frac{0.1}{n}$ so that the expected number of noise edges is at most $0.05 \cdot (n-1)$ and much smaller than $n-2$, the number of edges in $\bm{F}_n$, the misclustering error rate of any estimator is still bounded away from 0. More generally, Theorem~\ref{thm: impossible} implies that to achieve weak recovery, we would need the number of noise edges to be of smaller order than the number of tree edges, which is too restrictive for real-world data. 

Rather than working in the unrealistic setting necessary to obtain weak recovery over all the nodes, we make a simple but important observation: Theorem~\ref{thm: impossible} does not preclude the possibility that we can correctly cluster a specific (and possibly random) \emph{subset} of the nodes $V_0$, that is, we may have a $V_0 \subset V(\bm{G}_n)$ such that the local misclustering error rate $\mathbb{E}\bigl[\frac{1}{|V_0|}d_{V_0}(\hat{\ell}, \ell)\bigr]$ goes to 0.

Indeed, in the proof of Theorem~\ref{thm: impossible}, we see that the misclustering error is driven by nodes that are peripheral in the sense that they are late arriving and may have a long graph distance to the early arriving nodes. This conforms with our intuition that for real-world networks, nodes that are non-central and not well connected to the ``core" of the graph may be much more difficult to cluster than centrally located nodes. The concept of graph core is studied in~\cite{zhang2015identification,naik2021sparse,miao2023informative,yanchenko2025statistical}, although the precise notion of graph core that they use is quite different and not based on growth dynamics.

%{\color{purple}

In light of Theorem~\ref{thm: impossible} and our subsequent discussion, our recovery results focus on bounding the local misclustering error rate $\mathbb{E}\bigl[ \frac{1}{|V_0|} d_{V_0}(\hat{\ell}, \ell)\bigr]$ for various different subsets $V_0$. Loosely speaking, we consider subsets that are ``informative" in that the nodes of $V_0$ play important roles in the formation history of the network or if they carry strong structural signals that overcome the random perturbations introduced by the Erd\H{o}s--R\'enyi noise. More precisely, we consider three classes of $V_0$. The first is where $V_0$ consists of $L$ earliest arriving nodes for some fixed $L$. The second is where $V_0$ consists of nodes with sufficiently high degree. The third is where $V_0$ consists of nodes with a small graph distance from the initial root nodes of any of the communities. 

Before giving an overview of our recovery method and guarantees, we first give a brief summary of the proof of Theorem~\ref{thm: impossible}. 

\begin{remark}
\label{rem: impossible proof}
We give a high-level sketch of the proof of Theorem~\ref{thm: impossible} in this remark. Suppose we fix a particular node $u$ and know the community membership of all the nodes that arrived prior to $u$, then the task of identifying the community label of $u$ is equivalent to that of looking at all the edges that connect $u$ to the prior arriving nodes and picking out the single tree edge from the noise edges. 

We prove that if $u$ is a late-arriving node, then it is impossible to identify the tree edge. To be more specific, suppose $(u, v)$ is a tree edge; we show that when the arrival time of $u$ is greater than $n/2$, there is a high probability that there exists a noise edge $(u, v')$ such that $v$ and $v'$ have the same degree but belong to different communities. Swapping the roles $(u, v)$ and $(u, v')$ (so that $(u, v')$ becomes a tree edge and $(u, v)$ becomes a noise edge) produces an outcome with the same likelihood, making the two cases indistinguishable to any estimation method. 

\end{remark}

\subsection{Overview of Recovery Guarantees}

% Our method
Our community recovery method, described in detail in Section~\ref{sec: algorithm}, first prunes the graph to reveal a set of $K$ connected components (we call these anchor components) and then uses the anchor components to propagate the community labels to the remaining nodes. Although this method produces a global clustering of all the nodes, we know from Theorem~\ref{thm: impossible} that the overall global misclustering error rate of any method is bounded away from 0. We therefore study the behavior of the estimated clustering local to a specific subset $V_0$ of the nodes. In order for the local misclustering error $d_{V_0}(\hat{\ell}, \ell)$ to be small, the subset $V_0$ needs to consist of ``central" nodes. Our main recovery results, Theorem~\ref{thm: first L} and Theorem~\ref{thm: layer 1/2}, look at several different subsets. We summarize the conclusion here. 

Let $\bm{G}_n$ be a random graph distributed according to the planted forest model $\mathrm{PF}(\alpha, \theta, \ell, \pi)$ where $\alpha > -1$. We assume that the noise level is not too high, namely that there exists an arbitrarily small constant $\delta > 0$ and constant $C_0>0$ such that $\theta \leq C_0 n^{- \frac{1+\alpha}{2+\alpha} - \delta}$. Some control of the noise level is necessary since the community estimation problem is clearly infeasible if $\theta$ is too large. We also assume that the community sizes are roughly balanced in that there exists a constant $H \geq K$ such that each community has size $n_k \geq \frac{n}{H}$. This is to exclude the case where one community contains a vast majority of the nodes, so that one can obtain a low misclustering error by simply outputting the trivial membership function where every node is placed in the same community. We then have the following: 

\begin{itemize}
\item[1.] (Informal statement of Theorem~\ref{thm: first L}) Fix any $L \in \mathbb{N}$ and let $V^{L}_0$ be the set of $L$ earliest arriving nodes of each tree/community so that $|V_0| = KL$. Then, for sufficiently large $n$, 
\[
d_{V^{L}_0}(\hat{\ell}, \ell) = 0, \quad \text{ with high probability.}
\]
\item[2.] (Informal statement of the first part of Theorem~\ref{thm: layer 1/2}) Let $V^{(1)}_0$ be the set of nodes whose tree-distance to the root nodes is at most 1 (i.e direct children). Then, 
\[
|V^{(1)}_0| \asymp_p n^{\frac{1}{2+\alpha}} \,\, \text{ and } \,\, \frac{d_{V^{(1)}_0}(\hat{\ell}, \ell)}{|V^{(1)}_0|} \lesssim n^{ - \frac{1+\alpha}{2+\alpha} - \delta}, \quad \text{with high probability}.
\]
\item[3.] (Informal statement of the second part of Theorem~\ref{thm: layer 1/2}) Let $V^{(2)}_0$ be the set of nodes whose tree-distance to the root nodes is at most 2. When $\alpha = 0$, 
\[
|V^{(2)}_0| \asymp_p n^{\frac{1}{2}}\log n, \quad \frac{d_{V^{(2)}_0}(\hat{\ell},\ell)}{|V^{(2)}_0 |} \lesssim n^{-\delta}, \quad \text{ with high probability}.
\]
\end{itemize}

\begin{remark}
\label{rem: K estimation 1}
For the most part, we assume that the number of communities $K$ is known and does not increase with $n$. Although we do provide a method that can provably estimate $K$ on the planted forest model (see Remark~\ref{rem: K estimation 2} and  Algorithm~\ref{alg: theory align} in Section~\ref{secsec: The simple Two-Step Algorithm} of the appendix), the estimator is tailored to the specifics of the model and is in practice difficult to tune. Estimating the number of communities on large real-world network data is known to be a challenging problem \citep{le2022estimating,ma2021determining,jin2023optimal,wang2017likelihood}; estimators that work in specific theoretical settings may not work well in practice. We defer a comprehensive study of how to best estimate $K$ on growing networks to future work. 
\end{remark}

\section{A Two-Stage Algorithm: SPAR}
\label{sec: algorithm}

In this section, we propose a two–stage procedure for community detection in the network growth model, which we call Selective Pruning with Anchor-based Recovery (SPAR). In stage one, we iteratively prune the graph based on degree and extract $K$ connected components which we refer to as anchor components. In stage two, we propagate the community labels to the entire network, either based on distance to the anchor components or on more complex model-based criteria. 

The high-level rationale is as follows. If we knew the noise edges, we could remove them to immediately separate the communities, but since they are unknown, we aim to remove the noise edges indirectly by pruning the low-degree nodes. To ensure that we do not prune too many nodes and shatter the community structure, we proceed iteratively and at every step, try to preserve the core connected components formed by the earliest arriving nodes of each community, which we know will also contain the highest degree nodes of that community-tree \citep{senizergues2021geometry, contat2024eve}. In this way, we will eventually remove all the nodes that act as noise-induced bridges between the communities and reveal the $K$ core connected components of the communities. We can then take these $K$ connected components as anchor components to assign cluster memberships to all the nodes.

Our method requires knowledge of $\alpha$ and $K$. The $\alpha$ parameter can be estimated via the EM algorithm derived in \cite{crane2024root} (see Section S3.1 therein). Estimation of $K$ is much more difficult. Although we provide a theoretically valid method to estimate $K$ on the planted forest model based on removing low degree nodes and examining the number of connected components that remain (c.f. Remark~\ref{rem: K estimation 1} and~\ref{rem: K estimation 2}), it is difficult to use on real data. An alternative approach is to take $K=2$ and perform hierarchical clustering. We adopted this second approach in our case study and were able to recover meaningful communities without extensive tuning.

\subsection{Stage I: Selective Pruning}

The goal of the first stage is to separate the graph into at least $K$ connected components by pruning low-degree nodes. A simple way to do this is to set a threshold $T$ and remove all nodes whose degree is less than $T$. The problem with this approach is that it is very sensitive to the choice of $T$ and even if we have a good way of choosing $T$, we could end up with more than $K$ components. This latter issue can be resolved by keeping only those components that contain at least one high-degree node, since we know that the earliest arriving nodes of each community must contain high-degree nodes. 

To circumvent the need to choose $T$, which carries the risk that choosing too large a threshold may remove too many nodes and fragment the community structure, we propose an iterative graph pruning procedure in Algorithm~\ref{alg: step I}. The algorithm has three integer input parameters $K, \tau, Q$ and either outputs a set of $K$ anchor connected components, each of which has size at least $Q$ and maximum-degree at least $\tau$, or outputs failure if no qualified anchor components can be found. In the failure case, we can decrease the input parameter; see Remarks~\ref{rem:anchor-degree-threshold} and \ref{remark: shrinkage parameter}. 

The algorithm is iterative and performs three operations at each step: we first check if the current graph can already be separated into at least $K$ connected components, each of which is of size at least $Q$ and whose maximum degree is at least the anchor-degree threshold $\tau$. If not, then we construct a set of candidate nodes for removal, remove the lowest degree nodes among the candidate set to get the pruned graph, and then go on to the next iteration. We note that the degree is always computed with respect to the input graph $\bm{G}_n$ and not updated as we prune the graph; this is mainly to make theoretical analysis more tractable -- updating the degree at each iteration may also work well in practice. 

The removal candidate set is designed to capture all the nodes that might have noise edges which connect different community-trees. It consists of all the nodes on all the cycles of the graph as well as any node whose removal would split the graph into a larger number of qualified connected components. Our specific definition of the removal candidate set endows Algorithm~\ref{alg: step I} the useful property that it is more conservative than simple degree thresholding in the sense that when both algorithms output $K$ connected components, Algorithm~\ref{alg: step I} will never remove more nodes than simple degree thresholding; this property is deterministic and does not depend on the probabilistic model. We prove this property in Lemma~\ref{lemma improved} of Section~\ref{secsec: The simple Two-Step Algorithm} of the appendix, where we also formalize the simple degree thresholding algorithm.

\begin{algorithm}[H]
\caption{Step I Selective Pruning}
\label{alg: step I}
\KwIn{Graph $\boldsymbol{G}_n=(V,E)$; community count $K$; anchor-degree threshold $\tau\in\mathbb{N}$; anchor-size threshold $Q\in\mathbb{N}$}
\KwOut{$K$ community anchors $\{V_1,\dots,V_K\}$}

\BlankLine  
  \textbf{Initialize:} Let $\bar{\boldsymbol{G}} := \boldsymbol{G}_n$.
  
  %, and denote 
%$\tilde{V}$ as the candidate removal set from Algorithm~\ref{alg: candidate removal nodes},  when applied to $\boldsymbol{G}_n$.

  \While{TRUE}{

  Set $K_c := | \{ \bm{g} \in \mathcal{C}(\bar{\bm{G}}) \,:\, | V(\bm{g})| \geq Q,\; \operatorname{deg}_{\bm{G}_n}(V(\bm{g})) \geq \tau \}|$.

  \If{$K_c \geq K$}{
    \textbf{return} the $K$ vertex sets of $\mathcal{C}(\bar{\bm{G}})$ with the largest maximum degree.
  }
  
  %\If{$\boldsymbol{\bar G}$ admits a decomposition $\boldsymbol{\bar G}=\oplus_j g_j\ $, Compute $K_c := | \{ g_j \,:\, |V(g_j)| \ge Q \text{satisfied that}\ \deg_{\boldsymbol{G}_n}(V(g_j))\ge \tau \}|$, and $K_c\ge K$}{
  %    \textbf{return} $K$ such components with largest degree $\{V(g_{j_1}),\ldots,V(g_{j_K})\}$\;
  %}

  Set $\tilde{V}' := \{ v \in V(\bm{\bar{G}}) \,:\, v \text{ belongs to at least one cycle of $\bar{\bm{G}}$} \}$.

  Set $\tilde{V}'' := \biggl\{ v \in V(\bm{\bar{G}}) \,:\,  |\{ \bm{g} \in \mathcal{C}(\bar{\bm{G}} \backslash \{v\}) \,:\, |V(\bm{g})| \geq Q, \, \operatorname{deg}_{\bm{G}_n}(V(\bm{g})) \geq \tau \}| > K_c  \biggr\}$.

  Set $\tilde{V} = \tilde{V}' \cup \tilde{V}''$.

  \If{$\tilde{V} = \emptyset$}{
    \textbf{\Return FAIL}.
  }

    {   

  Set $D^* := \min \{ \operatorname{deg}_{\bm{G}_n}(v') \,:\, v' \in \tilde{V} \}$. 
    
  Remove $\bigl\{ v \in \tilde{V} \,:\, \operatorname{deg}_{\bm{G}_n}(v) = D^* \bigr\}$ and update $\boldsymbol{\bar G}$ accordingly.
    
  %Remove the nodes from $\tilde{V}$ with smallest degree with respect the initial graph $\bm{G}_n$ and update $\boldsymbol{\bar G}$ accordingly\;

  }
  %Update $\tilde V$ as the candidate removal set of $\boldsymbol{\bar G}$ using Algorithm~\ref{alg: candidate removal nodes}\;
  }
  
  %\textbf{return} \textbf{FAIL}\;
\end{algorithm}

\begin{remark}
\label{rem:anchor-degree-threshold}
The choice of the parameter $\tau$ in Algorithm~\ref{alg: step I} depends on the degree distribution of the input graph. 
If input graph is distributed according to the planted forest model, then one can determine the anchor-degree threshold $\tau$ via the asymptotic degree distribution described in Lemma~\ref{lemma: G}. In practice, however, we give a heuristic for selecting $\tau$ through numerical approximation, as detailed in Algorithm~\ref{alg: choosing the core degree threshold} in Section~\ref{secsec: Choosing the Core Degree Threshold} of the appendix. The intuitive idea is to set $\tau$ to be the degree of the $M$-th highest degree node in the graph, where we choose $M$ with guidance from Monte Carlo simulation. We note that output is relatively insensitive to the choice of $M$. 
\end{remark}

\begin{remark}
\label{remark: shrinkage parameter}
The success of Algorithm~\ref{alg: step I} does not depend critically on the choice of the anchor-size threshold $Q$. We can take $Q = 1$ in all of our theoretical analysis although we find it still helpful to have an appropriately large $Q$ to improve practical performance. If $Q$ is chosen to be too large, Algorithm~\ref{alg: step I} may fail to give an output, which occurs when $\bar{\bm{G}}$ becomes a disjoint collection of trees such that no tree can be split, by the removal of a single node, into two trees both of size at least $Q$ and with maximum degree at least $\tau$. In such cases, we relax the size threshold by setting $Q \leftarrow Q-1$ or $Q \leftarrow \lfloor \zeta Q \rfloor$ for some $\zeta \in (0,1)$, and repeat until the procedure successfully outputs a set of $K$ components. 
%If the algorithm fails for $Q = 1$, we may decrease the core degree threshold $\tau$, or, in step 13, remove only a single node with the smallest degree rather than a set. 
\end{remark}

\begin{figure}[H]
    \centering
    % First subfigure
    \begin{subfigure}[b]{0.32\linewidth}
        \centering
        \vspace{-0.4cm}
        \includegraphics[width=1\linewidth]{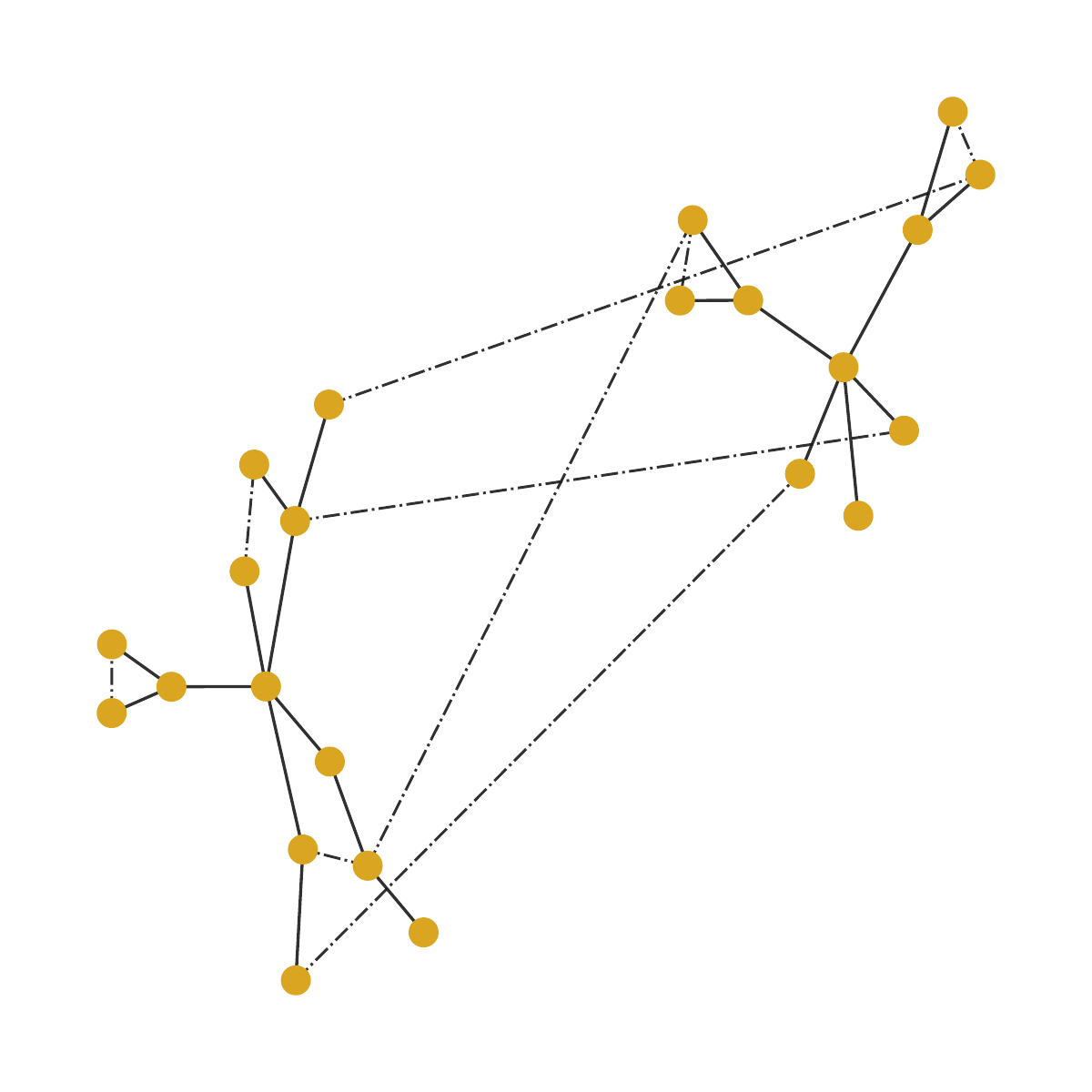}
        \subcaption{Origin Graph}
        \label{fig: improved algorithm a}
    \end{subfigure}
    \hfill
    % Second subfigure
    \begin{subfigure}[b]{0.32\linewidth}
        \centering
        \vspace{-0.4cm}
        \includegraphics[width=1\linewidth]{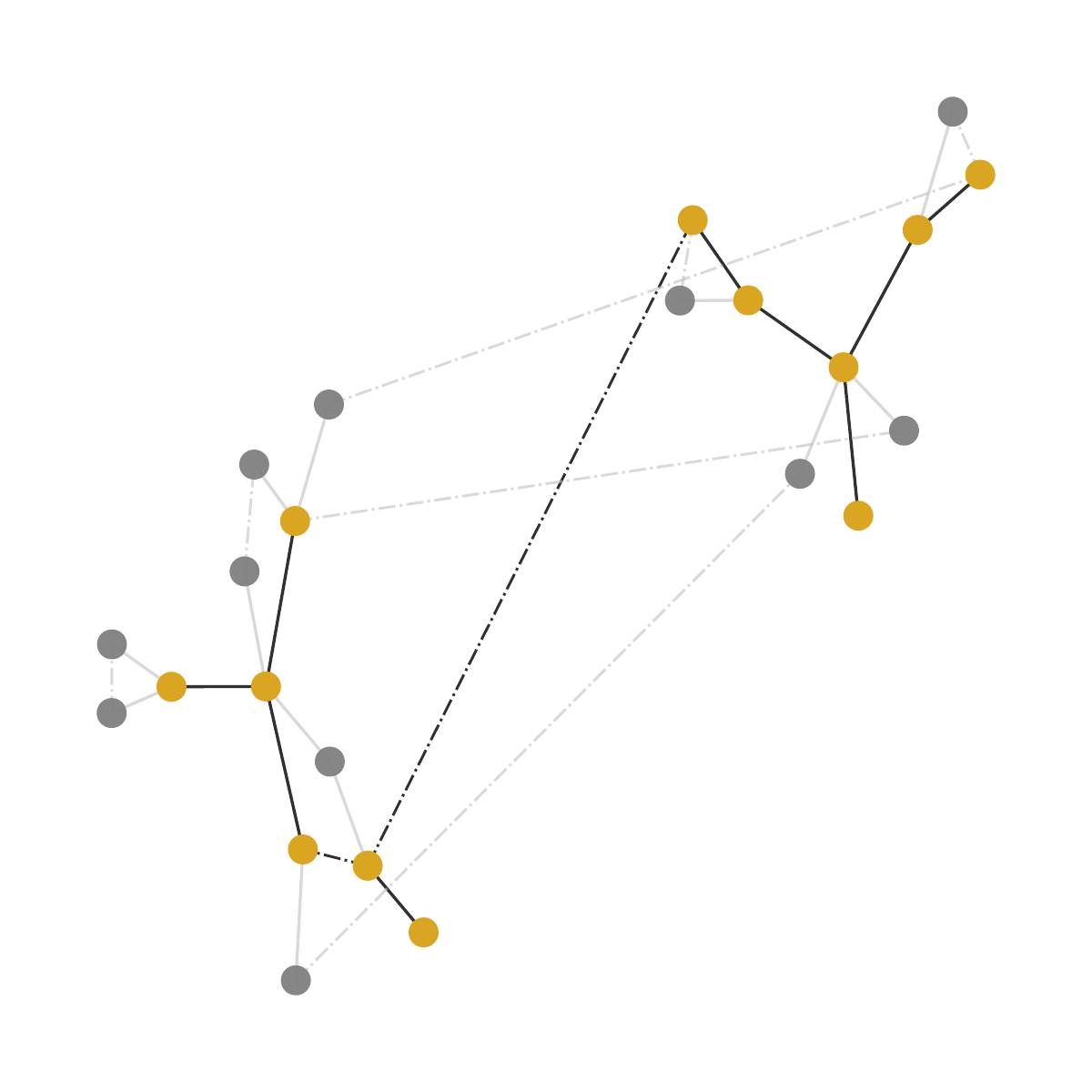} % Add your image here
        \subcaption{Graph after Iteration 1}
        \label{fig: improved algorithm b}
    \end{subfigure}
    % Third subfigure
    \begin{subfigure}[b]{0.32\linewidth}
        \centering
        \vspace{-0.4cm}
        \includegraphics[width=1\linewidth]{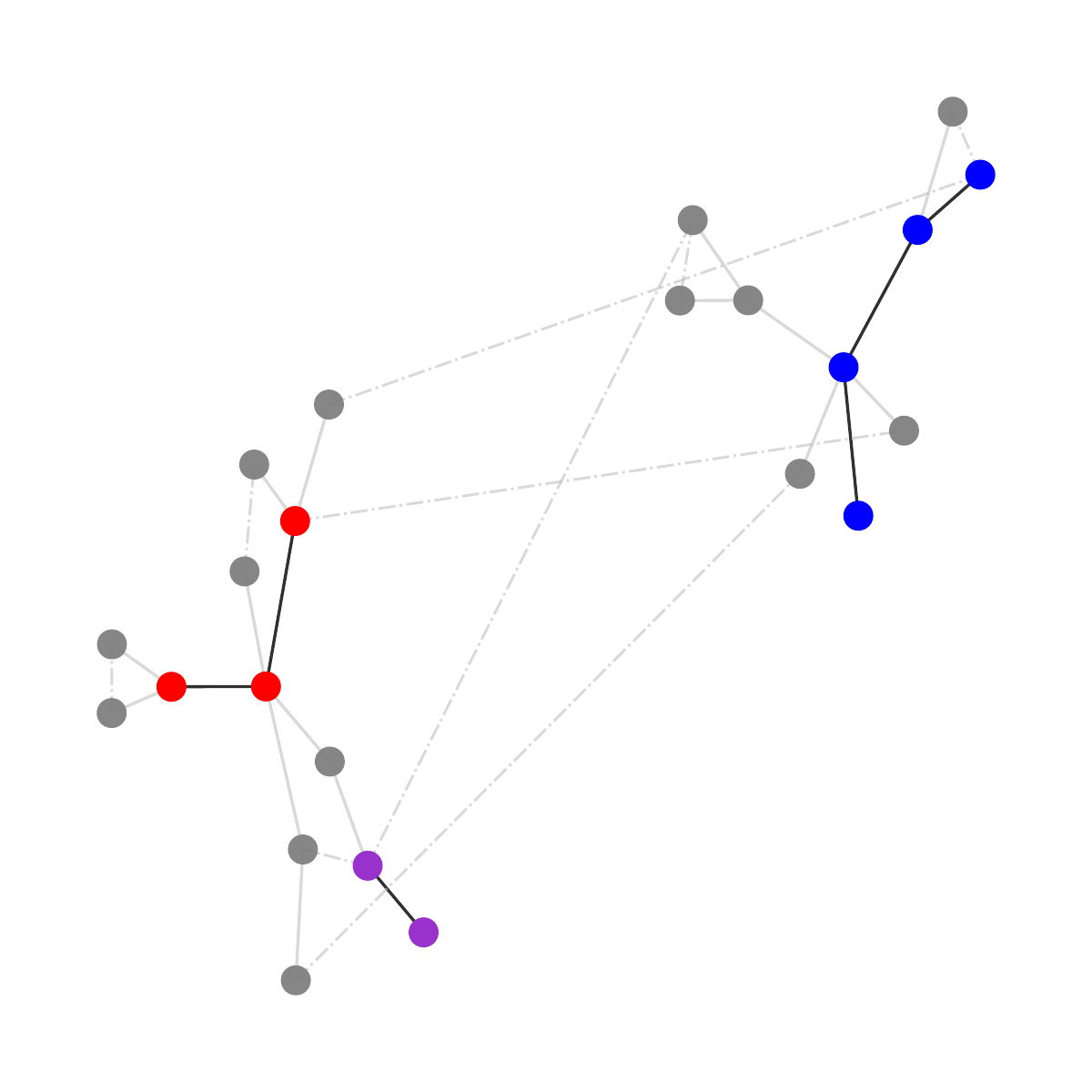} % Add your image here        
        \subcaption{Graph after Iteration 2}
        \label{fig: improved algorithm c}
    \end{subfigure}
    \caption{Forest community detection via Algorithm~\ref{alg: step I} 
        on a toy example with $\tau=5$, $K=2$, and $Q=3$.
        \textbf{(a)} Initial graph containing two planted tree structures.
        \textbf{(b)} Graph after the first pruning iteration.
        \textbf{(c)} Final graph after the second pruning iteration; red and blue anchors consist of nodes with degree $\geq 5$ and size $\geq 3$.}
\label{fig: improved algorithm}
\end{figure}

\begin{example}
To illustrate the procedure of Algorithm~\ref{alg: step I}, we consider a toy example shown in Figure~\ref{fig: improved algorithm a}. 
The graph contains two planted trees of sizes 13 and 10, respectively, with 4 between-tree noise edges and 5 within-tree noise edges (drawn as dotted lines).

In the first iteration, all degree-2 nodes are removed, as they belong to the cycle-based candidate removal set $\tilde{V}'$ identified by Algorithm~\ref{alg: step I}, yielding the graph in Figure~\ref{fig: improved algorithm b}. 
In contrast, the split-based candidate removal set $\tilde{V}''$ selected by Algorithm~\ref{alg: step I} is empty in this iteration.

In the second iteration, no cycles remain in the graph, so the cycle-based candidate removal set $\tilde{V}'$ is empty. 
The split-based candidate removal set $\tilde{V}''$ contains 4 nodes, including 3 nodes of degree 3. 
Removing these three degree-3 nodes produces the graph shown in Figure~\ref{fig: improved algorithm c}. 
The resulting red and blue anchors are the two components that satisfy both the minimum degree ($\tau = 5$) and minimum size ($Q = 3$) requirements.
\end{example}

%We note that Algorithm~\ref{alg: step I} is an improvement upon the simple algorithm described in the beginning of Section~\ref{secsec: The simple Two-Step Algorithm} and formally stated as Algorithm~\ref{alg: theory align} in the appendix. To be precise, if Algorithm~\ref{alg: step I} and Algorithm~\ref{alg: theory align} output the same number of components, e.g. if we obtain $\hat{K}$ from Algorithm~\ref{alg: theory align} and use it as the input $K$ in Algorithm~\ref{alg: step I}, then each core component produced by Algorithm~\ref{alg: step I} is a superset of a component outputted by the simple Algorithm~\ref{alg: theory align}. 
%When $\hat K = K$, the refined procedure produces cores that are at least as large as those obtained by Algorithm~\ref{alg: theory align}, as established in Lemma~\ref{lemma improved} in Section~\ref{secsec: The simple Two-Step Algorithm} of the supplementary material.

\begin{remark}
\label{rem:runtime}
Algorithm~\ref{alg: step I} can be implemented relatively efficiently. To compute $\tilde{V}'$ which consists of all the nodes that belong to at least one cycle in the graph $\bar{\bm{G}}$, we can first compute a spanning tree/forest of $\bar{\bm{G}}$ (say via Kruskal's algorithm \citep{kruskal1956shortest}). Each edge not in the spanning tree/forest induces a cycle in the graph. We can then identify all the nodes in a cycle by iterating over the edges not in the spanning tree/forest. The overall procedure has a runtime of $O(|V|\cdot |E|)$. To compute the set $\tilde{V}''$, we can run depth-first-search (DFS) on $\bar{\bm{G}}\backslash \{v\}$ for each node $v \notin \tilde{V}'$, giving us a runtime of $O(|V|^2)$. 
Moreover, by Lemma~\ref{lemma improved}, the quantity $D^*$ increases strictly at each iteration, and hence the number of iterations is bounded by $\operatorname{deg}(\bm{G}_n)$. Therefore, the overall runtime of Algorithm~\ref{alg: step I} is $O(|V|\cdot |E|\cdot \operatorname{deg}(\bm{G}_n))$.
\end{remark}

\begin{remark}
\label{remark: different pruning criterion}
Degree is not the only criterion one could use for pruning. Other centrality measures that correlate with node arrival order — eigenvector/betweenness centrality, or the posterior root probabilities of \cite{crane2024root} are also plausible. These alternatives are much harder to analyze however. Preferential attachment provides a rigorous way to quantify the relationship between the degree and the arrival order (see e.g. \cite{pekoz2017joint,senizergues2021geometry}), and this link can be shown to degrade gracefully under noise (c.f. Lemma~\ref{lemma R}). 
\end{remark}
%These alternatives are harder to justify, however. Preferential attachment gives a direct mechanistic explanation for why degree tracks arrival order, and this link can be quantified and shown to degrade gracefully under noise. The corresponding relationships for other centralities are far less tractable, which in turn makes theoretical guarantees for community detection substantially harder to establish.

% {
% \color{blue}
% Alternatives to degree-based pruning can also be considered. Other centrality measures that may correlate with node arrival order, such as eigenvector, betweenness, or closeness centrality, could be used instead. In addition, pruning may also be guided by posterior root probabilities derived in \cite{crane2024root}. However, compared with degree-based pruning, these alternatives are generally less interpretable, since the preferential attachment mechanism provides a direct explanation for the strong relationship between degree centrality and node arrival order. Moreover, unlike degree centrality, it is difficult to quantify the relationship between these alternative centrality measures and node arrival time, or to characterize how this relationship changes under different noise levels. As a result, establishing theoretical guarantees for their effectiveness in community detection becomes substantially more challenging.
% }
% 

\subsection{Stage II: Anchor-based Recovery}

The second stage of SPAR uses the anchor components $\{V_1, \ldots, V_K\}$ from the first stage to assign community labels to all the nodes. We present two different approaches for this. The first and simplest is Algorithm~\ref{alg: distance recovery} which uses graph distance. The second is Algorithm~\ref{alg: Model recovery} which uses posterior probabilities from a specified model, possibly computed through Monte Carlo approximations.

%In Stage II of the two–stage algorithm, we also take the {\color{blue} anchor components} $\{V_1,\ldots, V_K\}$ from Stage I as input.
%We present three recovery variants. 
%The Algorithm~\ref{alg: distance recovery} require recovery threshold $\tau^\prime$ as additional input, and Algorithm~\ref{alg: monte carlo sample recovery} requires Monte Carlo samples $\{\boldsymbol{F}_n^{(j)}\}$. 
%We also add an Algorithm~\ref{alg: RRR} to Section~\ref{secsec: RRR} of the supplementary material which requires the graph-pruning history $\mathcal{B}=(B_1,\ldots,B_T)$ as additional input.

\subsubsection{Distance-based Recovery}

Given two connected nodes $u$ and $v$ in a generic graph $\bm{g}$, we define the graph distance between $u$ and $v$ to be
\begin{align}
\mathrm{dist}_{\bm{g}}(u, v) := \text{length of the shortest-path connecting $u$ and $v$}.
\end{align}
If $u$ and $v$ are disconnected, we define $\mathrm{dist}_{\bm{g}}(u,v) = \infty$. For two subsets of nodes $U, V$, we define 
\begin{align*}
\mathrm{dist}_{\bm{g}}(U, V) = \min \{ \mathrm{dist}_{\bm{g}}(u, v) \,:\, u \in U, v \in V \}.
\end{align*}

% {
% \color{blue}
% Before introducing the specific algorithm, we first formally define the graph distance $\text{dist}(\cdot, \cdot)$:
% \begin{defin}
% \label{defin distance}
% For a graph $g=(V, E)$ and subsets $V_1, V_2\subseteq V$, the graph distance between $V_1$ and $V_2$ is given by the length of the shortest path between the sets $V_1$ and $V_2$ if they are connected.
% Otherwise, it is infinity.
% That is
% %
% \begin{equation*}
% \operatorname{dist}_g(V_1, V_2)= \begin{cases}\left|e\right|, & \text { if } e\subset E \text { is the shortest path connecting } V_1 \text { and } V_2 \\ \infty, & V_1 \text { and } V_2 \text { are not connected }\end{cases}
% \end{equation*}
% %
% \end{defin}
% }

One way to perform distance-based recovery is to assign each node $u$ the community label of the closest anchor component. However, large anchor components may sometimes still contain misclassified nodes so that, to further improve the reliability of the recovery method, we instead assign each node $u$ the community label of the anchor component whose high-degree nodes are nearest to $u$. To be specific, let $\tau' > 0$ be a tuning parameter and define the high-degree sub-component $\tilde{V}_k := \{ u \in V_k \,:\, \text{deg}_{\bm{G}_n}(u) \geq \tau' \}$. We then define a modified distance between a node $u$ and a component $V_k$ as
\[
D_k(u) = \text{dist}_{\bm{G}_n}(u, \tilde{V}_k),
\]
and assign $u$ the label $k'$ such that $D_{k'}(u)$ is minimized. 

In practice, we can choose $\tau'$ by setting $\tau^\prime=\tau/2$, where $\tau$ is the anchor-degree threshold in Algorithm~\ref{alg: step I}, or let it be the degree of the last batch of removed nodes from the first stage, i.e. the value $D^*$ at the end of Algorithm~\ref{alg: step I}, which is what we use in Theorem~\ref{thm: first L}. 

%The most direct method of assigning community label to a node $u$ based on anchor components $\{V_1, \ldots, V_K\}$ is to find the component closest to $u$, say $V_{k}$, and then assign $u$ to community $k$. In order to facilitate our theoretical analysis, we  

%The distance-based recovery follows the same labeling step of Algorithm~\ref{alg: theory align}, with one modification: we introduce a recovery  $\tau^\prime$ so that only vertices with degree at least $\tau^\prime$ participate in distance recovery.
%In practice, we set $\tau^\prime=0$ or, if we use Algorithm~\ref{alg: step I} to construct the community {\color{blue} anchors}, then we can set $\tau'$ as the degree of the last batch of removed nodes, i.e., the value of $D^*$ at the end of the algorithm. 
%This thresholding aligns the procedure with the theoretical guarantees and prevents the use of low-degree nodes during recovery.

\begin{algorithm}[H]
\caption{Distance Recovery}
\label{alg: distance recovery}
\KwIn{graph $\boldsymbol{G}_n=\left(V,E\right)$; community  anchors $V_1, \cdots, V_K$; recovery threshold $\tau^\prime \in\mathbb{N}$.}
\KwOut{$\hat{\ell}:V\to\{1,\dots,K\}$}

\For{$i=1$ \KwTo $K$}{
    $\tilde{V}_i:= \left\{u\in V_i: \deg_{\boldsymbol{G}_n}\left(u\right)\ge \tau^\prime\right\}$\;

  }
\textbf{Labeling:}\:
\For{$u\in V$}{
    $D_k\left(u\right)= \mathrm{dist}_{\boldsymbol{G}_n}\bigl(u, \tilde{V}_k \bigr)$ for all $k\in[K]$;\quad
    $\hat{\ell}(u):= \operatorname{argmin}_{k\in[K]} D_k\left(u\right)$ with uniform random tie–breaking if non-unique\;
}
\Return $\hat{\ell}$\;
\end{algorithm}

\subsubsection{Model-based Recovery}

By taking the random-community planted forest model specified in Definition~\ref{defin: RC-PF} with a random arrival ordering $\pi$ that is generated uniformly at random from the set of permutations of the $n$ nodes, we can write down the joint probability
\[
\mathbb{P}(\bm{G}_n, \bm{F}_n, \pi) = \mathbb{P}(\bm{G}_n \,|\, \bm{F}_n) \mathbb{P}(\bm{F}_n \,|\, \pi) \mathbb{P}(\pi)
\]
where, on the right-hand side, the first term is determined by the Erd\"{o}s--R\'{e}nyi noise, the second term is determined by the affine preferential attachment mechanism, and the third term is $\frac{1}{n!}$. 

For a fixed node $u$ and a set of anchor components $\{V_1, \ldots, V_K\}$, we may then define the conditional probability that node $u$ belongs to community $k$ conditional on the observed graph $\bm{G}_n$ and the event that for every $i \in \{1, 2, \ldots, K\}$, each anchor component $V_i$ contains the root node $\pi_i$ of community-tree $\bm{T}^i$:
\begin{align*}
p_k(u) &:= \mathbb{P}\bigl( \ell(u) = k \,\big|\, \pi_1 \in V^1, \ldots, \pi_K \in V^K, \, \bm{G}_n \bigr) \\
&= \frac{ \mathbb{P}\bigl( u \in V(\bm{T}^k), \, \pi_1 \in V^1, \ldots, \pi_K \in V^K  \,\big|\, \bm{G}_n \bigr) }{
\mathbb{P}\bigl( \pi_1 \in V^1, \ldots, \pi_K \in V^K \, \big|\, \bm{G}_n \bigr)}.
\end{align*}

The conditional probability $p_k(u)$ does not have a tractable analytic form but we can compute a Monte Carlo approximation using samples from the posterior distribution of the trees $\bm{T}^1, \ldots, \bm{T}^K$ and of the ordering $\pi$ given the overall graph $\bm{G}_n$. Concretely, suppose we have $M$ samples $\{( \bm{T}^{1, (j)}, \pi_1^{(j)}), \ldots, (\bm{T}^{K, (j)}, \pi_K^{(j)})\}_{j=1}^M$ from the posterior distribution given $\bm{G}_n$, where $\bm{T}^{1, (j)}$ is the first community-tree and $\pi_1^{(j)}$ is the root node of the first community-tree in sample $j$. Then, we compute
\begin{align*}
\hat{p}_k(u) = \frac{\sum_{j=1}^M \mathbf{1}\left\{\pi^{(j)}_i \in V_i \quad \text{for all $i\in \left[K\right]$}\right\} \mathbf{1}\left\{u \in \bm{T}^{k, (j)}\right\}}{\sum_{j=1}^M \mathbf{1}\left\{\pi^{(j)}_i \in V_i \quad \text{for all $i\in \left[K\right]$}\right\}}.
\end{align*}

We then assign node $u$ the community label $k$ which maximizes the posterior probability $\hat{p}_k(u)$. The overall procedure is summarized in Algorithm~\ref{alg: Model recovery}. It is worth noting that in addition to giving a hard community assignment $\hat{\ell}(u)$, we may also use the vector $(\hat p_1(u), \ldots, \hat p_K(u))$ to get a rough sense of community recovery uncertainty for node $u$. 

To generate the posterior samples, we use the Gibbs sampler proposed in \cite{crane2024root}, with a slight modification to accommodate the fact that the model in \cite{crane2024root} initializes each community-tree as a single node with a self-loop whereas we initialize each community-tree as two nodes connected by a single edge.

\begin{algorithm}[H]
\caption{Model Recovery}
\label{alg: Model recovery}
\KwIn{Anchor connected components} $V_1,\ldots,V_K$; $M$ posterior samples $\{( \bm{T}^{1, (j)}, \pi_1^{(j)}), \ldots, (\bm{T}^{K, (j)}, \pi_K^{(j)})\}_{j=1}^M$
\KwOut{Soft assignment $\tilde{\ell}:V\to\Delta^{K-1}$ with $\tilde{\ell}\left(u\right)=\left(\hat p_u\left(1\right),\ldots,\hat p_u\left(K\right)\right)$; hard labels $\hat{\ell}:V\to[K]$}

\For{$u\in V, k\in[K]$} 
{
Set
\begin{equation*}
\hat{p}_k(u) = \frac{\sum_{j=1}^M \mathbf{1}\left\{\pi^{(j)}_i \in V_i \quad \text{for all $i\in \left[K\right]$}\right\} \mathbf{1}\left\{u \in \bm{T}^{k, (j)}\right\}}{\sum_{j=1}^M \mathbf{1}\left\{\pi^{(j)}_i \in V_i \quad \text{for all $i\in \left[K\right]$}\right\}}.
\end{equation*}
Set $\tilde{\ell}(u):=(\hat p_1(u),\ldots,\hat p_K(u))$ and $\hat{\ell}(u):= \arg\max_{k\in[K]} \hat p_k(u)$
(break ties uniformly at random).
} 
\Return $\tilde{\ell},\ \hat{\ell}$\;
\end{algorithm}

\begin{figure}[H]
    \centering
    % First subfigure: Hard classification
    \begin{subfigure}[b]{0.45\textwidth}
        \centering
        \vspace{-0.4cm}
        \includegraphics[width=\linewidth]{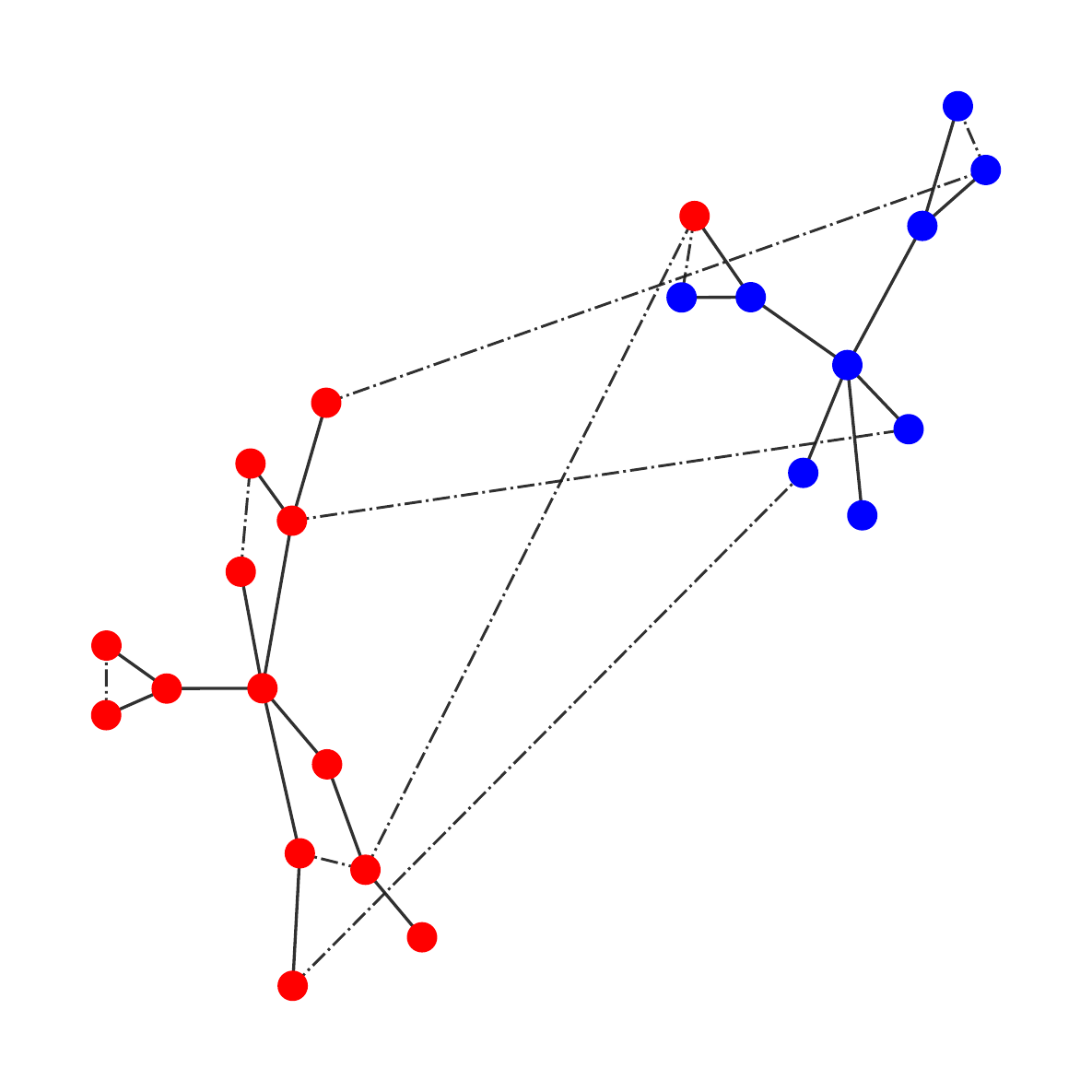}
        \caption{Hard classification results}
        \label{fig: hard_classification}
    \end{subfigure}
    \hfill
    % Second subfigure: Soft classification
    \begin{subfigure}[b]{0.5\textwidth}
        \centering
        \vspace{-0.4cm}
        \includegraphics[width=\linewidth]{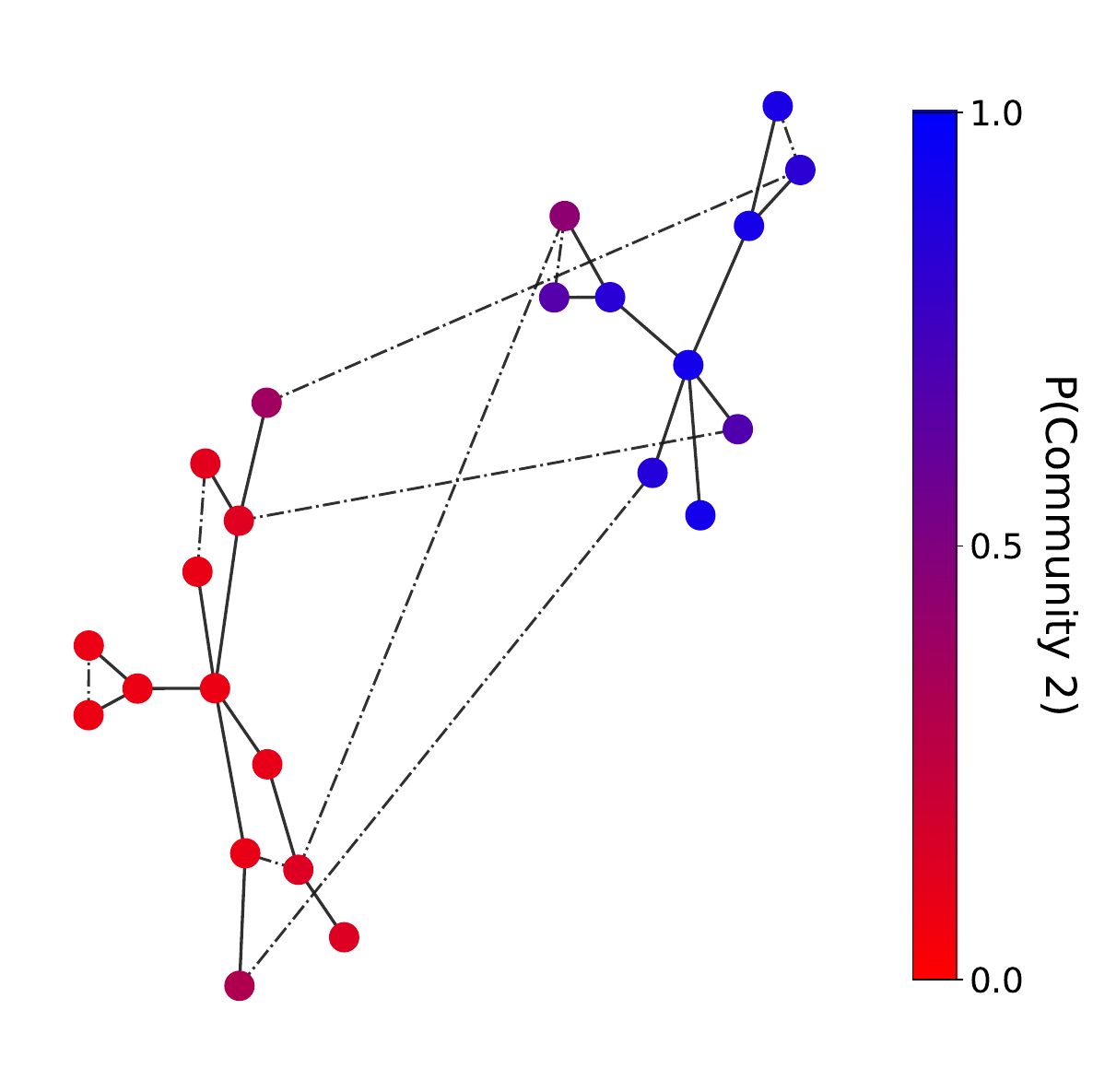}
        \caption{Soft classification results}
        \label{fig: soft_classification}
    \end{subfigure}
    
    \caption{
        Forest community detection results using Algorithm~\ref{alg: Model recovery} 
        on the toy example with input anchors from Figure~\ref{fig: improved algorithm c}. 
        \textbf{(a)} Hard classification outcomes, where each node is assigned to the community with the highest posterior probability: $\hat{\ell}(u) = \arg\max_{i\in[2]} \hat{p}_i(u)$. 
        \textbf{(b)} Soft classification results, visualized through node colors corresponding to posterior probability vectors $\tilde{\ell}(u) = (\hat{p}_1(u),\hat{p}_2(u))$, with the colorbar indicating probability values.
    }
    \label{fig: MCMC_results}
\end{figure}

\begin{example}
We evaluate the performance of Algorithm~\ref{alg: Model recovery} using the anchors obtained from the example in Figure~\ref{fig: improved algorithm c}. 
To generate posterior samples,
we draw $500$ posterior samples from the \textsc{RC-PF} model under the Gibbs sampling framework as described in \cite{crane2024root}, with parameters $\alpha=0$, and $K=2$.
Among these, $319$ samples have roots falling within the identified anchors, indicating a moderate alignment between the core extraction method and the sampling mechanism. 
This alignment supports the suitability of Algorithm~\ref{alg: Model recovery} for community recovery in this setting.

Figure~\ref{fig: hard_classification} presents hard classification results, where each node is assigned to the class with the higher posterior probability. 
The soft classification results in Figure~\ref{fig: soft_classification} encode the full posterior probability distribution into a color mapping, providing detailed insight into which nodes are confidently classified versus those with higher uncertainty. 
Boundary nodes connected by noisy edge to the opposite cluster have greater classification ambiguity, which is consistent with our intuition to prove Theorem~\ref{thm: impossible}.
\end{example}

\section{Theoretical Analysis}
\label{sec: theory}

In this section, we provide upper bounds on the misclustering error rate of our procedure local to specific subsets of nodes. We study three subsets of central nodes: earliest-arriving nodes, high-degree nodes, and nodes that are within tree-distance 1 or 2 of the root node, i.e. children or grandchildren of the root nodes. We give the proof of all the results in Section~\ref{sec: theory proofs} of the appendix. 

% {
% \color{blue}
% In this section, we establish subset-consistent community detection guarantees for SPAR. 
% Specifically, we study recovery on several informative subsets of the network, including the earliest-arriving nodes, high-degree nodes, and nodes in the first few layers of each preferential attachment tree. {\color{orange} [Be more specific]}

Before stating our results, we introduce some useful notation. For any pair of labeled graphs $\boldsymbol{g} \subset \boldsymbol{g}'$ and any $C>0$, we define the following high degree subsets:

\begin{equation}
\label{define VC}
V^C_{\bm{g}'}(\bm{g}) := \{ u \in V(\bm{g}) \,|\, \mathrm{deg}_{\bm{g}'}(u) \geq C \},\quad
    V^{C,\alpha}_{\boldsymbol{g}^\prime}(\boldsymbol{g}) := \bigl\{ u \in V(\boldsymbol{g}) \;\big|\; \operatorname{deg}_{\boldsymbol{g}^\prime}(u) \ge C {|V(\boldsymbol{g}^\prime)|}^{\frac{1}{2+\alpha}} \bigr\}.
\end{equation}
That is, $V^C_{\bm{g}'}(\bm{g})$ is the set of nodes in $\bm{g}$ whose degree, computed with respect to $\bm{g}'$, is at least $C$. When $\boldsymbol{g} = \boldsymbol{g}^\prime$, we write these more concisely as $V^C(\bm{g})$ and $V^{C,\alpha}(\boldsymbol{g})$ respectively. 
The scale factor of $|V(\boldsymbol{g}^{\prime})|^{\frac{1}{2+\alpha}}$ in the definition of $V^{C,\alpha}_{\boldsymbol{g}^\prime}(\boldsymbol{g})$ takes on that particular form because, under affine preferential attachment mechanism, the maximum degree of a graph with $n$ nodes grows at rate $n^{\frac{1}{2+\alpha}}$ (see Lemma~\ref{lemma deg bound} in Section~\ref{secsec: technical lemmas} of the appendix).

% {
% \color{blue}
% The normalization $|V(\boldsymbol{g}^\prime)|^{\frac{1}{2+\alpha}}$ in $V^{C,\alpha}_{\boldsymbol{g}^\prime}(\boldsymbol{g})$ reflects the characteristic degree scale under affine preferential attachment. 
% In particular, the highest-degree vertices in an size $n$ $\mathrm{APA}(\alpha,\pi)$ preferential attachment grow at rate $n^{\frac{1}{2+\alpha}}$ (see Lemma~\ref{lemma deg bound} in Section~\ref{secsec: technical lemmas} of the supplementary material).  
% {\color{orange} [It is very good to explain the $|V(\bm{g}')|^{\frac{1}{2+\alpha}}$ term in the definition. It might be easier to say that the highest degree nodes have degree of order $n^{\frac{1}{2+\alpha}}$]}
% }

We now introduce the assumptions required for our theoretical results.

\begin{assumption}
\label{assumption: bound}
Define $n_k:=\bigl| \left\{u \mid u\in V(\bm{G}_n),\, \ell(u)=k\right\}\bigr|$ as the size of community $k$ and let $H \in [K, \infty)$. We assume that $\frac{n_k}{n} \ge \frac{1}{H}$ for every $k \in [K]$. 
\end{assumption}

Assumption~\ref{assumption: bound} says that the $K$ communities should be roughly balanced in the sense that the ratio of sizes between the largest and the smallest community should be less than a constant factor. Our numerical experiments (c.f. Section~\ref{sec: additional simulation} in the appendix) show that the recovery performance of SPAR deteriorates with the severity of community-size imbalance. This is potentially because very small communities may get accidentally removed in the first pruning stage of our algorithm.

% {
% \color{blue}
% {\color{orange} [Avoid this justification if at all possible.] }
% In our setting, Assumption~\ref{assumption: bound} ensures that the degree-based pruning step in SPAR does not eliminate small communities before the planted trees are separated. 
% Our numerical experiments also shows that the recovery performance decreases with more severe community-size imbalance.
% When Assumption~\ref{assumption: bound} is substantially violated, the current misclustering error becomes less informative, as classifying vertices in communities of vastly different sizes can involve substantially different levels of difficulty.
% A community-size-weighted classification error may therefore be more appropriate in such settings, and we leave its development for future work.{\color{orange} [This explanation should be more intuitive for the reader. It makes sense if you know the algorithm and the theory very well but it is hard to digest for a new reader.]}
% }

\begin{assumption}
\label{assumption: sparsity}
We assume that there exist $\delta > 0$ and $C_0 > 0$ such that $\theta \leq C_0 n^{-\frac{1+\alpha}{2+\alpha}-\delta}$.
\end{assumption}

Assumption~\ref{assumption: sparsity} controls the magnitude of the Erd\H{o}s--R\'enyi perturbation. To give an intuitive interpretation of the tolerance bound $n^{ - \frac{1+\alpha}{2+\alpha}}$, we observe that if the noise probability $\theta = o(n^{ - \frac{1+\alpha}{2+\alpha}})$, then the expected degree contributed by the Erd\H{o}s--R\"{e}nyi random edges for a particular node is equal to $(n-1)\theta
=o(n^{\frac{1}{2+\alpha}})$. For early arriving nodes, this is negligible compared to the degree generated by the preferential attachment mechanism, which is of order $n^{\frac{1}{2+\alpha}}$. We can therefore still use the degree to recover the early-arriving node of each of the communities. Based on our numerical experiments, we conjecture that $n^{-\frac{1+\alpha}{2+\alpha}}$ is a tight threshold in that if $\theta = \Omega(n^{-\frac{1+\alpha}{2+\alpha} + \delta})$ for some $\delta > 0$, then the planted forest graph is indistinguishable from pure ER graph. It is also worth noting that, in the linear preferential attachment case $(\alpha=0)$, our noise tolerance bound for $\theta$ coincides with that appearing in Theorem~11 of \cite{crane2024root}.
In fact, our results suggest that the analysis in \cite{crane2024root} may potentially be extended from the linear preferential attachment setting to the general affine regime $\alpha>-1$.

Our first result bounds the misclustering error rate on the first $L$ earliest-arriving nodes in each of the communities, i.e., the subset $\cup_{i=1}^K \pi_{1:L}\left(\bm{T}^i\right)$. We define the constant $C_1(\varepsilon, H) := \gamma(\frac{\varepsilon}{3H}, 1)$ where $\gamma(\cdot, \, \cdot)$ is the function specified in Lemma~\ref{lemma order--degree} via equation~\eqref{order to degree}. 

% Let $\boldsymbol{G}_n \sim \mathrm{PF}(\alpha, \theta, \ell, \pi)$ and suppose Assumptions~\ref{assumption: bound} and \ref{assumption: sparsity} hold. 
% Fix any $\varepsilon > 0$ and fixed $L\in \mathbb{N}$. Let $\hat{\ell}(\cdot)$ be the estimated community membership function obtained from Algorithm~\ref{alg: step I} followed by Algorithm~\ref{alg: distance recovery} with the following parameter settings: we let the anchor-degree threshold 
% $\tau = C_1(\varepsilon, H) n^{\frac{1}{2+\alpha}}$, the anchor-size threshold $Q > 0$ be arbitrary, and the recovery threshold $\tau^\prime$ to be any non-negative number less than 
% $D^*$ outputted by Algorithm~\ref{alg: step I}. Then, with probability at least $1-\varepsilon + \eta_n$, where the $\eta_n = o(1)$ as $n \rightarrow \infty$ and depending only on $\varepsilon, L, Q, \alpha, H, \delta$, 

\begin{thy}
\label{thm: first L}
Let $\boldsymbol{G}_n \sim \mathrm{PF}(\alpha, \theta, \ell, \pi)$ and suppose Assumptions~\ref{assumption: bound} and \ref{assumption: sparsity} hold. 
Fix any $\varepsilon > 0$ and fixed $L\in \mathbb{N}$. Let $\hat{\ell}(\cdot)$ be the estimated community membership from Algorithm~\ref{alg: step I} followed by Algorithm~\ref{alg: distance recovery} with $\tau = C_1(\varepsilon, H) n^{\frac{1}{2+\alpha}}$, any fixed $Q\in\mathbb{N}$, and $\tau^\prime$ as any non-negative number less than 
$D^*$ outputted by Algorithm~\ref{alg: step I}. Then, with probability at least $1-\varepsilon + \eta_n$, where the $\eta_n = o(1)$ as $n \rightarrow \infty$ and depending only on $\varepsilon, L, Q, \alpha, H, \delta$, 
\begin{equation}
\label{thm: first L eqn 1}
d_{\cup_{i=1}^K \pi_{1:L}\left(\bm{T}^i\right)}(\hat{\ell}, \ell) = 0.
\end{equation}
\end{thy}

We note that Theorem~\ref{thm: first L} remains valid for any $\tau$ of order $n^{\frac{1}{2+\alpha}}$ that is bounded above by $C_1(\varepsilon, H) n^{\frac{1}{2+\alpha}}$.
We fix a particular choice of $\tau$ only to simplify the statement of the theorem and avoid uninteresting technical complications in the proof.

\begin{remark}
\label{remark: Ln to inf}
The conclusion of Theorem~\ref{thm: first L} still holds if we allow $L \equiv L_n$ to be a sequence that diverges with $n$. To see this, for each $L$, let $\eta(L,n)$ denote the $\eta_n$ term appearing in the probability bound of the theorem with the convention that $\eta(0,n):=0$, and define $L_n = \max \{ L \in \mathbb{N} \,:\, |\eta(L,n)|<\varepsilon \}$. We note that $L_n \rightarrow \infty$ as, for all fixed integer $M$, we have $\max_{L \in [M]} |\eta(L, n)| \rightarrow 0$ so that $\lim_{n \rightarrow \infty} L_n \geq M$. Then, by Theorem~\ref{thm: first L},  $d_{\cup_{i=1}^K \pi_{1:L_n}(\bm{T}^i)}(\hat{\ell}, \ell) \neq 0$ with probability at most $\varepsilon+\eta(L,n)<2\varepsilon$. Unfortunately, we cannot determine how fast $L_n$ is allowed to diverge because it requires the Berry--Esseen type result on the convergence of the empirical degree distribution among the high-degree nodes to the its asymptotic limit.

\begin{remark}
\label{rem: K estimation 2}
We state our results assuming that the number of communities $K$ is given. In our theoretical setting where $\bm{G}_n$ is generated by the planted forest model, we can estimate $K$ by carefully selecting two parameters $\tau_1 > \tau_2 > 0$, then removing all nodes of degree less than $\tau_2$, and finally letting $\hat{K}$ be the number of remaining connected components that contain a node of degree at least $\tau_1$; we describe this formally in Algorithm~\ref{alg: theory align}. By Lemma~\ref{lemma early hist containment} in Section~\ref{secsec: thm first L} of the appendix, if we choose $\tau_1 = \gamma(\frac{\varepsilon}{3H}, 1) n^{\frac{1}{2+\alpha}}$ and $\tau_2 = \tilde{\gamma}(\varepsilon, 1, 1) n^{\frac{1}{2+\alpha}}$ (where $\tilde{\gamma}$ is defined in the proof of Lemma~\ref{lemma early hist containment}), then $\hat{K} = K$ with probability at least $1 - \varepsilon + o(1)$. In practice however, the $\tau_2$ parameter is critical to the estimation result and not much easier to tune than $K$.
\end{remark}

% We then show that $L_n\rightarrow \infty$. 
% Fix any $M\in \mathbb{N}$, there exist a constant $N_M$ such that $|\eta(M,n)|< \varepsilon$ for $n\geq N_M$ since $\eta(M,n)=o(1)$.
% Therefore, for any $n\geq N_M$, $L_n:=\max_{L} \{|\eta(L,n)|<\varepsilon\}\geq M$.
% As $M$ is arbitrary, the divergence result follows as desired. 

% {\color{orange}
% Theorem holds for a sequence $L_n$ such that $L_n \rightarrow \infty$. But we cannot determine how fast $L_n$ diverges because it requires the Berry--Esseen type result on the convergence of the empirical degree distribution among the high-degree nodes to the its asymptotic limit. Could reference the specific lemma (Lemma 12??).

% Write the $o(1)$ term as $\eta(L, n)$. Need $\eta(L_n, n) \rightarrow 0$. $L_n:=\operatorname{max}_{L} \{|\eta(L,n)|<\varepsilon\}$
% }
\end{remark}

One somewhat restrictive aspect of Theorem~\ref{thm: first L} is that the target set $\cup_{i=1}^K \pi_{1:L}\left(\bm{T}^i\right)$ depends on the unobserved arrival ordering and therefore cannot be directly identified from the observed graph $\bm{G}_n$. However, we know that the early arriving nodes roughly coincide with the high-degree nodes of the graph (c.f. Lemma~\ref{lemma order--degree} in Section~\ref{secsec: thm first L} of the appendix) so that we may extend the recovery guarantee of Theorem~\ref{thm: first L} to the subset of high-degree vertex set $V^{\tilde C,\alpha}(\bm{G}_n)$, which is readily computable from the observed graph.

\begin{corollary}
\label{cor: first L}
Under the same setting as Theorem~\ref{thm: first L}, let $V^{\tilde{C}, \alpha}(\bm{G}_n)$ denote the set of high-degree nodes defined in \eqref{define VC} for any fixed $\tilde{C}>0$.
Then, the estimator $\hat{\ell}$ obtained using graph pruning through Algorithm~\ref{alg: step I} then with label recovery through Algorithm~\ref{alg: distance recovery}, with the same $\tau, \tau', Q$ as specified in Theorem~\ref{thm: first L}, satisfies, with probability at least $1-2\varepsilon + o(1)$ where the $o(1)$ term depends only on $\varepsilon, \tilde{C}, Q, \alpha, H, \delta$, 
\begin{equation}
\label{cor: first L eqn 1}
d_{V^{\tilde{C}, \alpha}\left(\bm{G}_n\right)}(\hat{\ell}, \ell) = 0.
\end{equation}
\end{corollary}

By an argument identical to that of Remark~\ref{remark: Ln to inf}, we can show that Corollary~\ref{cor: first L} holds if we allow $\tilde{C} \equiv \tilde{C}_n$ to be a sequence that goes to 0 as $n \rightarrow \infty$. 

% {
% \color{blue}
% By an argument similar Remark~\ref{remark: Ln to inf}, one can construct a threshold sequence $\tilde C_n \rightarrow 0$ with unknown convergence rate that still satisfied $d_{V^{\tilde{C_n}, \alpha}\left(\bm{G}_n\right)}(\hat{\ell}, \ell) = 0$ with high probability.
% We leave the problem of finding growth rate of sequence $L_n$ or $\tilde{C}_n$ for future work.

The second part of our analysis considers vertex sets defined through distances from the root nodes. We define these by layers so that layer-1 nodes consist of all direct children of the $K$ root nodes, layer-2 nodes consist of all children and grandchildren of the root nodes, etc. We formalize these definitions below:

% These layer-based subsets have explicitly growing size and naturally correspond to the distance-based recovery mechanism in SPAR. 
% Before presenting the next theorem, we first define the $s$-th layer $\mathcal{L}_s(\cdot)$.

\begin{defin}
\label{define mismatch layer}
For $\boldsymbol{G}_n \sim \mathrm{PF}(\alpha, \theta, \ell, \pi)$ and $s \in \mathbb{N}$, we define $\mathcal{L}_s\left(\boldsymbol{F}_n\right)$ as the nodes in $s$ layer of forest $\boldsymbol{F}_n$, that is
\begin{align*}
   \mathcal{L}_s(\boldsymbol{T}^i) &:= \biggl\{ v \in V(\boldsymbol{T}^i) \,\bigg|\, \operatorname{dist}_{\boldsymbol{T}^i}\bigl(v, \pi_1(\boldsymbol{T}^i)\bigr)=s \biggr\} \\
   \mathcal{L}_s(\boldsymbol{F}_n) &:= \cup_{i=1}^K \mathcal{L}_s(\boldsymbol{T}^i).    %\mathcal{L}_s\left(\boldsymbol{F}_n\right):=\left\{v\;\middle|\; \operatorname{dist}_{\boldsymbol{T}^i}\left(v, \pi_1\left(\boldsymbol{T}^i\right)\right)=s, v\in V\left(\boldsymbol{T}^i\right),\quad  \text{for some $i\in \left[K\right]$}\right\}. 
\end{align*}

\end{defin}

These vertex sets have cardinalities that diverge with $n$. From Lemma~\ref{lemma layer size} in the appendix, we have that if $\bm{G}_n$ is a planted forest random graph with linear preferential attachment, i.e $\alpha = 0$, then, for each $s \in \mathbb{N}$,
\[
\mathbb{E}| \mathcal{L}_s(\bm{F}_n)| \asymp \sqrt{n} \log^{s-1} n. 
\]
For a general $\alpha > -1$, we have from Lemma~\ref{lemma deg bound} in the appendix that
\begin{align}
| \mathcal{L}_1(\bm{F}_n) | \asymp_p n^{\frac{1}{2+\alpha}}. \label{eq:layer-1 size bound}
\end{align}

Although the number of layer-$s$ nodes is increasing with $n$, their proximity to the root node allows us to still bound the misclustering error rate on the high probability event that the $K$ root nodes are correctly dispersed among the $K$ anchor components. 

% {
% \color{blue}
% As shown in Lemmas~\ref{lemma layer size} and \ref{lemma deg bound}, the first layer satisfies $
% |\mathcal{L}_1(\bm{G}_n)| \asymp n^{\frac{1}{2+\alpha}}$
% under affine preferential attachment. In the linear preferential attachment case, the second layer additionally satisfies $|\mathcal{L}_2(\bm{G}_n)| \asymp n^{\frac{1}{2}}\log n$.
% Therefore, these layers form growing subsets of the network while still remaining close to the root nodes in graph distance.
% This property is crucial for the distance-based recovery step in SPAR. Once the root nodes $\pi_1^1,\ldots,\pi_1^K$ are correctly matched to the anchor sets $V_{\sigma(1)},\ldots,V_{\sigma(K)}$ for some $\sigma\in S(K)$, Algorithm~\ref{alg: distance recovery} can correctly propagate labels to the majority nearby layer 1/2 nodes with high probability when the noise level satisfies Assumption~\ref{assumption: sparsity}.
% }

\begin{thy}
\label{thm: layer 1/2}
For $\boldsymbol{G}_n \sim \mathrm{PF}(\alpha, \theta, \ell, \pi)$, and suppose Assumptions~\ref{assumption: bound} and \ref{assumption: sparsity} hold. 
Fix any $\varepsilon > 0$ and $c_0 \in (0, 1]$. 
Let $\hat{\ell}(\cdot)$ be the estimated community membership from Algorithm~\ref{alg: step I} followed by Algorithm~\ref{alg: distance recovery} with $\tau  := C_1(\frac{\varepsilon}{6}, H)n^{\frac{1}{2+\alpha}}$, any fixed $Q>0$, and any $\tau^\prime \in [c_0 \tau, \tau]$.
Then, there exists a constant $C_2 \equiv C_2(\varepsilon, c_0) <\infty$, with probability at least $1-\varepsilon+o(1)$ (with the $o(1)$ term depending only on $\varepsilon, Q, \alpha, H, \delta, c_0$),
\begin{equation}    
\label{thm: layer 1/2 eqn 1}
\frac{d_{\mathcal{L}_1\left(\bm{F}_n\right)}(\hat{\ell},\ell)}{\left|\mathcal{L}_1\left(\boldsymbol{F}_n\right)\right|}\le  C_2 n^{-\frac{1+\alpha}{2+\alpha}-\delta}.
\end{equation}
Moreover, if $\alpha=0$, then there exists a constant $C_3 \equiv C_3(\varepsilon, c_0) <\infty$ such that with probability at least $1-\varepsilon+o(1)$ (with the $o(1)$ term depending only on $\varepsilon, Q, H, \delta, c_0$), 
\begin{equation}
\label{thm: layer 1/2 eqn 2}
\frac{d_{\mathcal{L}_2\left(\bm{F}_n\right)}(\hat{\ell},\ell)}{\left|\mathcal{L}_2\left(\boldsymbol{F}_n\right)\right|}\le  C_3 n^{-\delta}.
\end{equation}
\end{thy}

\begin{remark}
It is worth noting that in the layer-1 bound~\eqref{thm: layer 1/2 eqn 1}, when $\alpha \geq 0$, we have by~\eqref{eq:layer-1 size bound} that
\[
d_{\mathcal{L}_1\left(\bm{F}_n\right)}(\hat{\ell},\ell)\le  C_2 n^{-\frac{1+\alpha}{2+\alpha}-\delta} \left|\mathcal{L}_1\left(\boldsymbol{F}_n\right)\right| = O_p( n^{-\frac{\alpha}{2+\alpha} - \delta}) 
\]
so that, when $n$ is large enough, we have $d_{\mathcal{L}_1\left(\bm{F}_n\right)}(\hat{\ell},\ell) < 1$ so that $d_{\mathcal{L}_1\left(\bm{F}_n\right)}(\hat{\ell},\ell) = 0$ necessarily. In other words, when $\alpha \geq 0$, the SPAR algorithm can perfectly recover the community label of the layer-1 nodes with high probability. When $\alpha \in (-1, 0)$, the right-hand side of~\eqref{thm: layer 1/2 eqn 1} appears to decrease as $\alpha$ increases, suggesting that larger $\alpha$ leads to better recovery. However, this interpretation is incorrect as, for larger $\alpha$, the noise tolerance bound becomes smaller and the size $|\mathcal{L}_1(\bm{F}_n)|$ becomes smaller so that the recovery problem is substantially easier. 
\end{remark}

\begin{remark}
Our layer-2 recovery result~\eqref{thm: layer 1/2 eqn 2} currently pertains only to the linear preferential attachment case $(\alpha=0)$. The main obstacle to extending this to a general $\alpha$ is that we do not have bounds on the number of layer-2 and layer-3 nodes, i.e. $|\mathcal{L}_2(\bm{F}_n)|$ and $|\mathcal{L}_3(\bm{F}_n)|$, for a general value of $\alpha$. Our proof technique for analyzing the number of layer-$s$ nodes in the $\alpha = 0$ setting unfortunately does not directly generalize beyond linear preferential attachment. Nevertheless, we conjecture that $|\mathcal{L}_s(\bm{G}_n)|$ is of order $n^{\frac{1}{2+\alpha}} \log^{s-1} n$ for any $\alpha > -1$ in which case our bound~\eqref{thm: layer 1/2 eqn 2} can be generalized.

% The main obstacle to extending the result to general affine preferential attachment models is that the proof relies on a precise asymptotic growth rate for the number of layer-2/3 nodes in each tree. 
% Although we establish such results for linear preferential attachment (see Lemma~\ref{lemma layer size} in Section~\ref{secsec: technical lemmas} of the appendix), the underlying proof technique does not appear to extend to the more general affine preferential attachment setting. {\color{orange} [Make it clear we proved the LPA bound]}

% Nevertheless, one may expect a similar extension of \eqref{thm: layer 1/2 eqn 2} to hold for $\alpha>-1$ if an asymptotic estimate of the form $|\mathcal{L}_s(\bm{G}_n)| \asymp n^{\frac{1}{2+\alpha}}\log^{s-1} n$
% can be established.
\end{remark}

\begin{remark}
Bounding the misclustering error rate for layer-$s$ nodes where $s \geq 3$ becomes substantially more challenging, even in the linear preferential attachment case $(\alpha=0)$. The main difficulty is that the misclassification errors accumulate across layers. Indeed, if a parent node at layer $s$ is misclassified, then we cannot meaningfully bound the misclassification error of all of its children at layer $s+1$. Our strong recovery guarantee for the layer-1 nodes allows us to obtain weaker but still non-trivial bounds for the layer-2 nodes, but these are not strong enough to yield a similar guarantee for the layer-3 nodes and beyond. 

%We were able to obtain non-trivial recovery bounds for layer-2 mainly because we have a strong recovery guarantee for layer-1. 

% %because controlling the probability of misclassifying node $u$ requires even stronger control on the misclassification probability of its parent. 

% We were able to obtain meaningful recovery bounds for the layer-2 nodes because the recovery guarantee for 

% A node $u$ tends to be misclassified if its tree parent $v$ is already assigned an incorrect label.
% For layer-2 nodes, \eqref{thm: layer 1/2 eqn 1} implies that the incorrectly labeled nodes count converges to $0$, making such error propagation vanishing for layer-2 nodes. 
% In contrast, when recovering layer-3 nodes, errors inherited
% from misclassified layer-2 nodes, whose number may growth at the rate of $n^{\frac{1}{2}-\delta}$ can no longer be ignored.
\end{remark}

Our theoretical results are developed for the affine preferential–attachment regime, where the existing random graph theory is sufficiently well developed to support rigorous analysis.
We expect these results to remain valid under mild deviations from this setting, such as when late–arriving nodes attach with a random but bounded number of edges, or when the attachment mechanism strengthens the advantage of early arrivals (e.g., exponential attachment).
However, extending the theory beyond the affine case would require new technical tools and is left for future work.

\section{Simulation}
\label{sec: simulation}
We simulate networks
$\boldsymbol{G}_n \sim \mathrm{PF}(\alpha,\theta,\ell,\pi)$ over $B=200$ independent repetitions.
We consider balanced component sizes $n_1 = n_2 \in \{500,1000,1500,2000,2500\}$ and attachment parameters $\alpha \in \{2,0,-0.5\}$.
The noise level is set to $\theta = c n^{-1}$, where $c$ is chosen so that $\theta = 0.01$ when $n=1000$.
We repeat the same simulation design for noise level set as $\theta = c' n^{-3/4}$ and $\theta = c'' n^{-1/2}$ where $c'$ and $c''$ are again set so that $\theta = 0.01$ when $n = 1000$. 

To facilitate comparisons of misclassification rates across different values of $n$ and $\theta$ for a fixed $\alpha$, while reducing variability induced by preferential attachment dynamics, we adopt the following construction.
For each Monte Carlo repetition and each choice of component sizes $n_1=n_2 \in \{500,1000,1500,2000,2500\}$, we first generate a single large preferential attachment forest $\bm{f}_{5000}$ consisting of two trees $\bm{t}^{1}_{2500}$ and $\bm{t}^{2}_{2500}$.
We then extract the induced subgraph formed by the first $n_i$ vertices in each tree to obtain the target forest
$\bm{f}_{n_1+n_2} := \oplus_{i=1}^{2} \bm{t}^{i}_{n_i}$.
Finally, we superimpose an Erd\H{o}s--R\'enyi random graph $\bm{r}_{n_1+n_2} \sim \text{Erd\H{o}s--R\'enyi}(\theta)$ with the corresponding noise level, yielding the observed network
$\bm{g}_{n_1+n_2} := \bm{f}_{n_1+n_2} + \bm{r}_{n_1+n_2}$.

For each simulated network, we apply Algorithm~\ref{alg: step I} followed by the distance-based recovery procedure in Algorithm~\ref{alg: distance recovery}. The size threshold is fixed at $Q=n/10$
\footnote{Our theoretical results are established for a fixed threshold $Q$ and do not depend on how $Q$ scales with the network size.
In contrast, in numerical analysis, we find that initializing with a larger $Q$ and then gradually shrinking it can improve classification performance when the noise level $\theta$ is small.
This strategy reduces the likelihood of selecting clusters without roots during the graph-pruning step, as it encourages the algorithm to retain larger, more stable components in the early stages of pruning.
}
with a shrinkage factor $\zeta=0.8$ as suggested in Remark~\ref{remark: shrinkage parameter}.
For the core degree threshold $\tau$, we employ Algorithm~\ref{alg: choosing the core degree threshold}, which is based on a rank-calibrated simulation procedure.
Specifically, we compute the 0.95 quantile of the maximum rank of the highest-degree vertex across planted components, using $H=5$ and $200$ Monte Carlo simulations.
For $n=1000$ and $\theta=0.01$, this procedure yields thresholds of $5$ for $\alpha=-0.5$ and $7$ for both $\alpha=0$ and $\alpha=2$. As $n$ increases, the corresponding $0.95$ quantiles range from $5$ to $8$.
For simplicity and consistency across all experimental settings, we therefore fix $\tau$ to be the degree of the 10th highest-degree vertex in the observed graph. 
For the recovery step in Algorithm~\ref{alg: distance recovery}, the recovery threshold $\tau'$ is defined as the minimum of the last removed node degree and $\tau/2$.

We summarize the misclassification rates for the $10/50$ first-arriving nodes, 
the $10/50$ highest-degree nodes, and nodes in Layers~0--2, together with the 
overall graph, under different decay rate of $\theta$: 
$\theta \asymp n^{-1}$ in Table~\ref{tab:mean_error_scale_n1}, 
$\theta \asymp n^{-3/4}$ in Table~\ref{tab:mean_error_scale_n12}, and 
$\theta \asymp n^{-1/2}$ in Table~\ref{tab:mean_error_scale_n34}. 
To better illustrate the decreasing trend of the misclassification rate across 
the $200$ Monte Carlo replications, we further display the empirical 
distributions using boxplots under different $\theta$ decay rates and values 
of $\alpha$ in Figures~\ref{fig: first 50 3 times 3}--\ref{fig: layer 1 3 times 3}.

\begin{table}[H]
\centering
\caption{Mean misclassification rate with Algorithm~\ref{alg: step I} and distance recovery method for $\theta \asymp n^{-1}$ (with $\theta=0.01$ for $n=1000$).}
{
\begin{tabular}{c|c|cc|cc|ccc|c}
\hline
\hline
\multirow{3}{*}{$\alpha$} & \multirow{3}{*}{n}
& \multicolumn{8}{c}{Misclassification Rate for Different Subsets} \\
\cline{3-10}
& 
& \multicolumn{2}{c|}{First Arriving}
& \multicolumn{2}{c|}{Highest Degree}
& \multicolumn{3}{c|}{Structural Position}
& \multirow{2}{*}{Overall} \\
\cline{3-9}
& 
& $L=10$ & $L=50$
& $L=10$ & $L=50$
& Root (Layer 0) & Layer 1 & Layer 2
& \\
\hline
\multirow{5}{*}{2}
& 1000 & 0.27 & 0.37 & 0.29 & 0.39 & 0.19 & 0.25 & 0.34 & 0.46 \\
& 2000 & 0.22 & 0.33 & 0.25 & 0.37 & 0.17 & 0.21 & 0.31 & 0.46 \\
& 3000 & 0.20 & 0.31 & 0.24 & 0.35 & 0.14 & 0.19 & 0.29 & 0.46 \\
& 4000 & 0.20 & 0.30 & 0.22 & 0.34 & 0.13 & 0.20 & 0.30 & 0.46 \\
& 5000 & 0.18 & 0.28 & 0.20 & 0.32 & 0.13 & 0.18 & 0.28 & 0.46 \\
\hline
\multirow{5}{*}{0}
& 1000 & 0.16 & 0.25 & 0.14 & 0.29 & 0.08 & 0.16 & 0.26 & 0.39 \\
& 2000 & 0.12 & 0.20 & 0.09 & 0.23 & 0.06 & 0.13 & 0.23 & 0.39 \\
& 3000 & 0.09 & 0.17 & 0.09 & 0.19 & 0.04 & 0.11 & 0.20 & 0.39 \\
& 4000 & 0.08 & 0.16 & 0.07 & 0.17 & 0.04 & 0.09 & 0.20 & 0.39 \\
& 5000 & 0.06 & 0.13 & 0.05 & 0.15 & 0.03 & 0.07 & 0.18 & 0.38 \\
\hline
\multirow{5}{*}{-0.5}
& 1000 & 0.11 & 0.18 & 0.07 & 0.21 & 0.05 & 0.13 & 0.21 & 0.32 \\
& 2000 & 0.07 & 0.14 & 0.05 & 0.15 & 0.03 & 0.10 & 0.18 & 0.31 \\
& 3000 & 0.05 & 0.10 & 0.04 & 0.12 & 0.03 & 0.07 & 0.16 & 0.31 \\
& 4000 & 0.04 & 0.08 & 0.03 & 0.10 & 0.02 & 0.05 & 0.15 & 0.31 \\
& 5000 & 0.03 & 0.07 & 0.03 & 0.09 & 0.01 & 0.04 & 0.14 & 0.31 \\
\hline
\hline
\end{tabular}
}
\label{tab:mean_error_scale_n1}
\end{table}

\begin{table}[H]
\centering
\caption{Mean misclassification rate with Algorithm~\ref{alg: step I} and distance recovery method for $\theta \asymp n^{-3/4}$ (with $\theta=0.01$ for $n=1000$).}
{
\begin{tabular}{c|c|cc|cc|ccc|c}
\hline
\hline
\multirow{3}{*}{$\alpha$} & \multirow{3}{*}{n}
& \multicolumn{8}{c}{Misclassification Rate for Different Subsets} \\
\cline{3-10}
& 
& \multicolumn{2}{c|}{First Arriving}
& \multicolumn{2}{c|}{Highest Degree}
& \multicolumn{3}{c|}{Structural Position}
& \multirow{2}{*}{Overall} \\
\cline{3-9}
& 
& $L=10$ & $L=50$
& $L=10$ & $L=50$
& Root (Layer 0) & Layer 1 & Layer 2
& \\
\hline
\multirow{5}{*}{2}
& 1000 & 0.27 & 0.37 & 0.28 & 0.38 & 0.19 & 0.25 & 0.34 & 0.46 \\
& 2000 & 0.26 & 0.36 & 0.26 & 0.37 & 0.18 & 0.24 & 0.34 & 0.47 \\
& 3000 & 0.23 & 0.34 & 0.25 & 0.37 & 0.15 & 0.23 & 0.33 & 0.47 \\
& 4000 & 0.21 & 0.32 & 0.24 & 0.36 & 0.15 & 0.20 & 0.32 & 0.47 \\
& 5000 & 0.20 & 0.32 & 0.24 & 0.37 & 0.13 & 0.20 & 0.32 & 0.48 \\
\hline
\multirow{5}{*}{0}
& 1000 & 0.15 & 0.25 & 0.14 & 0.29 & 0.07 & 0.15 & 0.25 & 0.39 \\
& 2000 & 0.13 & 0.23 & 0.12 & 0.26 & 0.09 & 0.15 & 0.24 & 0.40 \\
& 3000 & 0.10 & 0.20 & 0.10 & 0.23 & 0.06 & 0.12 & 0.23 & 0.40 \\
& 4000 & 0.09 & 0.18 & 0.08 & 0.20 & 0.04 & 0.12 & 0.22 & 0.41 \\
& 5000 & 0.09 & 0.17 & 0.08 & 0.19 & 0.04 & 0.11 & 0.22 & 0.41 \\
\hline
\multirow{5}{*}{-0.5}
& 1000 & 0.10 & 0.18 & 0.08 & 0.22 & 0.03 & 0.12 & 0.21 & 0.31 \\
& 2000 & 0.09 & 0.16 & 0.06 & 0.18 & 0.04 & 0.12 & 0.20 & 0.32 \\
& 3000 & 0.07 & 0.12 & 0.04 & 0.14 & 0.02 & 0.09 & 0.18 & 0.32 \\
& 4000 & 0.05 & 0.11 & 0.05 & 0.13 & 0.02 & 0.08 & 0.18 & 0.33 \\
& 5000 & 0.05 & 0.10 & 0.04 & 0.11 & 0.02 & 0.06 & 0.17 & 0.33 \\
\hline
\hline
\end{tabular}
}
\label{tab:mean_error_scale_n34}
\end{table}

\begin{table}[H]
\centering
\caption{Mean misclassification rate with Algorithm~\ref{alg: step I} and distance recovery method for $\theta \asymp n^{-1/2}$ (with $\theta=0.01$ for $n=1000$).}
{
\begin{tabular}{c|c|cc|cc|ccc|c}
\hline
\hline
\multirow{3}{*}{$\alpha$} & \multirow{3}{*}{n}
& \multicolumn{8}{c}{Misclassification Rate for Different Subsets} \\
\cline{3-10}
& 
& \multicolumn{2}{c|}{First Arriving}
& \multicolumn{2}{c|}{Highest Degree}
& \multicolumn{3}{c|}{Structural Position}
& \multirow{2}{*}{Overall} \\
\cline{3-9}
& 
& $L=10$ & $L=50$
& $L=10$ & $L=50$
& Root (Layer 0) & Layer 1 & Layer 2
& \\
\hline
\multirow{5}{*}{2}
& 1000 & 0.26 & 0.37 & 0.29 & 0.38 & 0.17 & 0.24 & 0.34 & 0.46 \\
& 2000 & 0.27 & 0.37 & 0.27 & 0.39 & 0.16 & 0.25 & 0.36 & 0.47 \\
& 3000 & 0.27 & 0.37 & 0.30 & 0.39 & 0.20 & 0.26 & 0.35 & 0.48 \\
& 4000 & 0.26 & 0.36 & 0.27 & 0.39 & 0.19 & 0.24 & 0.35 & 0.48 \\
& 5000 & 0.26 & 0.37 & 0.30 & 0.40 & 0.18 & 0.25 & 0.37 & 0.48 \\
\hline
\multirow{5}{*}{0}
& 1000 & 0.16 & 0.25 & 0.14 & 0.29 & 0.08 & 0.16 & 0.26 & 0.39 \\
& 2000 & 0.12 & 0.23 & 0.11 & 0.26 & 0.07 & 0.15 & 0.25 & 0.41 \\
& 3000 & 0.12 & 0.22 & 0.11 & 0.25 & 0.05 & 0.14 & 0.25 & 0.42 \\
& 4000 & 0.12 & 0.21 & 0.10 & 0.24 & 0.06 & 0.15 & 0.26 & 0.42 \\
& 5000 & 0.11 & 0.20 & 0.08 & 0.23 & 0.05 & 0.14 & 0.25 & 0.43 \\
\hline
\multirow{5}{*}{-0.5}
& 1000 & 0.11 & 0.19 & 0.08 & 0.21 & 0.07 & 0.13 & 0.21 & 0.32 \\
& 2000 & 0.09 & 0.16 & 0.06 & 0.19 & 0.05 & 0.13 & 0.21 & 0.33 \\
& 3000 & 0.08 & 0.15 & 0.05 & 0.17 & 0.06 & 0.12 & 0.21 & 0.34 \\
& 4000 & 0.08 & 0.14 & 0.05 & 0.16 & 0.04 & 0.11 & 0.21 & 0.35 \\
& 5000 & 0.07 & 0.13 & 0.05 & 0.15 & 0.04 & 0.10 & 0.20 & 0.35 \\
\hline
\hline
\end{tabular}
}
\label{tab:mean_error_scale_n12}
\end{table}

\begin{figure}[H]
    \centering
        {\includegraphics[width=.77\linewidth]{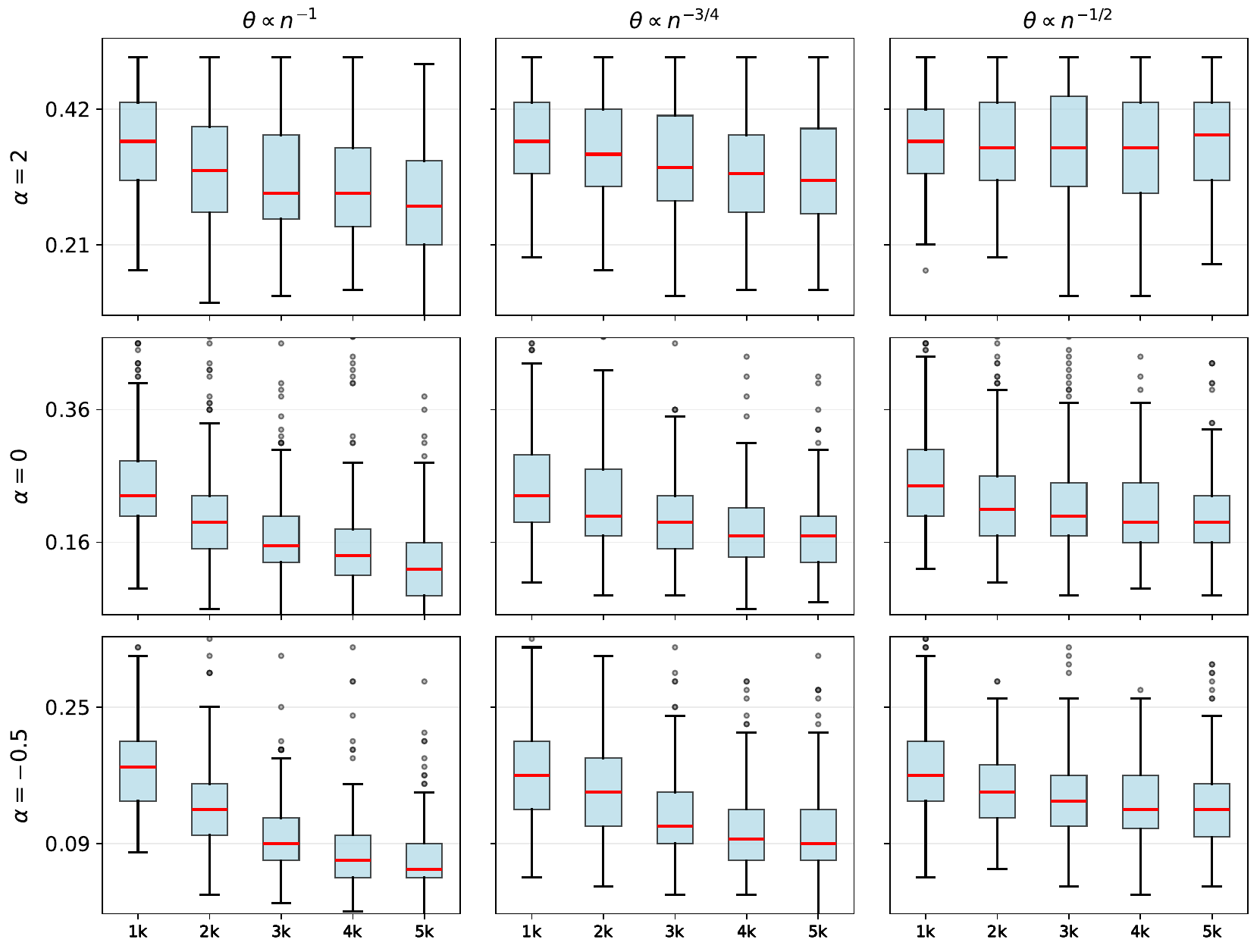}}
    \caption{Boxplots of misclassification rates for the first $50$ arriving vertices in each tree, computed over $200$ simulated networks. }
    \label{fig: first 50 3 times 3}
\end{figure}

%Each row corresponds to a different $\alpha$ ($-0.5$, $0$, and $2$ from top to bottom), and each column corresponds to a different noise scaling regime ($\theta \propto n^{-1}$, $n^{-3/4}$, and $n^{-1/2}$ from left to right). The $x$-axis represents the network size $n$, while the $y$-axis represents the misclassification rate.

\begin{figure}[H]
    \centering
        {\includegraphics[width=.77\linewidth]{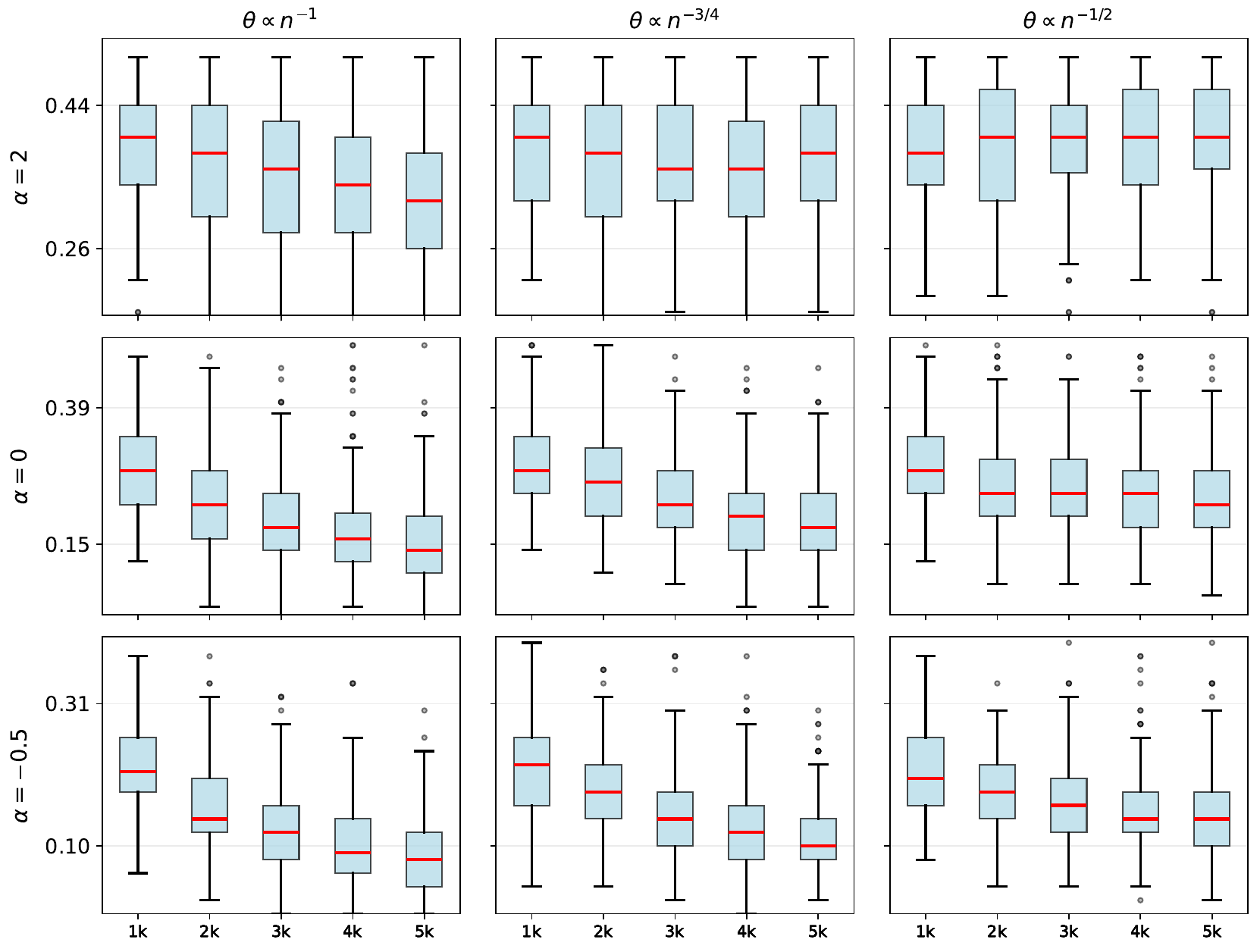}}
    \caption{Boxplots of misclassification rates for the $50$ highest degree vertices in the network, computed over $200$ simulated networks.}
    \label{fig: high degree 50 3 times 3}
\end{figure}

%Each row corresponds to a different $\alpha$ ($-0.5$, $0$, and $2$ from top to bottom), and each column corresponds to a different noise scaling regime ($\theta \propto n^{-1}$, $n^{-3/4}$, and $n^{-1/2}$ from left to right). The $x$-axis represents the network size $n$, while the $y$-axis represents the misclassification rate.

\begin{figure}[H]
    \centering
        {\includegraphics[width=.77\linewidth]{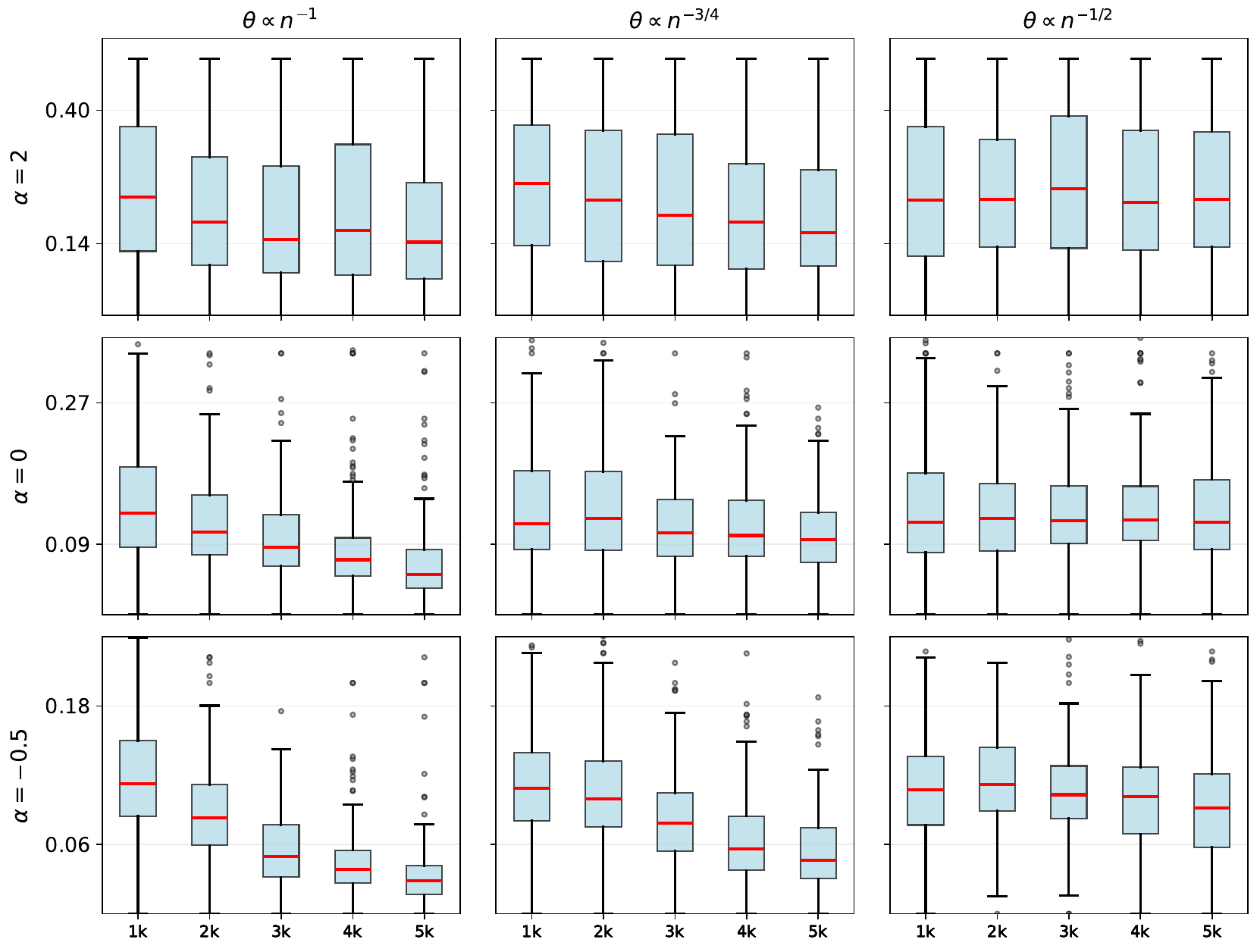}}
    \caption{Boxplots of misclassification rates for the layer-1 vertices in the network, computed over $200$ simulated networks.}
    \label{fig: layer 1 3 times 3}
\end{figure}

%Each row corresponds to a different $\alpha$ ($-0.5$, $0$, and $2$ from top to bottom), and each column corresponds to a different noise scaling regime ($\theta \propto n^{-1}$, $n^{-3/4}$, and $n^{-1/2}$ from left to right). The $x$-axis represents the network size $n$, while the $y$-axis represents the misclassification rate.

From Tables~\ref{tab:mean_error_scale_n1}--\ref{tab:mean_error_scale_n12} 
and Figures~\ref{fig: first 50 3 times 3}--\ref{fig: layer 1 3 times 3}, we observe that, for a fixed level of network size $n$ and noise level $\theta$, the misclassification rate in general improves as $\alpha$ decreases. 
These simulation results, which show better performance for smaller values of $\alpha$, are consistent with our theoretical findings that smaller $\alpha$ permits a slower decreasing rate of the noise level $\theta$.
They also align with the underlying intuition: when $\alpha$ is smaller, the attachment mechanism places relatively more weight on early high-degree vertices, increasing the likelihood that later-arriving nodes connect to them. 
As a result, the subsets become more informative, leading to improved classification accuracy.

\begin{table}[H]
\centering
\caption{Observed decreasing trends in misclassification rates for first-arriving nodes as n increases under different noise regimes.}
\begin{tabular}{c|c|c|c|c}
\hline
$\alpha$ 
& Critical rate $n^{-\frac{1+\alpha}{2+\alpha}}$
& $\theta \asymp n^{-1}$
& $\theta \asymp n^{-3/4}$
& $\theta \asymp n^{-1/2}$ \\
\hline
$2$
& $n^{-3/4}$
&  clearly decreasing
&  clearly decreasing
&  not decreasing \\
\hline
$0$
& $n^{-1/2}$
& clearly decreasing
& clearly decreasing
& clearly decreasing\\
\hline
$-0.5$
& $n^{-1/3}$
& clearly decreasing
& clearly decreasing
& clearly decreasing \\
\hline
\end{tabular}
\label{tab: trend}
\end{table}

We next investigate how the misclassification rate varies with $n$, while keeping $\alpha$ and the noise decay rate fixed, and interpret the results in light of the noise bound in Assumption~\ref{assumption: sparsity}.
Because our simulations are restricted to moderate sample sizes ($n \le 5000$), the asymptotic convergence of the $95\%$ quantile to zero predicted by theory is not directly observable in Figures~\ref{fig: first 50 3 times 3}--\ref{fig: layer 1 3 times 3}. We therefore focus on whether a decreasing trend with respect to $n$ can be detected.
The observed patterns of misclassification rate for first arriving nodes are summarized in Table~\ref{tab: trend}. Overall, the empirical behavior is broadly consistent with the theoretical predictions.
Interestingly, when $\theta$ lies at the critical threshold $n^{-\frac{1+\alpha}{2+\alpha}}$.
For example, $\theta \asymp n^{-1/2}$ with $\alpha=0$ or $\theta \asymp n^{-3/4}$ with $\alpha=2$, the misclassification rate sometimes still exhibits a decreasing tendency. 
This observation suggests that, in Assumption~\ref{assumption: sparsity}, the exponent gap parameter $\delta$ may potentially be taken to be zero.

One additional pattern emerge from Tables~\ref{tab:mean_error_scale_n1}--\ref{tab:mean_error_scale_n12}.
When $\theta \asymp n^{-1}$ (Table~\ref{tab:mean_error_scale_n1}), the overall mean misclassification rate remains stable or decreases slightly as $n$ increases. 
In contrast, under $\theta \asymp n^{-3/4}$ and $\theta \asymp n^{-1/2}$ (Tables~\ref{tab:mean_error_scale_n34} and \ref{tab:mean_error_scale_n12}), the overall misclassification rate exhibits a mildly increasing trend.
These observations suggest that the asymptotic behavior of the overall misclassification rate for our method may depend on the noise level $\theta$. A more detailed understanding of this dependence is beyond the scope of the current theoretical analysis and is left for future investigation.
Additional simulations for multiple unbalanced planted trees are reported in Section~\ref{sec: additional simulation}.

\section{Case Study}
\label{sec: case study}
We now apply our SPAR algorithm to perform an extensive analysis of a statistician co-authorship network constructed by \cite{ji2016coauthorship}.
In this network, each node corresponds to a statistician and two nodes $u$ and $v$ have an edge between them if they have co-authored 1 or more publications in either Journal of Royal Statistical Society: Series B, Journal of the American Statistical Association, Annals of Statistics, or Biometrika from 2002 to 2013. 
We consider only the largest connected component which has $n=2263$ nodes and $m=4388$ edges. 
As a preliminary analysis of the dataset, we first apply Algorithm~\ref{alg: step I} to obtain the community anchors where the number of communities is set to be $K = 2$. Then, we apply both the distance-based recovery method (Algorithm~\ref{alg: distance recovery}) and the Model-based recovery method (Algorithm~\ref{alg: Model recovery}) and compare their outputs. 
To select the core degree threshold $\tau$, we first estimate $\alpha = 0$ using the Expectation–Maximization method described in Section S3.1 of the supplementary material of \cite{crane2024root}.
Given this estimate, the core degree threshold is set to the degree of the 10th highest-degree node, following Algorithm~\ref{alg: choosing the core degree threshold}.
The size threshold is chosen as $\lfloor n/10 \rfloor = 226$, and the shrinkage parameter is set to $\zeta = 0.8$ (see Remark~\ref{remark: shrinkage parameter}).

We give the intermediate estimate of community anchors produced by the graph pruning procedure (Algorithm~\ref{alg: step I}) in Figure~\ref{fig: cores}. 
In each graph, the ten highest-degree nodes are labeled by their name. The procedure reveals two prominent anchor components: 
The first core consists primarily of researchers specializing in Bayesian statistics, while the second core corresponds to a large group of high-dimensional statistics researchers.

\begin{figure}[H]
    \centering
    % First subfigure
    \begin{subfigure}[b]{0.45\linewidth}
        \centering
        {\includegraphics[width=\linewidth]{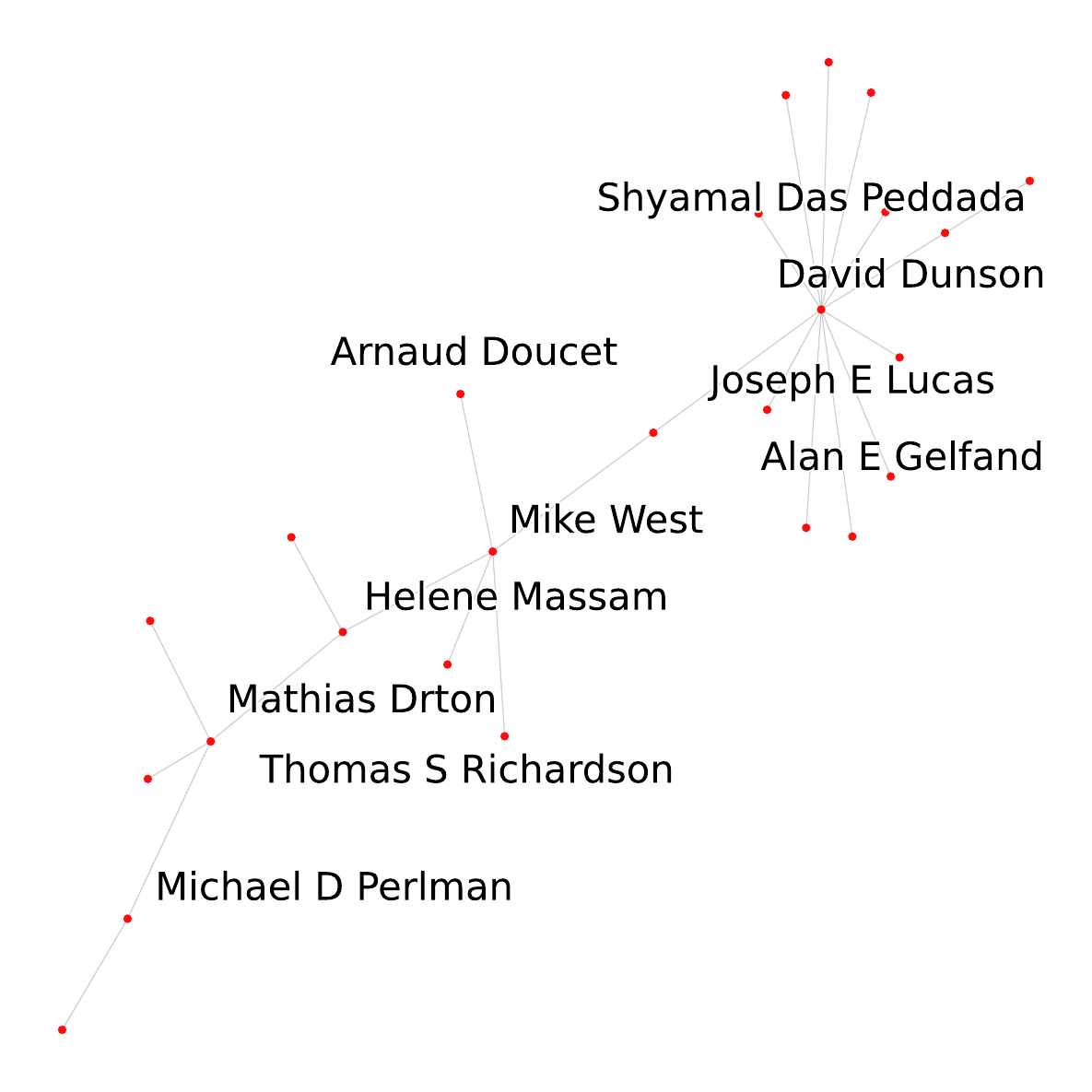}}
        \subcaption{Community Core 1}
    \end{subfigure}
    \hfill
    % Second subfigure
    \begin{subfigure}[b]{0.45\linewidth}
        \centering
        {\includegraphics[width=\linewidth]{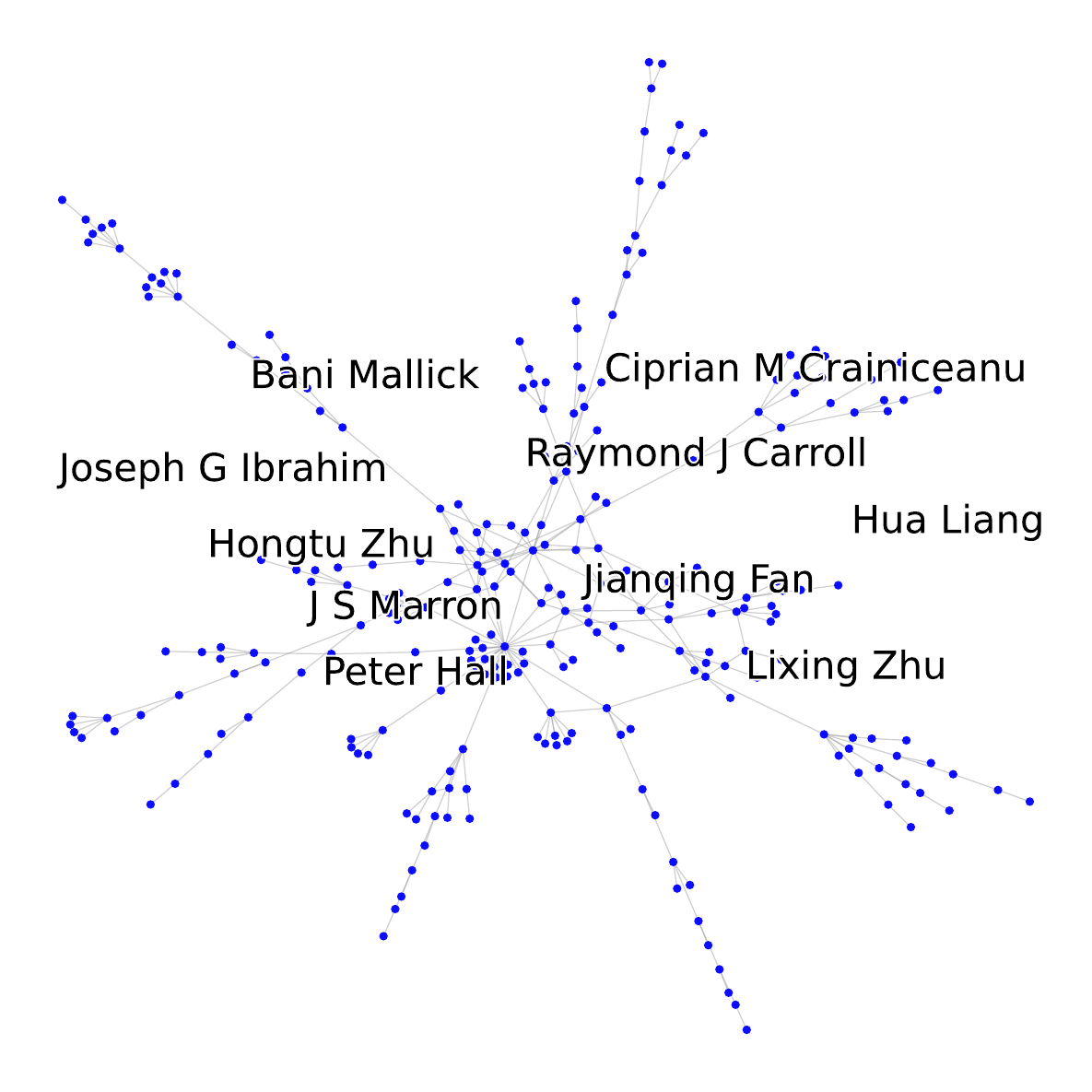}} % Add your image here
        \subcaption{Community Core 2}
    \end{subfigure}
    \caption{Community cores obtained from the co-authorship network of \cite{ji2016coauthorship} using the graph pruning procedure in Algorithm~\ref{alg: step I}. The parameters are set to $(K,\tau ,Q,\zeta)=(2,\text{10th highest degree},226,0.8)$.}
    \label{fig: cores}
\end{figure}

To compare the two different recovery methods, we first present the distance-based recovery results from Algorithm~\ref{alg: distance recovery} in Figure~\ref{fig: distance recovery}. We use the same recovery threshold as in Section~\ref{sec: simulation}, where $\tau'$ is defined as the minimum of two quantities: first is the degree of the last removed node of Algorithm~\ref{alg: step I}, that is, the value of $D^*$ at the conclusion of Algorithm~\ref{alg: step I} and the second is $\tau/2$ where $\tau$ is the core degree threshold from Algorithm~\ref{alg: step I}.
We also report the Monte Carlo recovery results from Algorithm~\ref{alg: Model recovery} in Figure~\ref{fig: mc recovery}, based on $1000$ Monte Carlo samples generated from the implementation of \cite{crane2024root}. 
Among these samples, $842$ root pairs are assigned to the two selected cores identified in Figure~\ref{fig: cores}. For both figures, only the largest connected components of the recovered clusters are shown. The communities obtained from the two methods are highly consistent, with a $98.1\%$ agreement in node classification. 
The main differences occur around the cluster containing James O. Berger and several of his coauthors. This may reflect a limitation of the method near the boundary between communities. From a research perspective, Berger’s work on theoretical Bayesian statistics lies between the two groups, making his position in the network naturally ambiguous.

%The cores produced by the graph pruning stage (Algorithm~\ref{alg: step I}) contain most of the high-degree vertices in the final estimated communities. 

%The selected cores and the final classified communities exhibit similar labels, both corresponding to high-degree nodes in the recovered subnetwork. 

%If high-degree nodes are viewed as informative vertices in the network, this suggests that the cores identified by the graph pruning procedure in Algorithm~\ref{alg: step I} capture the structurally important regions of the final detected communities.

\begin{figure}[H]
    \centering
    % First subfigure
    \begin{subfigure}[b]{0.45\linewidth}
        \centering
        {\includegraphics[width=\linewidth]{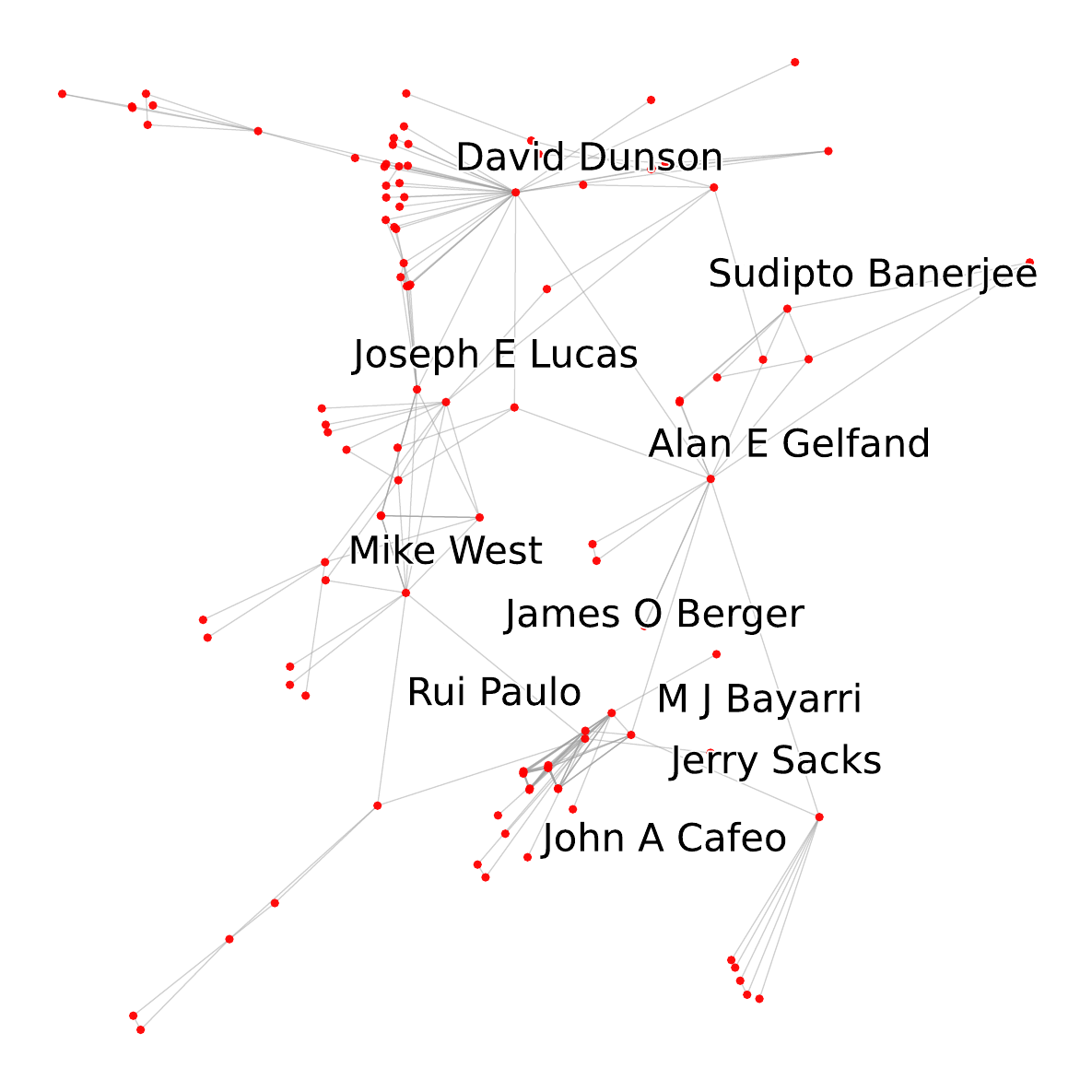}}
        \subcaption{Community 1}
    \end{subfigure}
    \hfill
    % Second subfigure
    \begin{subfigure}[b]{0.45\linewidth}
        \centering
        {\includegraphics[width=\linewidth]{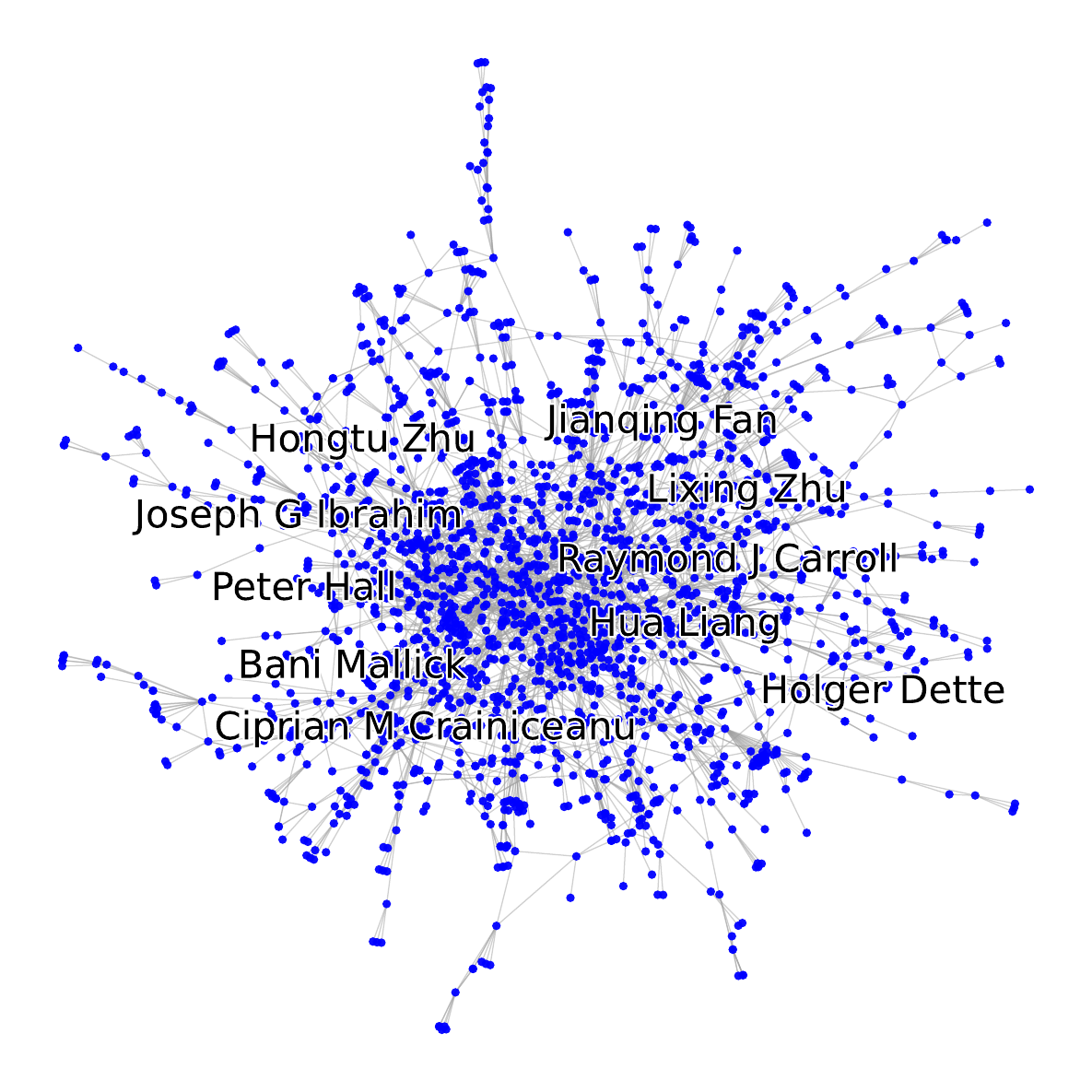}} % Add your image here
        \subcaption{Community 2}
    \end{subfigure}
    \caption{Communities obtained by the distance-based recovery procedure in Algorithm~\ref{alg: distance recovery} with threshold $\tau'=\tau/2$, using the anchor components identified in Figure~\ref{fig: cores}.}
    \label{fig: distance recovery}
\end{figure}

\begin{figure}[H]
    \centering
    % First subfigure
    \begin{subfigure}[b]{0.45\linewidth}
        \centering
        {\includegraphics[width=\linewidth]{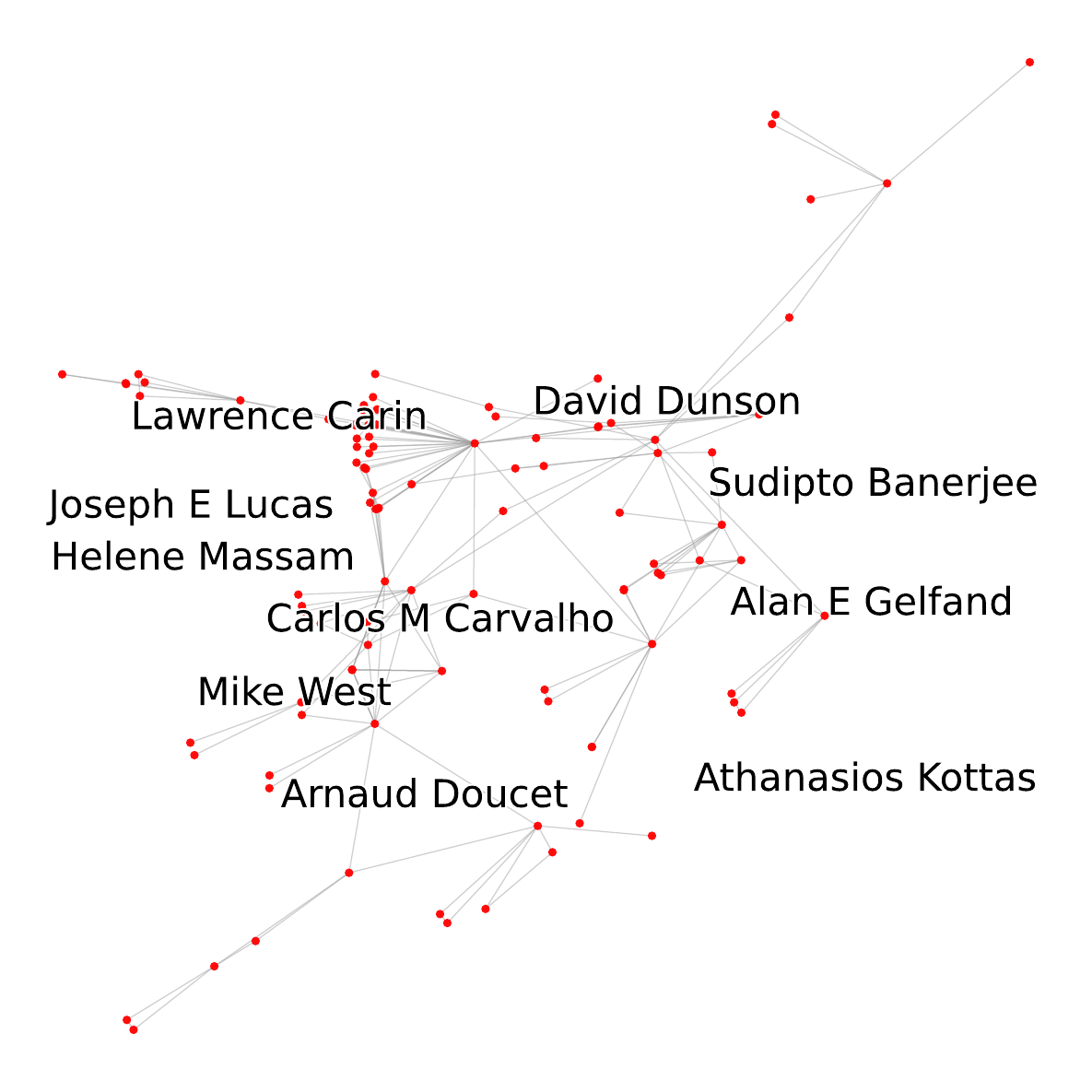}}
        \subcaption{Community 1}
    \end{subfigure}
    \hfill
    % Second subfigure
    \begin{subfigure}[b]{0.45\linewidth}
        \centering
        {\includegraphics[width=\linewidth]{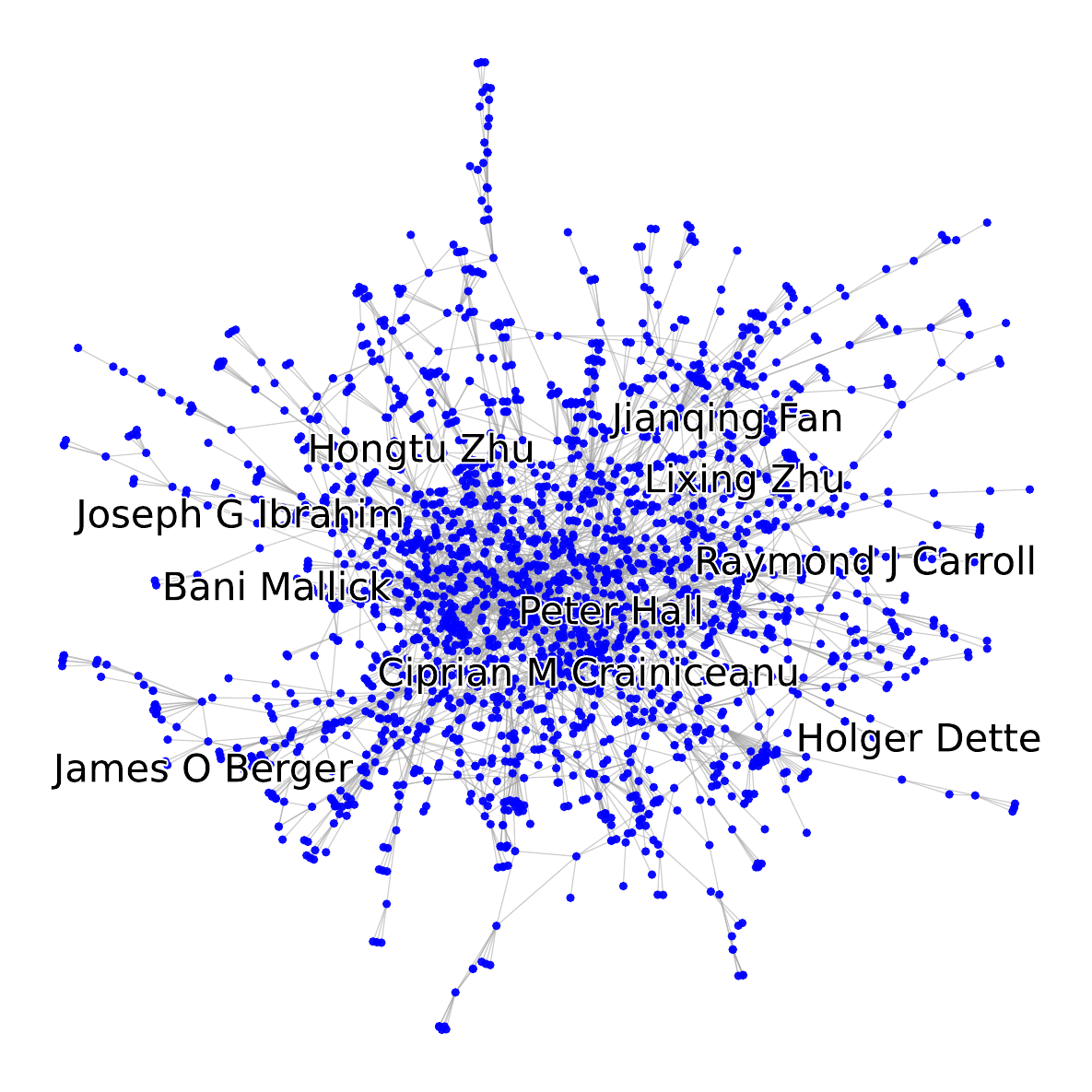}} % Add your image here
        \subcaption{Community 2}
    \end{subfigure}
    \caption{Communities obtained by the Model-based recovery procedure in Algorithm~\ref{alg: Model recovery}, based on $1000$ Monte Carlo samples and using the anchor components identified in Figure~\ref{fig: cores}.}
    \label{fig: mc recovery}
\end{figure}

While the classification results in Figures~\ref{fig: cores}--\ref{fig: mc recovery} suggest that our method produces clusters with coherent research themes, a direct observation is that some recovered cores or communities still contain multiple hidden clusters, particularly within Community~2 in Figure~\ref{fig: distance recovery}. 
This motivates further partitioning of the two detected communities. To avoid estimating the number of communities $K$ and to capture actual hierarchical structure in research communities, we apply a hierarchical clustering procedure with tuning parameters $(K,\tau,Q)=(2,\text{10th highest degree},10)$ for Algorithm~\ref{alg: step I} and recovery threshold $\tau' := \min\left\{\text{last removed node degree},\tau/2\right\}$ for Algorithm~\ref{alg: distance recovery}.

Specifically, we take the two communities produced by the distance recovery stage, extract the largest connected component in each community, and reapply the community detection procedure using the same tuning parameters. We do this recursively and continue as long as the component can be separated into two cores, each containing at least $10$ nodes. 
Once a cluster can no longer be further divided by the procedure, it becomes a terminal (``leaf'') cluster in the hierarchical tree. We rank the resulting clusters from $1$ to $20$ according to the order in which they appear as leaf clusters during the process. 
The hierarchical structure of the clusters is illustrated by the dendrogram in Figure~\ref{fig: dendrogram}.
The resulting clusters are shown in Figures~\ref{fig: h clustering I} and~\ref{fig: h clustering II}, where we also highlight the core within each cluster by coloring the edges of the core in yellow while all other edges are colored as gray. 

%It is worth mentioning that, as only the largest connected component is retained at each step, some isolated vertices may remain unclassified.

Within most of the recovered communities in Figures~\ref{fig: h clustering I} and~\ref{fig: h clustering II}, the researchers share similar research themes or institutional affiliations. 
For example, Community~5 can be interpreted as a “Bayesian computation” cluster, while Community~19 appears to represent a “biostatistics (government)” cluster. 
The results are not perfect, as we can see from Community~20, which emerges last in the hierarchical clustering procedure and appears to pick up all the leftover nodes.
Beyond similarities within individual communities, the dendrogram in Figure~\ref{fig: dendrogram} also reveals structural relationships among communities that are close in the hierarchical tree.
For example, Communities~1 and~2 originate from a larger “Bayesian” cluster and later split into “Bayesian modeling” and “Bayesian machine learning” groups, respectively. Similarly, Communities~13 and~14 arise from the same parent cluster and correspond to “theoretical statistics (European)” and “theoretical statistics (United States)”, respectively.

\begin{figure}
    \centering
    \includegraphics[width=0.5\linewidth, angle=180]{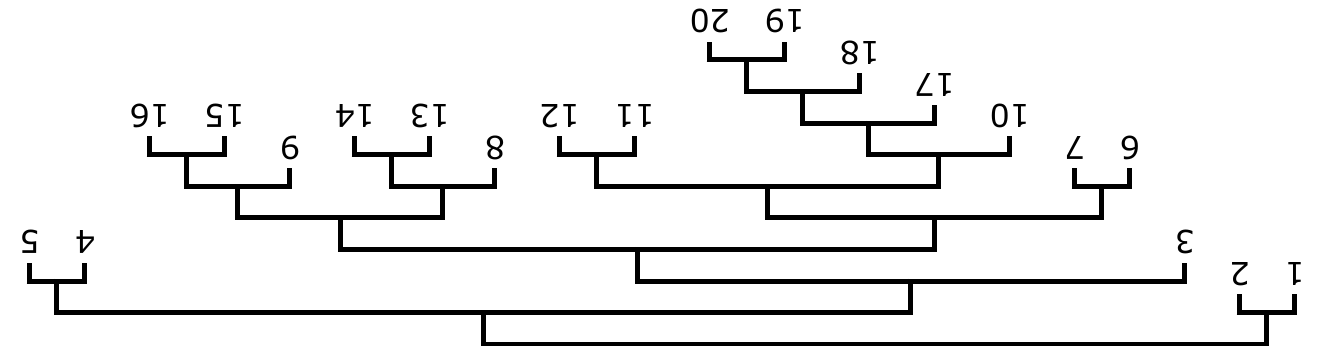}
    \caption{Hierarchical dendrogram from iterative core pruning}
    \label{fig: dendrogram}
\end{figure}

\begin{figure}[H]
    \centering
    % First subfigure
    \begin{subfigure}[b]{0.3\linewidth}
        \centering
        {\includegraphics[width=\linewidth]{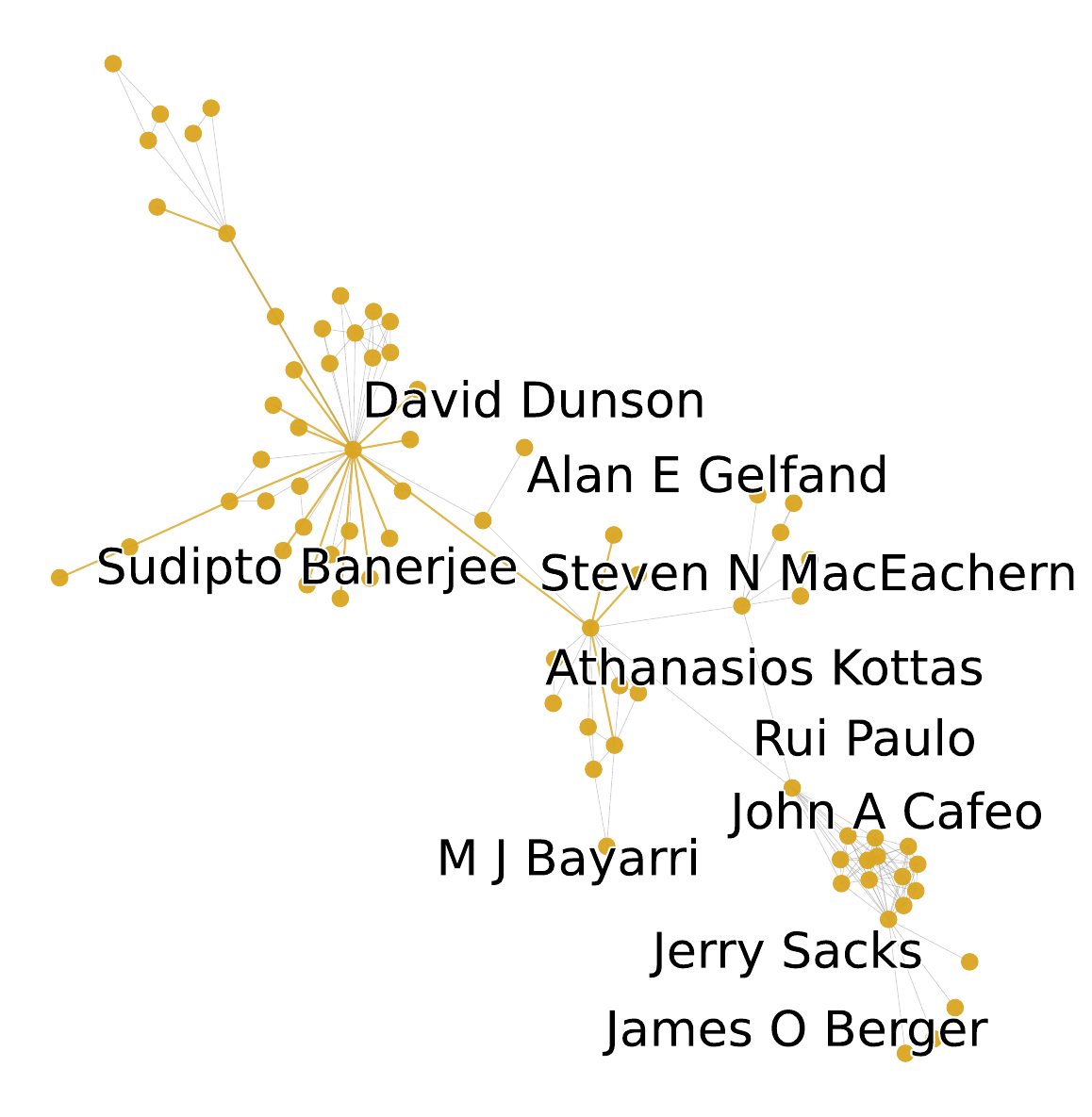}}
        \subcaption{Community 1 \\ (Bayesian Modeling)}
    \end{subfigure}
    \hfill
    % Second subfigure
    \begin{subfigure}[b]{0.3\linewidth}
        \centering
        {\includegraphics[width=\linewidth]{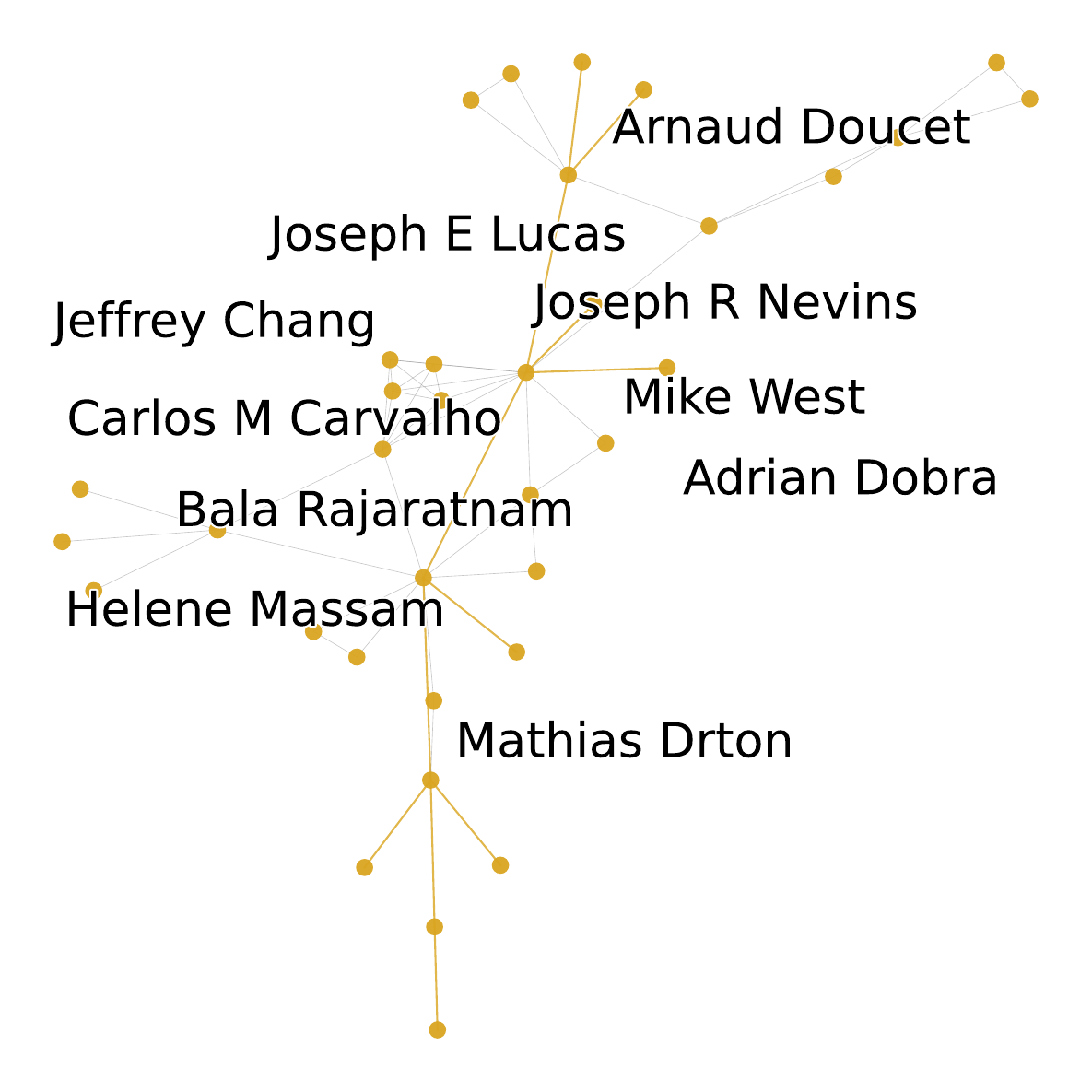}} % Add your image here
        \subcaption{Community 2 \\ (Bayesian ML)}
    \end{subfigure}
    \hfill
    % Second subfigure
    \begin{subfigure}[b]{0.3\linewidth}
        \centering
        {\includegraphics[width=\linewidth]{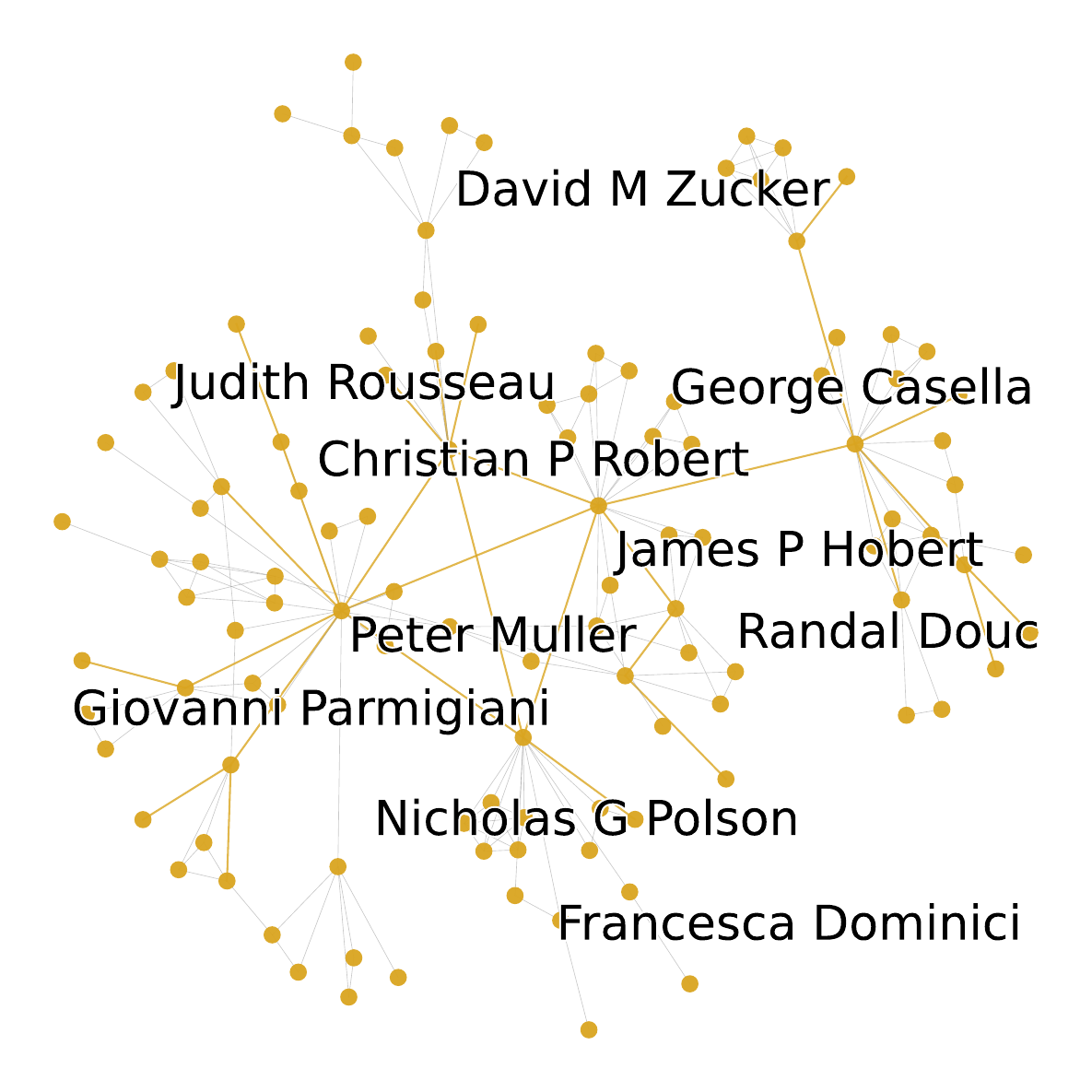}} % Add your image here
        \subcaption{Community 3 \\ (Bayesian Theory)}
    \end{subfigure}
    \hfill
    % First subfigure
    \begin{subfigure}[b]{0.3\linewidth}
        \centering
        {\includegraphics[width=\linewidth]{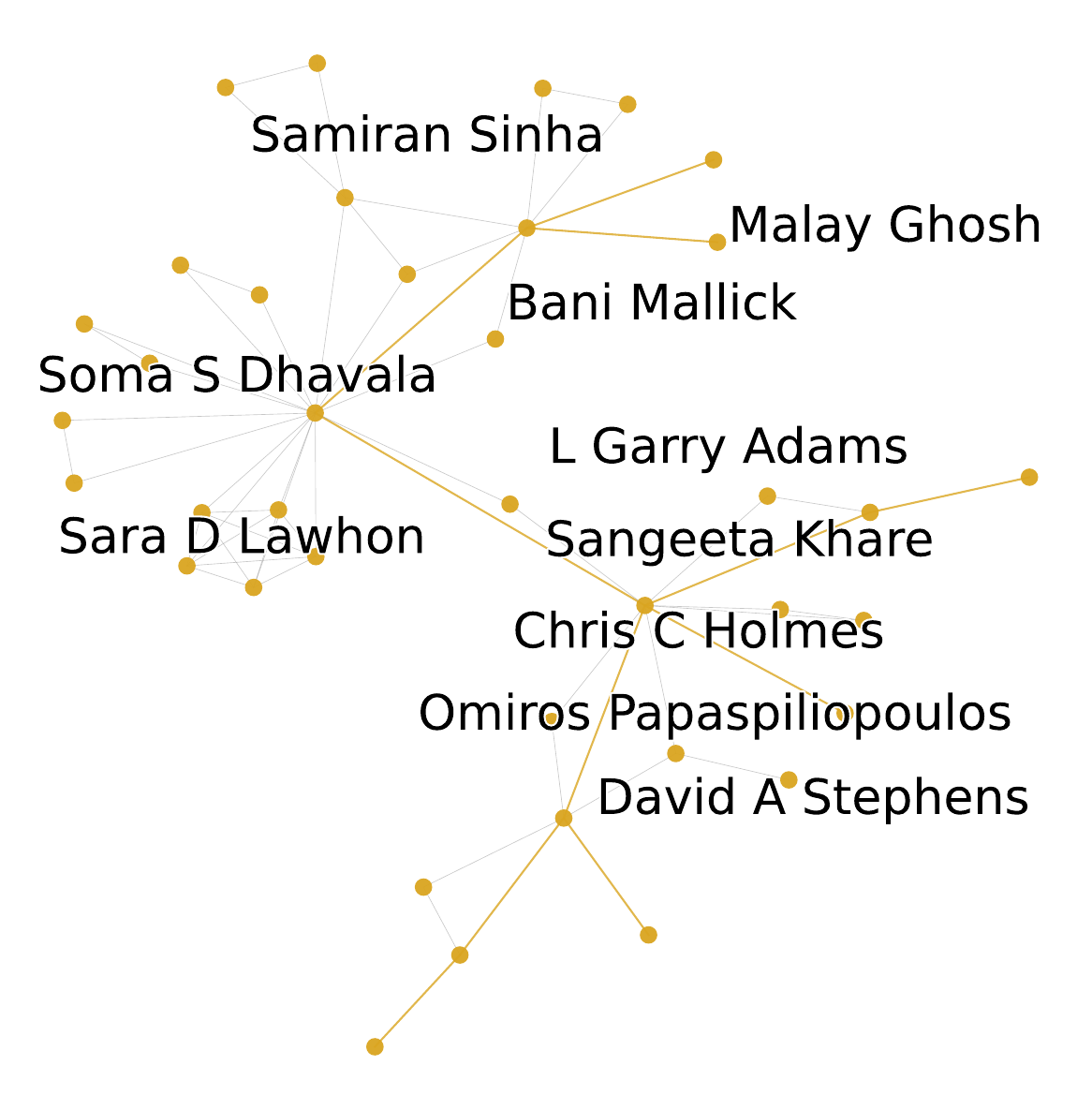}}
        \subcaption{Community 4 \\ (Bayesian Computation 1)}
    \end{subfigure}
    \hfill
    % Second subfigure
    \begin{subfigure}[b]{0.3\linewidth}
        \centering
        {\includegraphics[width=\linewidth]{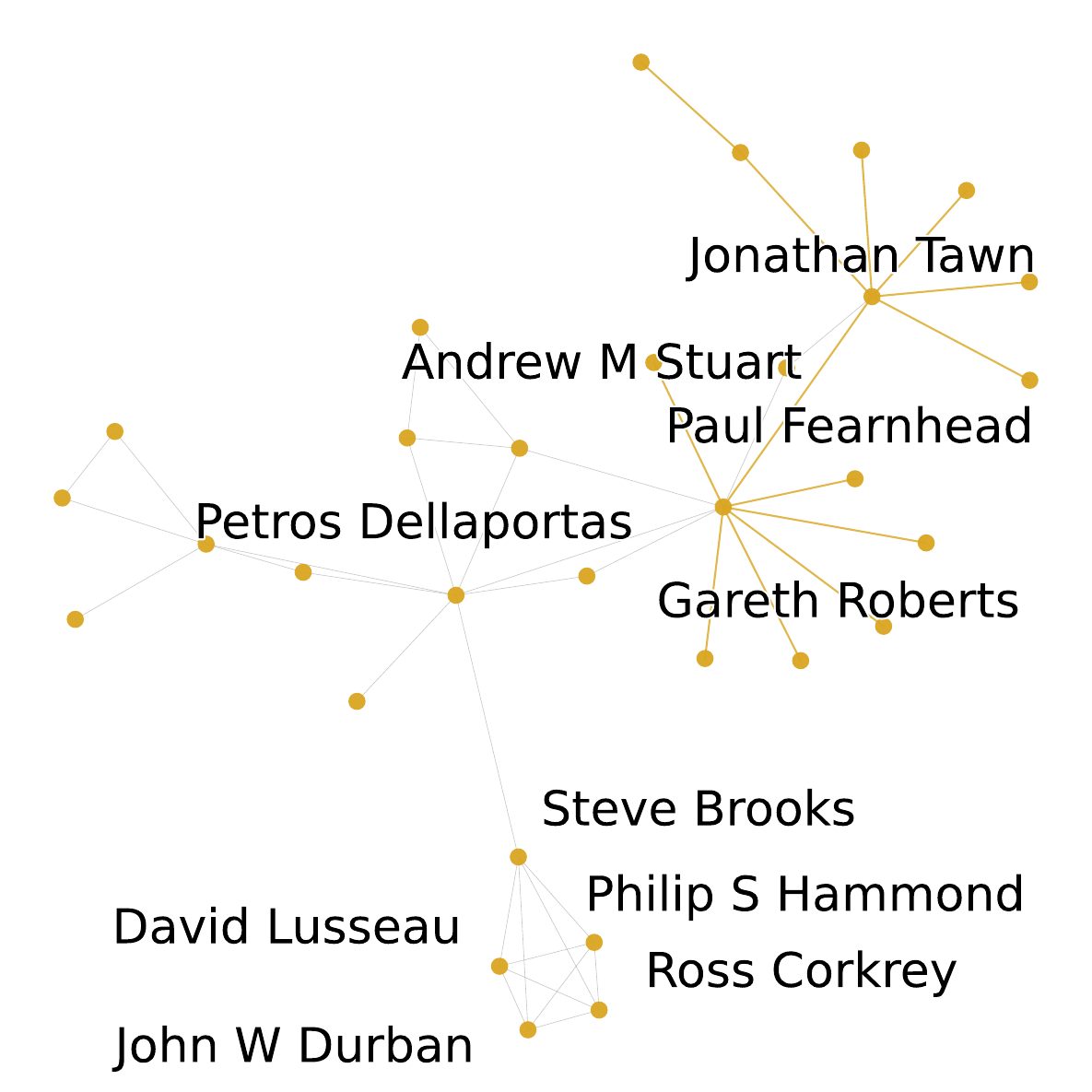}} % Add your image here
        \subcaption{Community 5 \\ (Bayesian Computation 2)}
    \end{subfigure}
    \hfill
    % Second subfigure
    \begin{subfigure}[b]{0.3\linewidth}
        \centering
        {\includegraphics[width=\linewidth]{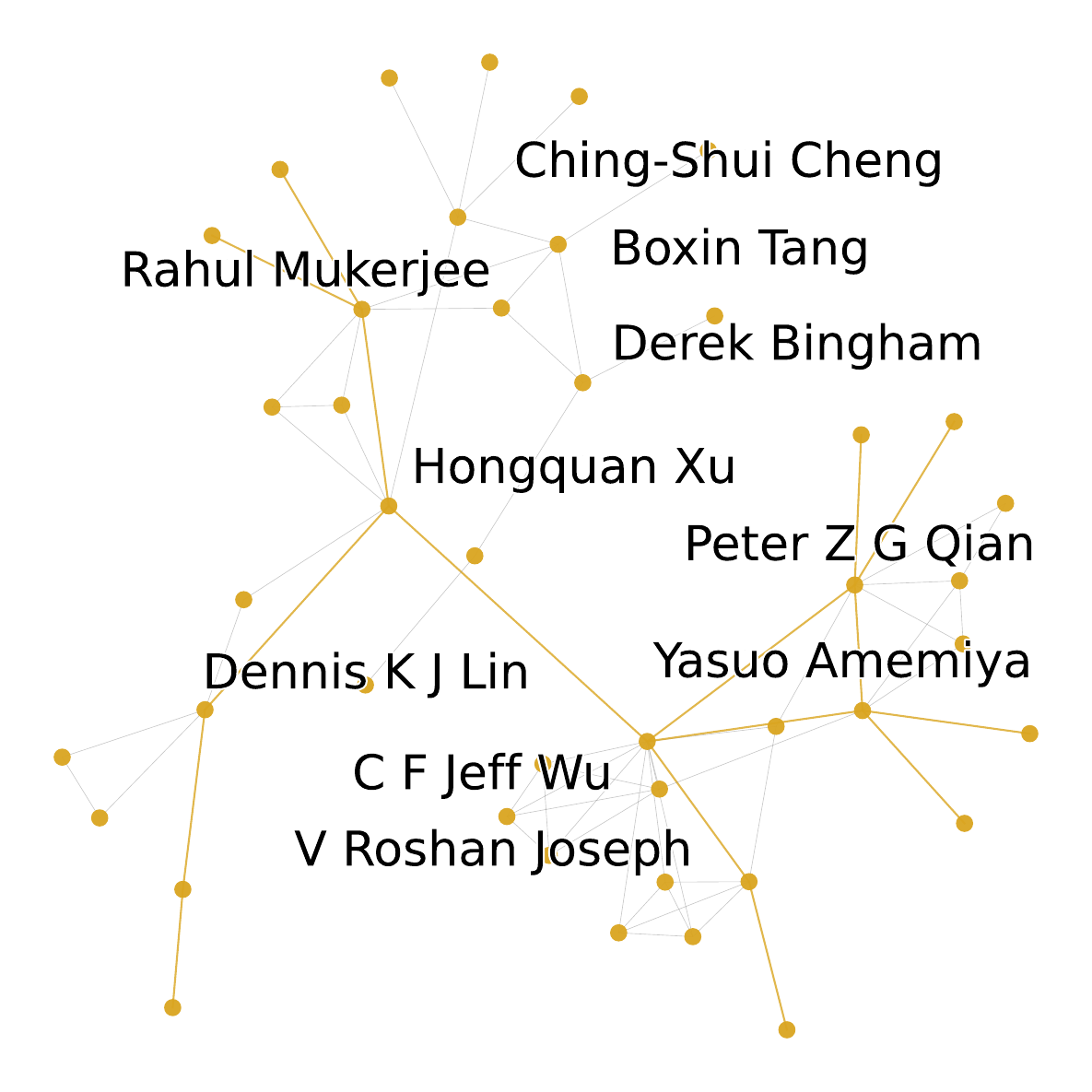}} % Add your image here
        \subcaption{Community 6 \\ (Experimental Design)}
    \end{subfigure}
    \hfill
    % Second subfigure
    \begin{subfigure}[b]{0.3\linewidth}
        \centering
        {\includegraphics[width=\linewidth]{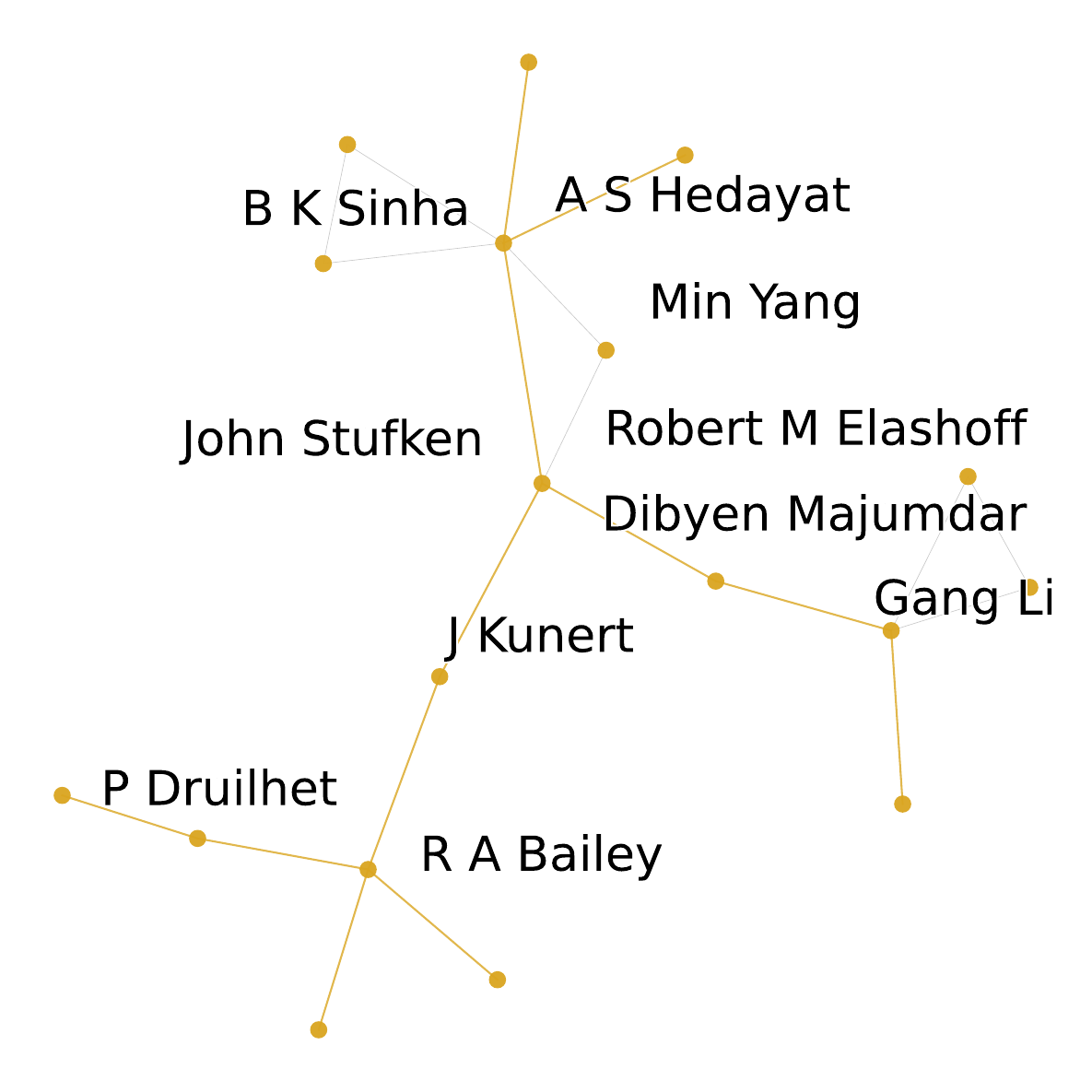}} % Add your image here
        \subcaption{Community 7}
    \end{subfigure}
    \hfill
        % First subfigure
    \begin{subfigure}[b]{0.3\linewidth}
        \centering
        {\includegraphics[width=\linewidth]{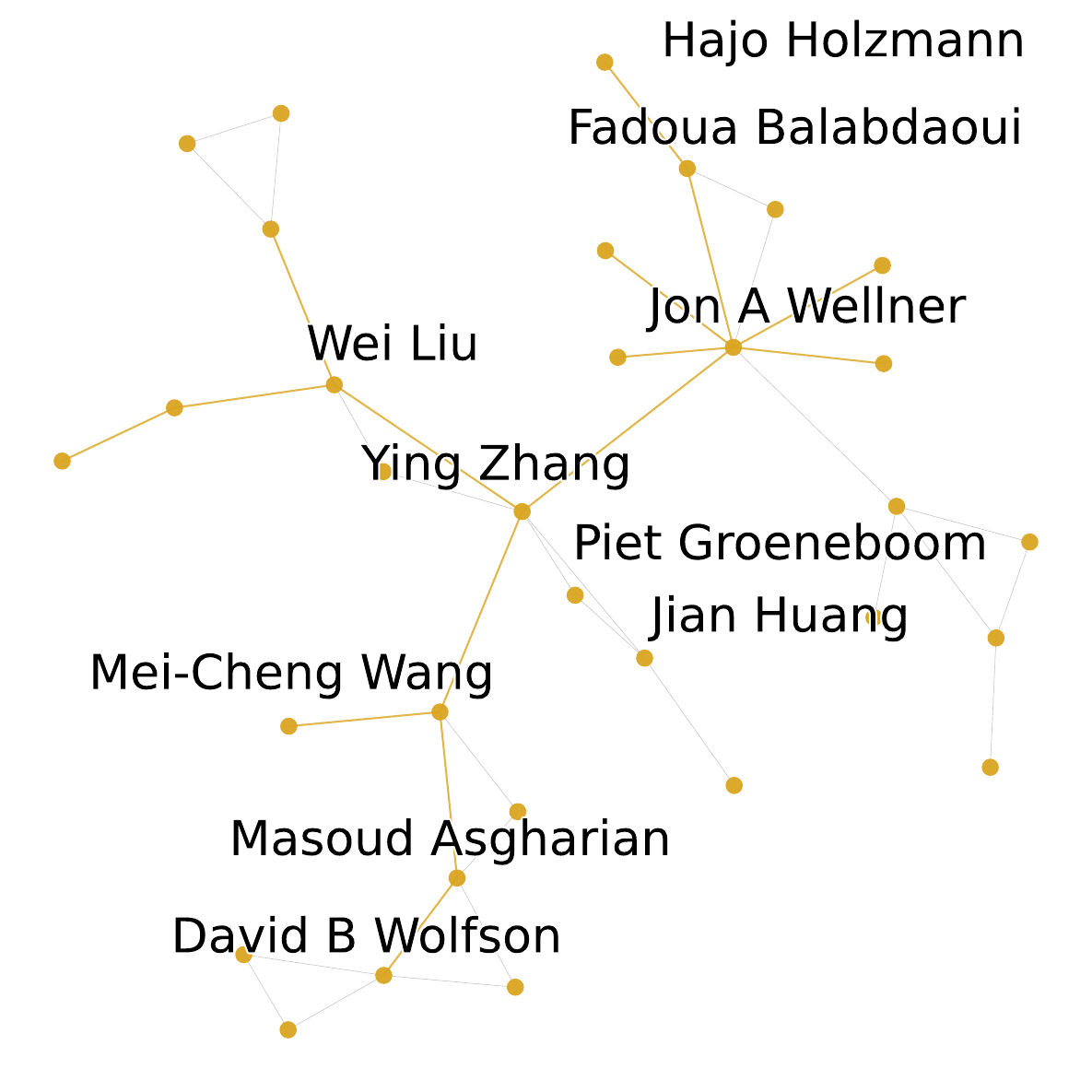}}
        \subcaption{Community 8}
    \end{subfigure}
    \hfill
    % Second subfigure
    \begin{subfigure}[b]{0.3\linewidth}
        \centering
        {\includegraphics[width=\linewidth]{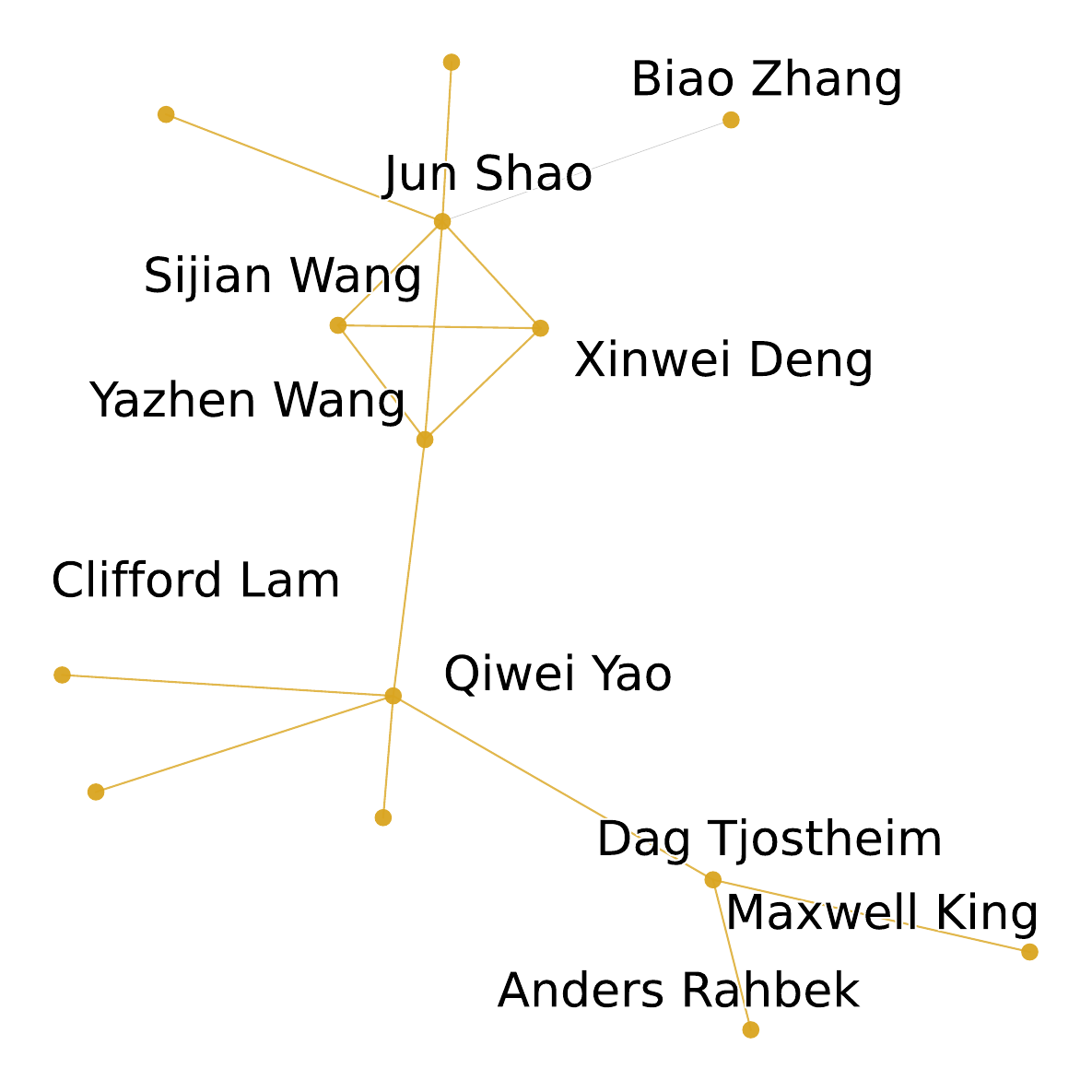}} % Add your image here
        \subcaption{Community 9}
    \end{subfigure}
    \hfill
    % Second subfigure
    \begin{subfigure}[b]{0.3\linewidth}
        \centering
        {\includegraphics[width=\linewidth]{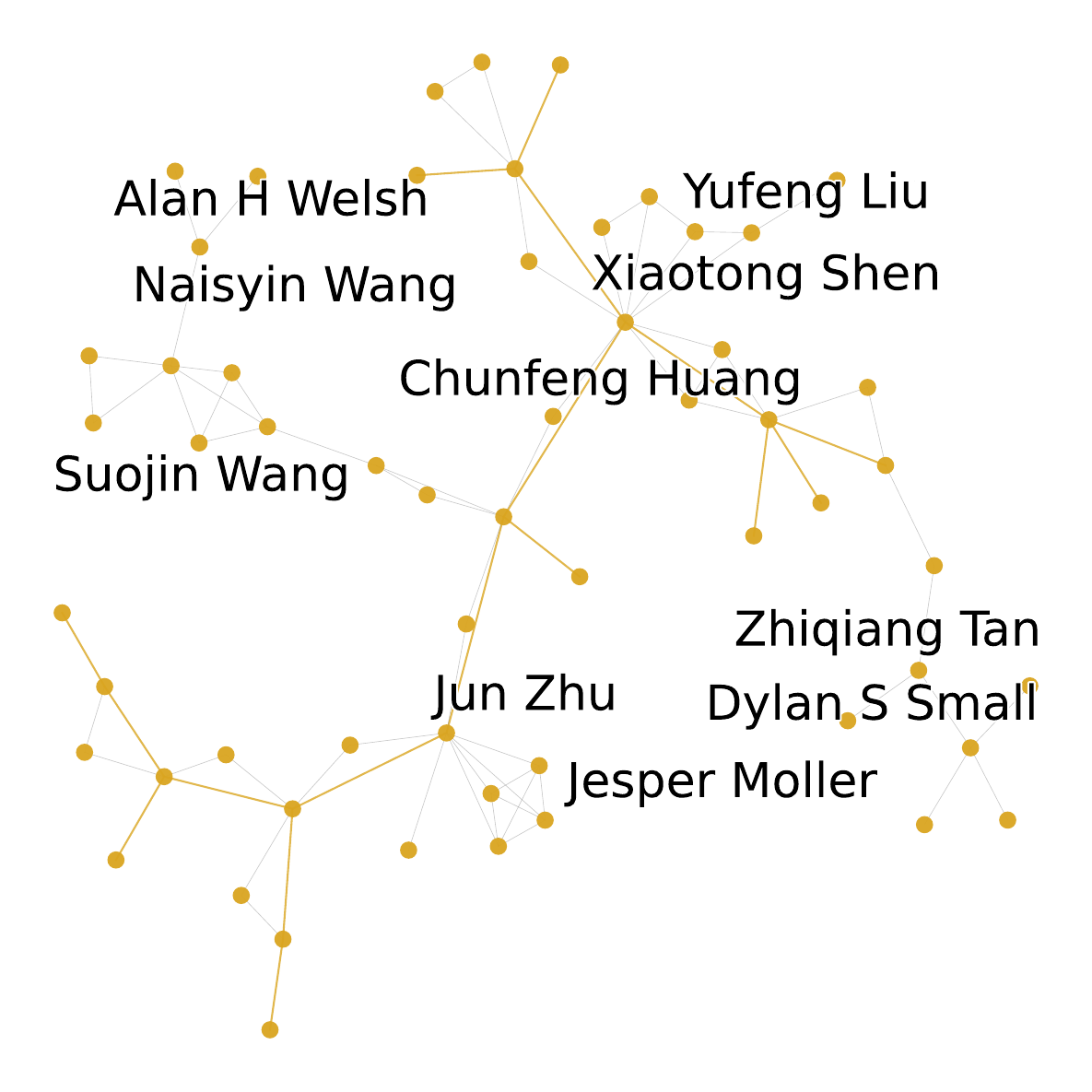}} % Add your image here
        \subcaption{Community 10}
    \end{subfigure}
    \hfill
    % First subfigure
    \begin{subfigure}[b]{0.3\linewidth}
        \centering
        {\includegraphics[width=\linewidth]{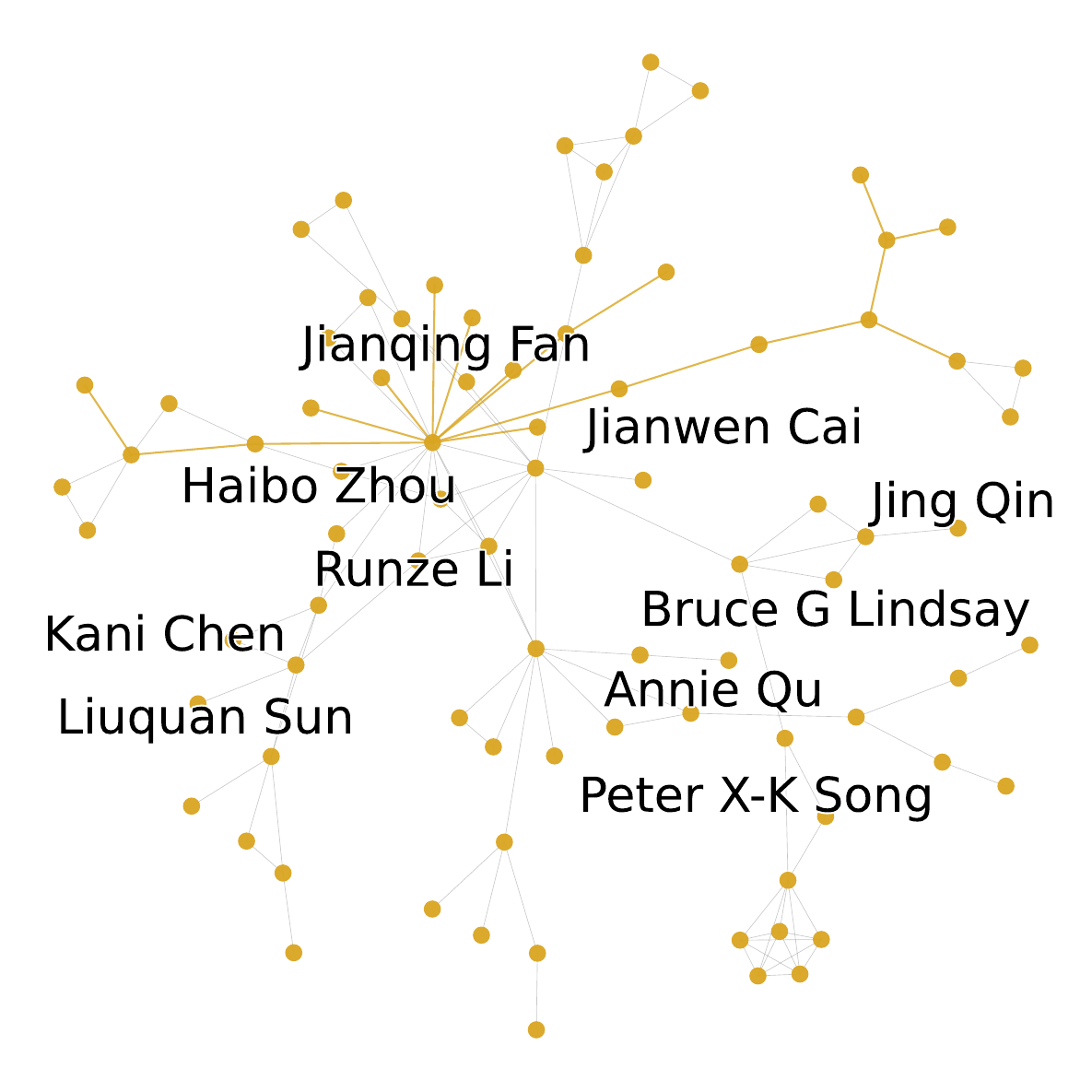}}
        \subcaption{Community 11}
    \end{subfigure}
    \hfill
    % Second subfigure
    \begin{subfigure}[b]{0.3\linewidth}
        \centering
        {\includegraphics[width=\linewidth]{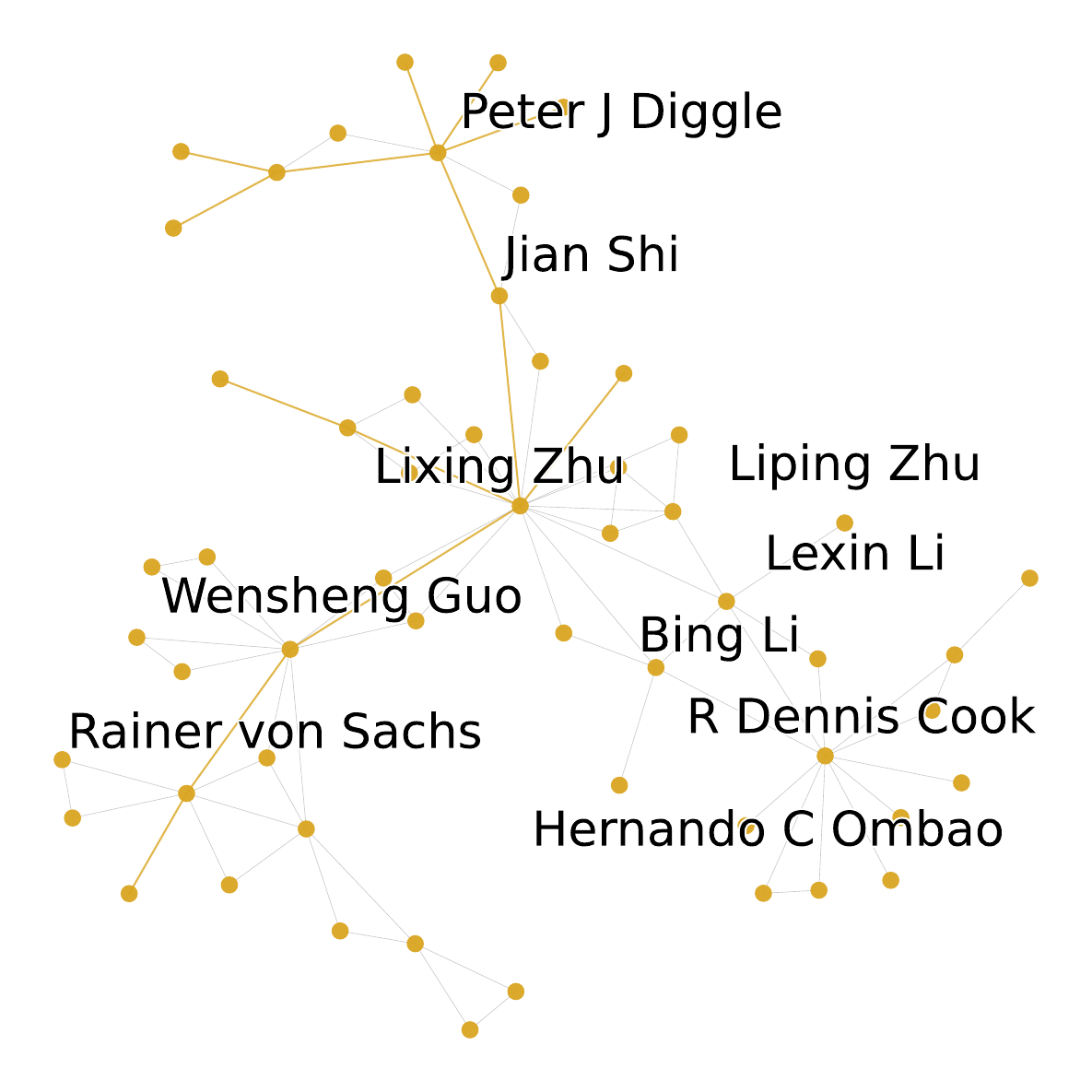}} % Add your image here
        \subcaption{Community 12}
    \end{subfigure}
        \caption{Communities obtained by distance recovery in the hierarchical procedure (Part I)}
        \label{fig: h clustering I}
\end{figure}
\begin{figure}[H]
    \hfill
    % Second subfigure
    \begin{subfigure}[b]{0.3\linewidth}
        \centering
        {\includegraphics[width=\linewidth]{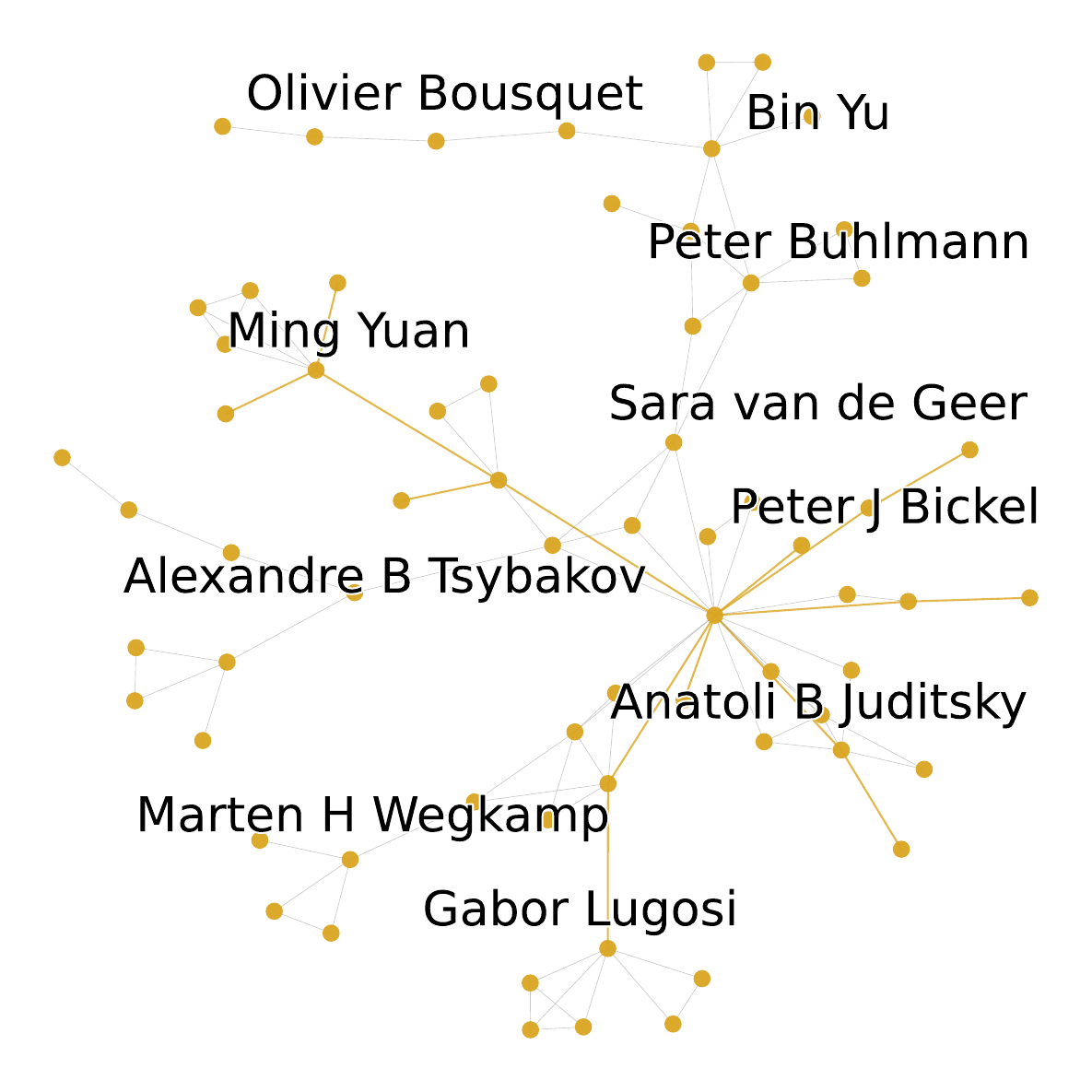}} % Add your image here
        \subcaption{Community 13 \\ (European Theoreticians)}
    \end{subfigure}
    \hfill
    % Second subfigure
    \begin{subfigure}[b]{0.3\linewidth}
        \centering
        {\includegraphics[width=\linewidth]{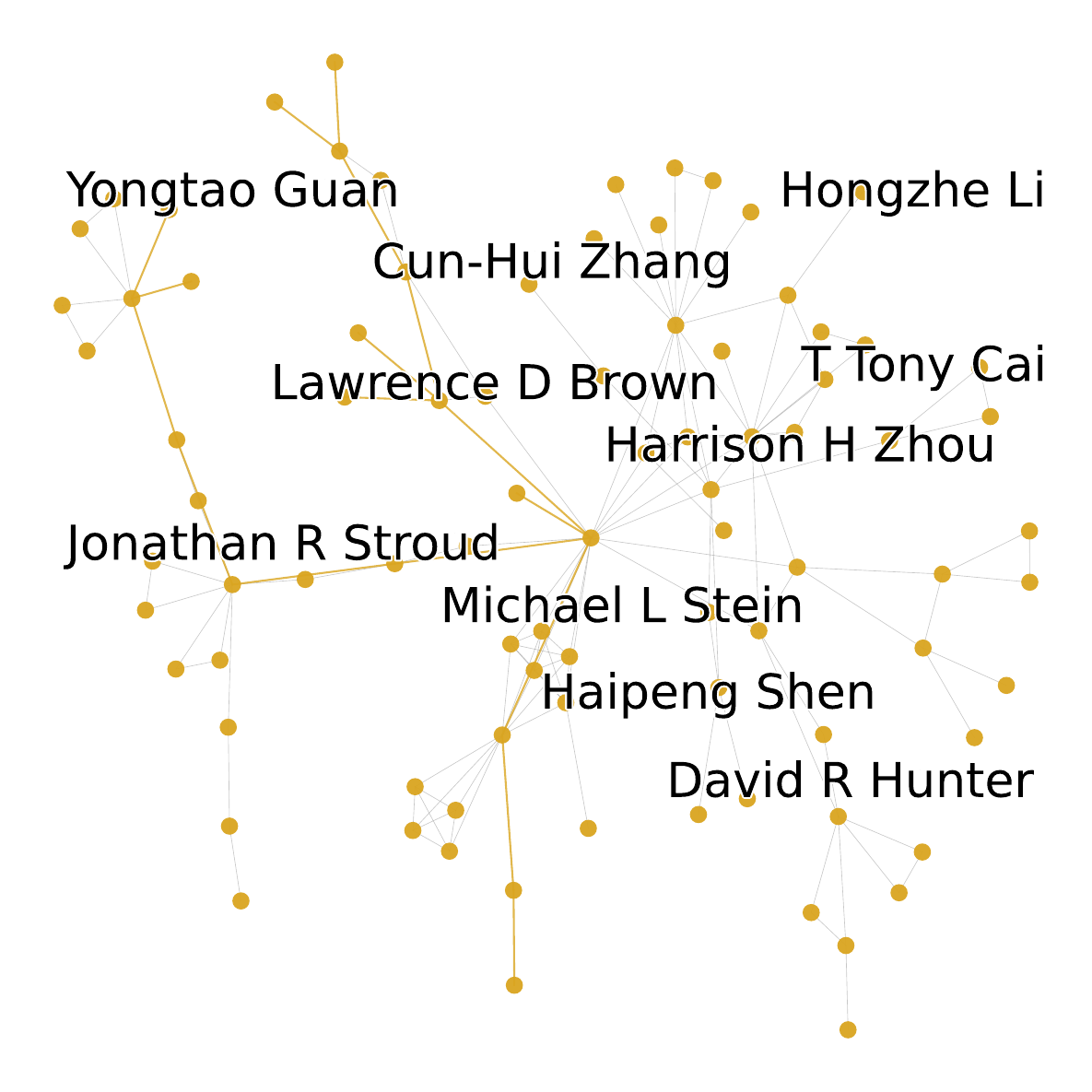}} % Add your image here
        \subcaption{Community 14 \\ (U.S. Theoreticians)}
    \end{subfigure}
    \hfill
    % Second subfigure
    \begin{subfigure}[b]{0.3\linewidth}
        \centering
        {\includegraphics[width=\linewidth]{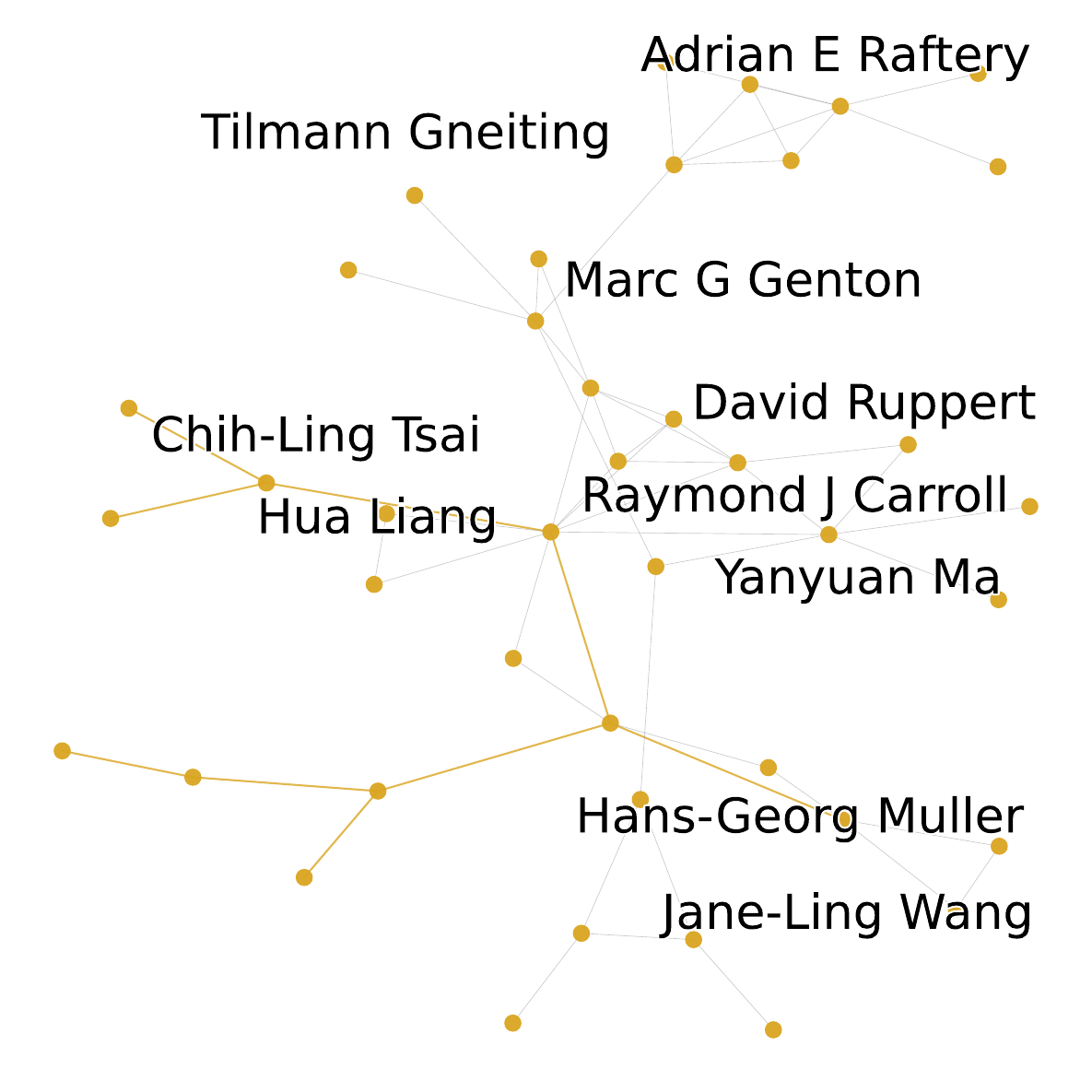}} % Add your image here
        \subcaption{Community 15}
    \end{subfigure}
    \hfill
    % Second subfigure
    \begin{subfigure}[b]{0.3\linewidth}
        \centering
        {\includegraphics[width=\linewidth]{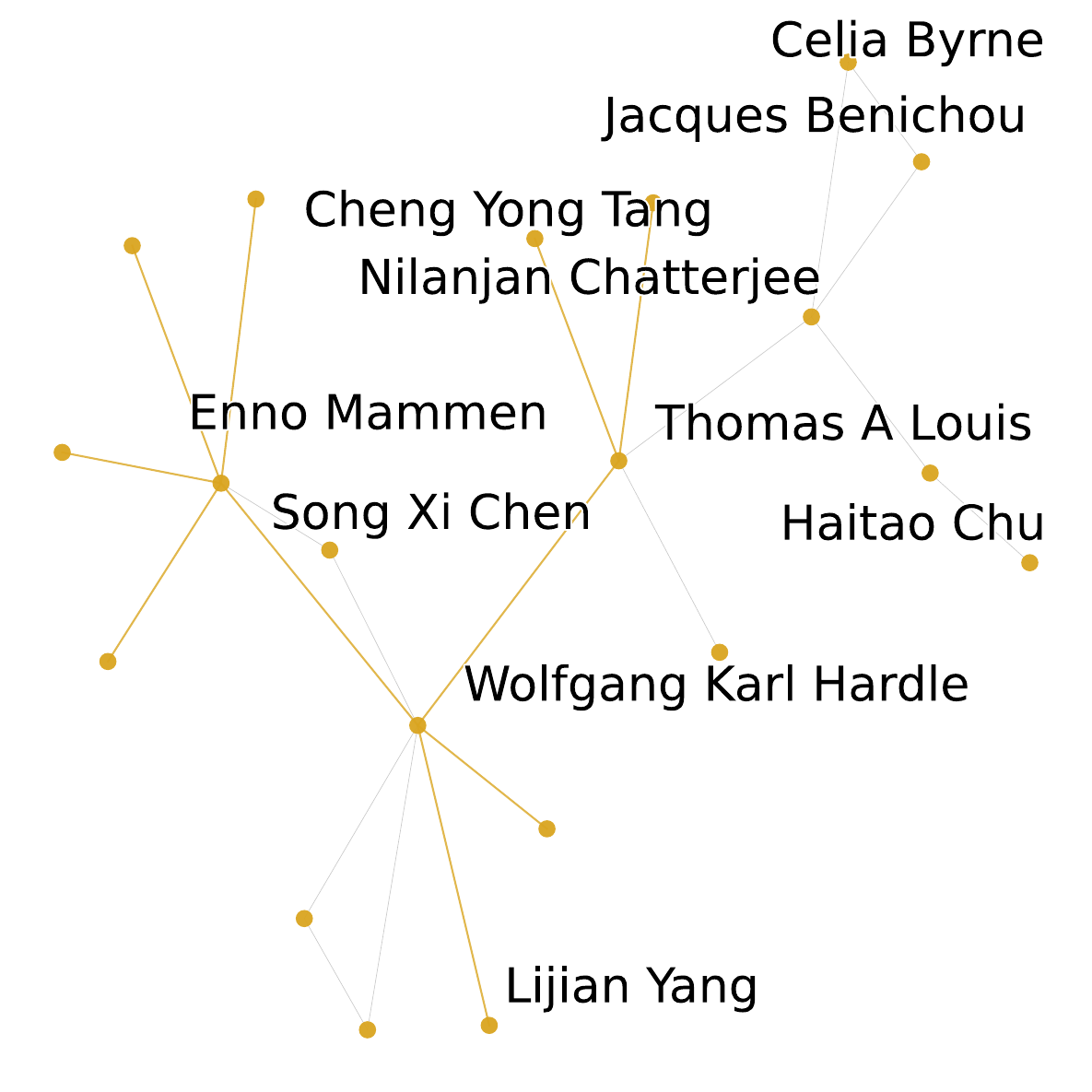}} % Add your image here
        \subcaption{Community 16}
    \end{subfigure}
    \hfill
    % First subfigure
    \begin{subfigure}[b]{0.3\linewidth}
        \centering
        {\includegraphics[width=\linewidth]{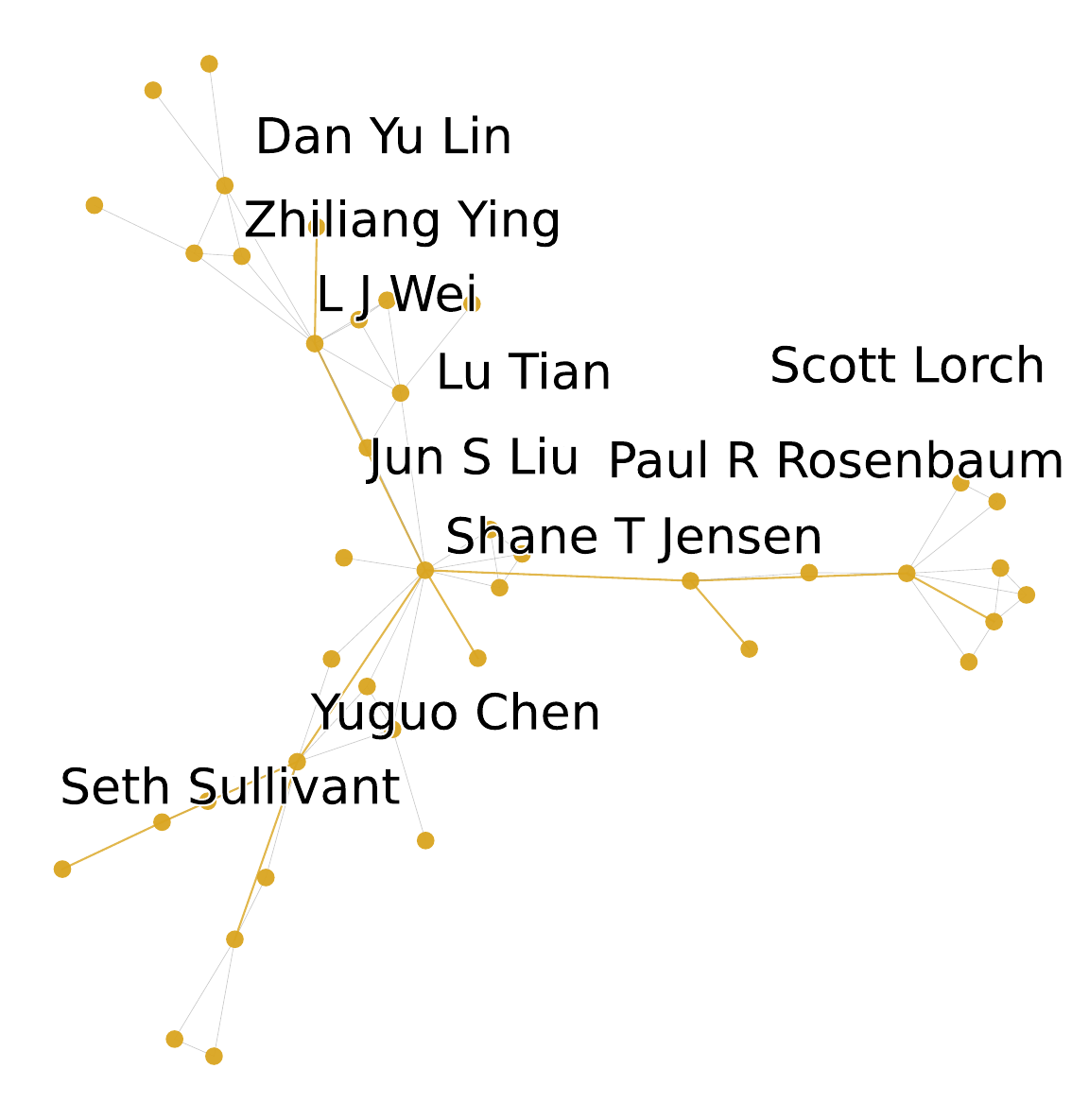}}
        \subcaption{Community 17 \\ (Biostatistics)}
    \end{subfigure}
    \hfill
    % Second subfigure
    \begin{subfigure}[b]{0.3\linewidth}
        \centering
        {\includegraphics[width=\linewidth]{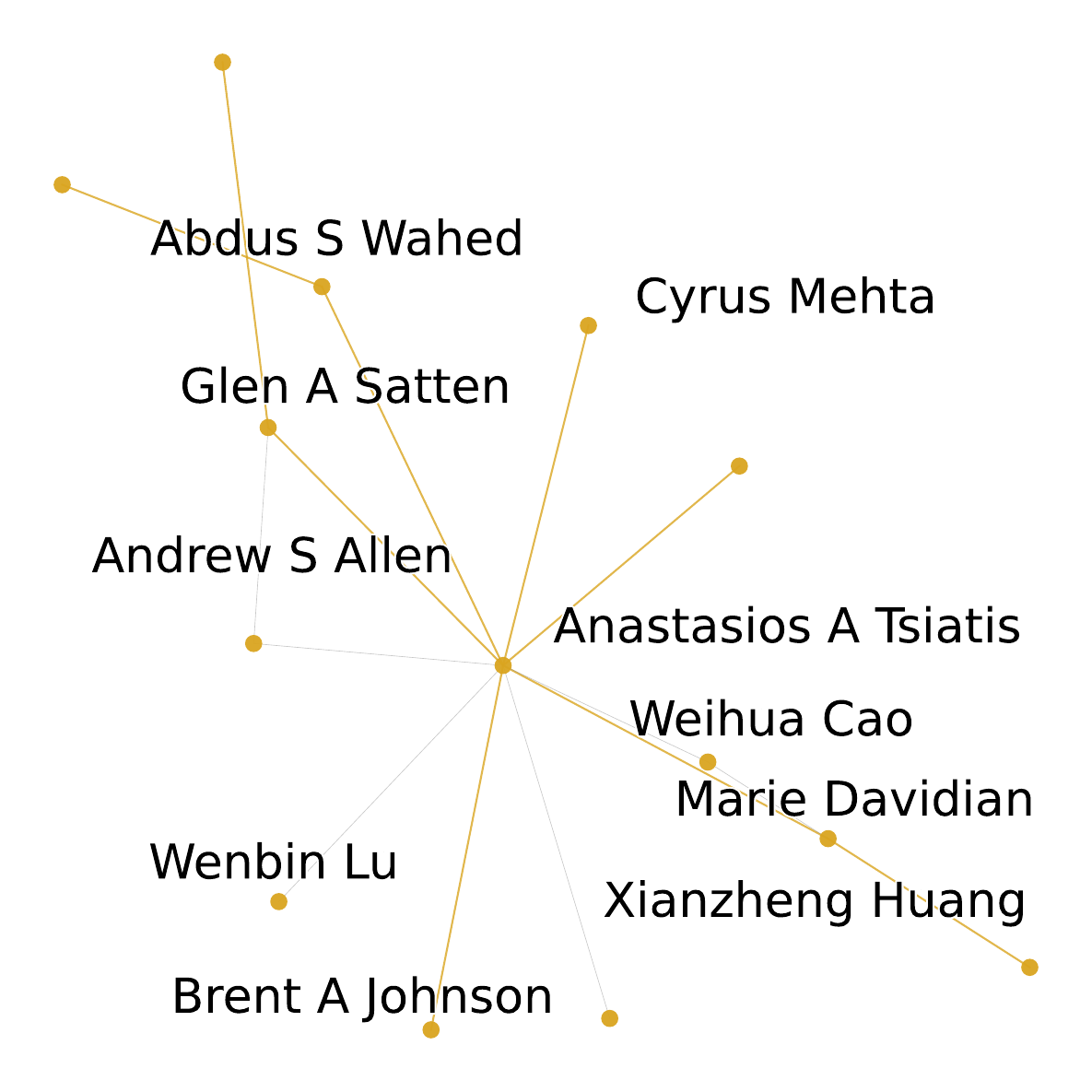}} % Add your image here
        \subcaption{Community 18 \\ (NCSU)}
    \end{subfigure}
    \hfill
    % Second subfigure
    \begin{subfigure}[b]{0.3\linewidth}
        \centering
        {\includegraphics[width=\linewidth]{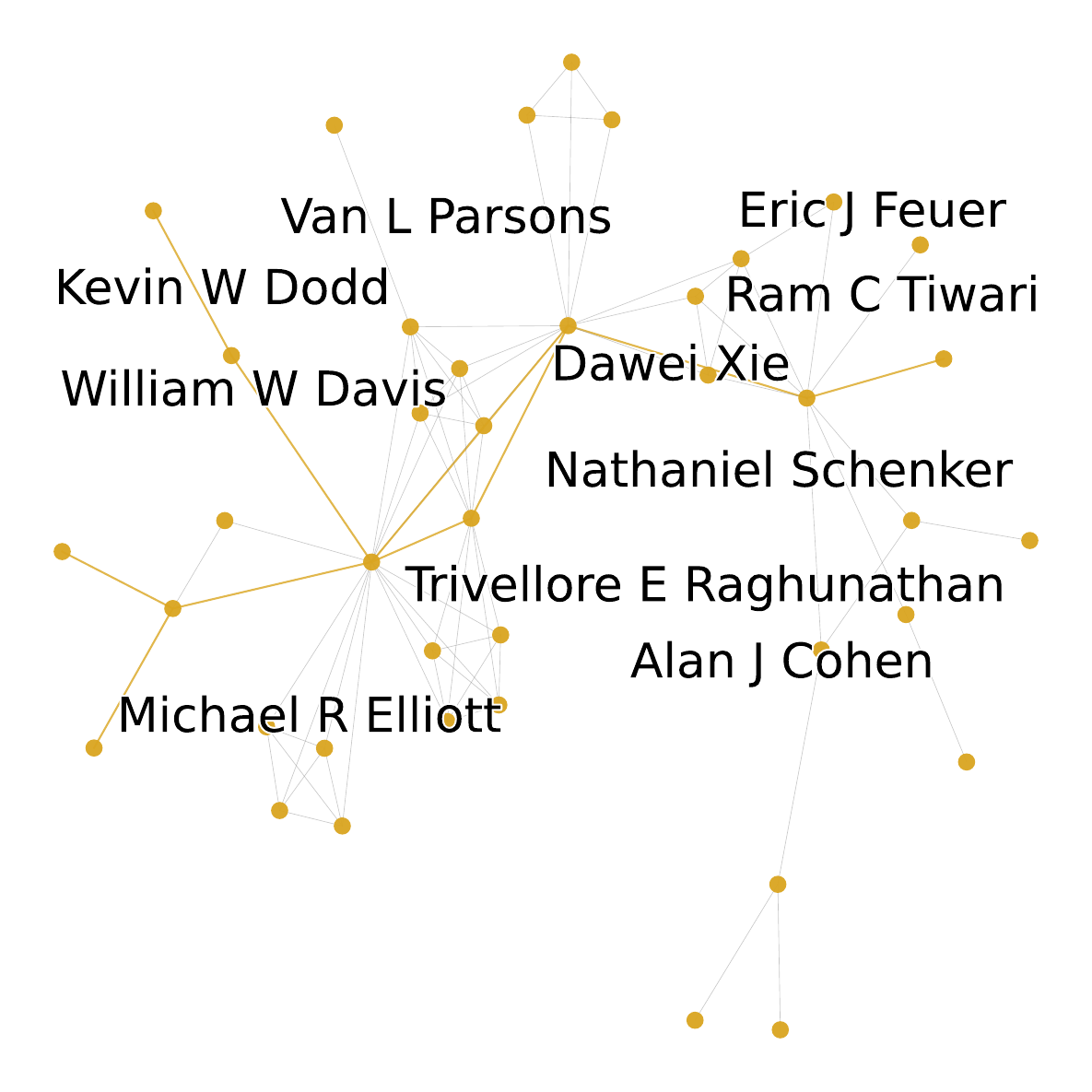}} % Add your image here
        \subcaption{Community 19 \\ (Biostatistics in government)}
    \end{subfigure}
    \hfill
    % Second subfigure
    \begin{subfigure}[b]{0.3\linewidth}
        \centering
        {\includegraphics[width=\linewidth]{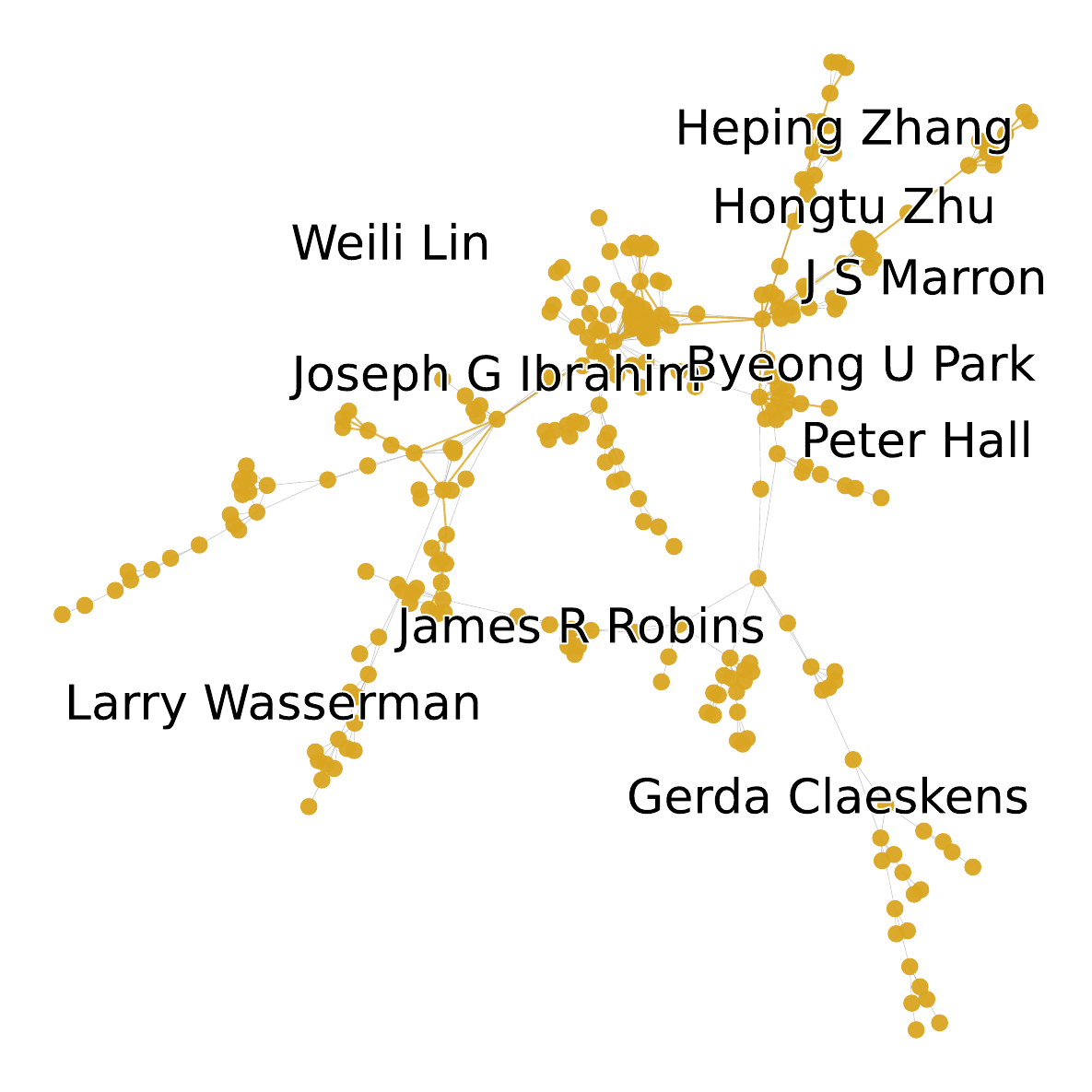}} % Add your image here
        \subcaption{Community 20}
    \end{subfigure}
    \caption{Communities obtained by distance recovery in the hierarchical procedure (Part II)}
    \label{fig: h clustering II}
\end{figure}

\newpage
\paragraph{Reproducibility.}
The code used to generate all simulation results and real-data analyses in this paper is publicly available at \texttt{https://github.com/JasonSmilingKnight/Forest-Community-Detection}. The repository also contains instructions and scripts for reproducing all numerical experiments and figures reported in the paper.

\section{Discussion}
This paper studies the planted forest model for networks formed by the growth of $K$ preferential attachment trees embedded in Erd\"{o}s--R\'{e}nyi noise and studies community detection on these randomly growing networks.
We propose a two-stage algorithm, SPAR, and establish local recovery guarantees with respect to several subsets of nodes, including early-arriving, high-degree, and layer-1 and layer-2 nodes, i.e. nodes within 1 or 2 distance away from the root nodes respectively. To the best of our knowledge, this is the first work on community detection with a rigorous theoretical guarantee for this class of randomly growing networks.

Several important questions remain open.
From a methodological perspective, an important challenge is estimating the community counts $K$, which is mostly assumed to be known in this paper. Although Remark~\ref{rem: K estimation 2} provides a consistent estimator, it requires careful tuning that is difficult to perform in practice. 
Developing a more stable estimator of $K$ would improve the practical applicability of the proposed method.
Another direction is to explore alternative pruning criteria in the first stage of SPAR, as mentioned in Remark~\ref{remark: different pruning criterion}.
Recent work on learning arrival-order-related statistics using graph neural networks \citep{xin2025toping,li2021propagation} may provide a promising avenue for improving the empirical performance of community detection within the framework of SPAR.  

From a theoretical perspective, an alternative to studying local community recovery is to replace the current definition of misclustering error with a criterion that assigns node-specific weights according to their structural importance in the network, such as arriving order. This may allow us to recover global consistency in community estimation, albeit with an alternative error metric. Another interesting problem is to allow for very small communities, e.g. communities whose size is $o(n)$. Recovering small communities may require modifying the error metric as well, possibly weighing the error on each community by the inverse of the community size. \\

\noindent \textbf{Acknowledgement:} This work is supported by the United States National Science Foundation grants DMS-2113671 and DMS-2311299, as well as the United States National Institute of Health grant 1R01GM157610-01.

\clearpage
\bibliographystyle{abbrvnat}
\bibliography{reference}      % Bibliography 

@article{wang2017likelihood,
  title={LIKELIHOOD-BASED MODEL SELECTION FOR STOCHASTIC BLOCK MODELS1},
  author={Wang, YX Rachel and Bickel, Peter J},
  journal={The Annals of Statistics},
  volume={45},
  number={2},
  pages={500--528},
  year={2017}
}

@article{jin2023optimal,
  title={Optimal estimation of the number of network communities},
  author={Jin, Jiashun and Ke, Zheng Tracy and Luo, Shengming and Wang, Minzhe},
  journal={Journal of the American Statistical Association},
  volume={118},
  number={543},
  pages={2101--2116},
  year={2023},
  publisher={Taylor \& Francis}
}

@article{ma2021determining,
  title={Determining the number of communities in degree-corrected stochastic block models},
  author={Ma, Shujie and Su, Liangjun and Zhang, Yichong},
  journal={Journal of machine learning research},
  volume={22},
  number={69},
  pages={1--63},
  year={2021}
}

@article{le2022estimating,
  title={Estimating the number of communities by spectral methods},
  author={Le, Can M and Levina, Elizaveta},
  journal={Electronic Journal of Statistics},
  volume={16},
  number={1},
  pages={3315--3342},
  year={2022},
  publisher={The Institute of Mathematical Statistics and the Bernoulli Society}
}

@article{kruskal1956shortest,
  title={On the shortest spanning subtree of a graph and the traveling salesman problem},
  author={Kruskal, Joseph B},
  journal={Proceedings of the American Mathematical society},
  volume={7},
  number={1},
  pages={48--50},
  year={1956},
  publisher={JSTOR}
}

@article{contat2024eve,
  title={Eve, Adam and the preferential attachment tree},
  author={Contat, Alice and Curien, Nicolas and Lacroix, Perrine and Lasalle, Etienne and Rivoirard, Vincent},
  journal={Probability Theory and Related Fields},
  volume={190},
  number={1},
  pages={321--336},
  year={2024},
  publisher={Springer}
}

@article{bubeck2017finding,
  title={Finding {A}dam in random growing trees},
  author={Bubeck, S{\'e}bastien and Devroye, Luc and Lugosi, G{\'a}bor},
  journal={Random Structures \& Algorithms},
  volume={50},
  number={2},
  pages={158--172},
  year={2017},
  publisher={Wiley Online Library}
}

@article{wang2016botnet,
  title={Botnet detection based on anomaly and community detection},
  author={Wang, Jing and Paschalidis, Ioannis Ch},
  journal={IEEE Transactions on Control of Network Systems},
  volume={4},
  number={2},
  pages={392--404},
  year={2016},
  publisher={IEEE}
}

@article{velickovic2018deep,
  title={Deep graph infomax},
  author={Velickovic, Petar and Fedus, William and Hamilton, William L and Li{\`o}, Pietro and Bengio, Yoshua and Hjelm, R Devon},
  journal={stat},
  volume={1050},
  pages={21},
  year={2018}
}

@article{de2020unveiling,
  title={Unveiling the hierarchical structure of music by multi-resolution community detection},
  author={de Berardinis, Jacopo and Vamvakaris, Michalis and Cangelosi, Angelo and Coutinho, Eduardo},
  journal={Transactions of the International Society for Music Information Retrieval},
  volume={3},
  number={1},
  pages={82--97},
  year={2020},
  publisher={Ubiquity Press, Ltd.}
}

@article{yu2007graph,
  title={Graph-based consensus clustering for class discovery from gene expression data},
  author={Yu, Zhiwen and Wong, Hau-San and Wang, Hongqiang},
  journal={Bioinformatics},
  volume={23},
  number={21},
  pages={2888--2896},
  year={2007},
  publisher={Oxford University Press}
}

@article{akbaritabar2020italian,
  title={Italian sociologists: A community of disconnected groups},
  author={Akbaritabar, Aliakbar and Traag, Vincent Antonio and Caimo, Alberto and Squazzoni, Flaminio},
  journal={Scientometrics},
  volume={124},
  number={3},
  pages={2361--2382},
  year={2020},
  publisher={Springer}
}

@article{traag2019louvain,
  title={From Louvain to Leiden: guaranteeing well-connected communities},
  author={Traag, Vincent A and Waltman, Ludo and Van Eck, Nees Jan},
  journal={Scientific reports},
  volume={9},
  number={1},
  pages={5233},
  year={2019},
  publisher={Nature Publishing Group UK London}
}

@article{krzakala2013spectral,
  title={Spectral redemption in clustering sparse networks},
  author={Krzakala, Florent and Moore, Cristopher and Mossel, Elchanan and Neeman, Joe and Sly, Allan and Zdeborov{\'a}, Lenka and Zhang, Pan},
  journal={Proceedings of the National Academy of Sciences},
  volume={110},
  number={52},
  pages={20935--20940},
  year={2013},
  publisher={National Academy of Sciences}
}

@article{shah2011rumors,
  title={Rumors in a network: Who's the culprit?},
  author={Shah, Devavrat and Zaman, Tauhid},
  journal={IEEE Transactions on information theory},
  volume={57},
  number={8},
  pages={5163--5181},
  year={2011},
  publisher={IEEE}
}

@book{durrett2019probability,
  title={Probability: theory and examples},
  author={Durrett, Rick},
  volume={49},
  year={2019},
  publisher={Cambridge university press}
}

@article{dong2022wavefront,
  title={Wavefront-based multiple rumor sources identification by multi-task learning},
  author={Dong, Ming and Zheng, Bolong and Li, Guohui and Li, Chenliang and Zheng, Kai and Zhou, Xiaofang},
  journal={IEEE Transactions on Emerging Topics in Computational Intelligence},
  volume={6},
  number={5},
  pages={1068--1078},
  year={2022},
  publisher={IEEE}
}

@article{li2021propagation,
  title={Propagation source identification of infectious diseases with graph convolutional networks},
  author={Li, Liang and Zhou, Jianye and Jiang, Yuewen and Huang, Biqing},
  journal={Journal of biomedical informatics},
  volume={116},
  pages={103720},
  year={2021},
  publisher={Elsevier}
}

@article{ji2022co,
  title={Co-citation and co-authorship networks of statisticians},
  author={Ji, Pengsheng and Jin, Jiashun and Ke, Zheng Tracy and Li, Wanshan},
  journal={Journal of Business \& Economic Statistics},
  volume={40},
  number={2},
  pages={469--485},
  year={2022},
  publisher={Taylor \& Francis}
}

@book{van2024random,
  title={Random graphs and complex networks},
  author={Van Der Hofstad, Remco},
  volume={2},
  year={2024},
  publisher={Cambridge university press}
}

@article{banerjee2023degree,
  title={Degree centrality and root finding in growing random networks},
  author={Banerjee, Sayan and Huang, Xiangying},
  journal={Electronic Journal of Probability},
  volume={28},
  pages={1--39},
  year={2023},
  publisher={The Institute of Mathematical Statistics and the Bernoulli Society}
}

@article{banerjee2022root,
  title={Root finding algorithms and persistence of Jordan centrality in growing random trees},
  author={Banerjee, Sayan and Bhamidi, Shankar},
  journal={The Annals of Applied Probability},
  volume={32},
  number={3},
  pages={2180--2210},
  year={2022},
  publisher={Institute of Mathematical Statistics}
}

@article{barabasi1999emergence,
  title={Emergence of scaling in random networks},
  author={Barab{\'a}si, Albert-L{\'a}szl{\'o} and Albert, R{\'e}ka},
  journal={science},
  volume={286},
  number={5439},
  pages={509--512},
  year={1999},
  publisher={American Association for the Advancement of Science}
}

@article{li2022hierarchical,
  title={Hierarchical community detection by recursive partitioning},
  author={Li, Tianxi and Lei, Lihua and Bhattacharyya, Sharmodeep and Van den Berge, Koen and Sarkar, Purnamrita and Bickel, Peter J and Levina, Elizaveta},
  journal={Journal of the American Statistical Association},
  volume={117},
  number={538},
  pages={951--968},
  year={2022},
  publisher={Taylor \& Francis}
}

@article{deng2024distributed,
  title={Distributed Pseudo-Likelihood Method for Community Detection in Large-Scale Networks},
  author={Deng, Jiayi and Huang, Danyang and Zhang, Bo},
  journal={ACM Transactions on Knowledge Discovery from Data},
  volume={18},
  number={7},
  pages={1--25},
  year={2024},
  publisher={ACM New York, NY}
}

@article{cline2007integration,
  title={Integration of biological networks and gene expression data using Cytoscape},
  author={Cline, Melissa S and Smoot, Michael and Cerami, Ethan and Kuchinsky, Allan and Landys, Nerius and Workman, Chris and Christmas, Rowan and Avila-Campilo, Iliana and Creech, Michael and Gross, Benjamin and others},
  journal={Nature protocols},
  volume={2},
  number={10},
  pages={2366--2382},
  year={2007},
  publisher={Nature Publishing Group UK London}
}

@article{ben2025inference,
  title={Inference in balanced community modulated recursive trees},
  author={Ben-Hamou, Anna and Velona, Vasiliki},
  journal={Bernoulli},
  volume={31},
  number={1},
  pages={457--483},
  year={2025},
  publisher={Bernoulli Society for Mathematical Statistics and Probability}
}

@article{hajek2019community,
  title={Community recovery in a preferential attachment graph},
  author={Hajek, Bruce and Sankagiri, Suryanarayana},
  journal={IEEE Transactions on Information Theory},
  volume={65},
  number={11},
  pages={6853--6874},
  year={2019},
  publisher={IEEE}
}

@article{cerqueira2023pseudo,
  title={A pseudo-likelihood approach to community detection in weighted networks},
  author={Cerqueira, Andressa and Levina, Elizaveta},
  journal={arXiv preprint arXiv:2303.05909},
  year={2023}
}

@article{amini2024hierarchical,
  title={Hierarchical stochastic block model for community detection in multiplex networks},
  author={Amini, Arash and Paez, Marina and Lin, Lizhen},
  journal={Bayesian Analysis},
  volume={19},
  number={1},
  pages={319--345},
  year={2024},
  publisher={International Society for Bayesian Analysis}
}

@article{zhen2023community,
  title={Community detection in general hypergraph via graph embedding},
  author={Zhen, Yaoming and Wang, Junhui},
  journal={Journal of the American Statistical Association},
  volume={118},
  number={543},
  pages={1620--1629},
  year={2023},
  publisher={Taylor \& Francis}
}

@article{xin2025toping,
  title={TopInG: Topologically Interpretable Graph Learning via Persistent Rationale Filtration},
  author={Xin, Cheng and Xu, Fan and Ding, Xin and Gao, Jie and Ding, Jiaxin},
  journal={arXiv preprint arXiv:2510.05102},
  year={2025}
}

@article{yanchenko2025statistical,
  title={Statistical inference for core-periphery structures},
  author={Yanchenko, Eric and Sengupta, Srijan and Mukherjee, Diganta},
  journal={arXiv preprint arXiv:2508.04730},
  year={2025}
}

@article{zhang2015identification,
  title={Identification of core-periphery structure in networks},
  author={Zhang, Xiao and Martin, Travis and Newman, Mark EJ},
  journal={Physical Review E},
  volume={91},
  number={3},
  pages={032803},
  year={2015},
  publisher={APS}
}

@article{naik2021sparse,
  title={Sparse networks with core-periphery structure},
  author={Naik, Cian and Caron, Fran{\c{c}}ois and Rousseau, Judith},
  journal={Electronic Journal of Statistics},
  volume={15},
  pages={1814--1868},
  year={2021}
}

@article{miao2023informative,
  title={Informative core identification in complex networks},
  author={Miao, Ruizhong and Li, Tianxi},
  journal={Journal of the Royal Statistical Society Series B: Statistical Methodology},
  volume={85},
  number={1},
  pages={108--126},
  year={2023},
  publisher={Oxford University Press US}
}

@article{senizergues2021geometry,
  title={Geometry of weighted recursive and affine preferential attachment trees},
  author={S{\'e}nizergues, Delphin},
  journal={Electronic Journal of Probability},
  volume={26},
  pages={1--56},
  year={2021},
  publisher={The Institute of Mathematical Statistics and the Bernoulli Society}
}

@article{rudas2007random,
  title={Random trees and general branching processes},
  author={Rudas, Anna and T{\'o}th, B{\'a}lint and Valk{\'o}, Benedek},
  journal={Random Structures \& Algorithms},
  volume={31},
  number={2},
  pages={186--202},
  year={2007},
  publisher={Wiley Online Library}
}

@article{crane2021inference,
  title={Inference on the history of a randomly growing tree},
  author={Crane, Harry and Xu, Min},
  journal={Journal of the Royal Statistical Society Series B: Statistical Methodology},
  volume={83},
  number={4},
  pages={639--668},
  year={2021},
  publisher={Oxford University Press}
}

@article{pekoz2017joint,
  title={Joint degree distributions of preferential attachment random graphs},
  author={Pek{\"o}z, Erol and R{\"o}llin, Adrian and Ross, Nathan},
  journal={Advances in Applied Probability},
  volume={49},
  number={2},
  pages={368--387},
  year={2017},
  publisher={Cambridge University Press}
}

@article{crane2024root,
  title={Root and community inference on the latent growth process of a network},
  author={Crane, Harry and Xu, Min},
  journal={Journal of the Royal Statistical Society Series B: Statistical Methodology},
  volume={86},
  number={4},
  pages={825--865},
  year={2024},
  publisher={Oxford University Press US}
}

@article{ji2016coauthorship,
  title={COAUTHORSHIP AND CITATION NETWORKS FOR STATISTICIANS},
  author={Ji, Pengsheng and Jin, Jiashun},
  journal={The Annals of Applied Statistics},
  pages={1779--1812},
  year={2016},
  publisher={JSTOR}
}

@article{xu2020optimal,
  title={Optimal rates for community estimation in the weighted stochastic block model},
  author={Xu, Min and Jog, Varun and Loh, Po-Ling},
  journal={The Annals of Statistics},
  volume={48},
  number={1},
  pages={183--204},
  year={2020},
  publisher={JSTOR}
}

@inproceedings{dawkins2021diffusion,
  title={Diffusion source identification on networks with statistical confidence},
  author={Dawkins, Quinlan E and Li, Tianxi and Xu, Haifeng},
  booktitle={International Conference on Machine Learning},
  pages={2500--2509},
  year={2021},
  organization={PMLR}
}

@article{addario2025leaf,
  title={Leaf stripping on uniform attachment trees},
  author={Addario-Berry, Louigi and Brandenberger, Anna and Briend, Simon and Broutin, Nicolas and Lugosi, G{\'a}bor},
  journal={Random Structures \& Algorithms},
  volume={67},
  number={1},
  pages={e70023},
  year={2025},
  publisher={Wiley Online Library}
}

\newpage

\clearpage

\setcounter{section}{0}
\setcounter{equation}{0}
\def\theequation{S\arabic{section}.\arabic{equation}}
\def\thesection{S\arabic{section}}
\def\thetheorem{S\arabic{theorem}}

\renewcommand\thefigure{S\arabic{figure}} % Changing the figure numbering
\setcounter{figure}{0} % Resetting the figure counter

\renewcommand\thetable{S\arabic{table}} % Changing the table numbering
\setcounter{table}{0} % Resetting the table counter

\begingroup
\centering
{\LARGE\bf Supplementary Material for Community Detection on a Randomly Growing Network\par}
\par
\endgroup

\vspace{1em}

\label{sec: appendix for algorithm}

\section{Proofs for Section~\ref{sec: model}}
\label{sec: appendix model}
Before proving Theorem~\ref{thm: impossible}, we introduce some notation to simplify the proof process. 

Let $\boldsymbol{G}_{n}:=\bm{F}_n+\bm{R}_n$ generated from $\mathrm{PF}(\alpha, \theta, \pi, 2)$.
Define the parent function $\operatorname{pa}: V(\boldsymbol{F}_{n}) \to V(\boldsymbol{F}_{n})$ such that $\operatorname{pa}(v) = u$ indicates that node $v$ is attached to node $u$ during the generation process of $\boldsymbol{F}_{n}$, and we denote $\operatorname{pa}\left(\pi_1\right):=\pi_1$ as a convention.

For $t \in [n]$, we let $(e_1, \ldots, e_{h_t})$ be the set of noise edges incident between the new node $\pi_t$ and previous nodes $\pi_{1:(t-1)}$; note that the number of such edges $h_t$ is also random. Conditional on $h_t$, we define a random node $q(\pi_t)$ where we choose an edge $e'$ in $(e_1, \ldots, e_{h_t})$ uniformly at random and let $q(\pi_t)$ be the endpoint of $e'$ that is in $\pi_{1:(t-1)}$. If $h_t = 0$, then we let $q(\pi_t) = \emptyset$. 

%We further assume that noise edges arrive sequentially with each new node. 
%Specifically, when node $\pi_t$ arrives, we first generate the total number of noise edges between $\pi_t$ and the existing node set $\pi_{1:\left(t-1\right)}$. 
%These noise edges are then attached one by one between $\pi_t$ and nodes in $\pi_{1:\left(t-1\right)}$.

%Let the $s$-th noise-edge parent of node $\pi_t$ be denoted by $\operatorname{pa-N-s}(\pi_t)$.
%If the total number of noise edges incident to $\pi_t$ is $h_t$, we denote the full noise-parent sequence by
%$\operatorname{pa-N}(\pi_t):=\left(\operatorname{pa-N-1}(\pi_t),\ldots,\operatorname{pa-N-h_t}(\pi_t)\right)$.

For an integer $t \in [n]$, denote by $\boldsymbol{T}^1_t$ and $\boldsymbol{T}^2_t$ the two tree components of the forest $\boldsymbol{F}_{n}$ at time $t$.
With $v_1, v_2 \in \pi_{1:(t-1)}$, 
we define the ``balance attachment" event $\mathcal{A}_t$ as follows:
\begin{align}
\nonumber
\mathcal{A}_t\bigl(\left\{v_1, v_2\right\}\bigr)& := \bigl\{ v_1 \in V(\bm{T}_{t-1}^1), v_2 \in V(\bm{T}_{t-1}^2),  \left\{\operatorname{pa}(\pi_t), q(\pi_t)\right\}=\{v_1, v_2\}, \\
\nonumber
&\quad \quad\operatorname{pa}(\pi_s) \notin \{v_1, v_2\},\,\, \forall s \neq t \bigr\} \\
\label{eqn: balance attachment}
\mathcal{A}_t &:= \cup_{\left\{v_1, v_2\right\} \in \pi_{1:(t-1)}} \mathcal{A}_t\bigl(\left\{v_1, v_2\right\}\bigr). 
\end{align}
The event $\mathcal{A}_t$ collects all growth histories for which the $t$-th node encounters an ambiguous attachment: the parent and a noise endpoint of $\pi_t$ are two candidate nodes $v_1$ and $v_2$ from the two trees, respectively. Under this event, either node could serve as the parent of $\pi_t$, and neither $v_1$ nor $v_2$ receives any additional child nodes in the growth process. We note that if $\{v_1, v_2\} \neq \{v'_1, v'_2\}$, then $\mathcal{A}_t(\{v_1, v_2\})$ and $\mathcal{A}_t(\{v_1', v_2'\})$ are disjoint. 

\begin{proof}[Proof of Theorem~\ref{thm: impossible}]

Let $\bar\sigma\in \operatorname{argmin}_{\sigma \in S_2} n^{-1}\sum_{t=1}^{\lceil n/2\rceil}\mathbbm{1}\left\{\hat{\ell}\left(\pi_t\right)\neq \sigma\circ \ell\left(\pi_t\right)\right\}$.
By Definition~\ref{defin mismatch}, we have
\begin{equation}
\label{thm impossible: eqn1}
\begin{aligned}
\mathbb{E}\biggl[\frac{d\bigl(\hat{\ell},\ell\bigr)}{n}\biggr]
&=
\mathbb{E}\biggl[
\min_{\sigma \in S_2}
n^{-1}\sum_{u\in V\bigl(\bm{G}_n\bigr)}
\mathbbm{1}\Bigl\{\hat{\ell}\bigl(u\bigr)\neq \sigma\circ \ell\bigl(u\bigr)\Bigr\}
\biggr] \\
&=
\mathbb{E}\biggl[
\min_{\sigma \in S_2}
n^{-1}\sum_{u\in \pi_{1:n}}
\mathbbm{1}\Bigl\{\hat{\ell}\bigl(u\bigr)\neq \sigma\circ \ell\bigl(u\bigr)\Bigr\}
\biggr] \\
% &\stackrel{(a)}{=}
% \mathbb{E}\biggl[
% \mathbb{E}\biggl[
% \min_{\sigma \in S_2}
% n^{-1}\sum_{t=1}^n
% \mathbbm{1}\Bigl\{\hat{\ell}\bigl(\pi_t\bigr)\neq \sigma\circ\ell\bigl(\pi_t\bigr)\Bigr\}
% \,\bigg|\,
% \ell\bigl(\pi_{1:\lceil n/2\rceil}\bigr)
% \biggr]
% \biggr] \\
% &\stackrel{(b)}{\geq}
% \mathbb{E}\biggl[
% \mathbb{E}\biggl[
% \min\Biggl\{
% n^{-1}\sum_{t=\lceil n/2\rceil+1}^n
% \mathbbm{1}\Bigl\{\hat{\ell}\bigl(\pi_t\bigr)\neq \bar{\sigma}\circ\ell\bigl(\pi_t\bigr)\Bigr\},
% \frac{1}{4}
% \Biggr\}
% \,\bigg|\,
% \ell\bigl(\pi_{1:\lceil n/2\rceil}\bigr)
% \biggr]
% \biggr] \\
&\stackrel{(a)}{\geq}
\mathbb{E}\biggl[
\min\Biggl\{
n^{-1}\sum_{t=\lceil n/2\rceil+1}^n
\mathbbm{1}\Bigl\{\hat{\ell}\bigl(\pi_t\bigr)\neq \bar{\sigma}\circ\ell\bigl(\pi_t\bigr)\Bigr\},
\frac{1}{4}
\Biggr\}
\biggr] \\
&\stackrel{(b)}{\geq}
\frac{1}{4}
\mathbb{E}\biggl[
n^{-1}\sum_{t=\lceil n/2\rceil+1}^n
\mathbbm{1}\Bigl\{\hat{\ell}\bigl(\pi_t\bigr)\neq \bar{\sigma}\circ\ell\bigl(\pi_t\bigr)\Bigr\}
\biggr]
\min\Biggl\{
\mathbb{E}\biggl[
n^{-1}\sum_{t=\lceil n/2\rceil+1}^n
\mathbbm{1}\Bigl\{\hat{\ell}\bigl(\pi_t\bigr)\neq \bar{\sigma}\circ\ell\bigl(\pi_t\bigr)\Bigr\}
\biggr],
\frac{1}{2}
\Biggr\}\\
&\stackrel{(c)}{\geq}
\frac{1}{12}
\min_{t > \lceil n/2 \rceil}\mathbb{P}\left(\hat{\ell}\left(\pi_{t}\right)\neq \bar{\sigma}\circ \ell\left(\pi_{t}\right)\right)
\min\Biggl\{\frac{1}{3}
\min_{t > \lceil n/2 \rceil}\mathbb{P}\left(\hat{\ell}\left(\pi_{t}\right)\neq \bar{\sigma}\circ \ell\left(\pi_{t}\right)\right),
\frac{1}{2}
\Biggr\}\\
&\stackrel{(d)}{\geq} \frac{1}{36}
\Bigl(\min_{t > \lceil n/2 \rceil}\mathbb{P}\left(\hat{\ell}\left(\pi_{t}\right)\neq \bar{\sigma}\circ \ell\left(\pi_{t}\right)\right)\Bigr)^2.
\end{aligned}
\end{equation}
where inequality (a) follows from Lemma~\ref{lemma: split alignment}, inequality (b) follows from Lemma~\ref{lemma: truncated lower bound}, inequality (c) follows from the fact that the integer set
$[\lceil n/2\rceil+1,\, n]$ contains more than $n/3$ integers when $n \geq 4$, and
inequality (d) follows from the fact that $\frac{1}{3} \mathbb{P}\left(\hat{\ell}\left(\pi_{t}\right)\neq \bar{\sigma}\circ \ell\left(\pi_{t}\right)\right)\le \frac{1}{3}<\frac{1}{2}$.

% where the identities (a) and (c) follow from the tower property of conditional expectation.
% Step (b) and Step (d) follow from Lemma~\ref{lemma: split alignment} and Lemma~\ref{lemma: truncated lower bound}, respectively.
% Step (e) follows from the fact that the integer set
% $[\lceil n/2\rceil+1,\, n]$ contains more than $n/3$ integers when $n \geq 4$.
% Step (f) follows from the fact that $\frac{1}{3} \mathbb{P}\left(\hat{\ell}\left(\pi_{t}\right)\neq \bar{\sigma}\circ \ell\left(\pi_{t}\right)\right)\le \frac{1}{3}<\frac{1}{2}$.
% }

For a fixed integer $t$ satisfying $t \ge \lceil n/2\rceil$,
let the $\lambda$ be chosen as Lemma~\ref{lemma balance attachment set bound}. 
Recall the definition of a ``balance attachment" event $\mathcal{A}_t$ in \eqref{eqn: balance attachment}. Now, we lower bound $\mathbb{P}\left(\hat{\ell}\left(\pi_{t}\right)\neq 
\bar{\sigma}\circ l\left(\pi_{t}\right)\right)$ by restricting ourselves to the event $\mathcal{A}_t$:
\begin{equation}
\label{thm: impossible eqn2}
\begin{aligned}
&\mathbb{P}(\hat{\ell}(\pi_t) \neq \bar{\sigma} \circ\ell(\pi_t)) \geq  \mathbb{P}\left(\left\{\hat{\ell}(\pi_t) \neq \bar{\sigma} \circ \ell(\pi_t) \right\}\cap \mathcal{A}_t\right) \\
& \stackrel{(a)}{=} 
\sum_{\{v_1, v_2\} \in \pi_{1:(t-1)}} \mathbb{P}\left(\left\{\hat{\ell}(\pi_t) \neq \bar{\sigma} \circ \ell(\pi_t) \right\}\cap \mathcal{A}_t\left(\{v_1, v_2\}\right)\right) \\
&\stackrel{(b)}{=} \sum_{\{v_1, v_2\} \in \pi_{1:(t-1)}} \mathbb{E}\biggl[ \mathbb{E}\biggl(  \mathbbm{1}\{\hat{\ell}(\pi_t) \neq \bar{\sigma} \circ\ell(\pi_t) \} \mathbbm{1}_{\mathcal{A}_t\left( \{v_1, v_2\}\right)} \,\bigg|  \left\{\operatorname{pa}(\pi_t), q(\pi_t)\right\}=\{v_1, v_2\},\, \{\text{pa}(\pi_s)\}_{s \in [n]\backslash{t}}, \bm{G}_n \biggr)  \biggr]\\
& \stackrel{(c)}{=} \sum_{\{v_1, v_2\} \in \pi_{1:(t-1)}} \mathbb{E}\biggl[ \mathbb{P}\biggl(  \hat{\ell}(\pi_t) \neq \bar{\sigma} \circ\ell(\pi_t) \,\bigg|  \left\{\operatorname{pa}(\pi_t), q(\pi_t)\right\}=\{v_1, v_2\},\, \{\text{pa}(\pi_s)\}_{s \in [n]\backslash{t}}, \bm{G}_n \biggr) \mathbbm{1}_{\mathcal{A}_t\left( \{v_1, v_2\}\right)} \biggr] \\
& \stackrel{(d)}{\geq } \sum_{\{v_1, v_2\} \in \pi_{1:(t-1)}} \mathbb{E}\biggl[ \min_{i \in \{1, 2\}} \mathbb{P}\biggl(\text{pa}(\pi_t) = v_i \,\bigg| \left\{\operatorname{pa}(\pi_t), q(\pi_t)\right\}=\{v_1, v_2\},\, \{\text{pa}(\pi_s)\}_{s \in [n]\backslash{t}}, \bm{G}_n \biggr) \mathbbm{1}_{\mathcal{A}_t(\{v_1, v_2\})} \biggr] \\
& \stackrel{(e)}{=} \sum_{\{v_1, v_2\} \in \pi_{1:(t-1)}} \frac{1}{2} \mathbb{E}\bigl[\mathbbm{1}_{\mathcal{A}_t(\{v_1, v_2\})} \bigr]  \stackrel{(f)}{=} \frac{1}{2} \mathbb{P}( \mathcal{A}_t),
\end{aligned}
\end{equation}
where the identities (a) and (f) follow from the fact that the events $\mathcal{A}_t(\{v_1, v_2\})$ are disjoint for different unordered pairs $\{v_1, v_2\}$. 
Step (b) follows from the tower property of conditional expectation. Step (c) follows because $\mathbbm{1}_{\mathcal{A}_t(\{v_1, v_2\})}$ is fixed after we condition on $\left\{\operatorname{pa}(\pi_t), q(\pi_t)\right\}=\{v_1, v_2\},\, \{\text{pa}(\pi_s)\}_{s \in [n]\backslash{t}}, \bm{G}_n$. 
For (d), the estimator $\hat{\ell}(\pi_t)$ is fixed given the observed graph $\bm{G}_n$, whereas the random true label $l(\pi_t)$ depends only on whether $\operatorname{pa}(\pi_t) = v_1$ or $\operatorname{pa}(\pi_t) = v_2$, 
The random permutation $\bar \sigma$ is also fixed conditional on ${\ell}(\pi_{1: \lceil n/2 \rceil})$ (which is fully determined by $\{\text{pa}(\pi_t)\}_{t=1}^{\lceil n/2 \rceil}$) and on $\bm{G}_n$  by its definition.
For equality (e), we show that the two conditional probabilities corresponding to $v_1$ and $v_2$ are equal on the event $\mathcal{A}_t(\{v_1, v_2\})$.  
By Bayes' rule and the attachment probability in \eqref{defin: attachment rule 2}, we have
\begin{align}
&\mathbb{P}\biggl(\text{pa}(\pi_t) = v_1 \,\bigg| \left\{\operatorname{pa}(\pi_t), q(\pi_t)\right\}=\{v_1, v_2\},\, \{\text{pa}(\pi_s)\}_{s \in [n]\backslash{t}}, \bm{G}_n \biggr)\mathbbm{1}_{\mathcal{A}_t(\{v_1, v_2\})} \label{eq:pa_prob_overall} \\
&=\frac{\mathbb{P}\biggl(\text{pa}(\pi_t) = v_1, \, \left\{\operatorname{pa}(\pi_t), q(\pi_t)\right\}=\{v_1, v_2\},\, \{\text{pa}(\pi_s)\}_{s \in [n]\backslash{t}}, \bm{G}_n \biggr)}{\mathbb{P}\biggl(\left\{\operatorname{pa}(\pi_t), q(\pi_t)\right\}=\{v_1, v_2\},\, \{\text{pa}(\pi_s)\}_{s \in [n]\backslash{t}}, \bm{G}_n \biggr)} \mathbbm{1}_{\mathcal{A}_t(\{v_1, v_2\})} \nonumber  \\
&=\frac{\mathbb{P}\biggl(\text{pa}(\pi_t) = v_1, \, q(\pi_t)=v_2,\, \{\text{pa}(\pi_s)\}_{s<t}\biggr)\mathbb{P}\biggl(\bm{G}_n, \{\text{pa}(\pi_s)\}_{s>t} \,\bigg| \text{pa}(\pi_t) = v_1, \{\text{pa}(\pi_s)\}_{s<t}\biggr)}{\mathbb{P}\biggl(\left\{\operatorname{pa}(\pi_t), q(\pi_t)\right\}=\{v_1, v_2\},\, \{\text{pa}(\pi_s)\}_{s \in [n]\backslash{t}}, \bm{G}_n \biggr)} \mathbbm{1}_{\mathcal{A}_t(\{v_1, v_2\})} \nonumber
\end{align}
We first see that 
\begin{align}
&\mathbb{P}\bigl( \operatorname{pa}(\pi_t) = v_1, q(\pi_t) = v_2, \{ \operatorname{pa}(\pi_s) \}_{s < t} \bigr) \mathbbm{1}_{\mathcal{A}_t(\{v_1, v_2\})}  \nonumber \\
&= \mathbb{P}\bigl( q(\pi_t) = v_2 \,\big|\, \operatorname{pa}(\pi_t) = v_1,  \{ \operatorname{pa}(\pi_s) \}_{s < t} \bigr) 
\mathbb{P}\bigl( \operatorname{pa}(\pi_t) = v_1 \,|\, \{ \operatorname{pa}(\pi_s) \}_{s < t} \bigr) \nonumber \\
&\qquad \qquad \mathbb{P}( \{\operatorname{pa}(\pi_s) \}_{s < t} ) \mathbbm{1}_{\mathcal{A}_t(\{v_1, v_2\})}\label{eq:pa_prob_decomp1}
\end{align}
On the event $\mathcal{A}_t(\{v_1, v_2\})$, we have that $h_t > 0$ so that the first term of~\eqref{eq:pa_prob_decomp1} is equal to $\frac{1}{t-2}$ by symmetry. The second term is equal to $\frac{\alpha}{2  (t-2) + \alpha ( t - 1)}$. The third term does not depend on $v_1, v_2$. Therefore, we have that
\begin{align}
&\mathbb{P}\bigl( \operatorname{pa}(\pi_t) = v_1, q(\pi_t) = v_2, \{ \operatorname{pa}(\pi_s) \}_{s < t} \bigr) \mathbbm{1}_{\mathcal{A}_t(\{v_1, v_2\})} \\
&\qquad \qquad =  
\mathbb{P}\bigl( \operatorname{pa}(\pi_t) = v_2, q(\pi_t) = v_1, \{ \operatorname{pa}(\pi_s) \}_{s < t} \bigr) \mathbbm{1}_{\mathcal{A}_t(\{v_1, v_2\})}.
\end{align}

Next, we observe that 
\begin{align}
& \mathbb{P}\biggl(\bm{G}_n, \{\text{pa}(\pi_s)\}_{s>t} \,\bigg| \text{pa}(\pi_t) = v_1, \{\text{pa}(\pi_s)\}_{s<t}\biggr)\mathbbm{1}_{\mathcal{A}_t(\{v_1, v_2\})} \nonumber \\
& = \mathbb{P}\bigl( \bm{G}_n \,|\, \{\operatorname{pa}(\pi_s)\}_{s \in [n] \backslash \{t\}}, \operatorname{pa}(\pi_t) = v_1 \bigr) \mathbb{P}\bigl( \{\text{pa}(\pi_s)\}_{s>t} \,\big|\, \text{pa}(\pi_t) = v_1, \{\text{pa}(\pi_s)\}_{s<t}\bigr)\mathbbm{1}_{\mathcal{A}_t(\{v_1, v_2\})}
\label{eq:pa_prob_decomp2}
\end{align}

Looking at the first term of~\eqref{eq:pa_prob_decomp2}, we have that
\[
\mathbb{P}\bigl( \bm{G}_n \,|\, \{\operatorname{pa}(\pi_s)\}_{s \in [n] \backslash \{t\}}, \operatorname{pa}(\pi_t) = v_1 \bigr) \mathbbm{1}_{\mathcal{A}_t(\{v_1, v_2\})} = \mathbb{P}\bigl( \bm{G}_n \,|\, \{\operatorname{pa}(\pi_s)\}_{s \in [n] \backslash \{t\}}, \operatorname{pa}(\pi_t) = v_2 \bigr) \mathbbm{1}_{\mathcal{A}_t(\{v_1, v_2\})}
\]
since the probability is determined by the Erd\H{o}s--R\'enyi random edges and since, on the event $\mathcal{A}_t(\{v_1, v_2\})$, both edges $(\pi_t, v_1)$ and $(\pi_t, v_2)$ are present in $\bm{G}_n$. 

Similarly, for the second term of~\eqref{eq:pa_prob_decomp2}, we also have
\begin{align*}
&\mathbb{P}\bigl( \{\text{pa}(\pi_s)\}_{s>t} \,\big|\, \text{pa}(\pi_t) = v_1, \{\text{pa}(\pi_s)\}_{s<t}\bigr)\mathbbm{1}_{\mathcal{A}_t(\{v_1, v_2\})}\\
& \qquad \qquad = \mathbb{P}\bigl( \{\text{pa}(\pi_s)\}_{s>t} \,\big|\, \text{pa}(\pi_t) = v_2, \{\text{pa}(\pi_s)\}_{s<t}\bigr)\mathbbm{1}_{\mathcal{A}_t(\{v_1, v_2\})}
\end{align*}
because, on the event $\mathcal{A}_t(\{v_1, v_2\})$, $\text{pa}(\pi_s) \notin \{v_1, v_2\}$ for all $s > t$. 

Returning to~\eqref{eq:pa_prob_overall}, we thus have that
\begin{equation*}
\begin{aligned}
&\mathbb{P}\biggl(\text{pa}(\pi_t) = v_1 \,\bigg| \left\{\operatorname{pa}(\pi_t), q(\pi_t)\right\}=\{v_1, v_2\},\, \{\text{pa}(\pi_s)\}_{s \in [n]\backslash{t}}, \bm{G}_n \biggr) \mathbbm{1}_{\mathcal{A}_t(\{v_1, v_2\})} \\
&=\mathbb{P}\biggl(\text{pa}(\pi_t) = v_2 \,\bigg| \left\{\operatorname{pa}(\pi_t), q(\pi_t)\right\}=\{v_1, v_2\},\, \{\text{pa}(\pi_s)\}_{s \in [n]\backslash{t}}, \bm{G}_n \biggr) \mathbbm{1}_{\mathcal{A}_t(\{v_1, v_2\})} .    
\end{aligned}
\end{equation*}
Since these two probabilities sum to $\mathbbm{1}_{\mathcal{A}_t(\{v_1, v_2\})}$, each must equal $\frac{1}{2} \mathbbm{1}_{\mathcal{A}_t(\{v_1, v_2\})}$, which proves (e).

%follows from the symmetry of the construction conditional on the information that 
%$\left\{\operatorname{pa}(\pi_t), \operatorname{pa-N-1}(\pi_t)\right\}= \{v_1,v_2\}$.

% {
% \color{red}
% By Lemma~\ref{lemma balance attachment set bound}, we already have $\mathbb{P}\left( \mathcal{A}_t\right)\geq 6\lambda$.
% Together with \eqref{thm impossible: eqn1} and \eqref{thm: impossible eqn2}, we have
% %
% \begin{equation*}
% \mathbb{E}\biggl[\frac{d^{\text{Ham}}\bigl(\hat{\ell},\sigma\circ \ell\bigr)}{n}\biggr]\geq \frac{1}{3}  \min_{t\geq \lceil n/2 \rceil} \mathbb{P}\left(\left\{\hat{\ell}(\pi_t) \neq \sigma \circ \ell(\pi_t) \right\}\cap \mathcal{A}_t\right)\geq \frac{1}{3}\times \frac{1}{2} \times 6\lambda =\lambda.
% \end{equation*}
% %
% }

By Lemma~\ref{lemma balance attachment set bound}, we already have $\mathbb{P}\left( \mathcal{A}_t\right)\geq 6\lambda^{1/2}$.
Together with \eqref{thm impossible: eqn1} and \eqref{thm: impossible eqn2}, we have
\begin{equation*}
\begin{aligned}
\mathbb{E}\biggl[\frac{d\bigl(\hat{\ell},\ell\bigr)}{n}\biggr]&\geq
\frac{1}{36}
\Bigl(\min_{t > \lceil n/2 \rceil}\mathbb{P}\left(\hat{\ell}\left(\pi_{t}\right)\neq \bar{\sigma}\circ \ell\left(\pi_{t}\right)\right)\Bigr)^2\\
&\geq \frac{1}{36}
\Bigl(\min_{t > \lceil n/2 \rceil}\mathbb{P}\left(\left\{\hat{\ell}\left(\pi_{t}\right)\neq \bar{\sigma}\circ \ell\left(\pi_{t}\right)\right\}\cap \mathcal{A}_t\right)\Bigr)^2\\
&\geq \frac{1}{36}\times 36\lambda=\lambda.
\end{aligned}
\end{equation*}

The result follows as desired.
\end{proof}

\begin{lemma}
\label{lemma: split alignment}
Let $(a_1,\ldots,a_n),\, (b_1,\ldots,b_n)\in\{1,2\}^n$, and $\bar{\sigma}
\in
\arg\min_{\sigma\in S_2}
\sum_{t=1}^{\lceil n/2\rceil}
\mathbbm{1}\{a_t\neq\sigma(b_t)\}$.
Then
\begin{equation*}
\min_{\sigma\in S_2}
\frac1n
\sum_{t=1}^n
\mathbbm{1}\{a_t\neq\sigma(b_t)\}
\ge
\min\Biggl\{
\frac1n
\sum_{t=\lceil n/2\rceil+1}^{n}
\mathbbm{1}\{a_t\neq\bar{\sigma}(b_t)\},
\,
\frac14
\Biggr\}.    
\end{equation*}
\end{lemma}

\begin{proof}[Proof of Lemma~\ref{lemma: split alignment}]
Let $m=\lceil n/2\rceil$, $\sigma^*
\in
\arg\min_{\sigma\in S_2}
\sum_{t=1}^{n}
\mathbbm{1}\{a_t\neq\sigma(b_t)\}$.
If $\sigma^*=\bar{\sigma}$, then
\begin{equation*}
\sum_{t=1}^{n}
\mathbbm{1}\{a_t\neq\sigma^*(b_t)\}
\ge
\sum_{t=m+1}^{n}
\mathbbm{1}\{a_t\neq\bar{\sigma}(b_t)\},   
\end{equation*}
and the conclusion of the lemma follows immediately. On the other hand, suppose $\sigma^*\neq\bar{\sigma}$. Since $S_2$ contains only two permutations, we have that $a_t \neq \sigma^*(b_t)$ if and only if $a_t = \bar{\sigma}(b_t)$ so that
\begin{equation*}
\sum_{t=1}^{m}
\mathbbm{1}\{a_t\neq\sigma^*(b_t)\}
=
m-
\sum_{t=1}^{m}
\mathbbm{1}\{a_t\neq\bar{\sigma}(b_t)\}
\ge
\frac m2
\ge
\frac n4,
\end{equation*}
where the first inequality follows from the optimality of $\bar\sigma$. 
Therefore, we have
\begin{equation*}
\frac{1}{n}\sum_{t=1}^{n}
\mathbbm{1}\{a_t\neq\sigma^*(b_t)\}
\ge
\min\Biggl\{
\frac{1}{n}\sum_{t=m+1}^{n}
\mathbbm{1}\{a_t\neq\bar{\sigma}(b_t)\},
\frac 14
\Biggr\}.    
\end{equation*}
The lemma follows as desired.
\end{proof}

\begin{lemma}
\label{lemma: truncated lower bound}
Let $X$ be a random variable taking value on $[0, 1]$. Then
\begin{equation*}
\mathbb{E}\bigl[\min\{X, 1/4\}\bigr]
\ge
\frac{1}{4}\min\bigl\{\mathbb{E}X, 1/2\bigr\}\,
\mathbb{E}X.    
\end{equation*}
\end{lemma}

\begin{proof}[Proof of Lemma~\ref{lemma: truncated lower bound}]
Let $\mu:=\mathbb{E}X$ and set $a:=\mu/2$. Since $0\le X\le 1$,
\begin{equation*}
\mu
=
\mathbb{E}\bigl[X\mathbbm{1}\{X<a\}\bigr]
+
\mathbb{E}\bigl[X\mathbbm{1}\{X\ge a\}\bigr]
\le
a+\mathbb{P}(X\ge a).    
\end{equation*}
Hence $\mathbb{P}(X\ge a)\ge \mu-a=\mu/2$. Therefore,
\begin{equation*}
\mathbb{E}\bigl[\min\{X,1/4\}\bigr]
\ge
\min\{a,1/4\}\mathbb{P}(X\ge a)
\ge
\frac{\mu}{2}\min\left\{\frac{\mu}{2},\frac14\right\}= \frac{1}{4}\mu\min\{\mu,\frac{1}{2}\}.    
\end{equation*}
The desired bound follows.
\end{proof}

\begin{lemma}
\label{lemma balance attachment set bound}
Let $\boldsymbol{G}_{n} \sim \mathrm{PF}(\alpha, \theta, \pi, 2)$ with $\theta \geq \frac{c}{n}$ for some constant $c>0$.
For any fixed $t\in \mathbb{N}$ satisfying $t\ge \lceil n/2\rceil$, let $\mathcal{A}_t$ be defined in~\eqref{eqn: balance attachment}.
Then there exists a constant $\lambda= \lambda(c, \alpha)>0$ such that, for all sufficiently large $n$,
\begin{equation*}
    \mathbb{P}\left(\mathcal{A}_t\right)\geq 6\lambda^{\frac{1}{2}}.
\end{equation*}
\end{lemma}
\begin{proof}[Proof of Lemma~\ref{lemma balance attachment set bound}]
To derive a lower bound for the probability of the event $\mathcal{A}_t$,
we introduce a sequence of nested events $\Omega_3 \subset \Omega_2 \subset \Omega_1$, with $\Omega_3 \subset \mathcal{A}_t$, ensuring that all conditions in $\mathcal{A}_t$ defined in \eqref{eqn: balance attachment} are satisfied.

For each consecutive pair of events, we establish lower bounds for the corresponding conditional probabilities.
By combining these bounds, we obtain 
\begin{eqnarray*}
\mathbb{P}\left(\mathcal{A}_t\right)\geq \mathbb{P}\left(\Omega_3\right)=  \mathbb{P}\left(\Omega_3\mid \Omega_2\right) \mathbb{P}\left(\Omega_2\right)= \mathbb{P}\left(\Omega_3\mid \Omega_2\right) \mathbb{P}\left(\Omega_2\mid \Omega_1\right)
\mathbb{P}\left(\Omega_1\right),
\end{eqnarray*}
Then we bound $\mathbb{P}\left(\mathcal{A}_t\right)$ by bounding $\mathbb{P}\left(\Omega_3\mid \Omega_2\right), \mathbb{P}\left(\Omega_2\mid \Omega_1\right), \mathbb{P}\left(\Omega_1\right)$ separately.
The precise definitions of the events $\Omega_i$, $i\in[3]$, 
and their corresponding lower bounds will be introduced in the rest of the proof.
Let
\begin{equation*}
\Omega_1:= \left\{\text{For each}\  i=1, 2,\bigl|V\bigl(\bm{T}^i_{t-1}\bigr)\bigr|\geq  \frac{n}{5}, 
\bigl|\bigl\{v: \operatorname{deg}_{\bm{T}^i_{t-1}}(v)=1\bigr\}\bigr|\geq \frac{p(1)n}{10}\right\},
\end{equation*}
where $p(1)$ is defined in Lemma~\ref{lemma: degree 1 proportion}.
A standard P\'olya–urn argument implies that, as $t \to\infty$, the joint distribution of $\bigl(\frac{|V(\bm{T}^1_{t-1})|}{t-1}, \frac{|V(\bm{T}^2_{t-1})|}{t-1}\bigr)$ converges to $\operatorname{Dirichlet}(\frac{2+2\alpha}{2+\alpha}, \frac{2+2\alpha}{2+\alpha})$.
Consequently, by Lemma~\ref{lemma: degree 1 proportion}, the proportion of nodes with degree 1 in each tree converges to $p(1)$.
Since we also have $t\geq n/2$, the two convergence statements imply that there exists $\lambda_1>0$ such that $\mathbb{P}\left(\Omega_1\right)>\lambda_1$ for all sufficiently large $n$.

The definition of $\Omega_1$ ensures that there are a large number of degree 1 nodes in $\pi_{1:t-1}$.
Next, we construct an event $\Omega_2 \subset \Omega_1$ 
that captures the constraints of $\mathcal{A}_t$ for $\operatorname{pa}(\pi_{i})$ for $i \in [t]$ and $q(\pi_{t})$:
\begin{equation*}
\begin{aligned}
\Omega_2:=\Omega_1\cap \Bigl\{&\text{there exist}\ v_1\in V\left(\bm{T}^1_{t-1}\right), v_2\in V\left(\bm{T}^2_{t-1}\right), \text{such that for each}\ i< t,\\
& \ \operatorname{pa}\left(\pi_i\right)\notin \left\{v_1, v_2\right\}, \text{and}\ \operatorname{pa}\left(\pi_t\right)=v_1, q\left(\pi_t\right)=v_2
\Bigr\}.
\end{aligned}
\end{equation*}
This requires $\operatorname{pa}\left(\pi_{t}\right)$ and $\operatorname{q}\left(\pi_{t}\right)$ to be degree-1 node in $\bm{F}_{t-1}$.

We now proceed to derive a lower bound for $\mathbb{P}(\Omega_2 \mid \Omega_1)$.
By the generation rule in \eqref{defin: attachment rule 2} and noting that on the event $\Omega_1$ there are at least $\tfrac{p(1)n}{10}$ degree 1 nodes in $\bm{T}^1_{t-1}$,
the probability that $\operatorname{pa}(\pi_t)$ is a degree 1 node in $\bm{T}^1_{t-1}$ is at least $\frac{1+\alpha}{2\left(n-2\right)+n \alpha}\times \frac{p\left(1\right)n}{10}=\frac{p\left(1\right)\left(1+\alpha\right)}{10\left(2+\alpha\right)-40/n}$.

Similarly, since on $\Omega_1$ there are at least $\tfrac{p(1)n}{10}$ degree-1 nodes in $\bm{T}_{t-1}^2$ and $\theta\ge c/n$, the probability that $\pi_t$ is connected by at least one noise edge to a degree-1 node in $\bm{T}_{t-1}^2$ is at least
\begin{equation*}
1-\left(1-\frac{c}{n}\right)^{p(1)n/10}
\ge
1-\exp\left(-\frac{cp(1)}{10}\right),    
\end{equation*}
where the inequality follows from $1-x\le e^{-x}$.
Consequently,
\begin{equation*}
\mathbb{P}\bigl(\Omega_2\mid\Omega_1\bigr)
\ge
\frac{p(1)(1+\alpha)}
{10(2+\alpha)-40/n}
\left(1-\exp\left(-\frac{cp(1)}{10}\right)\right).
\end{equation*}

Then we construct $\Omega_3$ on the event of $\Omega_2$ to capture the restriction of the graph for $\pi_{(t+1):n}$:
\begin{equation*}
\Omega_3:=\Omega_2\cap \left\{\text{for each}\ i> t, \ \operatorname{pa}\left(\pi_i \right)\notin \left\{\operatorname{pa}(\pi_t), q(\pi_t)\right\} \right\},     
\end{equation*}
which represents the event that all nodes arriving after time $t$ 
are not attached to $\operatorname{pa}(\pi_t)$ or $q(\pi_t)$.
Therefore, by the attachment rule in~\eqref{defin: attachment rule 2} and the fact that node $\operatorname{deg}_{\bm{F}_{t}}\left(\operatorname{pa}(\pi_t)\right)=2, \operatorname{deg}_{\bm{F}_{t}}\left(q(\pi_t)\right)=1$ on the event $\Omega_2$,
we have, for all sufficiently large $n$,
\begin{eqnarray*}
\mathbb{P}\left(\Omega_3\mid \Omega_2\right)&=&\prod^{n}_{i=t+1}\left(1-\frac{ 2 + \alpha+1+\alpha}{2(i-2-1) + (i-1)\alpha}\right)=\prod^{n}_{i=t+1}\left(1-\frac{ 3 + 2\alpha}{2(i-3) + (i-1)\alpha}\right)\\
&\geq &
\left(1-\frac{ 3 + 2\alpha}{2(n/2-3) + (n/2-1)\alpha}\right)^{n-t} \geq \left(1-\frac{ 3 + 2\alpha}{n-6 + (n/2-1)\alpha}\right)^{n-3}\\
&=&  \exp\left(\left(n-3\right)\ln\left(1-\frac{ 3 + 2\alpha}{n-6 + (n/2-1)\alpha}\right)\right)\geq  \exp\left(-\frac{\left(12+8\alpha\right)\left(n-3\right)}{2\left(n-6\right) + \left(n-2\right)\alpha}\right)\\
&=&\exp\left(-\frac{\left(12+8\alpha\right)\left(n-3\right)}{\left(2+\alpha\right)\left(n-3\right) + \alpha-6}\right)\geq \exp\left(-\frac{24+16\alpha}{2+\alpha}\right).
\end{eqnarray*}
The first inequality comes from the fact that $t\geq n/2$.
In the third inequality, we apply the bound 
$\ln(1 - x) \ge -2x$ for $0\le x \le \frac{1}{2}$.

By the definition of $\mathcal{A}_t$ in~\eqref{eqn: balance attachment} 
and the construction of the events $\Omega_i$, $i \in [3]$, 
we observe that $\Omega_3 \subset \mathcal{A}_t$, 
since all the conditions required by $\mathcal{A}_t$ are satisfied on $\Omega_3$.  
Combining all the conditional probabilities derived above,
Let
\begin{equation*}
\lambda^{\frac{1}{2}}(c,\alpha)
:=
\frac{\lambda_1 p(1)}
{60}
\left(1-\exp\left(-\frac{cp(1)}{10}\right)\right)
\exp\left(-\frac{24+16\alpha}{2+\alpha}\right).
\end{equation*}
Therefore, for all sufficiently large $n$, we have
\begin{eqnarray*}
\mathbb{P}\left(\mathcal{A}_t\right)&\geq& \mathbb{P}\left(\Omega_3\right)\geq \mathbb{P}\left(\Omega_3\mid \Omega_2\right) \mathbb{P}\left(\Omega_2\mid \Omega_1\right)
\mathbb{P}\left(\Omega_1\right)\\
&\geq&  \frac{\lambda_1 p(1)(1+\alpha)}
{10(2+\alpha)-40/n}
\left(1-\exp\left(-\frac{cp(1)}{10}\right)\right)\exp\left(-\frac{24+16\alpha}{2+\alpha}\right)\geq 6\lambda^{\frac{1}{2}}.    
\end{eqnarray*}
The lemma follows as desired.
\end{proof}

\section{Appendix for Section~\ref{sec: algorithm}}

\subsection{Choosing the Core Degree Threshold}
\label{secsec: Choosing the Core Degree Threshold}
We set the core degree threshold by matching the order of top degrees: 
We simulate the planted forest $T$ times under the known parameters $(n, K,\alpha,\theta)$ and a chosen minimum component proportion $1/H$ ($H\ge K$).
In each run, we record the global degree rank of the weakest hub (i.e., the largest-degree node in each component; take the minimum across components) and take a high quantile (e.g., 95th percentile and above) of this rank, denoted $r$. 
On the observed graph, we sort vertices by degree and set the core degree threshold $\tau$ to the degree of the node at rank $r$.
The procedure is summarized in Algorithm~\ref{alg: choosing the core degree threshold}.

\begin{algorithm}[H]
\caption{Rank-based selection of core degree threshold $\tau$}
\label{alg: choosing the core degree threshold}
\KwIn{Graph $\boldsymbol{G}_n=\left(V,E\right)$; community count $K\in\mathbb{N}$; minimum component proportion $1/H$; rank quantile $q = 0.95$; simulation times $B\in \mathbb{N}$.}
\KwOut{Core degree threshold $\tau$.}
\For{$t=1$ \KwTo $B$}{
  Simulate graph $\bm{G}_n^{(t)}\sim \mathrm{PF}(\alpha,\theta,\ell,\pi)$, where the community label function $\ell$ is chosen so that 
  $|V^1|=\cdots=|V^{K-1}|=\frac{n}{H}$, and $|V^K|=n-\frac{n(K-1)}{H}$; find each planted tree’s highest degree node; record $r^{(t)}=\max_k \operatorname{rank}^{(t)}(\text{highest degree node in comp }k)$.
}
$r:= \operatorname{Quantile}_{q}\big(\{r^{(t)}\}_{t=1}^T\big)$\;
\Return $\tau:=\operatorname{rank}_{r}\left\{\operatorname{deg}\left(u\right)\mid u \in V\left(\bm{G}_n\right)\right\}$.
\end{algorithm}

Equipped with the core degree threshold, we now proceed to the graph pruning step.

\begin{remark}
Algorithm~\ref{alg: choosing the core degree threshold} assumes the graph parameters $(\alpha,\theta)$ are available.
If the graph parameters $(\alpha, \theta)$ are assumed to be unknown,
the noise level $\theta$ can still be easy to estimate, e.g.\ by 
$\hat{\theta}=\frac{|E(\boldsymbol{G}_n)|-n+1}{|E(\boldsymbol{G}_n)|}$.
By contrast, $\alpha$ can be harder to estimate.
A principled estimation procedure is given in the supplementary material of \cite{crane2024root}. 
In practice, one may use plug-in estimates or, for a conservative calibration in the simulations, a known upper bound for $\alpha$.
\end{remark}

\begin{remark}
With a chosen minimum proportion $1/H$. where $H\geq K$, the simulation setup we adopt here $\bm{G}_n^{(t)}\sim \mathrm{PF}(\alpha,\theta,\ell,\pi)$ with $\ell$ chosen as $|V^1|=\cdots=|V^{K-1}|=\frac{n}{H}$ and $|V^K|=n-\frac{n(K-1)}{H}$ represents the worst case: it maximizes imbalance across components and thus yields the largest (most conservative) values of $r^{(t)}$.
\end{remark}

\begin{remark}
A closely related alternative is to work with degrees directly instead of with degrees rank: in each simulation $t$, let  $\tau^{(t)}$ be the smallest top–degree across components, and set $\tau:= \operatorname{Quantile}_{q}\big(\{\tau^{(t)}\}_{t=1}^T\big)$.
\end{remark}

\subsection{The Simple Degree Thresholding Algorithm}
\label{secsec: The simple Two-Step Algorithm}
\begin{algorithm}[H]
\caption{Simple Forest Community Detection Method}
\label{alg: theory align}
\KwIn{Graph $\boldsymbol{G}_n=\left(V,E\right)$; core degree threshold $\tau_1 \in \mathbb{N}$; pruning degree threshold $\tau_2\in\mathbb{N}$}
\KwOut{$\hat{K}$, core vertex set $\{V_1, \ldots, V_{\hat{K}}\}$, community label $\hat{\ell}:V\to\{1,\dots,\hat{K}\}$}

Remove $\left\{v \in V: \deg_{\boldsymbol{G}_n}\left(v\right)<  \tau_2\right\}$ from $\bm{G}_n$ to get $\bar{\bm{G}}$\;
Let $\hat{K}:= \Bigl|\bigl\{\, \boldsymbol{g} \in \mathcal{C}\left(\bar{\bm G}\right)\,:\, \deg_{\bm{G}_n}\left(V\left(\boldsymbol{g}\right)\right)\ge \tau_1 \bigr\}\Bigr|$\;
Let $\left\{{V}_1,\dots,{V}_{\hat{K}}\right\}:=\left\{V\left(\bm{g}\right)\,:\, \boldsymbol{g} \in \mathcal{C}\left(\bar{\bm G}\right)\ \text{and}\ \deg_{\boldsymbol{G}_n}\left(V\left(\bm g\right)\right)\ge \tau_1\right\}_{}$;

\BlankLine
\textbf{Labeling:}
\For{$u\in V$}{
    $D_k\left(u\right)= \min_{v\in {V}_k}\operatorname{dist}_{\boldsymbol{G}_n}\left(u,v\right)$ for all $k\in[\hat{K}]$;\quad
    $\hat{\ell}(u):= \operatorname{argmin}_{k\in[\hat{K}]} D_k\left(u\right)$ with uniform random tie–breaking if non-unique\;
}
%\Return $\hat{\ell}$\;
\end{algorithm}

\begin{example}

\begin{figure}[H]
    \centering
    % First subfigure
    \begin{subfigure}[b]{0.32\linewidth}
        \centering
        \vspace{-0.4cm}
        \includegraphics[width=1\linewidth]{numerical_output/figure1_original_graph.pdf}
        \subcaption{Origin Graph}
        \label{fig: simple algorithm a}
    \end{subfigure}
    \hfill
    % Second subfigure
    \begin{subfigure}[b]{0.32\linewidth}
        \centering
        \vspace{-0.4cm}
        \includegraphics[width=1\linewidth]{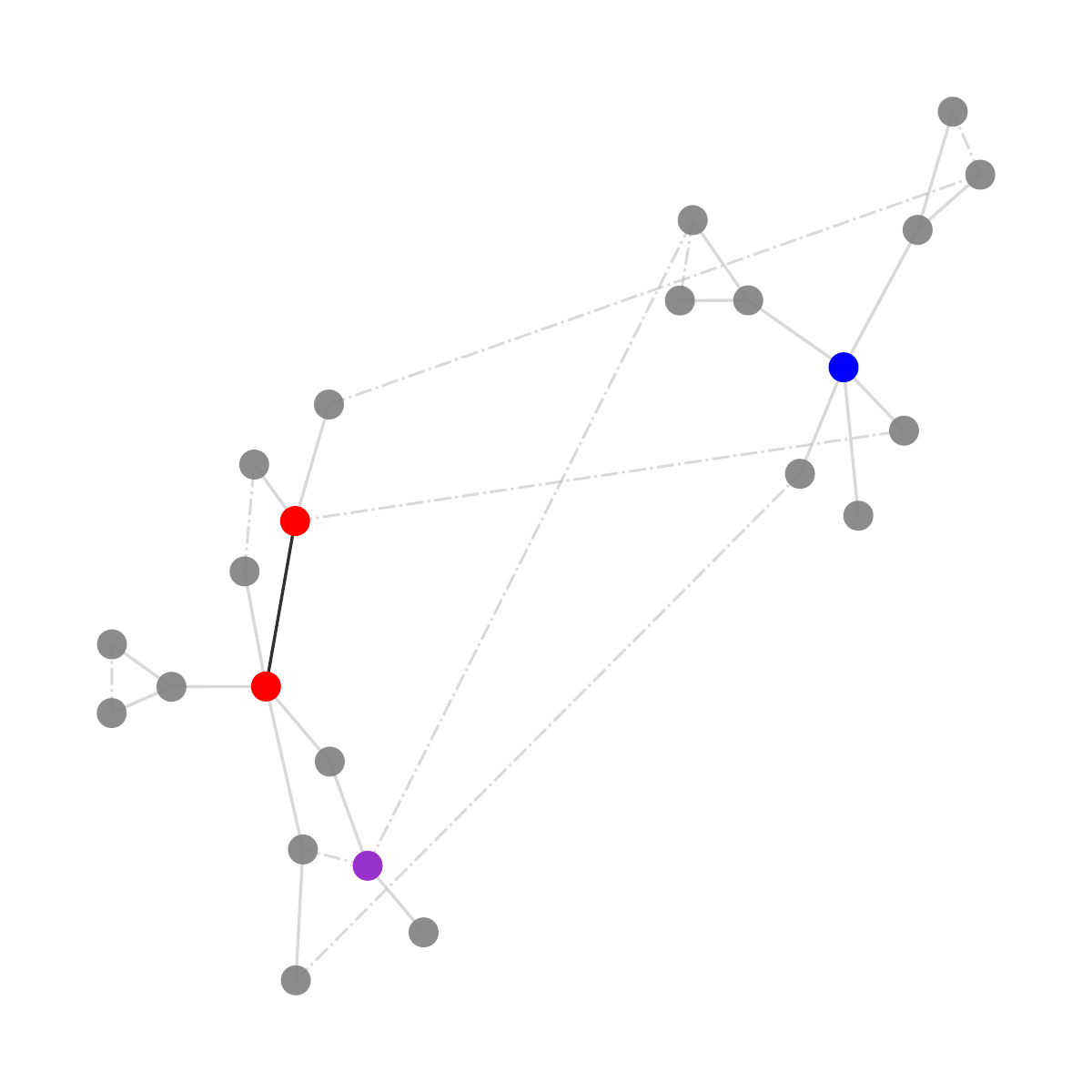} % Add your image here
        \subcaption{Graph after Degree Pruning}
        \label{fig: simple algorithm b}
    \end{subfigure}
    % Third subfigure
    \begin{subfigure}[b]{0.32\linewidth}
        \centering
        \vspace{-0.4cm}
        \includegraphics[width=1\linewidth]{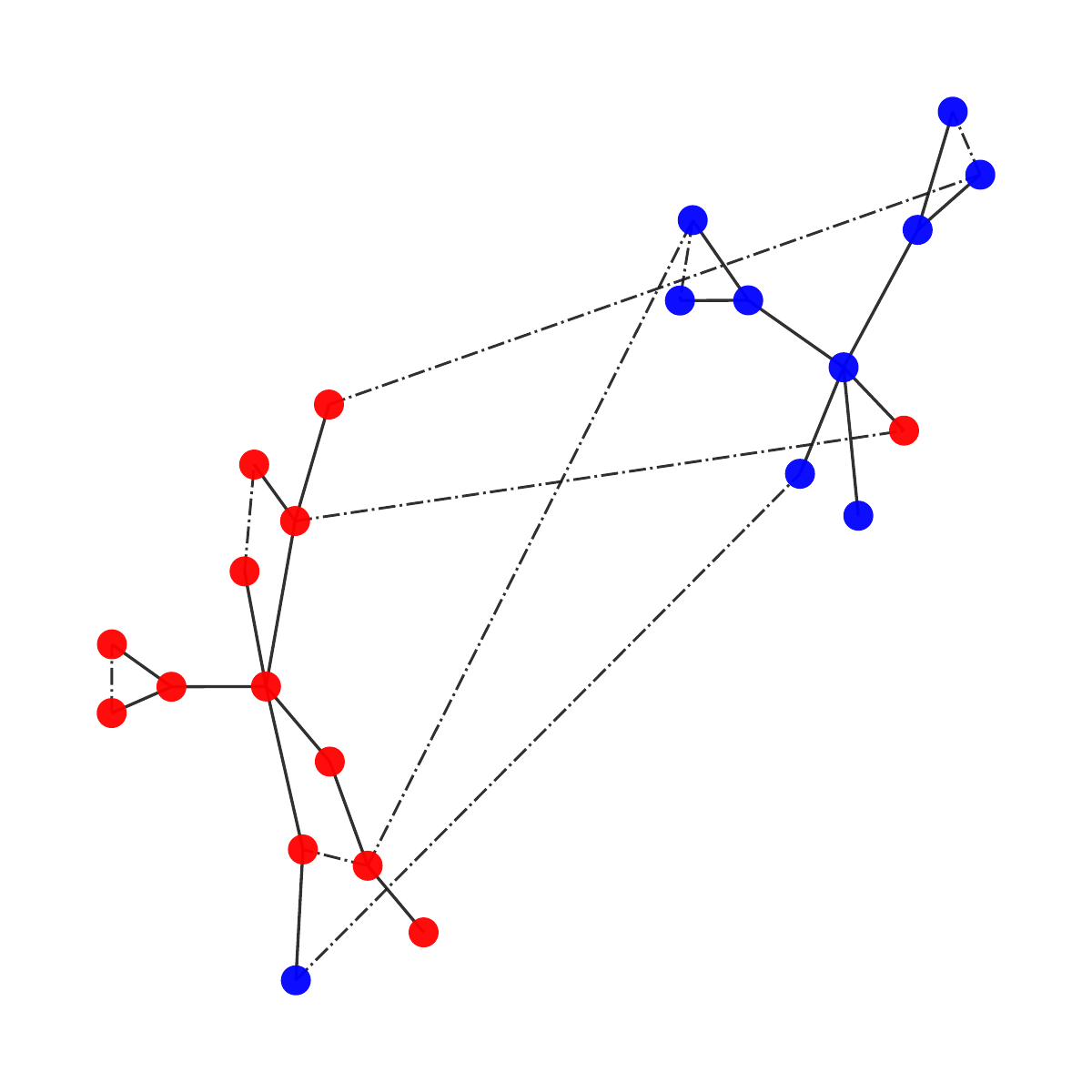} % Add your image here        
        \subcaption{Final Classification}
        \label{fig: simple algorithm c}
    \end{subfigure}
    \caption{Forest community detection via Algorithm~\ref{alg: theory align} 
        on a toy example with thresholds $\tau_1=5$ and $\tau_2=4$. 
        \textbf{(a)} Original graph containing two planted trees. 
        \textbf{(b)} Three cores obtained after degree-based pruning; 
        red/blue cores contain nodes of degree at least $5$. 
        \textbf{(c)} Final community assignment using red/blue cores for distance-based recovery.}
\label{fig: simple algorithm}
\end{figure}

We analyze the same synthetic example graph from Figure~\ref{fig: improved algorithm a} to illustrate Algorithm~\ref{alg: theory align}.
We set the core threshold to $\tau_1 = 5$ and the pruning threshold to $\tau_2 = 4$. After degree-based pruning (Figure~\ref{fig: simple algorithm b}), the graph decomposes into $M = 3$ isolated cores. Among these, only two cores satisfy the highest-degree requirement of $\tau_1 = 5$, which includes the most ``center" nodes in each planted tree.

The final distance recovery outcome, presented in Figure~\ref{fig: simple algorithm c} shows that the procedure is effective for this example: only 2 nodes are misclassified in the entire graph.
\end{example}

\begin{lemma}
\label{lemma improved}  
Let $\hat{K}$ be the estimated number of clusters and $V^1_1, \ldots, V^1_{\hat{K}}$ be the core vertex sets outputted by Algorithm~\ref{alg: theory align} with input parameter $\tau_1, \tau_2 > 0$; define $Q_{\max} := \min_{k \in [\hat{K}]} |V^1_k|$. Let $V^2_1, \ldots, V^2_{\hat{K}}$ be the output core vertex sets of Algorithm~\ref{alg: step I} with input $K = \hat{K}$, $\tau = \tau_1$, and any $Q \in \mathbb{N}$ such that $Q \leq Q_{\max}$. Then, with suitable reindexing, 

%Consider Algorithm~\ref{alg: theory align} with tuning parameters $(\tau_1, \tau_2)$ and Algorithm~\ref{alg: step I} with $(\tau^s_1, Q^s)$, both applied to the same graph $\bm{G}_n = (V, E)$.  
%Let the resulting core sets from Algorithm~\ref{alg: theory align} be $V_1, \dots, V_K$, and those from Algorithm~\ref{alg: step I} be $V^s_1, \dots, V^s_K$.  
%If $\tau_1=\tau^s_1, \hat{K}=K$ and $Q^s\le \min_{1\le i \le K} |V_i|$, then after a suitable reindexing, 
%
\begin{equation*}
    V^1_i \subseteq V^2_i  \quad \text{for all $i\in \bigl[ \hat{K} \bigr]$}.
\end{equation*}
i.e., each core vertex set from Algorithm~\ref{alg: theory align} is contained in the corresponding core vertex set returned by Algorithm~\ref{alg: step I}. 
Moreover, $D^*$ is strictly increasing in each iteration, and the final output of $D^*$ is smaller than $\tau_2$. 
\end{lemma}

The following direct corollary of Lemma~\ref{lemma improved} will be useful for our theoretical analysis in Section~\ref{sec: theory proofs}. 

\begin{corollary}
\label{cor: simple containment}
For any $\tau> 0$ and $Q, K \in \mathbb{N}$, let $\{V_1^2, \ldots, V_K^2\}$ and $D^*$ be the output of Algorithm~\ref{alg: step I} with input $\tau$, $Q$, and $K$. Let $0 < \tau' \leq \tau$ and define
\begin{align*}
V_1^1, \ldots, V_M^1 &= \mathcal{C}\bigl( \bm{G}_n \cap V^{\tau'}(\bm{G}_n) \bigr) \,\, \text{ and} \\
\mathcal{K} &= \biggl\{ i \in [M] \,:\, \mathrm{deg}_{\bm{G}_n}(V_i^1) \geq \tau \biggr\}.
\end{align*}
If $|\mathcal{K}| = K$ and $Q \leq |V_i^1|$ for all $i \in \mathcal{K}$, then $D^* < \tau'$ and there exist a bijection $\sigma \,:\, \mathcal{K} \rightarrow [K]$ such that
\[
V_i^1 \subseteq V_{\sigma(i)}^2, \quad \forall i \in \mathcal{K}.
\]

\end{corollary}

\subsection{Proofs for Section~\ref{sec: algorithm}}
\begin{proof}[Proof of Lemma~\ref{lemma improved}]
Suppose the pruning procedure in Algorithm~\ref{alg: step I} produce a graph sequence such that $\bm{G}^{(T)}\subset \bm{G}^{(T-1)} \cdots \subset \bm{G}^{(0)}=\bm{G}_n$, where $\bm{G}^{(t)}$ denotes the graph after $t$-th pruning iteration.
And the $D^*$ at each iteration will be denoted as $D^{(t)}$, while the final output of $D^*$ being $D^{(T)}$.
Let $\tilde{V}^{\prime(t)}, \tilde{V}^{\prime \prime(t)}$ be the cycle-based and split-based target set of $\bm{G}^{(t)}$ from Algorithm~\ref{alg: step I} at iteration $t$, respectively. In the last iteration $T$, it must be either that we return $\hat{K}$ vertex sets $\{V_1^2, \ldots, V_{\hat{K}}^2\}$ such that $| V_k^2| \geq Q_{\max}$ and $\operatorname{deg}_{\bm{G}_n}(V^2_k) \geq \tau_1$ for all $k \in [\hat{K}]$ or that $\tilde{V}^{\prime (T)}\cup\tilde{V}^{\prime \prime(T)}=\emptyset$ and Algorithm~\ref{alg: step I} fails. 

We first claim, for every $k \in [\hat{K}]$, that $V^1_k \subset V(\bm{G}^{(t)})$ for every $t \in [T]$. We will then finish the proof by showing that failure is impossible and that each $V^1_k$ is contained in a distinct $V^2_{k'}$. 

%In the following part, we first show that $\cup_{i=1}^K V_i \subset V(\bm{G}^{(t)})$ for every $t$, and then show that they are separately contained in isolated components of $\bm{G}^{(T)}$.

For the first claim, we use proof by induction on $t\le T$. 
The induction hypothesis is $\cup_{i=1}^{\hat{K}} V^1_i \subset V(\bm{G}^{(t)})$, which holds trivially for $t=0$.
Assume the hypothesis holds for some $t\le T$.

We observe that, since Algorithm~\ref{alg: step I} did not terminate at step $t$, it must be that there exists $k \neq k'$ such that $V^1_k$ and $V^1_{k'}$ are connected in $\bm{G}^{(t)}$. Indeed, if $V^1_1, \ldots, V^1_{\hat{K}}$ are all disconnected in $\bm{G}^{(t)}$, then, writing $V^{(t)}_1, \ldots, V^{(t)}_{\hat{K}}$ as the connected components of $\bm{G}^{(t)}$ containing $V_1^1, \ldots, V^1_{\hat{K}}$, we have that $|V^{(t)}_k| \geq |V^1_k| \geq Q_{\max}\geq Q$ and that $\operatorname{deg}_{\bm{G}_n}(V^{(t)}_k) \geq \operatorname{deg}_{\bm{G}_n}(V^1_k) \geq \tau_1$ and thus, the condition in line 4 of Algorithm~\ref{alg: step I} would be satisfied and the algorithm would terminate, creating a contradiction. 

%Because $\bm{G}^{(t)}$ does not yield $K$ qualified cores,
%WLOG, there must be a chain linking $V_1$ and $V_2$. 

Now we show that there exists a node $u \in \tilde{V}^{(t)} = \tilde{V}^{
\prime(t)} \cup \tilde{V}^{\prime \prime(t)}$ such that $\operatorname{deg}_{\bm{G}_n}(u) < \min\{ \operatorname{deg}_{\bm{G}_n}(v) \,:\, v \in \cup_{k=1}^K V^1_k \}$. To see this, consider any chain of nodes $\{u_1, \ldots, u_m\}$ connecting $V^1_k$ and $V^1_{k'}$ in $\bm{G}^{(t)}$. By line 1 of Algorithm~\ref{alg: theory align}, $V^1_k$ and $V^1_{k'}$ are disconnected after removing all nodes $u \in \bm{G}_n$ such that $\operatorname{deg}_{\bm{G}_n}(u) < \tau_2$. 
Thus, it must be that there is a node $u_i$ on the chain has degree $\operatorname{deg}_{\bm{G}_n}(u_i) < \tau_2$. On the other hand, every node in $\cup_{k=1}^K V^1_k$ must have degree (with respect to $\bm{G}_n$) at least $\tau_2$ by construction. Moreover, we note that $u_i \in \tilde{V}^{(t)}$ because either $u_i$ is in a cycle (and hence in $\tilde{V}^{\prime (t)}$) or if it is not in a cycle, then any path from a node in $V^1_k$ to a node in $V^1_{k'}$ must pass through $u_i$, and hence, removing $u_i$ would disconnect $V^1_k$ and $V^1_{k'}$ in $\bm{G}^{(t)}$ so that $u_i$ must be in $\tilde{V}^{\prime \prime(t)}$. Thus,

\begin{equation}
\label{lemma improved eqn 1}
    \min_{v \in \tilde{V}^{(t)}} \operatorname{deg}_{\bm{G}_n}\left(v\right) <\tau_2, 
\end{equation}
so that the removed nodes in $(t+1)$th pruning iteration have degree smaller than $\tau_2$ and cannot belong to $\cup_{i=1}^{\hat{K}} V^1_i$. This completes the induction and proves the claim that $\cup_{k=1}^{\hat{K}} V^1_k \subset V(\bm{G}^{(t)})$ for each $t \in [T]$.

Suppose for contradiction that the algorithm fails. Then it must be that there exists $k \neq k'$ such that $V^1_k$ and $V^1_{k'}$ are connected in $\bm{G}^{(T)}$ and that $\tilde{V}^{(T)} = \emptyset$. Indeed, using the same argument as before, if $V^1_1, \ldots, V^1_{\hat{K}}$ are all disconnected in $\bm{G}^{(T)}$, then the connected components of $\bm{G}^{(T)}$ containing $V_1^1, \ldots, V^1_{\hat{K}}$ would have size at least $Q_{\max}$ and maximum degree at least $\tau_1$ and satisfy condition in line 4 of Algorithm~\ref{alg: step I}. However, if there exists a pair $V^1_k$ and $V^1_{k'}$ that are connected in $\bm{G}^{(T)}$, then at least one of the vertex on the chain connected $V^1_k$ and $V^1_{k'}$ would be in $\tilde{V}^{(T)}$ and thus creating a contradiction. 

Therefore, $V\bigl(\bm{G}^{(T)}\bigr)$ consists of $\bar K$ disjoint components $V^2_1, \ldots, V^2_{\bar K}$ for $\bar K \geq \hat{K}$. 
By Algorithm~\ref{alg: step I} and Algorithm~\ref{alg: theory align}, $\operatorname{deg}_{\bm{G}_n}\left(V^2_k\right)\geq \tau_1$ and $\left\{v \in V: \operatorname{deg}_{\bm{G}_n}\left(u\right) \geq \tau_1\right\} \subset \cup_{i=1}^K V^1_i$ so that, for every $k \in [\bar{K}]$, there exists $k'$ such that $V^1_{k'} \cap V^2_k \neq \emptyset$ so that $\bar{K} = \hat{K}$. But since $V^1_k \subset V\bigl(\bm{G}^{(T)}\bigr)$, it must be that, after suitable reindexing, 
\[
\forall k \in [\hat{K}], \, V^1_k \subseteq V^2_k.
\]

We note that Equation~\eqref{lemma improved eqn 1} immediately implies that $D^{(T)} < \tau_2$ so that the final output of $D^*$ is smaller than $\tau_2$.

Now we show that $D^{(t)}> D^{(t-1)}$ for all $t \in \{2, 3, \ldots, T\}$. We note that the following must be true:
\begin{align}
\tilde{V}^{\prime (t)} \subseteq \tilde{V}^{\prime (t-1)}, \quad \tilde{V}^{\prime \prime (t)} \subseteq \tilde{V}^{\prime (t-1)} \cup \tilde{V}^{\prime \prime (t-1)},
\label{eq:d_increasing_eq1}
\end{align}
where the first inclusion holds because removing nodes from a graph cannot create new cycles. The second inclusion holds because if a node $u$ in $\tilde{V}^{\prime \prime (t)}$, then \begin{align}
| \{ \bm{g} \in \mathcal{C}( \bm{G}^{(t)} \backslash \{ u \}) \,:\, |V(\bm{g})| \geq Q, \, \operatorname{deg}_{\bm{G}_n}( V(\bm{g})) \geq \tau_1 \} | > | \{ \bm{g} \in \mathcal{C}( \bm{G}^{(t)}) \,:\, |V(\bm{g})| \geq Q, \, \operatorname{deg}_{\bm{G}_n}( V(\bm{g})) \geq \tau_1 \} |  \label{eq:d_increasing_eq2}.
\end{align}
Since $\bm{G}^{(t)}$ is a subgraph of $\bm{G}^{(t-1)}$, it must be that either $u$ belongs in a cycle in $\bm{G}^{(t-1)}$ (so that $u \in \tilde{V}^{\prime (t-1)}$) or, if $u$ does not belong in any cycle in $\bm{G}^{(t-1)}$, $u$ must satisfy~\eqref{eq:d_increasing_eq2} with $\bm{G}^{(t-1)}$ in place of $\bm{G}^{(t)}$. Since we remove the nodes in $\tilde{V}^{\prime (t-1)} \cup \tilde{V}^{\prime \prime (t-1)}$ with the smallest degree, we use~\eqref{eq:d_increasing_eq1} to conclude that
\[
D^{(t-1)} = \min_{v \in \tilde{V}^{\prime (t-1)} \cup \tilde{V}^{\prime \prime (t-1)}} \operatorname{deg}_{\bm{G}_n}(v) < \min_{v \in \tilde{V}^{\prime (t)} \cup \tilde{V}^{\prime \prime (t)}} \operatorname{deg}_{\bm{G}_n}(v) = D^{(t)}. 
\]

The lemma follows as desired.
\end{proof}

\section{Additional numerical results for multiple unbalanced trees}
\label{sec: additional simulation}
We adopt the same experimental setting as in Section~\ref{sec: simulation}, 
except that the network now consists of three planted trees of unequal sizes, 
with $n_1 = 3n_3$, $n_2 = 2n_3$, and 
$n_3 \in \{200, 400, 600, 800, 1000\}$.

For $n=1200$ and $\theta=0.01$, this procedure yields thresholds of $11$ for 
$\alpha=-0.5$, $12$ for $\alpha=0$, and $9$ for $\alpha=2$. 
As $n$ increases, the corresponding $0.95$ quantiles range from $9$ to $15$. 
For simplicity and consistency across all experimental settings, we therefore 
fix $\tau$ to be the degree of the 15th highest-degree vertex in the observed graph. 
All remaining tuning parameters are chosen in the same manner as in the 
balanced two-tree setting described in Section~\ref{sec: simulation}.

The corresponding numerical results are summarized in the tables and figures below. 
Overall, the patterns are similar to those observed in the simulations presented in Section~\ref{sec: simulation}. The main difference is that the misclassification rates are uniformly higher in this setting, reflecting the increased difficulty 
of the problem when three planted trees of unequal sizes are present.

\begin{table}[H]
\centering
\caption{Mean misclassification rate with Algorithm~\ref{alg: step I} and distance recovery method for $\theta \asymp n^{-1}$ under the multiple tree setting $n_1 = 3 n_3, n_2=2 n_3$ (with $\theta=0.01$ for $n=1200$).}
{
\begin{tabular}{c|c|cc|cc|ccc|c}
\hline
\hline
\multirow{3}{*}{$\alpha$} & \multirow{3}{*}{n}
& \multicolumn{8}{c}{Misclassification Rate for Different Subsets} \\
\cline{3-10}
& 
& \multicolumn{2}{c|}{First Arriving}
& \multicolumn{2}{c|}{Highest Degree}
& \multicolumn{3}{c|}{Structural Position}
& \multirow{2}{*}{Overall} \\
\cline{3-9}
& 
& $L=10$ & $L=50$
& $L=10$ & $L=50$
& Root (Layer 0) & Layer 1 & Layer 2
& \\
\hline
\multirow{5}{*}{2}
& 1200 & 0.43 & 0.55 & 0.37 & 0.52 & 0.29 & 0.39 & 0.51 & 0.61 \\
& 2400 & 0.39 & 0.51 & 0.33 & 0.49 & 0.26 & 0.34 & 0.46 & 0.61 \\
& 3600 & 0.37 & 0.49 & 0.33 & 0.48 & 0.25 & 0.32 & 0.45 & 0.61 \\
& 4800 & 0.35 & 0.47 & 0.32 & 0.47 & 0.23 & 0.30 & 0.44 & 0.61 \\
& 6000 & 0.33 & 0.46 & 0.30 & 0.45 & 0.23 & 0.28 & 0.42 & 0.61 \\
\hline
\multirow{5}{*}{0}
& 1200 & 0.30 & 0.43 & 0.24 & 0.41 & 0.18 & 0.26 & 0.39 & 0.54 \\
& 2400 & 0.25 & 0.38 & 0.18 & 0.35 & 0.16 & 0.24 & 0.36 & 0.54 \\
& 3600 & 0.21 & 0.33 & 0.16 & 0.30 & 0.11 & 0.21 & 0.33 & 0.53 \\
& 4800 & 0.19 & 0.31 & 0.14 & 0.27 & 0.12 & 0.19 & 0.33 & 0.53 \\
& 6000 & 0.18 & 0.29 & 0.14 & 0.26 & 0.10 & 0.19 & 0.31 & 0.53 \\
\hline
\multirow{5}{*}{-0.5}
& 1200 & 0.21 & 0.32 & 0.13 & 0.31 & 0.13 & 0.20 & 0.32 & 0.44 \\
& 2400 & 0.19 & 0.27 & 0.12 & 0.25 & 0.10 & 0.19 & 0.30 & 0.44 \\
& 3600 & 0.15 & 0.23 & 0.09 & 0.19 & 0.09 & 0.16 & 0.27 & 0.43 \\
& 4800 & 0.11 & 0.18 & 0.08 & 0.16 & 0.05 & 0.11 & 0.24 & 0.42 \\
& 6000 & 0.09 & 0.16 & 0.08 & 0.14 & 0.05 & 0.09 & 0.22 & 0.42 \\
\hline
\hline
\end{tabular}
}
\label{tab:mean_error_scale_n1 3}
\end{table}

\begin{table}[H]
\centering
\caption{Mean misclassification rate with Algorithm~\ref{alg: step I} and distance recovery method for $\theta \asymp n^{-3/4}$ under the multiple tree setting $n_1 = 3 n_3, n_2=2 n_3$ (with $\theta=0.01$ for $n=1200$).}
{
\begin{tabular}{c|c|cc|cc|ccc|c}
\hline
\hline
\multirow{3}{*}{$\alpha$} & \multirow{3}{*}{n}
& \multicolumn{8}{c}{Misclassification Rate for Different Subsets} \\
\cline{3-10}
& 
& \multicolumn{2}{c|}{First Arriving}
& \multicolumn{2}{c|}{Highest Degree}
& \multicolumn{3}{c|}{Structural Position}
& \multirow{2}{*}{Overall} \\
\cline{3-9}
& 
& $L=10$ & $L=50$
& $L=10$ & $L=50$
& Root (Layer 0) & Layer 1 & Layer 2
& \\
\hline
\multirow{5}{*}{2}
& 1200 & 0.43 & 0.54 & 0.38 & 0.52 & 0.27 & 0.38 & 0.50 & 0.61 \\
& 2400 & 0.40 & 0.52 & 0.36 & 0.51 & 0.26 & 0.35 & 0.49 & 0.62 \\
& 3600 & 0.39 & 0.52 & 0.36 & 0.51 & 0.28 & 0.34 & 0.48 & 0.62 \\
& 4800 & 0.39 & 0.51 & 0.33 & 0.50 & 0.24 & 0.35 & 0.48 & 0.63 \\
& 6000 & 0.38 & 0.50 & 0.34 & 0.50 & 0.22 & 0.32 & 0.47 & 0.63 \\
\hline
\multirow{5}{*}{0}
& 1200 & 0.31 & 0.43 & 0.25 & 0.42 & 0.19 & 0.27 & 0.39 & 0.54 \\
& 2400 & 0.27 & 0.39 & 0.20 & 0.37 & 0.18 & 0.23 & 0.37 & 0.54 \\
& 3600 & 0.25 & 0.37 & 0.20 & 0.35 & 0.16 & 0.23 & 0.37 & 0.56 \\
& 4800 & 0.24 & 0.35 & 0.17 & 0.32 & 0.15 & 0.23 & 0.36 & 0.56 \\
& 6000 & 0.23 & 0.34 & 0.16 & 0.31 & 0.14 & 0.23 & 0.36 & 0.56 \\
\hline
\multirow{5}{*}{-0.5}
& 1200 & 0.22 & 0.33 & 0.15 & 0.32 & 0.13 & 0.21 & 0.33 & 0.45 \\
& 2400 & 0.19 & 0.29 & 0.12 & 0.26 & 0.11 & 0.19 & 0.31 & 0.45 \\
& 3600 & 0.16 & 0.26 & 0.10 & 0.22 & 0.10 & 0.19 & 0.30 & 0.45 \\
& 4800 & 0.16 & 0.24 & 0.10 & 0.20 & 0.09 & 0.17 & 0.29 & 0.45 \\
& 6000 & 0.12 & 0.20 & 0.09 & 0.18 & 0.06 & 0.14 & 0.27 & 0.45 \\
\hline
\hline
\end{tabular}
}
\label{tab:mean_error_scale_n34 3}
\end{table}

\begin{table}[H]
\centering
\caption{Mean misclassification rate with Algorithm~\ref{alg: step I} and distance recovery method for $\theta \asymp n^{-1/2}$ under the multiple tree setting $n_1 = 3 n_3, n_2=2 n_3$ (with $\theta=0.01$ for $n=1200$).}
{
\begin{tabular}{c|c|cc|cc|ccc|c}
\hline
\hline
\multirow{3}{*}{$\alpha$} & \multirow{3}{*}{n}
& \multicolumn{8}{c}{Misclassification Rate for Different Subsets} \\
\cline{3-10}
& 
& \multicolumn{2}{c|}{First Arriving}
& \multicolumn{2}{c|}{Highest Degree}
& \multicolumn{3}{c|}{Structural Position}
& \multirow{2}{*}{Overall} \\
\cline{3-9}
& 
& $L=10$ & $L=50$
& $L=10$ & $L=50$
& Root (Layer 0) & Layer 1 & Layer 2
& \\
\hline
\multirow{5}{*}{2}
& 1200 & 0.42 & 0.54 & 0.37 & 0.51 & 0.28 & 0.38 & 0.49 & 0.61 \\
& 2400 & 0.42 & 0.54 & 0.35 & 0.52 & 0.28 & 0.36 & 0.49 & 0.62 \\
& 3600 & 0.41 & 0.54 & 0.37 & 0.53 & 0.29 & 0.38 & 0.51 & 0.63 \\
& 4800 & 0.42 & 0.54 & 0.38 & 0.53 & 0.27 & 0.37 & 0.51 & 0.63 \\
& 6000 & 0.42 & 0.53 & 0.38 & 0.53 & 0.27 & 0.36 & 0.51 & 0.64 \\
\hline
\multirow{5}{*}{0}
& 1200 & 0.30 & 0.43 & 0.23 & 0.41 & 0.19 & 0.26 & 0.38 & 0.54 \\
& 2400 & 0.29 & 0.41 & 0.23 & 0.40 & 0.18 & 0.26 & 0.39 & 0.56 \\
& 3600 & 0.26 & 0.39 & 0.19 & 0.36 & 0.17 & 0.24 & 0.38 & 0.56 \\
& 4800 & 0.26 & 0.38 & 0.20 & 0.37 & 0.16 & 0.26 & 0.40 & 0.58 \\
& 6000 & 0.26 & 0.37 & 0.18 & 0.34 & 0.16 & 0.24 & 0.38 & 0.58 \\
\hline
\multirow{5}{*}{-0.5}
& 1200 & 0.21 & 0.32 & 0.15 & 0.31 & 0.12 & 0.20 & 0.32 & 0.44 \\
& 2400 & 0.20 & 0.30 & 0.12 & 0.29 & 0.11 & 0.20 & 0.33 & 0.46 \\
& 3600 & 0.19 & 0.29 & 0.13 & 0.25 & 0.10 & 0.20 & 0.33 & 0.48 \\
& 4800 & 0.18 & 0.28 & 0.12 & 0.24 & 0.11 & 0.20 & 0.33 & 0.48 \\
& 6000 & 0.17 & 0.26 & 0.10 & 0.22 & 0.08 & 0.19 & 0.32 & 0.48 \\
\hline
\hline
\end{tabular}
}
\label{tab:mean_error_scale_n12 3}
\end{table}

\begin{figure}[H]
    \centering
        {\includegraphics[width=.77\linewidth]{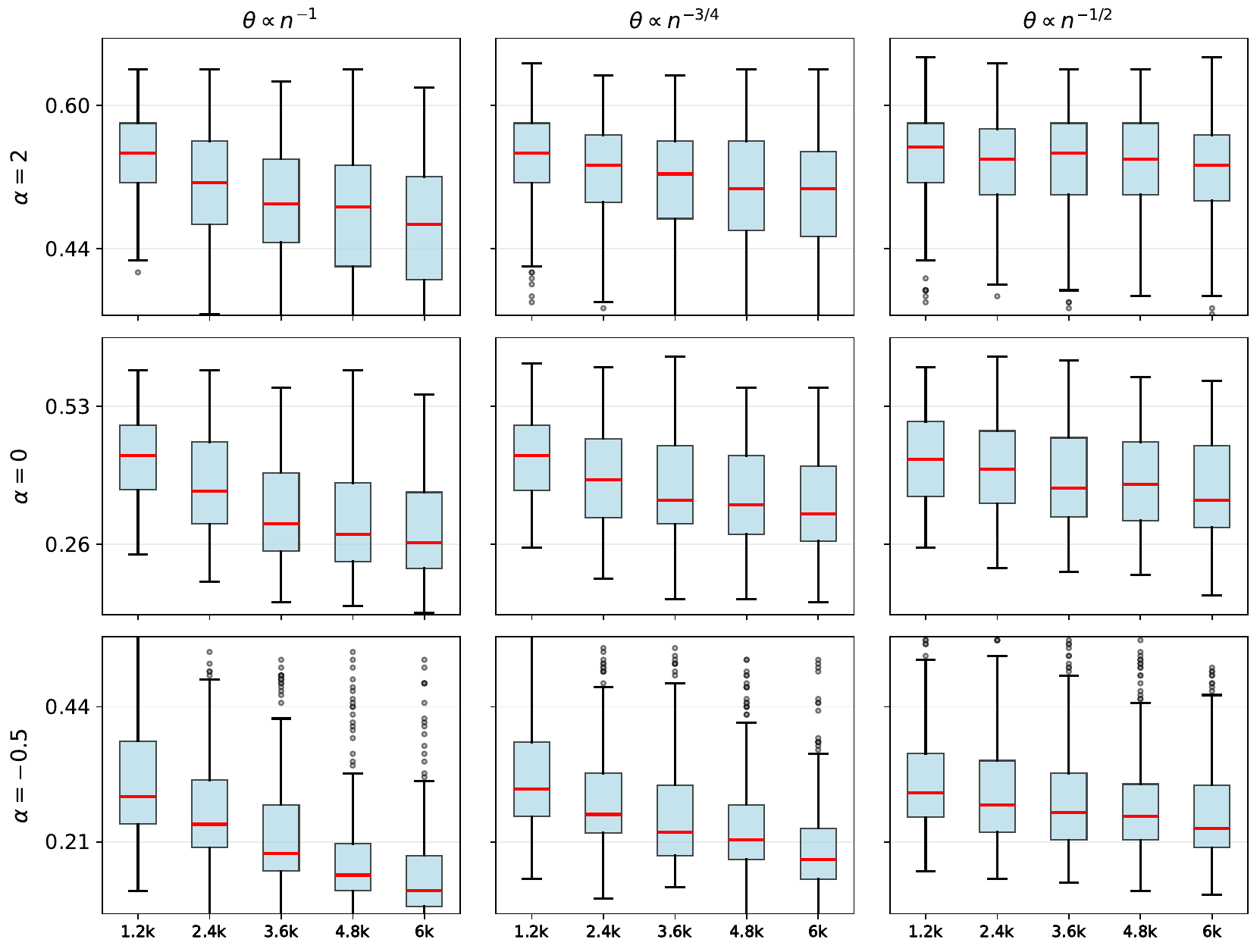}}
    \caption{Boxplots of misclassification rates for the first $50$ arriving vertices in each tree, computed over $200$ simulated networks.}
    \label{fig: 3 first 50 3 times 3}
\end{figure}

\begin{figure}[H]
    \centering
        {\includegraphics[width=.77\linewidth]{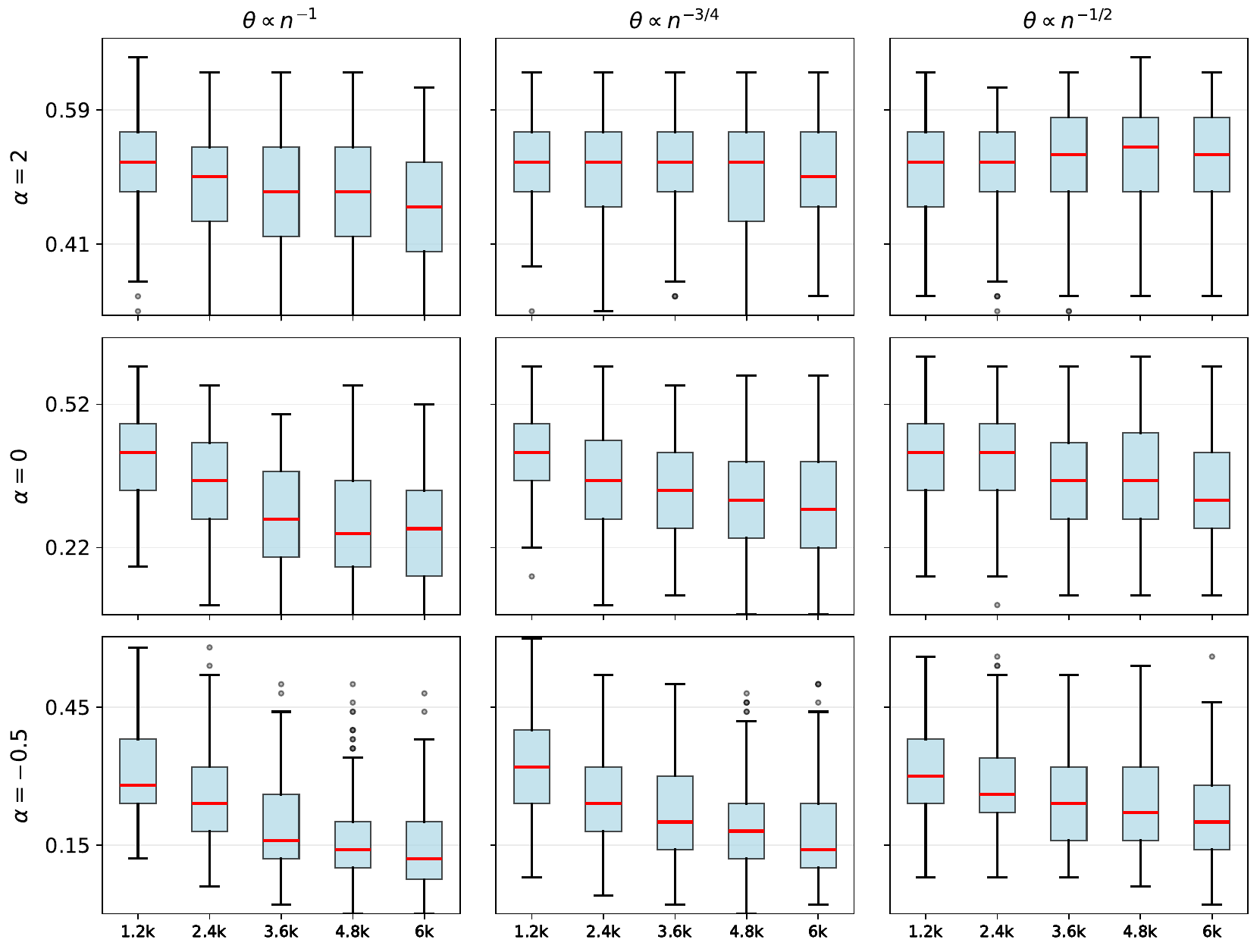}}
    \caption{Boxplots of misclassification rates for the $50$ highest degree vertices in the network, computed over $200$ simulated networks.}
    \label{fig: 3 high degree 50 3 times 3}
\end{figure}

\begin{figure}[H]
    \centering
        {\includegraphics[width=.77\linewidth]{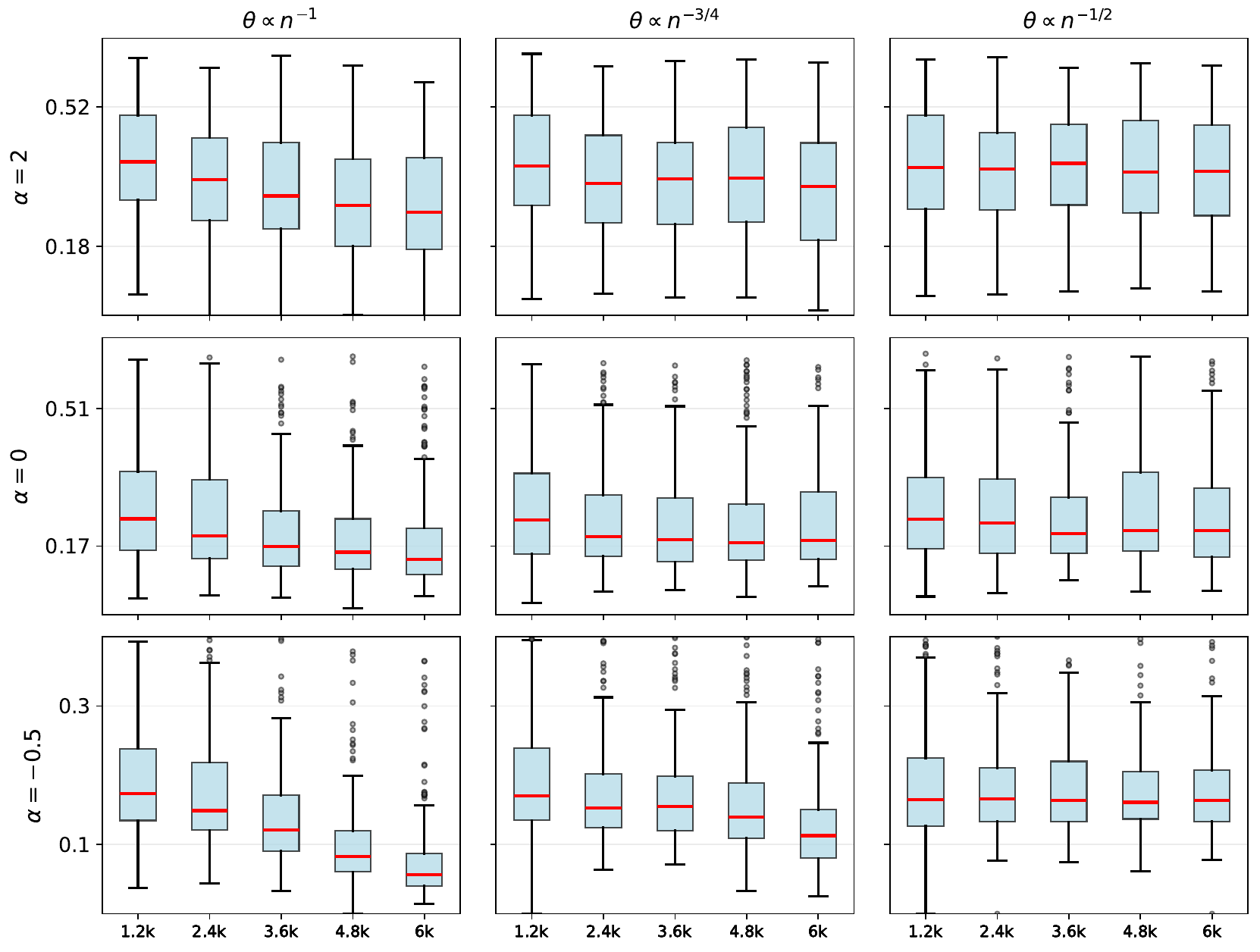}}
    \caption{Boxplots of misclassification rates for the layer-1 vertices in the network, computed over $200$ simulated networks.}
    \label{fig: 3 layer 1 3 times 3}
\end{figure}

\section{Proofs for Section~\ref{sec: theory}}
\label{sec: theory proofs}

\subsection{Useful intermediate results}
We first introduce several quantities that will be used repeatedly throughout the proof:
For $L \in \mathbb{N}$, $\varepsilon > 0$ and $c\in (0,1]$, we define
\begin{eqnarray*}
\tilde{L}(\varepsilon, c) &:=& L\biggl( \frac{\varepsilon}{3H}, c \gamma\biggl( \frac{\varepsilon}{3H}, 1 \biggr) \biggr),\\
\tilde{\gamma}(\varepsilon, L,c) & := & \gamma\biggl( \frac{\varepsilon}{3H},\,\, \tilde{L}(\varepsilon,c) \vee L \biggr).
\end{eqnarray*}
Here, the functions $\gamma(\varepsilon, L)$ and $L(\varepsilon, \gamma)$ are defined in Lemma~\ref{lemma order--degree}.

\subsection{Proof of Theorem~\ref{thm: first L}}
\label{secsec: thm first L}

\begin{proof}

Let $\{V_1^2,\dots,V_K^2\}$ and $D^*$ be the output of Algorithm~\ref{alg: step I} (and thus the input of Algorithm~\ref{alg: distance recovery}) with $\tau := C_1 n^{\frac{1}{2+\alpha}} =\gamma\left(\frac{\varepsilon}{3H}, 1\right)n^{\frac{1}{2+\alpha}}$. Let $\hat{\ell}$ be the estimated community labels outputted by Algorithm~\ref{alg: distance recovery} with $\tau' < D^*$. 

%For any $i \in [K]$, define $\tilde{V}^2_i := \{ u \in V_i^2 \,:\, \text{deg}_{\bm{G}_n}(u) \geq \tau'\}$. 

Let us define
\begin{align*}
V_1^1, \ldots, V_M^1 &= \mathcal{C}\bigl( \bm{G}_n \cap V^{\tilde{\gamma}(\varepsilon,\, L\vee Q, 1), \alpha}(\bm{G}_n)\bigr) \,\, \text{ and} \\
\mathcal{K} &= \biggl\{ i \in [M] \,: \, \text{deg}_{\bm{G}_n}(V_i^1) \geq \gamma\bigl(\frac{\varepsilon}{3H}, 1\bigr) n^{\frac{1}{2+\alpha}} \biggr\}.
\end{align*}

Without loss of generality, we may suppose that $\mathcal{K} = \{1, 2, \ldots, |\mathcal{K}|\}$. Define then the event
\[
\Omega:=\left\{ |\mathcal{K}| = K \text{ and } \exists \sigma \in S_K \text{ such that }\pi_{1:(L \vee Q)}(\boldsymbol{T}^{\sigma\left(i\right)})\subset V^1_i,\, \text{for all $i\in \left[K\right]$}\right\}.
\]

We claim that $d_{\cup_{i=1}^K \pi_{1:(L \vee Q)}(\bm{T}^i)}(\hat{\ell}, \ell) = 0$ on the event $\Omega$. To prove this, we work on the event $\Omega$ for the remainder of the proof. Since $|\mathcal{K}| = K$ and $Q \leq | \pi_{1:(L\vee Q)}( \bm{T}^{\sigma(i)})| \leq | V_i^1| $ for all $i \in \mathcal{K}$, we may conclude using Corollary~\ref{cor: simple containment} that there exists $\sigma' \in S_K$ such that
\begin{align}
\forall i \in [K], \quad V_i^1 \subset V_{\sigma'(i)}^2. \label{eq: first L key 1}
\end{align}

Moreover, by Corollary~\ref{cor: simple containment}, we also have that $D^* \leq \tilde{\gamma}(\varepsilon, L \vee Q,1) n^{\frac{1}{2+\alpha}}$ so that $\tau' < \tilde{\gamma}(\varepsilon, L \vee Q, 1) n^{\frac{1}{2+\alpha}}$ as well since $\tau'$ is chosen to be less than $D^*$. Therefore, via the definition of $V_i^1$, we have that 
\begin{align}
\forall i \in [K],\quad V_i^1 \subset V^{\tilde{\gamma}(\varepsilon, L \vee Q, 1), \alpha}(\bm{G}_n) \subset V^{\tau'}({\bm{G}_n}). \label{eq: first L key 2}
\end{align}

From the definition of Algorithm~\ref{alg: distance recovery}, for any $i \in [K]$, we assign $\hat{\ell}(u) = i$ for all $u \in V_i^2 \cap V^{\tau'}(\bm{G}_n)$. Combining this with the definition of event $\Omega$, with~\eqref{eq: first L key 1}, and with~\eqref{eq: first L key 2}, we have that 
\[
\forall u \in \pi_{1:(L \vee Q)}(\bm{T}^i), \quad \hat{\ell}(u) = \sigma' (\sigma^{-1}(i)) \,\, \text{ and } \,\, \ell(u) = i.
\]
It holds therefore that
\begin{align*}
   d_{\cup_{i=1}^K \pi_{1:L}\left(\bm{T}^i\right)}(\hat{\ell}, \ell) 
   = 
   \min_{\sigma'' \in S_K} d_{\cup_{i=1}^K \pi_{1:L}\left(\bm{T}^i\right)}^{\text{Ham}}\left(\hat{\ell},\sigma'' \circ \ell\right)
   \le  
   d_{\cup_{i=1}^K \pi_{1:L}\left(\bm{T}^i\right)}^{\text{Ham}}\left(\hat{\ell},\sigma' \circ \sigma^{-1} \circ \ell\right)
   =0.
\end{align*}

It remains only to bound the probability of the event $\Omega$. By Lemma~\ref{lemma early hist containment}, we have $\mathbb{P}(\Omega) \geq 1 - \varepsilon + \eta_n$ where $\eta_n=o(1)$ depends only on $\varepsilon, L, Q, H, \alpha, \delta$. The Theorem thus follows as desired. 

\end{proof}

\begin{lemma}
\label{lemma order--degree}
For $\boldsymbol{G}_n \sim \mathrm{PF}(\alpha, \theta, \ell, \pi)$, under Assumption~\ref{assumption: bound} and \ref{assumption: sparsity}. For any $L\in \mathbb{N}$, and $\epsilon>0$, there exists $\gamma\left(\varepsilon, L\right) \in (0, \infty)$ such that, for any $i \in [K]$,
\begin{equation}
\label{order to degree}
     \liminf_{n \rightarrow \infty} \mathbb{P}\biggl\{\pi_{1:L}\bigl(\boldsymbol{T}^i\bigr)\subset V^{\gamma(\varepsilon, L), \alpha}_{\boldsymbol{G}_n}\bigl(\boldsymbol{T}^i \bigr)\biggr\} \ge 1-\varepsilon.
\end{equation}

Similarly, for any $C > 0$, and $\epsilon>0$, there exist an integer $L\left(\varepsilon, C\right)$ such that, for any $i \in [K]$,
\begin{equation}
\label{degree to order}
     \liminf_{n \rightarrow \infty} \mathbb{P}\biggl\{ V^{C,\alpha}_{\boldsymbol{G}_n}\bigl(\boldsymbol{T}^i\bigr)\subset\pi_{1:{L\left(\varepsilon, C\right)}}\bigl(\boldsymbol{T}^i\bigr)\biggr\}\ge 1-\varepsilon.
\end{equation}
\end{lemma}

\begin{proof}
We begin by proving the first claim~\eqref{order to degree}. Let $i \in [K]$ and let $\gamma(\varepsilon, L)$ be a positive number whose definition we defer to later in~\eqref{eq:gamma_defn}. Recall the definition $\pi_{1:{L}}\bigl(\boldsymbol{T}^i\bigr):=\bigl\{{\pi}_t\left(\boldsymbol{T}^i\right)\mid t \in [L]\bigr\}$ from Definition~\ref{defin: PF} and $V^{\gamma\left(\varepsilon, L\right), \alpha}_{\boldsymbol{G}_n}\bigl(\boldsymbol{T}^i\bigr):= \bigl\{ u \in V\left(\boldsymbol{T}^i\right) \;\big|\; \operatorname{deg}_{\boldsymbol{G}_n}\left(u\right) \ge \gamma\left(\varepsilon, L\right) n^{\frac{1}{2+\alpha}} \bigr\}$ from \eqref{define VC}, we have
\begin{equation}
\label{lemma 1: eqn 1}
     \mathbb{P}\biggl\{\pi_{1:L}\bigl(\boldsymbol{T}^i\bigr)\subset V^{\gamma\left(\varepsilon, L\right), \alpha}_{\boldsymbol{G}_n}\left(\boldsymbol{T}^i\right)\biggr\} \geq 
    \mathbb{P}\Biggl(\,\bigcap^L_{j=1}\biggl\{n^{-\frac{1}{2+\alpha}}\operatorname{deg}_{\boldsymbol{G}_n}\bigl(\pi_{j}\bigl(\boldsymbol{T}^i\bigr)\bigr) >  \gamma\left(\varepsilon, L\right) \biggr\}\Biggr).
\end{equation}
%Next, we apply Lemma~\ref{lemma: G} to link the expression in~\eqref{lemma 1: eqn 1} to the distribution of the random variables $Y_i$, for $i \in [L]$.

Let $(Y_1, Y_2, \ldots )$ be a sequence of random variables specified in Lemma~\ref{lemma: G}.
By~\eqref{lemma 1: eqn 1}, the fact that $\frac{n}{n_i} \leq H$, and by applying Lemma~\ref{lemma: G} in conjunction with the Portmanteau lemma on the set $\prod_{j=1}^{L} \left(H^{\frac{1}{2+\alpha}}\gamma\left(\varepsilon, L\right),\infty\right)
   \times \prod_{j=L+1}^{\infty} \mathbb{R}
\ \subset \mathbb{R}^\mathbb{N}$ which is open with respect to the $\ell_q$ topology for any $q \geq 2$, it follows that
\begin{eqnarray}
   \label{lemma 1 eqn 2}
   \nonumber
   \liminf_{n \rightarrow \infty} \mathbb{P}\left(\pi_{1:L}\left(\boldsymbol{T}^i\right)\subset  V^{\gamma\left(\varepsilon, L\right), \alpha}_{\boldsymbol{G}_n}\left(\boldsymbol{T}^i\right)\right)&\geq& \liminf_{n \rightarrow \infty} \mathbb{P}\Biggl(\,\bigcap^L_{j=1}\biggl\{n^{-\frac{1}{2+\alpha}}\operatorname{deg}_{\boldsymbol{G}_n}\bigl(\pi_{j}\bigl(\boldsymbol{T}^i\bigr)\bigr) >  \gamma\left(\varepsilon, L\right) \biggr\}\Biggr)\\
   \nonumber
   &\geq & \liminf_{n \rightarrow \infty} \mathbb{P}\Biggl(\,\bigcap^L_{j=1}\biggl\{n_i^{-\frac{1}{2+\alpha}}\operatorname{deg}_{\boldsymbol{G}_n}\bigl(\pi_{j}\bigl(\boldsymbol{T}^i\bigr)\bigr) >  H^{\frac{1}{2+\alpha}}\gamma\left(\varepsilon, L\right) \biggr\}\Biggr)\\
   \nonumber
   &\geq &\mathbb{P}\Biggl(\,\bigcap^L_{j=1}\left\{Y_{j, \alpha} > H^{\frac{1}{2+\alpha}}\gamma\left(\varepsilon, L\right) \right\}\Biggr)\\
   \nonumber
%   &=& 1- \mathbb{P}\left(\cup^L_{j=1}\left\{Y_{j, \alpha}\geq H^{\frac{1}{2+\alpha}}\gamma\left(\varepsilon, L\right) \right\}^c\right)\\
%   \nonumber
   &=& 1- \mathbb{P}\Biggl(\, \bigcup^L_{j=1}\left\{Y_{j, \alpha} \leq H^{\frac{1}{2+\alpha}}\gamma\left(\varepsilon, L\right) \right\}\Biggr)\\
   &\ge &1-\sum^L_{j=1} \mathbb{P}\left(Y_{j, \alpha} \leq H^{\frac{1}{2+\alpha}}\gamma\left(\varepsilon, L\right) \right).
\end{eqnarray}

To bound the final term $\sum_{j=1}^L \mathbb{P}\left(Y_{j, \alpha} \leq H^{\frac{1}{2+\alpha}} \gamma\left(\varepsilon, L\right)\right)$.
we note that, by Lemma~\ref{lemma: G}, each random variable $Y_{j, \alpha}$ has a density $q_j(\cdot)$ supported on $[0, \infty)$ with respect to the Lebesgue measure. Thus, the function $t \mapsto \sum_{j=1}^L \mathbb{P}(Y_{j, \alpha} \leq t )$ is continuous on $[0, \infty)$ and hence surjective on its image $[0, L]$. Therefore, for any $\varepsilon > 0$, we choose $\gamma\left(\varepsilon, L\right)$ to be the unique positive number such that  
\begin{equation}
\label{eq:gamma_defn}   
\sum_{j=1}^L \mathbb{P}\left(Y_{j, \alpha} \leq H^{\frac{1}{2+\alpha}}\gamma\left(\varepsilon,L\right)\right) = \varepsilon.  
\end{equation}
Putting this together with inequality~\eqref{lemma 1 eqn 2}, we conclude that
\begin{equation*}
  \liminf_{n \rightarrow \infty} \mathbb{P}\left(\pi_{1:L}\bigl(\boldsymbol{T}^i\bigr)\subset  V^{\gamma\left(\varepsilon, L\right), \alpha}_{\boldsymbol{G}_n}\bigl(\boldsymbol{T}^i\bigr)\right)\geq  1-\sum^L_{j=1} \mathbb{P}\left(Y_{j, \alpha} \leq H^{\frac{1}{2+\alpha}}\gamma\left(\varepsilon, L\right) \right) = 1-\varepsilon.
\end{equation*}
This completes the proof of the first claim~\eqref{order to degree} in the lemma.
We now proceed to prove the second claim~\eqref{degree to order}. We follow a similar argument. Let $L(\varepsilon, C)$ be a positive integer whose definition we defer to later in~\eqref{eq:L_defn}. Recall the definition $\pi_{1:L\left(\varepsilon, C\right)}\left(\boldsymbol{T}^i\right):=\bigl\{{\pi}_t\left(\boldsymbol{T}^i\right)\mid t \in [L(\varepsilon, C)] \bigr\}$ from Definition~\ref{defin: PF} and $V^{C,\alpha}_{\boldsymbol{G}_n}\left(\boldsymbol{T}^i\right):= \bigl\{ u \in V\left(\boldsymbol{T}^i\right) \mid \operatorname{deg}_{\boldsymbol{G}_n}\left(u\right) \ge Cn^{\frac{1}{2+\alpha}} \bigr\}$ from \eqref{define VC}, we have
\begin{eqnarray*}
    \mathbb{P}\left( V^{C,\alpha}_{\boldsymbol{G}_n}\bigl(\boldsymbol{T}^i\bigr)\subset\pi_{1:{L\left(\varepsilon, C\right)}}\bigl(\boldsymbol{T}^i\bigr)\right)&=& \mathbb{P}\left( \pi_{\left({L\left(\varepsilon, C\right)+1}\right):{n_i}}\bigl(\boldsymbol{T}^i\bigr)\cap V^{C,\alpha}_{\boldsymbol{G}_n}\bigl(\boldsymbol{T}^i\bigr)=\emptyset\right)\\
    &=& 
    \mathbb{P}\left(\, \bigcap^{n_i}_{j=L\left(\varepsilon, C\right)+1}\left\{n^{-\frac{1}{2+\alpha}}\operatorname{deg}_{\boldsymbol{G}_n}\bigl(\pi_{j}\bigl(\boldsymbol{T}^i\bigr)\bigr)<C \right\}\right)\\
    &\ge & \mathbb{P}\left(\, \bigcap^{n_i}_{j=L\left(\varepsilon, C\right)+1}\left\{n_i^{-\frac{1}{2+\alpha}}\operatorname{deg}_{\boldsymbol{G}_n}\bigl(\pi_{j}\bigl(\boldsymbol{T}^i\bigr)\bigr)< C \right\}\right).
\end{eqnarray*}
We fix $q > \left(2+\alpha\right) \vee \frac{1}{\delta}$, and again apply Lemma~\ref{lemma: G} in conjunction with the Portmanteau lemma, noting that the set $\prod_{j=1}^{L(\varepsilon,C)} \mathbb{R} 
   \times \prod_{j=L(\varepsilon,C)+1}^{\infty} (-\infty,C)
\ \subset \ \mathbb{R}^{\mathbb{N}}$ is open with respect to the $\ell_q$ topology, to obtain
\begin{eqnarray}
\label{lemma 1 eqn 3}
\nonumber
    \liminf_{n \rightarrow \infty} \mathbb{P}\left( V^{C,\alpha}_{\boldsymbol{G}_n}\bigl(\boldsymbol{T}^i\bigr)\subset\pi_{1:{L\left(\varepsilon, C\right)}}\bigl(\boldsymbol{T}^i\bigr)\right)&\ge &\liminf_{n \rightarrow \infty} \mathbb{P}\left(\,\bigcap^{n_i}_{j=L\left(\varepsilon, C\right)+1}\left\{n_i^{-\frac{1}{2+\alpha}}\operatorname{deg}_{\boldsymbol{G}_n}\bigl(\pi_{j}\bigl(\boldsymbol{T}^i\bigr)\bigr)< C \right\}\right)\\
    \nonumber
    &\geq& \mathbb{P}\left(\, \cap^\infty_{j=L\left(\varepsilon, C\right)+1}\left\{Y_{j, \alpha}< C \right\}\right)\\
    % \nonumber
    % &=& 1- \mathbb{P}\left(\cup^\infty_{j=L\left(\varepsilon, C\right)+1}\left\{Y_{j, \alpha}< C \right\}^c\right)\\
    \nonumber
    &=& 1- \mathbb{P}\left( \cup^\infty_{j=L\left(\varepsilon, C\right)+1}\left\{Y_{j, \alpha}\geq  C \right\}\right)\\
    \nonumber
    &=& 1- \mathbb{P}\left(\cup^\infty_{j=L\left(\varepsilon, C\right)+1}\left\{Y^q_j\geq  C^q \right\}\right)\\
    \nonumber
    &\ge& 1-\mathbb{P}\left(\sum_{j=L\left(\varepsilon, C\right)+1}^\infty Y^q_j\geq  C^q \right)\\
    &\ge &1- \frac{\sum_{j=L\left(\varepsilon, C\right)+1}^\infty \mathbb{E} Y^q_j}{C^q},
\end{eqnarray}
where the last inequality in \eqref{lemma 1 eqn 3} comes from Markov inequality.
By Lemma~\ref{lemma: G}, we have that $\sum_{j=1}^\infty \mathbb{E}\bigl(Y^q_j\bigr)<\infty$. Therefore, we choose $L\left(\varepsilon,C\right)$ to be the smallest positive integer such that 
\begin{align}
\sum_{j=L(\varepsilon, C)}^\infty \mathbb{E} Y^q_j \le C^q \varepsilon. \label{eq:L_defn}
\end{align}
Putting this together with \eqref{lemma 1 eqn 3}, we have that
\begin{equation*}
     \liminf_{n \rightarrow \infty} \mathbb{P}\left( V^{C,\alpha}_{\boldsymbol{G}_n}\bigl(\boldsymbol{T}^i\bigr)\subset\pi_{1:{L\left(\varepsilon, C\right)}}\bigl(\boldsymbol{T}^i\bigr)\right)\geq 1- \frac{\sum_{j=L\left(\varepsilon, C\right)+1}^\infty \mathbb{E} Y^q_j}{C^q}\geq 1-\varepsilon.
\end{equation*}
The lemma follows as desired.
\end{proof}

\begin{lemma}
\label{lemma: noise}
Let $\boldsymbol{G}_n \sim \mathrm{PF}(\alpha, \theta, \ell, \pi)$, and suppose Assumptions~\ref{assumption: bound} and \ref{assumption: sparsity} hold. Then for any fixed $C > 0$, the following result holds:
\begin{equation*}
     \lim_{n \rightarrow \infty} \mathbb{P}\left(\bigl|E_{\boldsymbol{R}_n}\bigl(V^{C,\alpha}\left(\boldsymbol{G}_n\right)\bigr)\bigr|=0\right)=1.
\end{equation*}
\end{lemma}
\begin{proof}

% Since $V^{C,\alpha}\left(\boldsymbol{G}_{n}\right)=\cup_{i=1}^K V^{C,\alpha}_{\boldsymbol{G}_n}\left(\boldsymbol{T}^i\right)$, it follows that 
% %
% \begin{equation}
% \label{lemma 2 eqn 1}
% \mathbb{P}\left(\left|E_{\boldsymbol{R}_n}\left(V^{C,\alpha}\left(\boldsymbol{G}_n\right)\right)\right|=0\right)=\mathbb{P}\left(\left|E_{\boldsymbol{R}_n}\left(\cup_{i=1}^K V^{C,\alpha}_{\boldsymbol{G}_n}\left(\boldsymbol{T}^i\right)\right)\right|=0\right) 
% \end{equation}
% %
Fix $\varepsilon > 0$ arbitrarily. Suppose $n \geq H L(\varepsilon, C)$ (where $L(\varepsilon, C)$ is defined in Lemma~\ref{lemma order--degree}) so that $\min_{i \in [K]} n_i \geq L(\varepsilon, C)$ and the set $\pi_{1:L(\varepsilon, C)}(\boldsymbol{T}^i)$ is well-defined. Let $\Omega_i=\bigl\{ V^{C,\alpha}_{\boldsymbol{G}_n}\left(\boldsymbol{T}^i\right)\subset\pi_{1:{L\left(\varepsilon, C\right)}}\left(\boldsymbol{T}^i\right)\bigr\}$ and $ \Omega:=\cap_{i=1}^K \Omega_i$. Using the fact that $V^{C,\alpha}\left(\boldsymbol{G}_{n}\right)=\cup_{i=1}^K V^{C,\alpha}_{\boldsymbol{G}_n}\left(\boldsymbol{T}^i\right)$, we have
\begin{eqnarray} 
\label{lemma 2 eqn 2}
&&\mathbb{P}\left(\bigl|E_{\boldsymbol{R}_n}\bigl(V^{C,\alpha}\left(\boldsymbol{G}_n\right)\bigr)\bigr|=0\right) \nonumber \\
&=&\mathbb{P}\left(\biggl|E_{\boldsymbol{R}_n}\left(\cup_{i=1}^K V^{C,\alpha}_{\boldsymbol{G}_n}\bigl(\boldsymbol{T}^i\bigr)\right)\biggr|=0\right) \nonumber \\
\nonumber
&\geq & \mathbb{P}\left(\left\{\left|E_{\boldsymbol{R}_n}\left(\cup_{i=1}^K V^{C,\alpha}_{\boldsymbol{G}_n}\bigl(\boldsymbol{T}^i\bigr)\right)\right|=0 \right\}\cap \Omega\right)\\
\nonumber
&\geq& \mathbb{P}\left(\left\{\left|E_{\boldsymbol{R}_n}\left(\cup_{i=1}^K \pi_{1:L\left(\varepsilon,C\right)}\bigl(\boldsymbol{T}^i\bigr)\right)\right|=0\right\}\cap \Omega\right)\\
\nonumber
&=& \mathbb{P}\left(\left|E_{\boldsymbol{R}_n}\left(\cup_{i=1}^K \pi_{1:L\left(\varepsilon,C\right)}\bigl(\boldsymbol{T}^i\bigr)\right)\right|=0\right)-\mathbb{P}\left(\left\{\left|E_{\boldsymbol{R}_n}\left(\cup_{i=1}^K \pi_{1:L\left(\varepsilon,C\right)}\bigl(\boldsymbol{T}^i\right)\bigr)\right|=0\right\}\cap \Omega^c\right)\\
&\geq & \mathbb{P}\left(\left|E_{\boldsymbol{R}_n}\left(\cup_{i=1}^K \pi_{1:L\left(\varepsilon,C\right)}\bigl(\boldsymbol{T}^i\bigr)\right)\right|=0\right)-\mathbb{P}\left(\Omega^c\right).
\end{eqnarray}

We begin by bounding the first term of~\eqref{lemma 2 eqn 2}. Note that the set $\cup_{i=1}^K \pi_{1:L\left(\varepsilon,C\right)}\left(\boldsymbol{T}^i\right)$
is deterministic and not dependent on $\boldsymbol{R}_n$ by Definition~\ref{defin: PF}. Therefore, the number of random edges of $\boldsymbol{R}_n$ with both endpoints in this set satisfies
\begin{equation*}
    \left| E_{\boldsymbol{R}_n} \left( \cup_{i=1}^K \pi_{1:L\left(\varepsilon,C\right)}\bigl(\boldsymbol{T}^i\bigr) \right) \right|
\sim \operatorname{Bin} \left( \frac{L(\varepsilon,C)K \left\{ L(\varepsilon,C)K - 1 \right\}}{2}, \theta \right),
\end{equation*}
since each pair of nodes in the selected set contributes an edge independently with probability $\theta$. Therefore
\begin{eqnarray}
\label{lemma 2 eqn 3}
\nonumber
\mathbb{P}\left(\left|E_{\boldsymbol{R}_n}\left(\cup_{i=1}^K \pi_{1:L\left(\varepsilon,C\right)}\bigl(\boldsymbol{T}^i\bigr)\right)\right|=0\right)&=& 1- \mathbb{P}\left(\left|E_{\boldsymbol{R}_n}\left(\cup_{i=1}^K \pi_{1:L\left(\varepsilon,C\right)}\bigl(\boldsymbol{T}^i\bigr)\right)\right|\ge 1\right)\\
\nonumber
&\geq & 1- \mathbb{E}\left(\left|E_{\boldsymbol{R}_n}\left(\cup_{i=1}^K \pi_{1:L\left(\varepsilon,C\right)}\bigl(\boldsymbol{T}^i\bigr)\right)\right|\right)\\
\nonumber
&=& 1- \frac{L(\varepsilon,C)K\left(L(\varepsilon,C)K-1\right)\theta}{2}\\
&=&1- O\bigl(n^{-\frac{1+\alpha}{2+\alpha}-\delta}\bigr).
\end{eqnarray}
The last equality follows from Assumption~\ref{assumption: sparsity}.

Now we bound the second term $\mathbb{P}\left(\Omega^c\right)$ of~\eqref{lemma 2 eqn 2}. 
By Lemma~\ref{lemma order--degree}, we have
\begin{equation}
\label{lemma 2 eqn 4}
\mathbb{P}\left(\Omega^c\right)=\mathbb{P}\left(\cup_{i=1}^K \Omega_i^c\right)\le \sum_{i=1}^K \mathbb{P}\left(\Omega_i^c\right)= \sum_{i=1}^K \mathbb{P}\left(\left\{ V^{C,\alpha}_{\boldsymbol{G}_n}\bigl(\boldsymbol{T}^i\bigr)\subset\pi_{1:{L\left(\varepsilon, C\right)}}\bigl(\boldsymbol{T}^i\bigr)\right\}^c\right) \leq K\varepsilon + o(1).
\end{equation}
Combining equations~\eqref{lemma 2 eqn 2}–\eqref{lemma 2 eqn 4}, we have shown that
\begin{equation*}  \mathbb{P}\left(\left|E_{\boldsymbol{R}_n}\left(V^{C,\alpha}\left(\boldsymbol{G}_n\right)\right)\right|=0\right)\geq 1-K\varepsilon-o(1).
\end{equation*}
Since $\varepsilon>0$ can be arbitrary small, we have 
\begin{equation*}
     \lim_{n \rightarrow \infty} \mathbb{P}\left(\left|E_{\boldsymbol{R}_n}\left(V^{C,\alpha}\left(\boldsymbol{G}_n\right)\right)\right|=0\right)=1.
\end{equation*}
\end{proof}

\begin{lemma}
\label{lemma early hist containment}

Let $L \in \mathbb{N}$, $\varepsilon > 0$ and $c\in (0,1]$. Let $\bm{G}_n \sim \mathrm{PF}(\alpha, \theta, \ell, \pi)$ and suppose Assumptions~\ref{assumption: bound} and \ref{assumption: sparsity} hold. Define
\begin{align*}
V_1, \ldots, V_M &= \mathcal{C}(\bm{G}_n \cap V^{\tilde{\gamma}(\varepsilon, L,c), \alpha}(\bm{G}_n))\\
\mathcal{K} &:= \left\{ i \in [M] \,: \, \operatorname{deg}_{\boldsymbol{G}_n}(V_i) \ge \gamma\biggl( \frac{\varepsilon}{3H}, 1 \biggr) n^{\frac{1}{2+\alpha}} \right\}.
\end{align*}

Then, with probability at least $1 - \varepsilon + \eta_n$ where $\eta_n = o(1)$ and depends only on $\varepsilon, L, \alpha, H, c, \delta$.
\begin{align}
|\mathcal{K}| &= K,\quad  \text{ and } \nonumber
\exists \text{ bijection } \sigma \,:\, [K]  \rightarrow \mathcal{K} \text{ s.t.}\\
\quad \pi_1(\bm{T}^i) &\subset V^{\gamma( \frac{\varepsilon}{3H}, 1 ), \alpha}_{\bm{G}_n}(\bm{T}^i)
\subset V^{c\gamma(\frac{\varepsilon}{3H}, 1 ), \alpha}_{\bm{G}_n}(\bm{T}^i)
\subset \pi_{1:(L \vee \tilde{L}(\varepsilon,c))}( \bm{T}^{i}) \subset V_{\sigma(i)},\, \text{ for all $i \in [K]$}. \label{prop 2 eqn 1}
\end{align}

\end{lemma}

\begin{proof}
Define
\[
\tilde{\gamma}(\varepsilon, L, c) := \gamma\biggl( \frac{\varepsilon}{3H},\,\, \tilde{L}(\varepsilon, c) \vee L \biggr), \quad \text{ where } \tilde{L}(\varepsilon, c) := L\biggl( \frac{\varepsilon}{3H}, c \gamma\biggl( \frac{\varepsilon}{3H}, 1 \biggr) \biggr).
\]
Let us define a ``good" event $
    \Omega := \Omega_1\cap\Omega_2\cap \Omega_3$ where
\begin{eqnarray*}
   \Omega_1&:=& \left\{\pi_{1}(\boldsymbol{T}^i)\subset V^{\gamma( \frac{\varepsilon}{3H}, 1 ), \alpha}_{\boldsymbol{G}_n}(\boldsymbol{T}^i), \text{for all $i \in \left[K\right]$}\right\} \\
   \Omega_2&:=& \left\{V^{\gamma( \frac{\varepsilon}{3H}, 1 ), \alpha}_{\boldsymbol{G}_n}(\boldsymbol{T}^i) \subseteq V^{ c\gamma( \frac{\varepsilon}{3H}, 1 ), \alpha}_{\boldsymbol{G}_n}(\boldsymbol{T}^i)\subset \pi_{1:\tilde{L}(\varepsilon, c)}(\boldsymbol{T}^i), \text{for all $i \in \left[K\right]$}\right\}\\
   \Omega_3&:=& \left\{ \exists \text{ injective } \sigma \,:\, [K] \rightarrow [M] \text{ such that }\, \pi_{1:{\tilde{L}(\varepsilon, c)} \vee L }(\boldsymbol{T}^{i})\subset V_{\sigma(i)},  \text{for all $i \in \left[K\right]$}\right\}.
\end{eqnarray*}

By definition, on $\Omega \subset \Omega_3$, it holds that $\pi_{1:L}(\bm{T}^i) \subset V_{\sigma(i)}$ for some injective $\sigma \,:\, [K] \rightarrow [M]$. We now show that $| \mathcal{K} | = K$ on $\Omega$ and that $\sigma$ is a bijection mapping $[K]$ to $\mathcal{K}$. To that end, suppose we are on $\Omega$. Let $\sigma \,:\, [K] \rightarrow [M]$ be the injective mapping such that $\pi_{1:{\tilde{L}(\varepsilon, c)} \vee L }(\boldsymbol{T}^i)\subset V_{\sigma(i)}$ for all $i \in [K]$ whose existence is guaranteed on $\Omega_3$. Then, $\pi_1(\bm{T}^i) \in V_{\sigma(i)}$ and $\pi_1(\bm{T}^i) \in V_{\bm{G}_n}^{{\gamma( \frac{\varepsilon}{3H}, 1 )}, \alpha}(\bm{T}^i)$; hence, for all $i \in [K]$,
\[
\operatorname{deg}_{\bm{G}_n}(V_{\sigma(i)}) \geq \operatorname{deg}_{\bm{G}_n}(\pi_1(\bm{T}^i)) \geq {\gamma( \frac{\varepsilon}{3H}, 1 )} n^{\frac{1}{2+\alpha}}.
\]
It thus holds that $\sigma([K]) \subset \mathcal{K}$ and so $|\mathcal{K}| \geq K$. To show that $|\mathcal{K}| \leq K$, first note that, 
by the fact that $V(\bm{G}_n) = \cup_{i=1}^K V(\bm{T}^i)$ and the definition of $\Omega_2$ and $\Omega_3$, we have
\[
V^{{\gamma( \frac{\varepsilon}{3H}, 1 )}, \alpha}(\bm{G}_n) = \bigcup_{i=1}^K V_{\bm{G}_n}^{{\gamma( \frac{\varepsilon}{3H}, 1 )}, \alpha}(\bm{T}^i) \subset \bigcup_{i=1}^K V_{\bm{G}_n}^{{c\gamma( \frac{\varepsilon}{3H}, 1 )}, \alpha}(\bm{T}^i) \subset \bigcup_{i=1}^K \pi_{1:{\tilde{L}(\varepsilon, c)} }(\bm{T}^i) \subset \bigcup_{i=1}^K V_{\sigma(i)}.
\]
Therefore, if $i \notin \sigma([K])$, then $V_i \cap V^{{\gamma( \frac{\varepsilon}{3H}, 1 )}, \alpha}(\bm{G}_n) = \emptyset$ and hence $i \notin \mathcal{K}$ as well. We thus have that $\mathcal{K} \subset \sigma([K])$. Therefore, we conclude that $|\mathcal{K}| = K$ and that $\sigma$ is a bijection from $[K]$ to $\mathcal{K}$.

It remains to bound the probability of the event $\Omega$. We have by Lemma~\ref{lemma order--degree} that for each $i \in [K]$, there exists $\eta^i_n = o(1)$ (dependent only on $\varepsilon, L, H$) such that
\[
\mathbb{P}\bigl\{ \pi_1(\bm{T}^i) \in V_{\bm{G}_n}^{{\gamma( \frac{\varepsilon}{3H}, 1 )}, \alpha}(\bm{T}^i) \bigr\} \geq 1 - \frac{\varepsilon}{3H} + \eta^i_n.
\]
By a union bound and writing $\eta'_n = \sum_{i=1}^K \eta^i_n$, we have
\[
\mathbb{P}\bigl\{ \forall i \in [K], \, \pi_1(\bm{T}^i) \in V_{\bm{G}_n}^{{\gamma( \frac{\varepsilon}{3H}, 1 )}, \alpha}(\bm{T}^i) \bigr\} \geq 1 - \frac{K\varepsilon}{3H} + \eta'_n \geq 1 - \frac{\varepsilon}{3} + \eta'_n. 
\]
Thus, $\mathbb{P}(\Omega_1^c) \leq \frac{\varepsilon}{3} - \eta'_n$. Recall now that $\tilde{L}(\varepsilon, c) = L\bigl(\frac{\varepsilon}{3H}, {c\gamma\bigl( \frac{\varepsilon}{3H}, 1 \bigr)}\bigr)$; via a similar argument, by using the second part of Lemma~\ref{lemma order--degree} and a union bound, we also obtain that $\mathbb{P}( \Omega_2^c) \leq \frac{\varepsilon}{3} - \eta''_n$ where $\eta''_n \rightarrow 0$ (dependent only on $\varepsilon, L, H$) as $n \rightarrow \infty$. 
% {\color{purple}
% Now we apply Proposition~\ref{prop early hist containment} with $C_1 = \tilde{\gamma}(\varepsilon, L) = \gamma(\frac{\varepsilon}{3H}, \max(\tilde{L}(\varepsilon), L))$ to obtain that there exists $\eta'''_n$ (dependent on $\varepsilon, L$) that is $o(1)$ as $n \rightarrow \infty$ such that
% \[
% \mathbb{P}( \Omega^c_3 ) \leq \frac{\varepsilon}{3} - \eta'''_n. 
% \]
% }
% {

To upper bound $\mathbb{P}(\Omega_3^c)$, we first define $\mathcal{A}:=\mathcal{A}_1 \cap \mathcal{A}_2$, where
\begin{eqnarray*}
  \mathcal{A}_1&:=&\bigl\{ \pi_{1:{\tilde{L}(\varepsilon, c)} \vee L }(\boldsymbol{T}^i)\subset V^{{{\tilde{\gamma}(\varepsilon, L,c)}},\alpha}_{\boldsymbol{G}_n}(\boldsymbol{T}^i), \text{for all $i \in \left[K\right]$}\bigr\} \\
  \mathcal{A}_2 &:=& \left\{\left|E_{\boldsymbol{R}_n}\bigl(V^{{{\tilde{\gamma}(\varepsilon, L,c)}}, \alpha}(\boldsymbol{G}_n) \bigr)\right|=0\right\}.
\end{eqnarray*}
We now claim that $\mathcal{A}\subset \Omega_3$.
We verify the claim by showing that on $\mathcal{A}$, there exist $\sigma: [K]\rightarrow [M]$ such that $\pi_{1:({\tilde{L}(\varepsilon, c)} \vee L)}\left(\bm{T}^i\right)\subseteq V_{\sigma(i)}$ for all $i\in [K]$.

We note that since $V(\bm{G}_n) = \sqcup_{i=1}^K V(\bm{T}^i)$, we have $V^{{{\tilde{\gamma}(\varepsilon, L,c)}},\alpha}(\bm{G}_n) = \sqcup_{i=1}^K V_{\bm{G}_n}^{{{\tilde{\gamma}(\varepsilon, L,c)}}, \alpha}(\bm{T}^i)$. 
Moreover, on $\mathcal{A}_2$, for any $i \neq j$, the two sets $V^{{{\tilde{\gamma}(\varepsilon, L,c)}}, \alpha}_{\boldsymbol{G}_n}(\boldsymbol{T}^i)$ and $V^{{{\tilde{\gamma}(\varepsilon, L,c)}}, \alpha}_{\boldsymbol{G}_n}(\boldsymbol{T}^j)$ are disconnected in $\bm{G}_n$; to see this, note that
\[
E_{\boldsymbol{G}_n}\bigl(V^{{{\tilde{\gamma}(\varepsilon, L,c)}}, \alpha}_{\boldsymbol{G}_n}(\boldsymbol{T}^i),V^{{{\tilde{\gamma}(\varepsilon, L,c)}}, \alpha}_{\boldsymbol{G}_n}(\boldsymbol{T}^j)\bigr)=E_{\boldsymbol{R}_n}\bigl(V^{{{\tilde{\gamma}(\varepsilon, L,c)}}, \alpha}_{\boldsymbol{G}_n}(\boldsymbol{T}^i),V^{{{\tilde{\gamma}(\varepsilon, L,c)}}, \alpha}_{\boldsymbol{G}_n}(\boldsymbol{T}^j)\bigr)\subset E_{\boldsymbol{R}_n}\bigl(V^{{{\tilde{\gamma}(\varepsilon, L,c)}}, \alpha}_{\boldsymbol{G}_n}\left(\boldsymbol{G}_{n}\bigr)\right)=\emptyset,
\]
where the first equality follows because $\bm{T}^i$ and $\bm{T}^j$ are disconnected and the second equality follows from the definition of $\mathcal{A}_2$. 

Therefore, each $V^{{{\tilde{\gamma}(\varepsilon, L,c)}}, \alpha}_{\bm{G}_n}(\bm{T}^i)$ must correspond to a subset of $V_1, \ldots, V_M$; more precisely, there exists a mapping $\kappa \,:\, [M] \rightarrow [K]$ such that $V^{{{\tilde{\gamma}(\varepsilon, L,c)}}, \alpha}_{\bm{G}_n}(\bm{T}^i) = \sqcup_{j \in \kappa^{-1}(i)} V_j$. Since $\pi_{1:\tilde{L}(\varepsilon, c) \vee L }(\bm{T}^i)$ is connected in $\bm{T}^i$ and hence in $\bm{G}_n$, on event $\mathcal{A}_1$, it must be that there exists $j^*_i \in \kappa^{-1}(i)$ such that $\pi_{1:{\tilde{L}(\varepsilon, c)} \vee L }(\bm{T}^i) \subset V_{j^*_i}$. We may then define $\sigma(i) = j^*_i$ so that $\pi_{1:{\tilde{L}(\varepsilon, c)} \vee L }(\bm{T}^i) \subset V_{\sigma(i)}$; note that $\sigma$ is injective by definition. 
The claim $\mathcal{A}\subset \Omega_3$ follows as desired.

It remains only then to upper bound the probability of $\mathbb{P}(\mathcal{A}^c)$. 
Recalling that ${{\tilde{\gamma}(\varepsilon, L,c)}} = \gamma(\frac{\varepsilon}{3H}, {\tilde{L}(\varepsilon, c)} \vee L )$, we have by Lemma~\ref{lemma order--degree} that, for each $i \in [K]$, there exists $\mu^i_n$ (dependent on $\varepsilon, L$) that is $o(1)$ as $n \rightarrow \infty$ such that
\[
\mathbb{P}\bigl\{ \pi_{1:{\tilde{L}(\varepsilon, c)} \vee L }(\bm{T}^i) \subset V_{\bm{G}_n}^{{{\tilde{\gamma}(\varepsilon, L,c)}}, \alpha}(\bm{T}^i) \bigr\} \geq 1 - \frac{\varepsilon}{3H} + \mu^i_n.
\]
By an application of union bound, by defining $\mu'_n = \sum_{i=1}^K \mu^i_n$, and using the fact that $H \geq K$, we have that
\[
\mathbb{P}(\mathcal{A}_1)= \mathbb{P}\bigl\{\forall i \in [K],\,  \pi_{1:{\tilde{L}(\varepsilon, c)} \vee L }(\bm{T}^i) \subset V_{\bm{G}_n}^{{{\tilde{\gamma}(\varepsilon, L,c)}}, \alpha}(\bm{T}^i) \bigr\} \geq 1 - \frac{K\varepsilon}{3H} + \mu'_n \geq 1 -\frac{\varepsilon}{3} + \mu'_n.
\]

Now, by Lemma~\ref{lemma: noise}, there exists $\mu''_n$ (dependent only on $\varepsilon, L, H$) that is $o(1)$ as $n \rightarrow \infty$ such that $\mathbb{P}(\mathcal{A}_2) \geq 1 + \mu''_n$. Writing $\mu_n = \mu'_n + \mu''_n$ and using another union bound, we then have that $\mathbb{P}(\Omega_3)\geq \mathbb{P}(\mathcal{A}) \geq 1 - \frac{\varepsilon}{3} + \mu_n$ as desired.

By another application of the union bound, we have that $\mathbb{P}(\Omega) \geq 1 - \varepsilon + \eta_n$ where 
$\eta_n = \eta'_n + \eta''_n + \mu_n$. The proposition thus follows as desired. 

\end{proof}

\subsection{Proof of Corollary~\ref{cor: first L}}
\begin{proof}

Define the event $\Omega = \Omega_1 \cap \Omega_2$ where
\begin{eqnarray*}
 \Omega_1 &:=&\left\{d_{\cup_{i=1}^K \pi_{1:L(\varepsilon, \tilde{C})}\left(\bm{T}^i\right)}(\hat{\ell}, \ell) = 0\right\},\\  
 \Omega_2 &:= &\left\{V^{\tilde{C}, \alpha}\left(\bm{G}_n\right) \subset  \cup_{i=1}^K \pi_{1:L(K^{-1}\varepsilon, \tilde{C})}\bigl(\bm{T}^i\bigr)\right\}.
 \end{eqnarray*}

It is clear that on event $\Omega$, the conclusion of the Corollary~\eqref{cor: first L eqn 1} holds; indeed, by Definition~\ref{defin mismatch},
\begin{eqnarray*}
d_{V^{\tilde{C}, \alpha}\left(\bm{G}_n\right)}(\hat{\ell}, \ell)
&:=& 
\min_{\sigma \in S_K} \sum_{u\in V^{\tilde{C}, \alpha}\left(\bm{G}_n\right) }\mathbbm{1}\left\{{\ell} \left(u\right)\neq \sigma\circ \hat{\ell} \left(u\right)\right\}
\le \min_{\sigma \in S_K} \sum_{u\in \cup_{i=1}^K \pi_{1:L(K^{-1}\varepsilon, \tilde{C})}\left(\bm{T}^i\right) }\mathbbm{1}\left\{{\ell} \left(u\right)\neq \sigma\circ \hat{\ell} \left(u\right)\right\}\\
& = & d_{\cup_{i=1}^K \pi_{1:L(K^{-1}\varepsilon, \tilde{C})}\left(\bm{T}^i\right)}(\hat{\ell}, \ell)=0.
\end{eqnarray*}

We thus only need to lower bound the probability of $\Omega$. By Theorem~\ref{thm: first L}, we have $\mathbb{P}(\Omega_1^c)\le \varepsilon + \eta_n$ where $\eta_n = o(1)$ and depends only on $\varepsilon, \tilde{C}, Q, \alpha, H, \delta$.
By \eqref{degree to order} in Lemma~\ref{lemma order--degree}, we have 
\begin{equation*}
V^{\tilde{C}, \alpha}\left(\bm{G}_n\right):= \cup_{i=1}^K V_{\bm G_n}^{\tilde{C}, \alpha}\left(\bm{T}^i\right)\subset \cup_{i=1}^K \pi_{1:L\left(K^{-1}\varepsilon, \tilde{C}\right)}\left(\bm{T}^i\right).    
\end{equation*}
with probability at least $1-\varepsilon-\eta'_n$ where $\eta'_n = o(1)$ and depends only on $\varepsilon, \tilde{C}, Q, \alpha, H$. We thus have that $\mathbb{P}(\Omega_2^c) \leq \varepsilon + \eta'_n$. Therefore, we have
\begin{equation*}
\mathbb{P}\left(\Omega\right) \geq 1 - 2 \varepsilon - \eta_n - \eta'_n.
\end{equation*}

The Corollary then follow as desired.

\end{proof}

\subsection{Proof of Theorem~\ref{thm: layer 1/2}}
\begin{proof}

Let $V^2_1, \ldots, V^2_K$ be the output of Algorithm~\ref{alg: step I} (and thus the input of Algorithm~\ref{alg: distance recovery}). Recall that $\tau = \gamma(\frac{\varepsilon}{18H}, 1) n^{\frac{1}{2+\alpha}}$. 

Let us define
\begin{align*}
V_1^1, \ldots, V_M^1 &= \mathcal{C}\bigl( \bm{G}_n \cap V^{\tilde{\gamma}(\frac{\varepsilon}{6},\, Q, c_0), \alpha}(\bm{G}_n)\bigr) \,\, \text{ and} \\
\mathcal{K} &= \biggl\{ i \in [M] \,: \, \text{deg}_{\bm{G}_n}(V_i^1) \geq \gamma\bigl(\frac{\varepsilon}{18H}, 1\bigr) n^{\frac{1}{2+\alpha}} \biggr\}.
\end{align*}

It will be convenient to also define, for $j \in [K]$,
\begin{align*}
\tilde{V}^2_j := V^2_j \cap V^{\tau'}(\bm{G}_n) \subset V_j^2 \cap V^{c_0 \gamma(\frac{\varepsilon}{18H}, 1), \alpha}(\bm{G}_n).
\end{align*}

To prove the first claim of the Theorem, i.e.~\eqref{thm: layer 1/2 eqn 1}, we let $C_2$ be a positive constant whose definition will be specified later and also define the "good" event $\Omega=\Omega_1\cap \Omega_2$ where
\begin{align}
\Omega_1:=\Bigl\{& {|\mathcal{K}|} = K \text{ and } \exists \sigma \in S_K \text{ s.t. } \label{eq:omega1_defn}\\
&\pi_1(\bm{T}^{i}) \subset V^{{\gamma(\frac{\varepsilon}{18H}, 1)}, \alpha}_{\bm{G}_n}(\bm{T}^i) \subset V^{{c_0\gamma(\frac{\varepsilon}{18H}, 1)}, \alpha}_{\bm{G}_n}(\bm{T}^i) \subset \pi_{1:(Q \vee {\tilde{L}(\varepsilon/6, c_0)})}(\boldsymbol{T}^{i})\subset {V}^1_{\sigma(i)},\quad \forall i \in [K]\Bigr\} \nonumber \\ 
\Omega_2:= \Biggl\{&\frac{\Bigr|E_{\boldsymbol{R}_n}\bigl(V^{c_0\gamma\left(\frac{\varepsilon}{18 H}, 1\right), \alpha}(\boldsymbol{G}_n), \,\mathcal{L}_1\left(\boldsymbol{F}_n\right)\bigr)\Bigr|}{\left| \mathcal{L}_1\left(\boldsymbol{F}_n\right)\right|}\le \frac{C_2 n^{-\frac{1+\alpha}{2+\alpha}-\delta}}{2}\Biggr\}. \nonumber
\end{align}

\textbf{Part A:} We prove the first statement of the Theorem, i.e.~\eqref{thm: layer 1/2 eqn 1}. Let us first show that~\eqref{thm: layer 1/2 eqn 1} holds on the event $\Omega$. We proceed in 3 steps. Let $\sigma \in S_K$ be the permutation whose existence is guaranteed on $\Omega_1$.

\textbf{Step A1:} We claim that on $\Omega$, there exists $\sigma'' \in S_K$ such that
\begin{align}
\forall i \in [K], \quad \pi_1(\bm{T}^{i}) \subset V_{\sigma''(i)}^2 \cap V^{\tau'}(\bm{G}_n) = \tilde{V}_{\sigma''(i)}^2\label{eq:pi1_containment1}
\end{align}
To see this, we work on $\Omega$ and note that since $|\mathcal{K}| = K$ and $Q \leq |\pi_{1:(Q \vee {\tilde{L}(\varepsilon/6, c_0)})}(\bm{T}^{i}) | \leq |V_{\sigma(i)}^1|$ for all $i \in [K]$, we may apply Corollary~\ref{cor: simple containment} to conclude that there exists $\sigma' \in S_K$ such that 
\begin{align}
\forall i \in [K], \quad V_i^1 \subset V_{\sigma'(i)}^2. \label{eq:v1v2_containment1}
\end{align}
Since $\tau' \leq \gamma(\frac{\varepsilon}{18H}, 1) n^{\frac{1}{2+\alpha}}$, 
\begin{align*}
\pi_1(\bm{T}^i) &\subset V_{\bm{G}_n}^{\gamma(\frac{\varepsilon}{18H}, 1), \alpha}(\bm{T}^i) = V_{\bm{G}_n}^{\gamma(\frac{\varepsilon}{18H}, 1), \alpha}(\bm{T}^i) \cap V_{\sigma(i)}^1 \\
&\subset V_{\bm{G}_n}^{\gamma(\frac{\varepsilon}{18H}, 1), \alpha}(\bm{T}^i) \cap V_{\sigma'(\sigma(i))}^2 \subset V^{\tau'}(\bm{G}_n) \cap V_{\sigma'(\sigma(i))}^2. 
\end{align*}
We conclude then that~\eqref{eq:pi1_containment1} holds with $\sigma'' = \sigma' \circ \sigma$.

\textbf{Step A2:} We next claim that, on $\Omega_1$, for every $i,j \in [K]$ such that $j \neq \sigma''(i)$, $\tilde{V}^2_{j} \cap V(\bm{T}^i) =\emptyset$, i.e., 
\begin{align}
\text{for all $j \in [K]$},\quad \tilde{V}^2_{j}\subseteq V(\bm{T}^{\sigma^{\prime\prime -1}(j)}). 
\label{eq: tilde vj2 well classified 2}
\end{align}
To see this, let $i, j \in [K]$ such that $j \neq \sigma''(i)$; since $\tilde{V}_{j}^2:=V_{j}^2 \cap V^{\tau'}(\bm{G}_n)$, we have 
\begin{eqnarray}
\label{eq: tilde vj2 well classified}
\nonumber
\tilde{V}_{j}^2 \cap V(\bm{T}^i)&=& \bigr({V}_{j}^2 \cap V^{\tau'}(\bm{G}_n)\bigr) \cap V(\bm{T}^i)={V}_{j}^2  \cap \bigr( V(\bm{T}^i)\cap V^{\tau'}(\bm{G}_n)\bigr)\\
\nonumber
&=& {V}_{j}^2 \cap V^{\tau'}(\bm{T}^i) \stackrel{(a)}{\subseteq} {V}_{j}^2 \cap V_{\bm{G}_n}^{c_0 \gamma(\frac{\varepsilon}{18H},1)}(\bm{T}^i)\\
&\stackrel{(b)}{\subseteq} &{V}_{j}^2 \cap V^1_{\sigma(i)} \stackrel{(c)}{\subseteq} {V}_{j}^2 \cap V^2_{\sigma^{\prime\prime}(i)}=\emptyset,
\end{eqnarray}
where $(a)$ follows since $\tau' \geq c_0 \gamma(\frac{\varepsilon}{18H}, 1)$, where $(b)$ follows from the definition of $\Omega_1$, and where $(c)$ follows by~\eqref{eq:v1v2_containment1}.

\textbf{Step A3:} We now proceed to bound $d_{\mathcal{L}_1(\bm{F}_n)}(\hat{\ell}, \ell)$. 
Fix $i \in [K]$ and $u \in \mathcal{L}_1(\bm{T}^i)$. Then $u$ is a neighbor of $\pi_1(\bm{T}^i)$ and hence, by~\eqref{eq:pi1_containment1}, 
\begin{align*}
D_{\sigma''(i)}(u) = \min_{v \in \tilde{V}^2_{\sigma''(i)}} \text{dist}_{\bm{G}_n}(u, v) \leq \text{dist}_{\bm{G}_n}(u, \pi_1(\bm{T}^i)) = 1. 
\end{align*}

On the other hand, for any $j \neq \sigma''(i)$ and any $v \in \tilde{V}_j^2$, we know by~\eqref{eq: tilde vj2 well classified} that $u$ and $v$ are disconnected in $\bm{F}_n$, so that $\text{dist}_{\bm{G}_n}(u, v) = 1$ if and only if $(u, v)$ is an edge in $\bm{R}_n$. Therefore, 
\[
D_{j}(u)=\min_{v\in \tilde{V}^2_j} \text{dist}_{\bm{G}_n}(u, v) 
\begin{cases} = 1 & \text{ if there exists $v \in \tilde{V}_j^2$ such that $(u, v) \in E(\bm{R}_n)$} \\
\geq 2 & \text{ if $(u, v) \notin E(\bm{R}_n)$ for all $v \in \tilde{V}_j^2$}.
\end{cases}
\]
Via the definition of $\hat{l}(\cdot)$ in Algorithm~\ref{alg: distance recovery} then,
\begin{align*}
\mathbbm{1}\{\hat{l}(u)\neq \sigma''(i)\} &\le  \mathbbm{1}\biggl\{\exists j \neq \sigma''(i),\, D_{\sigma''(i)}(u) \geq D_{j}(u) \biggr\} = \mathbbm{1}\biggl\{\exists j \neq \sigma''(i),\, D_{j}(u) = 1 \biggr\}\\
&\leq \mathbbm{1} \biggl\{ \exists v \in \cup_{j=1}^K \tilde{V}_j^2,\; (u, v) \in E(\bm{R}_n) \biggr\} \leq |E_{\bm{R}_n}(u, \cup_{j=1}^K \tilde{V}_j^2)|.
\end{align*}

% \begin{align*}
% \mathbbm{1}\{\hat{l}(u)\neq \sigma''(i)\}&\le  \mathbbm{1}\{D_{\sigma''(i)}(u)\geq \min_{j\neq \sigma''(i)} D_{j}(u)\}\le   \mathbbm{1}\{1\geq \min_{j\neq \sigma''(i)} \{ 1 + \min_{v\in \tilde{V}^2_j} \mathbbm{1}\{ (u, v) \notin E(\bm{R}_n) \}\}\}\\
% &= \mathbbm{1}\{0\geq \min_{j\neq \sigma''(i)} \{  \min_{v\in \tilde{V}^2_j} \mathbbm{1}\{ (u, v) \notin E(\bm{R}_n) \}\}\}\le  \mathbbm{1}\{\min_{v\in \cup^K_{j=1} \tilde{V}^2_j} \mathbbm{1}\{ (u, v) \notin E(\bm{R}_n) \}=0\}\\
% &= \mathbbm{1}\{ E_{\bm{R}_n}(u, \cup^K_{j=1} \tilde{V}^2_j)\neq \emptyset\}\le |E_{\bm{R}_n}(u, \cup^K_{j=1} \tilde{V}^2_j)|
% \end{align*}
%
%Recall the expression of $d(\cdot,\cdot)$ in Definition~\ref{defin mismatch}, with the expression $\mathbbm{1}\{\hat{l}(u)\neq \sigma''(i)\}\le |E_{\bm{R}_n}(u, \cup^K_{j=1} \tilde{V}^2_j)|$, we have on $\Omega_2$

By the definition of $d(\cdot, \cdot)$ in~\eqref{defin mismatch}, on $\Omega$, we have by the definition of $\Omega_2$ that

\begin{equation*}
\begin{aligned}
\frac{d_{\mathcal{L}_1\left(\bm{F}_n\right)}\left(\hat{\ell},\ell\right)}{\left|\mathcal{L}_1\left(\boldsymbol{F}_n\right)\right|} &= \min_{\sigma \in S_K} \frac{d_{\mathcal{L}_1\left(\bm{F}_n\right)}^{\text{Ham}}\left(\hat{\ell},\sigma \circ \ell\right)}{\left|\mathcal{L}_1\left(\boldsymbol{F}_n\right)\right|} 
\leq \frac{d_{\mathcal{L}_1\left(\bm{F}_n\right)}^{\text{Ham}}\left(\hat{\ell},\sigma'' \circ \ell\right)}{\left|\mathcal{L}_1\left(\boldsymbol{F}_n\right)\right|} \\
&= \sum_{i=1}^K \sum_{u\in \mathcal{L}_1(\bm{T}^i)}
\frac{\mathbbm{1}\{\hat{l}(u)\neq \sigma''(i)\}}{\left|\mathcal{L}_1\left(\boldsymbol{F}_n\right)\right|} 
\le \sum_{u\in \mathcal{L}_1(\bm{F}_n)}\frac{|E_{\bm{R}_n}(u, \cup^K_{j=1} \tilde{V}^2_j)|}{\left|\mathcal{L}_1\left(\boldsymbol{F}_n\right)\right|}\\
&\stackrel{(a)}{\le} \frac{2\left|E_{\boldsymbol{R}_n}\bigl(\cup_{j=1}^K \tilde{V}^2_j, \,\mathcal{L}_1\left(\boldsymbol{F}_n\right)\bigr)\right|}{\left| \mathcal{L}_1\left(\boldsymbol{F}_n\right)\right|}
\le \frac{2\left|E_{\boldsymbol{R}_n}\bigl(V^{c_0\gamma\left(\frac{\varepsilon}{18 H}, 1\right), \alpha}(\boldsymbol{G}_n), \,\mathcal{L}_1\left(\boldsymbol{F}_n\right)\bigr)\right|}{\left| \mathcal{L}_1\left(\boldsymbol{F}_n\right)\right|}\le C_2 n^{-\frac{1+\alpha}{2+\alpha}-\delta}.
\end{aligned}
\end{equation*}
where $(a)$ follows from Lemma~\ref{lem:edge_count_inequality}. We have thus shown that~\eqref{thm: layer 1/2 eqn 1} holds on $\Omega$.

It remains to bound the probability of the event $\Omega_1, \Omega_2$. By applying Lemma~\ref{lemma early hist containment} (with $L \mapsto Q \vee \tilde{L}(\varepsilon, c_0)$ and $\varepsilon \mapsto \varepsilon/6$), we have 
\begin{align}
\mathbb{P}(\Omega_1) \geq 1 - \frac{\varepsilon}{6} + o(1) \label{eq:omega1_prob}
\end{align}
where the $o(1)$ sequence depends only on $\varepsilon, Q, H, \alpha, \delta$. To bound $\mathbb{P}\left(\Omega_2^c\right)$, we have by Lemma~\ref{lemma layer noise general} (with $C \mapsto c_0 \gamma(\frac{\varepsilon}{18H}, 1)$ and $s \mapsto 1$) that there exists 
$C_2 > 0$ such that
\begin{align*}
\mathbb{P}\biggl( 2 n^{\frac{1+\alpha}{2+\alpha} + \delta} \frac{\left|E_{\boldsymbol{R}_n}\bigl(V^{c_0\gamma\left(\frac{\varepsilon}{18 H}, 1\right), \alpha}(\boldsymbol{G}_n), \,\mathcal{L}_1\left(\boldsymbol{F}_n\right)\bigr)\right|}{\left| \mathcal{L}_1\left(\boldsymbol{F}_n\right)\right|} \geq C_2\biggr) \leq \frac{\varepsilon}{2}. 
\end{align*}
This immediately implies that $\mathbb{P}(\Omega_2) \geq 1 - \varepsilon/2$. Therefore, by an application of union bound, 
\begin{equation*}
\mathbb{P}\left(\Omega^c\right) \leq \mathbb{P}(\Omega_1^c) + \mathbb{P}(\Omega_2^c) \leq \frac{\varepsilon}{6} + \frac{\varepsilon}{2} + o(1). 
\end{equation*}
Thus, we have have completed the proof of the first claim~\eqref{thm: layer 1/2 eqn 1} of the theorem.
\vspace{.5in}

\textbf{Part B:} We now prove the second statement of the Theorem~\eqref{thm: layer 1/2 eqn 2}, in the linear preferential attachment regime where $\alpha=0$. We let $C_3$ be a positive number whose definition will be specified later. We also define the "good" event
$\tilde{\Omega}={\Omega}_1\cap \tilde{\Omega}_2\cap \tilde{\Omega}_3$, where $\Omega_1$ is defined in~\eqref{eq:omega1_defn} and $\tilde{\Omega}_2$ and $\tilde{\Omega}_3$ are defined below:
\begin{align*}
    \tilde{\Omega}_2 &:= \left\{\left|E_{\boldsymbol{R}_n}\left(V^{c_0\gamma\left(\frac{\varepsilon}{18 H}, 1\right), 0}\left(\boldsymbol{G}_n\right), \mathcal{L}_s\left(\boldsymbol{F}_n\right)\right)\right|=0, \quad\text{for $s=1,3$}\right\}\\
    \tilde{\Omega}_3 &:= \left\{\frac{\Bigl|E_{\boldsymbol{R}_n}\Bigl(V_1^{c_0\gamma\left(\frac{\varepsilon}{18 H}, 1\right), 0}\left(\boldsymbol{G}_n\right), \mathcal{L}_2\bigl(\boldsymbol{F}_n\bigr)\Bigr)\Bigr|+\Bigl|\bigcup_{\substack{i,j\in[K]\\ i\neq j}} \dot{E}_{\boldsymbol{R}_n}\Bigl(V^{c_0 \gamma(\frac{\varepsilon}{18H}, 1), 0}_{\bm{G}_n}(\bm{T}^j), \mathcal{L}_2(\boldsymbol{T}^{i})\Bigr)\Bigr|}{\bigl|\mathcal{L}_2\bigl(\boldsymbol{F}_n\bigr)\bigr|}\le \frac{C_3 n^{-\delta}}{2}\right\},
\end{align*}
where we define 
\begin{align*}
V_1^{c_0\gamma\left(\frac{\varepsilon}{18 H}, 1\right),0}\left(\boldsymbol{G}_n\right)&:=\left\{v \in V\left(\boldsymbol{G}_n\right) \bigg |  \operatorname{dist}_{\bm{F}_n}\left(v, V^{c_0\gamma\left(\frac{\varepsilon}{18 H}, 1\right), 0}\left(\boldsymbol{G}_n\right)\right)\le 1\right\}
\end{align*}
which consists of all nodes in $V^{c_0 \gamma(\frac{\varepsilon}{18H}, 1), 0}(\bm{G}_n)$ as well as their $\bm{F}_n$--neighbors. We also define 
\begin{align*}
\dot{E}_{\boldsymbol{R}_n}(A, B) &:= \bigl\{(u, v) \,\big|\, u \in A,\, v \in B,\, \text{ and } \exists w \in V(\bm{G}_n),\, (u, w), (w, v) \in E(\bm{R}_n) \bigr\},
\end{align*}
as the set of pairs $(u, v)$ with $u \in A$ and $v \in B$ such that $u$ and $v$ are connected by two noise edges. 

We first show that~\eqref{thm: layer 1/2 eqn 2} holds on $\tilde{\Omega}$. To that end, fix $i \in [K]$ and $u \in \mathcal{L}_2(\bm{T}^i)$. Since $\pi_1(\bm{T}^i) \in \tilde{V}_{\sigma''(i)}^2$ by~\eqref{eq:pi1_containment1}, we have $D_{\sigma''(i)}(u) \leq \text{dist}_{\bm{G}_n}(u, \pi_1(\bm{T}^i)) \leq 2$. Let $j \neq \sigma''(i)$ and recall that on $\Omega_1$, by~\eqref{eq: tilde vj2 well classified}, we have that $V(\bm{T}^i) \cap \tilde{V}^2_j = \emptyset$ so that
\begin{align*}
E_{\bm{G}_n}(u, \tilde{V}_j^2) = E_{\bm{R}_n}(u, \tilde{V}^2_j),
\end{align*}
that is, all edges from $u$ to $\tilde{V}_j^2$ must be noise edges. Moreover, we have that all the length-two path from $u$ to $\tilde{V}_j^2$ can be decomposed as follows:
\begin{align*}
\dot{E}(u, \tilde{V}_j^2) &:= \bigl\{ (u, v) \, \big|\, v \in \tilde{V}_j^2,\, \text{ and } \exists w \in V(\bm{G}_n), \, (u, w), (w, v) \in E(\bm{G}_n) \bigr\} \\
&= \underbrace{\bigl\{(u, v) \,\big|\, v \in \tilde{V}_j^2,\, \text{ and } \exists w \in V(\bm{T}^i), \, (u, w)\in E(\bm{F}_n), (w, v) \in E(\bm{R}_n) \bigr\}}_{=: \dot{E}_{\text{FR}}(u, \tilde{V}_j^2)} \cup \\
&\qquad \underbrace{\bigl\{(u, v) \,\big|\, v \in \tilde{V}_j^2,\, \text{ and } \exists w \in V(\bm{G}_n), \, (u, w) \in E(\bm{R}_n), (w, v) \in E(\bm{R}_n) \bigr\}}_{=: \dot{E}_{\text{RR}}(u, \tilde{V}_j^2) = \dot{E}_{\bm{R}_n}(u, \tilde{V}_j^2)} \cup \\
&\qquad \underbrace{\bigl\{(u, v) \,\big|\, v \in \tilde{V}_j^2,\, \text{ and } \exists w \in V(\bm{G}_n), \,(u, w) \in E(\bm{R}_n),\, (w, v) \in E(\bm{F}_n) \bigr \}}_{=: \dot{E}_{\text{RF}}(u, \tilde{V}_j^2)} \cup \\
&\qquad \underbrace{\bigl\{(u, v) \,\big|\, v \in \tilde{V}_j^2,\, \text{ and } \exists w \in V(\bm{G}_n), \,(u, w) \in E(\bm{F}_n),\, (w, v) \in E(\bm{F}_n) \bigr \}}_{=: \dot{E}_{\text{FF}}(u, \tilde{V}_j^2)}
\end{align*}

Since $u$ belongs to $\bm{T}^i$ and $\tilde{V}_j^2$ does not intersect $\bm{T}^i$ on $\Omega_1$ by~\eqref{eq: tilde vj2 well classified}, there cannot exist a node $w$ such that $(u, w)$ and $(w, v)$ are both edges in $\bm{F}_n$. Therefore,
\begin{align}
\dot{E}_{\text{FF}}(u, \tilde{V}_j^2) = \emptyset, \quad \text{ on } \Omega_1.
\end{align}

For any $w \in V(\bm{T}^i)$ such that $(u, w) \in E(\bm{F}_n)$, it must be that $w \in \mathcal{L}_1(\bm{T}^i) \cup \mathcal{L}_3(\bm{T}^i)$. Since $\tilde{V}^2_j \subset V^{c_0 \gamma(\frac{\varepsilon}{18H}, 1) 0}(\bm{G}_n)$, it must be that on $\tilde{\Omega}_2$, there are no edges in $E(\bm{R}_n)$ connecting such $w \in \mathcal{L}_1(\bm{T}^i) \cup \mathcal{L}_3(\bm{T}^i)$ in  to $\tilde{V}_j^2$. Therefore, 
\begin{align}
\dot{E}_{\text{FR}}(u, \tilde{V}_j^2) = \emptyset, \quad \text{ on } \tilde{\Omega}_2.
\end{align}

We thus have that $D_j(u) = \min_{v \in \tilde{V}_j^2} \text{dist}_{\bm{G}_n}(u, v) \leq 2$ if and only if 
\[
E_{\bm{R}_n}(u, \tilde{V}_j^2)\cup \dot{E}_{\text{RR}}(u, \tilde{V}_j^2) \cup \dot{E}_{\text{RF}}(u, \tilde{V}_j^2) \neq \emptyset.
\]

We thus have that, on the event $\tilde{\Omega}$, 
\begin{align*}
\mathbbm{1}\{ \hat{\ell}(u) \neq \sigma''(i)\} 
&\leq 
\mathbbm{1}\bigl\{ \exists j \neq \sigma''(i), \, D_j(u) \leq 2 \bigr\} \\
&\leq
\mathbbm{1}\bigl\{ \exists j \neq \sigma''(i), \, E_{\bm{R}_n}(u, \tilde{V}_j^2) \cup \dot{E}_{\text{RR}}(u, \tilde{V}_j^2) \cup \dot{E}_{\text{RF}}(u, \tilde{V}_j^2) \neq \emptyset \bigr\} \\
&\stackrel{(a)}{\leq} \mathbbm{1}\biggl\{ E_{\bm{R}_n}\bigl(u, V_1^{c_0 \gamma(\frac{\varepsilon}{18H}, 1), 0}(\bm{G}_n) \bigr) \neq \emptyset \biggr\} + \mathbbm{1}\bigl\{ \exists j \neq \sigma''(i), \, \dot{E}_{\text{RR}}(u, \tilde{V}_j^2) \neq \emptyset \bigr\}\\
&\leq \biggl| E_{\bm{R}_n}\bigl(u, V_1^{c_0 \gamma(\frac{\varepsilon}{18H}, 1), 0}(\bm{G}_n) \bigr) \biggr| + \sum_{j \neq \sigma''(i)} | \dot{E}_{\text{RR}}(u, \tilde{V}_j^2) |.
\end{align*}

where $(a)$ follows because 
\begin{align*}
E_{\bm{R}_n}\bigl(u, V_1^{c_0 \gamma(\frac{\varepsilon}{18H}, 1), 0}(\bm{G}_n) \bigr) = \emptyset \, \implies \, E_{\bm{R}_n}(u, \tilde{V}_j^2) = \emptyset \text{ and } \dot{E}_{\text{RF}}(u, \tilde{V}_j^2) = \emptyset.
\end{align*}
To see this, note that $\tilde{V}_j^2 \subset V_1^{c_0 \gamma(\frac{\varepsilon}{18H}, 1), 0}(\bm{G}_n)$ so that $E_{\bm{R}_n}(u, \tilde{V}_j^2) \subset E_{\bm{R}_n}\bigl(u, V_1^{c_0 \gamma(\frac{\varepsilon}{18H}, 1), 0}(\bm{G}_n) \bigr)$. Note also that for any $(u, v) \in \dot{E}_{\text{RF}}(u, \tilde{V}_j^2)$, there exists a node $w$ such that $(w, v) \in E(\bm{F}_n)$ so that $w \in V_1^{c_0 \gamma(\frac{\varepsilon}{18H}, 1), 0}(\bm{G}_n)$ and $(u, w)$ is an element in $E_{\bm{R}_n}\bigl(u, V_1^{c_0 \gamma(\frac{\varepsilon}{18H}, 1), 0}(\bm{G}_n) \bigr)$. Thus, we have
\begin{align*}
\frac{d_{\mathcal{L}_2(\bm{F}_n)}(\hat{\ell}, \ell)}{|\mathcal{L}_2(\bm{F}_n)|} 
&= 
\min_{\sigma \in S_K} \frac{ d_{\mathcal{L}_2(\bm{F}_n)}^{\mathrm{Ham}}( \hat{\ell}, \sigma \circ \ell) }{|\mathcal{L}_2(\bm{F}_n)|} 
\leq
\frac{d_{\mathcal{L}_2(\bm{F}_n)}^{\mathrm{Ham}}( \hat{\ell}, \sigma'' \circ \ell) }{|\mathcal{L}_2(\bm{F}_n)|} \\
&\leq \sum_{i=1}^K \sum_{u \in \mathcal{L}_2(\bm{T}^i)} \frac{ \mathbbm{1}\{ \hat{\ell}(u) \neq \sigma''(i)\}}{|\mathcal{L}_2(\bm{F}_n)|} \\
&\leq \sum_{i=1}^K \sum_{u \in \mathcal{L}_2(\bm{T}^i)} \frac{ \bigl| E_{\bm{R}_n}\bigl(u, V_1^{c_0 \gamma(\frac{\varepsilon}{18H}, 1), 0}(\bm{G}_n) \bigr) \bigr| + \sum_{j \neq \sigma''(i)} | \dot{E}_{\text{RR}}(u, \tilde{V}_j^2) | }{ | \mathcal{L}_2(\bm{F}_n) |} \\
&\stackrel{(a)}{\leq} \frac{1}{| \mathcal{L}_2(\bm{F}_n)|} \biggl\{ \sum_{u \in \mathcal{L}_2(\bm{F}_n)} \bigl| E_{\bm{R}_n}\bigl(u, V_1^{c_0 \gamma(\frac{\varepsilon}{18H}, 1), 0}(\bm{G}_n) \bigr) \bigr| + \sum_{i \neq j} \sum_{u \in \mathcal{L}_2(\bm{T}^i)} | \dot{E}_{\bm{R}_n}(u, V^{c_0 \gamma(\frac{\varepsilon}{18H}, 1), 0}(\bm{T}^j))| \biggr\}\\
&\stackrel{(b)}{\leq} \frac{1}{| \mathcal{L}_2(\bm{F}_n)|} \biggl\{ 2 \bigl| E_{\bm{R}_n}\bigl( \mathcal{L}_2(\bm{F}_n), V_1^{c_0 \gamma(\frac{\varepsilon}{18H}, 1), 0}(\bm{G}_n) \bigr) \bigr| + \sum_{i \neq j} 2 \bigl| \dot{E}_{\bm{R}_n}\bigl( \mathcal{L}_2(\bm{T}^i), V^{c_0 \gamma(\frac{\varepsilon}{18H}, 1), 0}(\bm{T}^j)\bigr) \bigr| \biggr\} \\
&\leq C_3 n^{-\delta},\quad \text{ on $\tilde{\Omega}_3$},
\end{align*}
where inequality $(a)$ follows from~\eqref{eq: tilde vj2 well classified 2} and where inequality $(b)$ follows from Lemma~\ref{lem:edge_count_inequality}; to argue that 
\[
\sum_{u \in \mathcal{L}_2(\bm{T}^i)} | \dot{E}_{\bm{R}_n}(u, V^{c_0 \gamma(\frac{\varepsilon}{18H}, 1), 0}(\bm{T}^j))| \leq 2 \bigl| \dot{E}_{\bm{R}_n}\bigl( \mathcal{L}_2(\bm{T}^i), V^{c_0 \gamma(\frac{\varepsilon}{18H}, 1), 0}(\bm{T}^j) \bigr|,
\]
we apply Lemma~\ref{lem:edge_count_inequality} with respect to a graph $\bm{g}'$ where $(u, v)$ forms an edge if $(u, v)$ belongs in the set $\dot{E}_{\bm{R}_n}\bigl( \mathcal{L}_2(\bm{T}^i), V^{c_0 \gamma(\frac{\varepsilon}{18H}, 1), 0}(\bm{T}^j) \bigr)$. 

We have thus shown that~\eqref{thm: layer 1/2 eqn 2} holds on the event $\tilde{\Omega}$. It remains to bound the probability of $\tilde{\Omega}$. Applying the first statement of Lemma~\ref{lemma layer noise} along with a union bound (with $C = c_0 \gamma(\frac{\varepsilon}{18H}, 1)$ and $s = 1$ and then $s=3$), we have that
\begin{align}
\mathbb{P}(\tilde{\Omega}_2^c) \leq 2 o(1),
\end{align}
where the $o(1)$ sequence depends only on $\alpha, c_0, H, \delta$. 

Noticing that $\bigl|\dot{E}_{\mathbb{R}_n}\bigl(A,B\bigr)\bigr|\leq \bigl|{E}^2_{\mathbb{R}_n}\bigl(A,B\bigr)\bigr|$ defined in Lemma~\ref{lemma layer noise}, and apply the second and third statement of the lemma (with $C = c_0 \gamma(\frac{\varepsilon}{18H}, 1)$ and $s = 2$), we may conclude that there exists $C_3 > 0$ such that
\[
\mathbb{P} \biggl( 4n^{2\delta} \frac{ \bigl| \cup_{i,j \in [K], i\neq j} \dot{E}_{\mathbb{R}_n}\bigl( V_{\bm{G}_n}^{c_0 \gamma( \frac{\varepsilon}{18H}, 1), 0}(\bm{T}^j), \mathcal{L}_2(\bm{T}^i) \bigr) \bigr|}{|\mathcal{L}_2(\bm{F}_n)} \geq C_3\biggr) \leq \frac{\varepsilon}{6}
\]
and that 
\[
\mathbb{P}\biggl( 4n^{\delta} \frac{ \bigl| E_{\bm{R}_n}\bigl( V_1^{c_0 \gamma(\frac{\varepsilon}{18H}, 1), 0}(\bm{G}_n), \mathcal{L}_2(\bm{F}_n) \bigr) \bigr|}{ | \mathcal{L}_2(\bm{F}_n) |}\geq C_3 \biggr) \leq \frac{\varepsilon}{6}.
\]
It holds therefore that $\mathbb{P}(\tilde{\Omega}_3^c) \leq \frac{\varepsilon}{3}$. Combining these with~\eqref{eq:omega1_prob} yields that $\mathbb{P}(\tilde{\Omega}^c) \leq \varepsilon + o(1)$ and the Theorem thus follows.

We have proved that $\mathbb{P}\left(\Omega_1^c\right) \le \varepsilon/3$ in the above proof of \eqref{thm: layer 1/2 eqn 1}. 
\end{proof}

\begin{lemma}
\label{lemma layer noise general}
Let $\boldsymbol{G}_n \sim \mathrm{PF}(\alpha, \theta, \ell, \pi)$, and suppose Assumptions~\ref{assumption: bound} and \ref{assumption: sparsity} hold. Then for any fixed $C > 0$ and $s\in \mathbb{N}$, the following result holds:
\begin{equation}
\label{lemma layer noise general  eqn 1}
\frac{\left|E_{\boldsymbol{R}_n}\left(V^{C, \alpha}\left(\boldsymbol{G}_n\right), \mathcal{L}_s\left(\boldsymbol{F}_n\right)\right)\right|}{\left| \mathcal{L}_s\left(\boldsymbol{F}_n\right)\right|}= O_p\left(n^{-\frac{1+\alpha}{2+\alpha}-\delta}\right),
\end{equation}
where $\mathcal{L}_s\left(\boldsymbol{F}_n\right)$ denotes all the nodes in layer $s$ of $\boldsymbol{T}^i$.
\end{lemma}
\begin{proof}
We fix $\varepsilon>0$ and define the ``good event'' $\Omega\left(t\right):=\left(\cap_{i=1}^K {\Omega}_i\right) \cap\left(\cap_{i=1}^K \breve{\Omega}_i\left(t\right)\right)$,
where
\begin{equation*}
   \Omega_i=\left\{ V^{C,\alpha}_{\boldsymbol{G}_n}\left(\boldsymbol{T}^i\right)\subset\pi_{1:{L\left(\varepsilon, C\right)}}\left(\boldsymbol{T}^i\right)\right\} 
\end{equation*}
and 
\begin{equation*}
\breve{\Omega}_i\left(t\right):=\left\{\frac{\left|E_{\boldsymbol{R}_n}\bigl(\pi_{1:{L\left(\varepsilon, C\right)}}(\boldsymbol{T}^i), \mathcal{L}_s\left(\boldsymbol{F}_n\right)\bigr)\right|}{\left| \mathcal{L}_s\left(\boldsymbol{F}_n\right)\right|}\le tn^{-\frac{1+\alpha}{2+\alpha}-\delta}\right\}.  
\end{equation*}
for any $t>0$. The function $L(\cdot)$ is defined in Lemma~\ref{lemma order--degree}.

We can observed that for fixed $t$, on the event ${\Omega}\left(t\right)$
\begin{equation*}
\begin{aligned}
\frac{\left|E_{\boldsymbol{R}_n}\left(V^{C, \alpha}\left(\boldsymbol{G}_n\right), \mathcal{L}_s\left(\boldsymbol{F}_n\right)\right)\right|}{\left| \mathcal{L}_s\left(\boldsymbol{F}_n\right)\right|}
&=\sum_{i=1}^K\frac{\left|E_{\boldsymbol{R}_n}\left(V_{\bm{G}_n}^{C, \alpha}\left(\boldsymbol{T}^i\right), \mathcal{L}_s\left(\boldsymbol{F}_n\right)\right)\right|}{\left| \mathcal{L}_s\left(\boldsymbol{F}_n\right)\right|}\\
&\le \sum_{i=1}^K\frac{\left|E_{\boldsymbol{R}_n}\left(\pi_{1:{L\left(\varepsilon, C\right)}}\left(\boldsymbol{T}^i\right), \mathcal{L}_s\left(\boldsymbol{F}_n\right)\right)\right|}{\left| \mathcal{L}_s\left(\boldsymbol{F}_n\right)\right|}\\
&\le K tn^{-\frac{1+\alpha}{2+\alpha}-\delta},    
\end{aligned}    
\end{equation*}
where the equality follows from the fact that $V^{C,\alpha}\left(\boldsymbol{G}_{n}\right)=\cup_{i=1}^K V^{C,\alpha}_{\boldsymbol{G}_n}\left(\boldsymbol{T}^i\right)$, and the two inequalities follow directly from the definition of $\Omega_i$ and $\breve{\Omega}_i(t)$.

%Then we claim that there exist $t$ depending only on $\varepsilon, C$, such that
%$\mathbb{P}\left({\Omega}^c\right)$ is bounded by $\left(2K+1\right)\varepsilon$ and hence \eqref{lemma layer noise general  eqn 1} holds. 

We now bound the probability of the event $\Omega(t)$. Note that by Lemma~\ref{lemma order--degree}, we have that 
\begin{align}
\forall i \in [K], \quad \mathbb{P}\left(\Omega^c_i\right)\le \varepsilon+o(1). \label{eq:omega_i_prob_bound}
\end{align}

Our next step then is to show that there exists $t > 0$ such that $\mathbb{P}\bigl(\breve{\Omega}_i(t)^c\bigr) \leq \varepsilon$. Note that the top $L(\varepsilon, C)$ nodes of each tree $\left\{\pi_{1:{L\left(\varepsilon, C\right)}}\left(\boldsymbol{T}^i\right), i \in \left[K\right]\right\}$ are deterministic and that the layer $s$ nodes $\mathcal{L}_s\left(\boldsymbol{F}_n\right)$
are determined entirely by $\bm{F}_n$ and hence independent of $\boldsymbol{R}_n$ by the Definition~\ref{defin: PF} of our model. Therefore, the random variable $|E_{\bm{R}_n}(\pi_{1:L(\varepsilon,C)}(\bm{T}^i), \mathcal{L}_s(\bm{F}_n))|$, conditioning on $\boldsymbol{F}_n$, is stochastic dominated by the binomial distribution $\operatorname{Bin} \left(\bigl|\pi_{1:{L(\varepsilon, C)}}\left(\boldsymbol{T}^i\right)\bigr| \left|\mathcal{L}_s\left(\boldsymbol{F}_n\right)\right|, \theta \right)$.
Combining with the fact that $\left|\pi_{1:{L\left(\varepsilon, C\right)}}\left(\boldsymbol{T}^i\right)\right|=L\left(\varepsilon, C\right)$, we have
\begin{equation*}    \mathbb{E}_{\bm{R}_n}\left(\frac{\left|E_{\boldsymbol{R}_n}\bigl(\pi_{1:{L\left(\varepsilon, C\right)}}\left(\boldsymbol{T}^i\right), \mathcal{L}_s\left(\boldsymbol{F}_n\right)\bigr)\right|}{\left| \mathcal{L}_s\left(\boldsymbol{F}_n\right)\right|}\right)
\le \bigl|\pi_{1:{L\left(\varepsilon, C\right)}}\bigl(\boldsymbol{T}^i\bigr)\bigr| \theta
\le L\left(\varepsilon, C\right)\theta.
\end{equation*}
By Markov's inequality and Assumption~\ref{assumption: sparsity}, for any $\varepsilon>0$, there exist an universal constant $C^1>0$ such that the following holds with probability at least $1-\varepsilon$:
\begin{equation*}   \frac{\left|E_{\boldsymbol{R}_n}\left(\pi_{1:{L\left(\varepsilon, C\right)}}\left(\boldsymbol{T}^i\right), \mathcal{L}_s\left(\boldsymbol{F}_n\right)\right)\right|}{\left| \mathcal{L}_s\left(\boldsymbol{F}_n\right)\right|}\le \varepsilon^{-1} L\left(\varepsilon, C\right)\theta\le  C^1\varepsilon^{-1} L\left(\varepsilon, C\right)n^{-\frac{1+\alpha}{2+\alpha}-\delta}. 
\end{equation*}
Therefore, we may take $t:=C^1\varepsilon^{-1} L\left(\varepsilon, C\right)$ to obtain that $\mathbb{P}\left(\breve{\Omega}^c_i\left(t\right)\right)\le \varepsilon$.

Combining the preceding bounds with a union bound, for all sufficiently large $n$, we have, for $t = C^1\varepsilon^{-1} L\left(\varepsilon, C\right)$,
\begin{eqnarray*}
\mathbb{P}\left({\Omega}\left(t\right)^c\right)
&\le & \sum_{i=1}^K \mathbb{P}\left({\Omega}^c_i\right)+\sum_{i=1}^K \mathbb{P}\left(\breve{\Omega}_i(t)^c \right)\le K\varepsilon+K\varepsilon+o\left(1\right)\le 2K \varepsilon + o(1).
\end{eqnarray*}
This completes the proof of \eqref{lemma layer noise general  eqn 1}.    
\end{proof}

\begin{lemma}[Degree sum of $\mathcal{L}_s(\boldsymbol{T}_n)$]
\label{lemma degree sum}
Let $\boldsymbol{T}_n \sim \mathrm{APA}(0, n)$. For $s\in \mathbb{N}$, let $\mathcal{L}_s(\boldsymbol{T}_n)$ denote the set of nodes at layer $s$. The sum of the degrees of all the nodes in $\mathcal{L}_s(\boldsymbol{T}_n)$ has the following form:
\begin{equation*}
\sum_{u\in \mathcal{L}_s\left(\boldsymbol{T}_{n}\right)}\operatorname{deg}_{\boldsymbol{T_n}}\left(u\right)=\left|\mathcal{L}_{s+1}\left(\boldsymbol{T}_{n}\right)\right|+\left|\mathcal{L}_{s}\left(\boldsymbol{T}_{n}\right)\right|.
\end{equation*}
\end{lemma}
\begin{proof}
This result can be directly observed from the network structure. For a node $u \in \mathcal{L}_s(\bm{T}_n)$, any edge incident on $u$ must either connect $u$ to one of its children or connect $u$ to its parent; we thus have that $\operatorname{deg}_{\bm{T}_n}(u) = | \text{children of $u$ in $\bm{T}_n$}| + 1$. We also note that any child of $u$ is in layer $s+1$. Therefore, 

\begin{align*}
\sum_{u\in \mathcal{L}_s\left(\boldsymbol{T}_{n}\right)} \operatorname{deg}_{\boldsymbol{T_n}}\left(u\right) &= 
\sum_{u\in \mathcal{L}_s\left(\boldsymbol{T}_{n}\right)}
 \bigl\{ | \text{children of $u$ in $\bm{T}_n$} | + 1 \bigr\} \\
 &= |\mathcal{L}_{s+1}(\bm{T}_n) | + | \mathcal{L}_s(\bm{T}_n)|.
\end{align*}

The Lemma follows as desired.

\end{proof}

\begin{lemma}[Difference equation for $\mathbb{E}\left(\left|\mathcal{L}_s\left(\boldsymbol{T}_n\right)\right|\right)$]
\label{lemma difference equation}
Let $\boldsymbol{T}_n\sim \mathrm{APA}(0, n)$.
Suppose $s,n \in \mathbb{N}$ are integers satisfying $n \ge 2$. let $\mathcal{L}_s\left(\boldsymbol{T}_n\right)$ be the nodes at layer $s$,
if $s\geq 2$
then the expectation obeys the following difference equation 
\begin{equation*}
\mathbb{E}\left(\left|\mathcal{L}_s\left(\boldsymbol{T}_{n+1}\right)\right|\right)-\mathbb{E}\left(\left|\mathcal{L}_s\left(\boldsymbol{T}_n\right)\right|\right)=\frac{\mathbb{E}\left(\left|\mathcal{L}_{s-1}\left(\boldsymbol{T}_n\right)\right|\right)+\mathbb{E}\left(\left|\mathcal{L}_s\left(\boldsymbol{T}_n\right)\right|\right)}{2 \left(n-1\right)}.
\end{equation*}
If $s = 1$ then the expectation obeys the following difference equation 
\begin{equation*}
\mathbb{E}\left(\left|\mathcal{L}_1\left(\boldsymbol{T}_{n+1}\right)\right|\right)-\mathbb{E}\left(\left|\mathcal{L}_1\left(\boldsymbol{T}_n\right)\right|\right)=\frac{\mathbb{E}\left(\left|\mathcal{L}_1\left(\boldsymbol{T}_n\right)\right|\right)}{2 \left(n-1\right)}.
\end{equation*}
\end{lemma}

\begin{proof}
To derive the difference equation for $s\geq 2$, we first consider the conditional distribution of $\left|\mathcal{L}_s\left(\boldsymbol{T}_{n+1}\right)\right|$ conditioning on $\boldsymbol{T}_n$. 
A direct observation is that $\left|\mathcal{L}_s\left(\boldsymbol{T}_{n+1}\right)\right|=\left|\mathcal{L}_s\left(\boldsymbol{T}_{n}\right)\right|+1$ if and only if the new node is attached to a node in layer $s-1$. 
Therefore, by the generation mechanism of preferential attachment tree and Lemma~\ref{lemma degree sum}
\begin{equation*}
\mathbb{P}\left(\left|\mathcal{L}_{s}\left(\boldsymbol{T}_{n+1}\right)\right|=\left|\mathcal{L}_{s}\left(\boldsymbol{T}_{n}\right)\right|+1 \mid\!\!\mid \boldsymbol{T}_n\right)=\frac{\sum_{u\in \mathcal{L}_{s-1}\left(\boldsymbol{T}_{n}\right)}\operatorname{deg}_{\boldsymbol{T_n}}\left(u\right)}{\sum_{u\in V\left(\boldsymbol{T}_{n}\right)}\operatorname{deg}_{\boldsymbol{T_n}}\left(u\right)}=\frac{\left|\mathcal{L}_{s-1}\left(\boldsymbol{T}_{n}\right)\right|+\left|\mathcal{L}_{s}\left(\boldsymbol{T}_{n}\right)\right|}{2\left(n-1\right)},
\end{equation*}
the second identity comes from Lemma~\ref{lemma degree sum}.
Then we proceed to the difference equation
\begin{eqnarray*}
\mathbb{E}\left(\left|\mathcal{L}_s\left(\boldsymbol{T}_{n+1}\right)\right|\right)-\mathbb{E}\left(\left|\mathcal{L}_s\left(\boldsymbol{T}_n\right)\right|\right)
&=&
\mathbb{E}\bigl\{\mathbb{E}(|\mathcal{L}_s(\boldsymbol{T}_{n+1})| \mid\!\!\mid \boldsymbol{T}_n ) \bigr\}-\mathbb{E}\left(\left|\mathcal{L}_s\left(\boldsymbol{T}_n\right)\right|\right)\\
&=& \mathbb{E} \bigl\{ \mathbb{P}\left(\left|\mathcal{L}_s\left(\boldsymbol{T}_{n+1}\right)\right|=\left|\mathcal{L}_s\left(\boldsymbol{T}_{n}\right)\right|+1 \mid\!\!\mid \boldsymbol{T}_n\right) \bigr\}\\
&=&\frac{\mathbb{E}\left(\left|\mathcal{L}_{s-1}\left(\boldsymbol{T}_n\right)\right|\right)+\mathbb{E}\left(\left|\mathcal{L}_s\left(\boldsymbol{T}_n\right)\right|\right)}{2 \left(n-1\right)}.   
\end{eqnarray*}
We now turn to the case $s=1$. The only change is that $\sum_{u\in \mathcal{L}_{0}\left(\boldsymbol{T}_{n}\right)}\operatorname{deg}_{\boldsymbol{T_n}}\left(u\right)=\left|\mathcal{L}_{1}\left(\boldsymbol{T}_{n}\right)\right|$ replaces $\left|\mathcal{L}_{0}\left(\boldsymbol{T}_{n}\right)\right|+\left|\mathcal{L}_{1}\left(\boldsymbol{T}_{n}\right)\right|$, here $\mathcal{L}_0(\bm{T_n})$ denotes the root node of tree $\bm{T_n}$.
Applying the same procedure as above, we obtain
\begin{equation*}
\mathbb{E}\left(\left|\mathcal{L}_1\left(\boldsymbol{T}_{n+1}\right)\right|\right)-\mathbb{E}\left(\left|\mathcal{L}_1\left(\boldsymbol{T}_n\right)\right|\right)=\frac{\mathbb{E}\left(\left|\mathcal{L}_1\left(\boldsymbol{T}_n\right)\right|\right)}{2 \left(n-1\right)}.
\end{equation*}
The lemma follows as desired.
\end{proof}

\begin{lemma}[Rate for $\mathbb{E}\left(\left|\mathcal{L}_s\left(\boldsymbol{T}_n\right)\right|\right)$]
\label{lemma layer size}
Let $\boldsymbol{T}_n \sim \mathrm{APA}(0, n)$ and let $\mathcal{L}_s(\boldsymbol{T}_n)$ denote the nodes at layer $s \in \mathbb{N}$. Then, for any $s\in \mathbb{N}$,
\begin{equation*}    \mathbb{E}\left(\left|\mathcal{L}_s\left(\boldsymbol{T}_n\right)\right|\right) \asymp  n^{1/2}\left(\log n\right)^{s-1}.
\end{equation*}
\end{lemma}
\begin{proof}
We use proof by induction on $s$. The first step is to show that $\mathbb{E}\left(\left|\mathcal{L}_1\left(\boldsymbol{T}_n\right)\right|\right)\asymp n^{1/2}$.
Denote $m_n:=\mathbb{E}\left(\left|\mathcal{L}_1\left(\boldsymbol{T}_n\right)\right|\right)$, by Lemma~\ref{lemma difference equation}, for $n\geq 2$
\begin{equation*}
    m_{n+1}-m_n=\mathbb{E}\left(\left|\mathcal{L}_1\left(\boldsymbol{T}_{n+1}\right)\right|\right)-\mathbb{E}\left(\left|\mathcal{L}_1\left(\boldsymbol{T}_n\right)\right|\right)=\frac{m_n}{2\left(n-1\right)}.
\end{equation*}
indicates that
\begin{equation}
\label{lemma 5 eqn 1}
   m_{n+1}=\frac{2n-1}{2n-2}m_{n}. 
\end{equation}
Putting it together with the boundary condition $m_2=1$, we have 
\begin{eqnarray*}
m_{n+1}&=&\frac{\left(2n-1\right)\left(2n-3\right)\cdots 3}{\left(2n-2\right)\left(2n-4\right)\cdots 2}\times m_2=\frac{\left(2n-1\right)!}{\left(\left(2n-2\right)\left(2n-4\right)\cdots 2\right)^2}=\frac{\left(2n-1\right)!}{2^{2n-2}\left(\left(n-1\right)!\right)^2}\\
& \asymp &\frac{\sqrt{2\pi}\sqrt{2n-1}\left(2n-1\right)^{2n-1}/e^{2n-1}}{2^{2n-2}\left(\sqrt{2\pi}\sqrt{n-1}\left(n-1\right)^{n-1}\right)^2/e^{2n-2}}\asymp \sqrt{n}.
\end{eqnarray*}
The asymptotic equivalent comes from the Stirling approximation.

Thus we have shown that $m_n := \mathbb{E}\left(\left|\mathcal{L}_1\left(\boldsymbol{T}_n\right)\right|\right) \asymp \sqrt{n}$, which in particular implies  that for $s=1$ we have $\mathbb{E}\left(\left|\mathcal{L}_{s}\left(\boldsymbol{T}_n\right)\right|\right)/m_n \asymp \left(\log n\right)^{s-1}$. For $s > 1$, assume as the induction hypothesis that $\mathbb{E}(|\mathcal{L}_{s-1}(\boldsymbol{T}_n)|)/m_n \asymp \bigl(\log n\bigr)^{s-2}$ for all sufficiently large $n$; we then proceed to establish the corresponding asymptotic equivalence for $s$.

By Lemma~\ref{lemma difference equation} we have the following equation for layer $s$
\begin{equation*}  \mathbb{E}\left(\left|\mathcal{L}_{s}\left(\boldsymbol{T}_{n+1}\right)\right|\right)=\mathbb{E}\left(\left|\mathcal{L}_s\left(\boldsymbol{T}_n\right)\right|\right)+\frac{\mathbb{E}\left(\left|\mathcal{L}_{s-1}\left(\boldsymbol{T}_n\right)\right|\right)+\mathbb{E}\left(\left|\mathcal{L}_s\left(\boldsymbol{T}_n\right)\right|\right)}{2 \left(n-1\right)} =  \frac{2n-1}{2 n-2}\mathbb{E}\left(\left|\mathcal{L}_s\left(\boldsymbol{T}_n\right)\right|\right)+\frac{\mathbb{E}\left(\left|\mathcal{L}_{s-1}\left(\boldsymbol{T}_n\right)\right|\right)}{2 n-2}. 
\end{equation*}
Divide both side with $m_{n+1}$ with the result \eqref{lemma 5 eqn 1}, we have 
\begin{align*}    \frac{\mathbb{E}\left(\left|\mathcal{L}_{s}\left(\boldsymbol{T}_{n+1}\right)\right|\right)}{m_{n+1}}
&=\frac{\mathbb{E}\left(\left|\mathcal{L}_{s}\left(\boldsymbol{T}_{n}\right)\right|\right)}{m_{n}}+\frac{\mathbb{E}\left(\left|\mathcal{L}_{s-1}\left(\boldsymbol{T}_n\right)\right|\right)}{\left(2n-2\right) m_{n+1}}
=\frac{\mathbb{E}\left(\left|\mathcal{L}_{s}\left(\boldsymbol{T}_{n}\right)\right|\right)}{m_{n}}+\frac{\mathbb{E}\left(\left|\mathcal{L}_{s-1}\left(\boldsymbol{T}_n\right)\right|\right)}{\left(2n-1\right) m_{n}} \\
&\asymp \frac{\mathbb{E}\left(\left|\mathcal{L}_{s}\left(\boldsymbol{T}_{n}\right)\right|\right)}{m_{n}}+ \frac{\left(\log n\right)^{s-2}}{2n-1} .
\end{align*}
Applying the above inequality repeatedly starting from $n=2$ and noting that $\frac{\mathbb{E}( | \mathcal{L}_s(\bm{T}_2)|)}{m_2} = 0$ since $s > 1$, we obtain
\begin{equation*}   \frac{\mathbb{E}\left(\left|\mathcal{L}_{s}\left(\boldsymbol{T}_{n}\right)\right|\right)}{m_{n}}\asymp \sum^{n-1}_{i=2}\frac{\left(\log i\right)^{s-2}}{2i-1} \asymp  \left(\log n\right)^{s-1}, 
\end{equation*}
where the last inequality follows from Lemma~\ref{lemma integration}. Therefore, $\mathbb{E}\left(\left|\mathcal{L}_{s}\left(\boldsymbol{T}_{n}\right)\right|\right)/ m_{n} \asymp \left(\log n\right)^{s-1}$. 
We have shown that the inductive hypothesis holds and the proof is thus complete. 

\end{proof}

\begin{lemma}
\label{lemma layer noise}
Let $\boldsymbol{G}_n \sim \mathrm{PF}(0, \theta ,\ell, \pi)$, and suppose Assumptions~\ref{assumption: bound} and \ref{assumption: sparsity} hold. Then for any fixed $C > 0$ and $s\in \mathbb{N}$, the following result holds:
\begin{equation}
\label{lemma layer noise eqn 1}
     \lim_{n \rightarrow \infty} \mathbb{P}\bigl\{\bigl|E_{\boldsymbol{R}_n}(V^{C,0}(\boldsymbol{G}_n), \mathcal{L}_s(\boldsymbol{F}_n))\bigr|=0\bigr\}=1,
\end{equation}
where $\mathcal{L}_s\left(\boldsymbol{F}_n\right)$ denotes all the nodes in layer $s$ of $\boldsymbol{T}^i$.
Furthermore, letting
\begin{align}
\label{lemma layer noise chain}
\nonumber
E^2_{\boldsymbol{R}_n}  \left(V^{C,0}_{\boldsymbol{G}_n}\bigl(\boldsymbol{T}^j\bigr), \mathcal{L}_s\bigl(\boldsymbol{T}^i\bigr)\right)&=\left\{\left(v_1, v_2, v_3\right)\bigg| \left(v_1, v_2\right)\in E_{\boldsymbol{R}_n}\left(V^{C,0}_{\boldsymbol{G}_n}\bigl(\boldsymbol{T}^j\bigr), V\left(\boldsymbol{G}_n\right)\right), \right. \\
& \qquad \qquad \qquad
\left. \left(v_2, v_3\right)\in E_{\boldsymbol{R}_n}\left(V\left(\boldsymbol{G}_n\right), \mathcal{L}_s\bigl(\boldsymbol{T}^i\bigr)\right), \right\}, 
\end{align}
we have
\begin{equation}
\label{lemma layer noise eqn 2}
\frac{\left|\bigcup_{\substack{i,j\in[K]\\ i\neq j}} E^2_{\boldsymbol{R}_n}\left(V^{C,0}_{\boldsymbol{G}_n}\left(\boldsymbol{T}^j\right), \mathcal{L}_s\left(\boldsymbol{T}^{i}\right)\right)\right|}{\left| \mathcal{L}_s\left(\boldsymbol{F}_n\right)\right|}= O_p\left(n^{-2\delta}\right).  
\end{equation}

And, letting
\begin{equation}
\label{lemma layer noise eqn 3}
V_1^C\left(\boldsymbol{G}_n\right):=\left\{ v\in V\left(\boldsymbol{G}_n\right) \bigg| \operatorname{dist}_{\bm{F}_n}\left(v, V^{C,0}\left(\boldsymbol{G}_n\right)\right)\le 1\right\},    
\end{equation}
we have
\begin{equation}
\label{lemma layer noise eqn 4}
\frac{\left|E_{\boldsymbol{R}_n}\left(V_1^C\left(\boldsymbol{G}_n\right), \mathcal{L}_s\left(\boldsymbol{F}_n\right)\right)\right|}{\left| \mathcal{L}_s\left(\boldsymbol{F}_n\right)\right|}= O_p\left(n^{-\delta}\right).  
\end{equation}
\end{lemma}

\begin{proof}
We first show \eqref{lemma layer noise eqn 1}. We fix $\varepsilon > 0$. 
Since $V^{C,0}\left(\boldsymbol{G}_{n}\right)=\cup_{i=1}^K V^{C,0}_{\boldsymbol{G}_n}\left(\boldsymbol{T}^i\right)$, it follows that 
\begin{equation}
\label{lemma layer noise eqn 5}
\mathbb{P}\left(\left|E_{\boldsymbol{R}_n}\left(V^{C,0}\left(\boldsymbol{G}_n\right), \mathcal{L}_s\left(\boldsymbol{F}_n\right)\right)\right|=0\right)=\mathbb{P}\left(\left|E_{\boldsymbol{R}_n}\left(\cup_{i=1}^K V^{C,0}_{\boldsymbol{G}_n}\left(\boldsymbol{T}^i\right), \mathcal{L}_s\left(\boldsymbol{F}_n\right)\right)\right|=0\right). 
\end{equation}

By Lemma~\ref{lemma layer size} and Markov's inequality, there exists a universal constant $\tilde{C} > 0$ such that $\left|\mathcal{L}_s\left(\boldsymbol{F}_n\right)\right|\le \tilde{C} \varepsilon^{-1} n^{1/2}\left(\log n\right)^{s-1}$ holds with probability higher than $1-\varepsilon$.
Therefore, we denote $\Omega_0=\left\{\left|\mathcal{L}_s\left(\boldsymbol{F}_n\right)\right|\le \tilde{C} \varepsilon^{-1} n^{1/2}\left(\log n\right)^{s-1}\right\}$, $\Omega_i=\left\{ V^{C,0}_{\boldsymbol{G}_n}\left(\boldsymbol{T}^i\right)\subset\pi_{1:{L\left(\varepsilon, C\right)}}\left(\boldsymbol{T}^i\right)\right\}$ for $1\le i\le K$, and $\Omega:=\cap_{i=0}^K \Omega_i$, we now aim to bound the right-hand side of \eqref{lemma layer noise eqn 5}  on the event $\Omega$:
\begin{eqnarray}
\label{lemma layer noise eqn 6}
\nonumber
&&\mathbb{P}\left(\left|E_{\boldsymbol{R}_n}\left(\cup_{i=1}^K V^{C,0}_{\boldsymbol{G}_n}\left(\boldsymbol{T}^i\right),\mathcal{L}_s\left(\boldsymbol{F}_n\right)\right)\right|=0\right)\\
\nonumber
&=&
\mathbb{P}\left(\left\{\left|E_{\boldsymbol{R}_n}\left(\cup_{i=1}^K V^{C,0}_{\boldsymbol{G}_n}\left(\boldsymbol{T}^i\right),\mathcal{L}_s\left(\boldsymbol{F}_n\right)\right)\right|=0\right\}\cap \Omega\right)+
\mathbb{P}\left(\left\{\left|E_{\boldsymbol{R}_n}\left(V^{C,0}\left(\boldsymbol{G}_n\right),\mathcal{L}_s\left(\boldsymbol{F}_n\right)\right)\right|=0\right\}\cap \Omega^c\right)\\
\nonumber
&\geq& \mathbb{P}\left(\left\{\left|E_{\boldsymbol{R}_n}\left(\cup_{i=1}^K \pi_{1:L\left(\varepsilon,C\right)}\left(\boldsymbol{T}^i\right),\mathcal{L}_s\left(\boldsymbol{F}_n\right)\right)\right|=0\right\}\cap \Omega\right)\\
\nonumber
&\geq & \mathbb{P}\left(\left\{\left|E_{\boldsymbol{R}_n}\left(\cup_{i=1}^K \pi_{1:L\left(\varepsilon,C\right)}\left(\boldsymbol{T}^i\right),\mathcal{L}_s\left(\boldsymbol{F}_n\right)\right)\right|=0\right\}\cap \Omega_0 \right)-\mathbb{P}\left(\Omega^c\setminus\Omega_0^c\right).
\\
&\geq & \mathbb{P}\left(\left\{\left|E_{\boldsymbol{R}_n}\left(\cup_{i=1}^K \pi_{1:L\left(\varepsilon,C\right)}\left(\boldsymbol{T}^i\right),\mathcal{L}_s\left(\boldsymbol{F}_n\right)\right)\right|=0\right\}\cap \Omega_0 \right)-\mathbb{P}\left(\Omega^c\right).
\end{eqnarray}

We begin by bounding the first term $\mathbb{P}\left(\left\{\left|E_{\boldsymbol{R}_n}\left(\cup_{i=1}^K \pi_{1:L\left(\varepsilon,C\right)}\left(\boldsymbol{T}^i\right),\mathcal{L}_s\left(\boldsymbol{F}_n\right)\right)\right|=0\right\}\cap \Omega_0 \right)$ in \eqref{lemma layer noise eqn 6}. 
Note that the set $\cup_{i=1}^K \pi_{1:L\left(\varepsilon,C\right)}\left(\boldsymbol{T}^i)\right)$ and $\mathcal{L}_s\left(\boldsymbol{F}_n\right)$
is determined entirely by $\bm{F}_n$, which is independent of $\boldsymbol{R}_n$ by Definition~\ref{defin: PF}. Therefore, the number of random edges induced by this set in $\boldsymbol{R}_n$ is stochastic dominant by $\operatorname{Bin} \left( L\left(\varepsilon,C\right)K \left|\mathcal{L}_s\left(\boldsymbol{F}_n\right)\right|, \theta \right)$ (conditioning on $\bm{F}_n$),
since each pair of nodes in the selected set contributes an edge independently with probability $\theta$. Therefore on the event $\Omega_0$
\begin{eqnarray}
\label{lemma layer noise eqn 7}
\nonumber
&&\mathbb{P}\left(\left\{\left|E_{\boldsymbol{R}_n}\left(\cup_{i=1}^K \pi_{1:L\left(\varepsilon,C\right)}\left(\boldsymbol{T}^i\right),\mathcal{L}_s\left(\boldsymbol{F}_n\right)\right)\right|=0\right\}\cap \Omega_0 \right)\\
\nonumber
&=& 1- \mathbb{P}\left(\left\{\left|E_{\boldsymbol{R}_n}\left(\cup_{i=1}^K \pi_{1:L\left(\varepsilon,C\right)}\left(\boldsymbol{T}^i\right),\mathcal{L}_s\left(\boldsymbol{F}_n\right)\right)\right|\ge 1\right\}\cap \Omega_0\right)\\
\nonumber
&\geq & 1- \mathbb{E}\left(\left|E_{\boldsymbol{R}_n}\left(\cup_{i=1}^K \pi_{1:L\left(\varepsilon,C\right)}\left(\boldsymbol{T}^i\right),\mathcal{L}_s\left(\boldsymbol{F}_n\right)\right)\right| \, \bigg |\!\!\biggl|\, \Omega_0\right)\\
\nonumber
&\ge & 1- \tilde{C} \varepsilon^{-1}L\left(\varepsilon,C\right)Kn^{1/2}\left(\log n\right)^{s-1}\theta\\
&=&1- O\left(\left(\log n\right)^{s-1}n^{-\delta}\right),
\end{eqnarray}
where the last inequality comes from the definition of the event $\Omega_0$, and the last identity follows from Assumption~\ref{assumption: sparsity}.

Then we bound the second term $\mathbb{P}\left(\Omega^c\right)$. 
By Lemma~\ref{lemma order--degree} and Lemma~\ref{lemma layer size}, we have
\begin{eqnarray}
  \label{lemma layer noise eqn 8}
  \nonumber
\mathbb{P}\left(\Omega^c\right)&=&\mathbb{P}\left(\cup_{i=0}^K \Omega_i^c\right)\le \sum_{i=0}^K \mathbb{P}\left(\Omega_i^c\right)\\
\nonumber
&=& \sum_{i=1}^K \mathbb{P}\left(\left\{ V^{C,0}_{\boldsymbol{G}_n}\left(\boldsymbol{T}^i\right)\subset\pi_{1:{L\left(\varepsilon, C\right)}}\left(\boldsymbol{T}^i\right)\right\}^c\right)+\mathbb{P}\left(\left|\mathcal{L}_s\left(\boldsymbol{F}_n\right)\right|\le \tilde{C}\varepsilon^{-1}n^{1/2}\left(\log n\right)^{s-1}\right)\\
&\le &\left(K+1\right)\varepsilon + o(1).
\end{eqnarray}
Combining equations~\eqref{lemma layer noise eqn 5}–\eqref{lemma layer noise eqn 8}, we have shown that
\begin{equation*}  \mathbb{P}\left(\left|E_{\boldsymbol{R}_n}\left(V^{C,0}\left(\boldsymbol{G}_n\right),\mathcal{L}_s\left(\boldsymbol{F}_n\right)\right)\right|=0\right)\geq 1-\left(K+1\right)\varepsilon-O\left(\left(\log n\right)^{s-1}n^{-\delta}\right)-o\left(1\right).
\end{equation*}
Since $\varepsilon>0$ can be arbitrary small, we have 
\begin{equation*}
     \lim_{n\rightarrow \infty} \mathbb{P}\left(\left|E_{\boldsymbol{R}_n}\left(V^{C,0}\left(\boldsymbol{G}_n\right),\mathcal{L}_s\left(\boldsymbol{F}_n\right)\right)\right|=0\right)=1.
\end{equation*}

We now turn to the proof of \eqref{lemma layer noise eqn 2}.  
From the expression of 
$E^{2}_{\boldsymbol{R}_n}\left(
V^{C,0}_{\boldsymbol{G}_n}\left(\boldsymbol{T}^j\right),  \mathcal{L}_s\left(\boldsymbol{T}^{i}\right)\right)$ 
given in \eqref{lemma layer noise chain}, and using the fact that the two sets 
$V^{C,0}_{\boldsymbol{G}_n}\left(\boldsymbol{T}^j\right)$ and 
$\mathcal{L}_s\left(\boldsymbol{T}^{i}\right)$ are disjoint for $i\neq j$, we obtain
\begin{align*}    \left|E^{2}_{\boldsymbol{R}_n}\left(
V^{C,0}_{\boldsymbol{G}_n}\left(\boldsymbol{T}^j\right),  \mathcal{L}_s\left(\boldsymbol{T}^{i}\right)\right)\right|&=\sum_{v\in V\left(\boldsymbol{G}_n\right)}\left|E_{\boldsymbol{R}_n}\left(
V^{C,0}_{\boldsymbol{G}_n}\left(\boldsymbol{T}^j\right),  v\right)\right|\cdot \left|E_{\boldsymbol{R}_n}\left(
v,  \mathcal{L}_s\left(\boldsymbol{T}^{i}\right)\right)\right|\\
&=\sum_{v\in V\left(\boldsymbol{G}_n\right)}\sum_{u_1\in V^{C,0}_{\boldsymbol{G}_n}\left(\bm{T}^j\right)}\sum_{u_2\in \mathcal{L}_s\left(\boldsymbol{T}^{i}\right)}\left|E_{\boldsymbol{R}_n}\left(
u_1,  v\right)\right|\cdot \left|E_{\boldsymbol{R}_n}\left(
v,  u_2\right)\right|.
\end{align*}
For any three different vertices $v, u_1, u_2$, the product $\left|E_{\boldsymbol{R}_n}\left(v,u_1\right)\right|\cdot \left|E_{\boldsymbol{R}_n}\left(v,u_2\right)\right|\sim \operatorname{Bin}\left(1,\theta^2\right)$.

Next, consider the size of $V^{C,0}_{\boldsymbol{G}_n}\left(\bm{T}^j\right)$
on the event $\Omega$. By definition of $\Omega$, it holds that $V^{C,0}_{\boldsymbol{G}_n}\left(\boldsymbol{T}^j\right)
\subset
\pi_{1:L\left(\varepsilon,C\right)}\left(\boldsymbol{T}^j\right)$.
Hence,
\begin{equation*}
\mathbbm{1}_{\Omega}
\Bigl|
E_{\bm R_n}^{2}
\bigl(
V_{\bm G_n}^{C,0}\bigl(\bm T^j\bigr),
\mathcal L_s \bigl(\bm T^i\bigr)
\bigr)
\Bigr|
\le
\Bigl|
E_{\bm R_n}^{2}
\bigl(
\pi_{1:L(\varepsilon,C)}\bigl(\bm T^j\bigr),
\mathcal L_s\bigl(\bm T^i\bigr)
\bigr)
\Bigr|.
\end{equation*}
Conditioning only on $\bm{F}_n$, the sets
$\pi_{1:L(\varepsilon,C)}(\boldsymbol{T}^j)$ and
$\mathcal{L}_s(\boldsymbol{T}^i)$ are deterministic. Therefore,
using the preceding expression for
$\left|E^{2}_{\boldsymbol{R}_n}(\cdot,\cdot)\right|$, we obtain
\begin{eqnarray*}
&&
\mathbb{E}_{\bm R_n}\Biggl(
\frac{
\mathbbm{1}_{\Omega}
\Bigl|
E_{\bm R_n}^{2}
\bigl(
V_{\bm G_n}^{C,0}\bigl(\bm T^j\bigr),
\mathcal L_s\bigl(\bm T^i\bigr)
\bigr)
\Bigr|
}{
\bigl|\mathcal L_s\bigl(\bm F_n\bigr)\bigr|
}
\,\Big|\,\bm F_n
\Biggr)
\\
&\le&
\mathbb{E}_{\bm R_n}\Biggl(
\frac{
\Bigl|
E_{\bm R_n}^{2}
\bigl(
\pi_{1:L(\varepsilon,C)}\bigl(\bm T^j\bigr),
\mathcal L_s\bigl(\bm T^i\bigr)
\bigr)
\Bigr|
}{
\bigl|\mathcal L_s\bigl(\bm F_n\bigr)\bigr|
}
\,\Big|\,\bm F_n
\Biggr)
\le
nL(\varepsilon,C)\theta^2
\frac{
\bigl|\mathcal L_s\bigl(\bm T^i\bigr)\bigr|
}{
\bigl|\mathcal L_s\bigl(\bm F_n\bigr)\bigr|
}.
\end{eqnarray*}
Summing the preceding bounds over all pairs $i,j\in[K]$ with $i\neq j$
gives
\begin{equation*}
\mathbb{E}_{\bm R_n}\Biggl(
\frac{
\mathbbm{1}_{\Omega}
\Bigl|
\bigcup_{\substack{i,j\in[K]\\ i\neq j}}
E_{\bm R_n}^{2}
\bigl(
V_{\bm G_n}^{C,0}\bigl(\bm T^j\bigr),
\mathcal L_s\bigl(\bm T^i\bigr)
\bigr)
\Bigr|
}{
\bigl|\mathcal L_s\bigl(\bm F_n\bigr)\bigr|
}
\,\Big|\,\bm F_n
\Biggr)
\le
nK L(\varepsilon,C)\theta^2.
\end{equation*}
Applying Markov's inequality conditional on $\bm{F}_n$, we obtain that,
with probability at least $1-\varepsilon$,
\begin{equation*}
\mathbbm{1}_{\Omega}
\frac{
\left|
\bigcup_{\substack{i,j\in[K]\\i\neq j}}
E^{2}_{\boldsymbol{R}_n}
\left(
V^{C,0}_{\boldsymbol{G}_n}\left(\boldsymbol{T}^j\right),
\mathcal{L}_s\left(\boldsymbol{T}^{i}\right)
\right)
\right|
}
{\left|\mathcal{L}_s\left(\boldsymbol{F}_n\right)\right|}
\le
\frac{nK L(\varepsilon,C)\theta^2}{\varepsilon}
\le
C^0n^{-2\delta},
\end{equation*}
where the last inequality follows from Assumption~\ref{assumption: sparsity}. Finally,
\begin{equation*}
\mathbb{P}\Biggl(
\frac{
\Bigl|
\bigcup_{\substack{i,j\in[K]\\ i\neq j}}
E_{\bm R_n}^{2}
\bigl(
V_{\bm G_n}^{C,0}\bigl(\bm T^j\bigr),
\mathcal L_s\bigl(\bm T^i\bigr)
\bigr)
\Bigr|
}{
\bigl|\mathcal L_s\bigl(\bm F_n\bigr)\bigr|
}
>
C^0n^{-2\delta}
\Biggr)
\le
\mathbb{P}\bigl(\Omega^c\bigr)+\varepsilon
\le
(K+2)\varepsilon.
\end{equation*}
This completes the proof of \eqref{lemma layer noise eqn 2}.

Then we turn to the proof of \eqref{lemma layer noise eqn 4}.
To this end, we introduce enlarged sets $\tilde{V}_i$ such that the $V^{C,0}_1\left(\boldsymbol{G}_n\right)$ in \eqref{lemma layer noise eqn 3} is the subset of $\cup_{i=1}^K \tilde{V}_i$ on the event $\Omega$:
\begin{equation*}
  \tilde{V}_{i}:=\left\{v \in V\left(\boldsymbol{G}_n\right) \bigg| \operatorname{dist}_{ \boldsymbol{T}^i}\left(v, \pi_{1:L\left(\varepsilon,C\right)}\left(\boldsymbol{T}^i\right)\right)\le 1\right\}.  
\end{equation*}

Then we define the ``good event'' $\tilde{\Omega}\left(t_1,t_2\right):=\left(\cap_{i=1}^K \bar{\Omega}_i\left(t_1\right)\right) \cap\left(\cap_{i=1}^K \breve{\Omega}_i\left(t_2\right)\right) \cap\Omega$,
where
\begin{equation*}
   \bar{\Omega}_i\left(t_1\right):=\left\{\operatorname{deg}_{\boldsymbol{T}^i}\left(\pi_{1:L\left(\varepsilon,C\right)}\left(\boldsymbol{T}^i\right)\right)\le t_1\sqrt{n}\right\} 
\end{equation*}
and 
\begin{equation*}
\breve{\Omega}_i\left(t_2\right):=\left\{\frac{\left|E_{\boldsymbol{R}_n}\left(\tilde{V_i}, \mathcal{L}_s\left(\boldsymbol{F}_n\right)\right)\right|}{\left| \mathcal{L}_s\left(\boldsymbol{F}_n\right)\right|}\le t_2 n^{-\delta}\right\}.  
\end{equation*}
for any $t_1, t_2>0$.

We can observed that for fixed $t_1, t_2$, on the event $\tilde{\Omega}\left(t_1,t_2\right)$ 
\begin{equation*}   \frac{\left|E_{\boldsymbol{R}_n}\left(V_1^C\left(\boldsymbol{G}_n\right), \mathcal{L}_s\left(\boldsymbol{F}_n\right)\right)\right|}{\left| \mathcal{L}_s\left(\boldsymbol{F}_n\right)\right|} \le \frac{\left|E_{\boldsymbol{R}_n}\left(\cup^K_{i=1}\tilde{V}_i, \mathcal{L}_s\left(\boldsymbol{F}_n\right)\right)\right|}{\left| \mathcal{L}_s\left(\boldsymbol{F}_n\right)\right|}\le K t_2 n^{-\delta}.
\end{equation*}

Then we claim that there exist $t_1, t_2>0$ depending only on $\varepsilon, C$, such that
$\mathbb{P}\left(\tilde{\Omega}^c\right)$ is bounded by $\left(3K+1\right)\varepsilon$ and hence \eqref{lemma layer noise eqn 4} holds.

Since a bound for $\mathbb{P}\left(\Omega^c\right)$ has already been established in the proof of \eqref{lemma layer noise eqn 1}, we try to control $\mathbb{P}\left(\bar{\Omega}_i^c\left(t_1\right)\right)$ and $\mathbb{P}\left(\breve{\Omega}_i^c\left(t_2\right) \cap \bar{\Omega}_i\left(t_1\right)\right)$ separately in the following proof.

By Lemma~\ref{lemma deg bound}, there exists some constant $t_1:=t_1\left(\varepsilon, C\right)$ such that $\mathbb{P}\left(\bar{\Omega}^c_i\left(t_1\right)\right)\le \varepsilon$ for every $i \in [K]$.
The next step is to bound $\mathbb{P}\left(\breve{\Omega}^c_i\left(t_2\right)\cap \bar{\Omega}_i\left(t_1\right)\right)$, We do so by showing that, on event $\bar{\Omega}_i(t_1)$, $\breve{\Omega}_i(t_2)$ stands with high probability for a suitable choice of $t_2$.

Note that the sets $\left\{\tilde{V}_i, i \in \left[K\right]\right\}$ and $\mathcal{L}_s\left(\boldsymbol{F}_n\right)$
is determined entirely by $\bm{F}_n$, which is independent of $\boldsymbol{R}_n$ by Definition~\ref{defin: PF}. Therefore, the number of random edges induced by this set in $\boldsymbol{R}_n$, conditioning on $\boldsymbol{F}_n$ is stochastic dominated by $\operatorname{Bin} \left(\left|\tilde{V_i}\right| \left|\mathcal{L}_s\left(\boldsymbol{F}_n\right)\right|, \theta \right)$.
Combining the fact that, on the event $\bar{\Omega}_i$, $\left|\tilde{V}_i\right|\le t_1 L\left(\varepsilon, C\right)\sqrt{n}$, 
\begin{equation*}    \mathbb{E}_{\bm{R}_n}\left(\frac{\left|E_{\boldsymbol{R}_n}\left(\tilde{V_i}, \mathcal{L}_s\left(\boldsymbol{F}_n\right)\right)\right|}{\left|\mathcal{L}_s\left(\boldsymbol{F}_n\right)\right|} \right)
\le \left|\tilde{V_i}\right| \theta
\le t_1 L\left(\varepsilon, C\right)n^{1/2}\theta.
\end{equation*}
By Markov inequality and Assumption~\ref{assumption: sparsity}, for any $\varepsilon>0$, there exist and a universal constant $C^1>0$ such that the following holds with probability at least $1-\varepsilon$:
\begin{equation*}   \frac{\left|E_{\boldsymbol{R}_n}\left(\tilde{V_i}, \mathcal{L}_s\left(\boldsymbol{F}_n\right)\right)\right|}{\left| \mathcal{L}_s\left(\boldsymbol{F}_n\right)\right|}\le \varepsilon^{-1} t_1 L\left(\varepsilon, C\right)n^{1/2}\theta\le  C^1\varepsilon^{-1} t_1 L\left(\varepsilon, C\right)n^{-\delta}. 
\end{equation*}
Therefore, takes $t_2:=C^1\varepsilon^{-1} t_1L\left(\varepsilon, C\right)$ the bound $\mathbb{P}\left(\breve{\Omega}^c_i\left(t_2\right)\cap \bar{\Omega}_i\left(t_1\right)\right)\le \varepsilon$ is satisfied.

Set $\bar{\Omega}_i:= \bar{\Omega}_i\left(t_1\right)$, $\breve{\Omega}_i:= \breve{\Omega}_i\left(t_2\right)$, $\tilde{\Omega}:=\tilde{\Omega}\left(t_1,t_2\right)$.
Combining the preceding bounds with an union bound, we 
\begin{eqnarray*}
\mathbb{P}\left(\tilde{\Omega}^c\right)&=& \mathbb{P}\left(\left(\left(\cap_{i=1}^K \bar{\Omega}_i\right) \cap\left(\cap_{i=1}^K \breve{\Omega}_i\right) \cap\Omega\right)^c\right)=\mathbb{P}\left(\left(\left(\cap_{i=1}^K \bar{\Omega}_i\right) \cap\left(\cap_{i=1}^K \left(\breve{\Omega}_i\cup  \bar{\Omega}^c_i \right)\right) \cap\Omega\right)^c\right)\\
\nonumber
&\le & \sum_{i=1}^K \mathbb{P}\left(\bar{\Omega}^c_i\right)+\sum_{i=1}^K \mathbb{P}\left(\breve{\Omega}_i^c \cap \bar{\Omega}_i\right)+ \mathbb{P}\left(\Omega^c\right)\le K\varepsilon+K\varepsilon+\left(K+1\right)\varepsilon\le \left(3K+1\right)\varepsilon.  
\end{eqnarray*}
This completes the proof of \eqref{lemma layer noise eqn 4}.

\end{proof}

\subsection{Technical Lemmas}
\label{secsec: technical lemmas}

\begin{lemma}[Theorem 2 in \cite{rudas2007random}]
\label{lemma: degree 1 proportion}
For $\boldsymbol{T}_{n} \sim \mathrm{APA}(\alpha, n)$ with $\alpha > -1$, as $n \to \infty$,
$$
\frac{\left|\left\{v: \operatorname{deg}_{\bm{T}_n}\left(v\right)=1\right\}\right|}{n}\xrightarrow{\ P\ } p\left(1\right)>0,
$$
that is, the proportion of degree-$1$ vertices converges in probability to a positive constant $p(1)$ that depends on $\alpha$.
\end{lemma}

\begin{lemma}[Proposition 5.5 in \cite{senizergues2021geometry}]
\label{lemma: T}
Let $q > 2+\alpha$. For $\boldsymbol{T}_{n}\sim \mathrm{APA}\left(\alpha, n\right)$, for $\alpha>-1$, when $n$ tends to infinity,
$$
n^{-\frac{1}{2+\alpha}}\left(\operatorname{deg}_{\boldsymbol{T}_n}\left(\pi_1\left(\boldsymbol{T}_n\right)\right), \operatorname{deg}_{\boldsymbol{T}_n}\left(\pi_2\left(\boldsymbol{T}_n\right)\right), \ldots, \operatorname{deg}_{\boldsymbol{T}_n}\left(\pi_n\left(\boldsymbol{T}_n\right)\right), 0,0, \ldots\right) \xrightarrow{a.s.}\left(Y_{1,\alpha}, Y_{2, \alpha}, Y_{3, \alpha}, \ldots\right),
$$
almost surely with respect to the $\ell_q$ metric where ($Y_{1, \alpha}, Y_{2, \alpha}, \ldots$) is a random sequence satisfying $\sum_{j=1}^{\infty} \mathbb{E} Y_{j,\alpha}^q<\infty$ and each random variable $Y_{j, \alpha}$ has a density with respect to the Lebesgue measure.
\end{lemma}

\begin{lemma}
   \label{lemma deg bound}
For $\boldsymbol{T}_{n}\sim \mathrm{APA}\left(\alpha, n\right)$, for $\alpha>-1$, for any fixed $L\in \mathbb{N}$, 
$$
\operatorname{deg}_{\boldsymbol{T}_n}\left(\pi_{1:L}\left(\boldsymbol{T}_n\right)\right)=O_p\left(n^{\frac{1}{2+\alpha}}\right).
$$
\end{lemma}
\begin{proof}
Let ($Y_{1,\alpha}, Y_{2,\alpha}, \ldots$) is a random sequence from Lemma~\ref{lemma: T},
we first show that for any $j \in\left[L\right], \varepsilon>0$, there exist a constant $C:=C(\varepsilon,L)$ such that $\mathbb{P}\left(Y_{j, \alpha}< C\right)>1-\frac{\varepsilon}{L}$.
By Lemma~\ref{lemma: T}, each $Y_{j, \alpha}$ admits a density $q_j(\cdot)$ supported on $[0,\infty)$ with respect to the Lebesgue measure. Consequently, for any $\varepsilon > 0$, we can select $C_j:=C_j\left(\varepsilon\right)$ such that $\mathbb{P}\left(Y_{j, \alpha}\geq  C_j\right)\le \frac{\varepsilon}{L}$, for any $j \in \left[L\right]$. We may then take $C(\varepsilon,L)=\max_{1\le j \le L}C_j(\varepsilon)$. 

We now consider the distribution of $\operatorname{deg}_{\boldsymbol{T}_n}\left(\pi_{j}\left(\boldsymbol{T}_n\right)\right)$. By  Lemma~\ref{lemma: T} and Portmanteau lemma, we have
\begin{equation}
\label{lemma deg bound eqn 1}
    \limsup_{n\rightarrow \infty}\mathbb{P}\left(n^{-\frac{1}{2+\alpha}}\operatorname{deg}_{\boldsymbol{T}_n}\left(\pi_j\left(\boldsymbol{T}_n\right)\right)\geq C\right)\le \mathbb{P}\left(Y_{j, \alpha} \geq C\right)\le \mathbb{P}\left(Y_{j, \alpha}\geq  C_j\right)\le \frac{\varepsilon}{L}.
\end{equation}
Then we proceed to show that $\operatorname{deg}_{\boldsymbol{T}_n}\left(\pi_{1:L}\left(\boldsymbol{T}_n\right)\right)=O_p\left(n^{\frac{1}{2+\alpha}}\right)$:
\begin{eqnarray*}
\liminf_{n\rightarrow \infty}\mathbb{P}\left(\operatorname{deg}_{\boldsymbol{T}_n}\left(\pi_{1:L}\left(\boldsymbol{T}_n\right)\right)< Cn^{\frac{1}{2+\alpha}}\right) &=& \liminf_{n\rightarrow \infty}\mathbb{P}\left(\max_{1\le j\le L}\operatorname{deg}_{\boldsymbol{T}_n}\left(\pi_{j}\left(\boldsymbol{T}_n\right)\right)< Cn^{\frac{1}{2+\alpha}}\right)  \\
&=& 1-\limsup_{n\rightarrow \infty}\mathbb{P}\left(\max_{1\le j\le L}\operatorname{deg}_{\boldsymbol{T}_n}\left(\pi_{j}\left(\boldsymbol{T}_n\right)\right)\geq  Cn^{\frac{1}{2+\alpha}}\right)\\
&=& 1-\limsup_{n\rightarrow \infty}\mathbb{P}\left(\cup^L_{j=1}\left\{\operatorname{deg}_{\boldsymbol{T}_n}\left(\pi_{j}\left(\boldsymbol{T}_n\right)\right)\geq Cn^{\frac{1}{2+\alpha}}\right\}\right)\\
&\ge & 1-\sum_{j=1}^L \limsup_{n\rightarrow \infty}\mathbb{P}\left(\operatorname{deg}_{\boldsymbol{T}_n}\left(\pi_{j}\left(\boldsymbol{T}_n\right)\right)\geq  Cn^{\frac{1}{2+\alpha}}\right)\\
&\ge &1-\sum_{j=1}^L \frac{\varepsilon}{L}=1-\varepsilon,
\end{eqnarray*}
where the first equality follows directly from the definition of $\deg_{\boldsymbol{T}_n}\left(\pi_{1:L}(\boldsymbol{T}_n)\right)$, while the final inequality is a consequence of \eqref{lemma deg bound eqn 1}.
Therefore, for all sufficiently large $n$
\begin{equation*}
\mathbb{P}\left(\operatorname{deg}_{\boldsymbol{T}_n}\left(\pi_{1:L}\left(\boldsymbol{T}_n\right)\right)< Cn^{\frac{1}{2+\alpha}}\right)\geq 1-\varepsilon-o\left(1\right)\ge 1-2\varepsilon,    
\end{equation*}
and the Lemma follows as desired.
\end{proof}

\begin{lemma}
\label{lemma R}
Let $q > \left(2+\alpha\right) \vee \frac{1}{\delta}$. For $\boldsymbol{G}_n \sim \mathrm{PF}(\alpha, \theta, \ell, \pi)$ with decomposition in \eqref{decomposition}, under Assumption~\ref{assumption: bound}, \ref{assumption: sparsity}. For any $i \in \left[K\right]$, as $n \rightarrow \infty$,
$$
n_i^{-\frac{1}{2+\alpha}}\left(\operatorname{deg}_{\boldsymbol{G}_n}\bigl(\pi_1\bigl(\boldsymbol{T}^i\bigr)\bigr)-\operatorname{deg}_{\boldsymbol{T}^i}\bigl(\pi_1\bigl(\boldsymbol{T}^i\bigr)\bigr), \ldots, \operatorname{deg}_{\boldsymbol{G}_n}\bigl(\pi_{n_i}\bigl(\boldsymbol{T}^i\bigr)\bigr)-\operatorname{deg}_{\boldsymbol{T}^i}\bigl(\pi_{n_i}\bigl(\boldsymbol{T}^i\bigr)\bigr), 0, \ldots\right) \xrightarrow{d}\left(0, 0, \ldots\right),
$$
in distribution with respect to the $\ell_q$ metric.
\end{lemma}

\begin{proof}
Indeed, we have
$$
\begin{aligned}
&\mathbb{E} \left\|n_i^{-\frac{1}{2+\alpha}}\left(\operatorname{deg}_{\boldsymbol{G}_n}\left(\pi_1\bigl(\boldsymbol{T}^i\bigr)\right)-\operatorname{deg}_{\boldsymbol{T}^i}\left(\pi_1\bigl(\boldsymbol{T}^i\bigr)\right), \ldots, \operatorname{deg}_{\boldsymbol{G}_n}\left(\pi_{n_i}\bigl(\boldsymbol{T}^i\bigr)\right)-\operatorname{deg}_{\boldsymbol{T}^i}\left(\pi_{n_i}\bigl(\boldsymbol{T}^i\bigr)\right), 0, \ldots\right) \right\|_q^q\\
&\le n_i^{-\frac{q}{2+\alpha}} \sum_{k=1}^{n_i} \mathbb{E}\operatorname{deg}_{\boldsymbol{R}_n}\left(\pi_1\bigl(\boldsymbol{T}^i\bigr)\right)^q\stackrel{(a)}{=} n_i^{1-\frac{q}{2+\alpha}} \mathbb{E}\left(\operatorname{Bin}(n-1, \theta)^q\right) \stackrel{(b)}{\leq} n_i^{1-\frac{q}{2+\alpha}}\left((2 \theta n)^q+C_q\right) \\
& \leq 2^q C^q_0 n_i^{1-\frac{q}{2+\alpha}} n^{\frac{q}{2+\alpha}-q \delta}+C_q n_i^{1-\frac{q}{2+\alpha}}\le 2^q C^q_0 H^{\frac{q}{2+\alpha}-1}n^{1-q \delta}+C_q H^{\frac{q}{2+\alpha}-1} n^{1-\frac{q}{2+\alpha}}
\end{aligned}
$$
where the equality (a) follows since $\operatorname{deg}_{\boldsymbol{R}_n}\left(\pi_j\bigl(\boldsymbol{T}^i\bigr)\right)$ is distributed as $\operatorname{Bin}(n-1, \theta)$ and where inequality (b) follows from Lemma~\ref{lemma cq}.

Since $q>\left(2+\alpha\right) \vee 1 / \delta$ by assumption, we have that, as $n \rightarrow \infty$,
$$
\mathbb{E}\left\|n_i^{-\frac{1}{2+\alpha}}\left(\operatorname{deg}_{\boldsymbol{G}_n}\left(\pi_1\bigl(\boldsymbol{T}^i\bigr)\right)-\operatorname{deg}_{\boldsymbol{T}^i}\left(\pi_1\bigl(\boldsymbol{T}^i\bigr)\right), \ldots, \operatorname{deg}_{\boldsymbol{G}_n}\left(\pi_{n_i}\bigl(\boldsymbol{T}^i\bigr)\right)-\operatorname{deg}_{\boldsymbol{T}^i}\left(\pi_{n_i}\bigl(\boldsymbol{T}^i\bigr)\right), 0, \ldots\right) \right\|^q_q \rightarrow 0
$$
and the Lemma follows as desired.
\end{proof}

\begin{lemma}
\label{lemma: G}
Let $q>\left(2+\alpha\right) \vee 1 / \delta$. For $\boldsymbol{G}_n \sim \mathrm{PF}(\alpha, \theta, \ell, \pi)$ with decomposition in \eqref{decomposition}, under  Assumption~\ref{assumption: bound} and \ref{assumption: sparsity}. For any $i \in \left[K\right]$, as $n \rightarrow \infty$,
$$
n_i^{-\frac{1}{2+\alpha}}\left(\operatorname{deg}_{\boldsymbol{G}_n}\left(\pi_1\bigl(\boldsymbol{T}^i\bigr)\right), \operatorname{deg}_{\boldsymbol{G}_n}\left(\pi_2\bigl(\boldsymbol{T}^i\bigr)\right), \ldots, \operatorname{deg}_{\boldsymbol{G}_n}\left(\pi_{n_i}\bigl(\boldsymbol{T}^i\bigr)\right), 0,0, \ldots\right) \xrightarrow{d}\left(Y_{1, \alpha}, Y_{2, \alpha}, Y_{3, \alpha}, \ldots\right),
$$
in distribution with respect to the $\ell_q$ metric where ($Y_{1, \alpha}, Y_{2, \alpha}, \ldots$) is a random sequence satisfying $\sum_{j=1}^{\infty} \mathbb{E} Y_{j, \alpha}^q<\infty$ and each random variable $Y_{j, \alpha}$ has a density with respect to the Lebesgue measure.
\end{lemma}
\begin{proof}
This lemma is a direct result of applying Slutsky Lemma with results with Lemma~\ref{lemma: T} and Lemma~\ref{lemma R}.
\end{proof}

\begin{lemma}[Lemma S8 in \cite{crane2024root}]
\label{lemma cq}
Let $X$ be a random variable with $\operatorname{Bin}(n, \theta)$ distribution. For any $q \geq 1, \theta \in[0,1]$ and any $n \in \mathbb{N}$, we have that
$$
\mathbb{E} X^q \leq(2 \theta n)^q+C_q,
$$
where $C_q>0$ is a constant that depends only on $q$.
\end{lemma}
\begin{proof}
Write $X$ as a random variable with the $\operatorname{Bin}(n, \theta)$ distribution. Then,
\begin{equation}
\label{lemma cq eqn 1}
\mathbb{E} X^q  =\int_0^{\infty} \mathbb{P}\left(X^q \geq t\right) d t \leq(2 \theta n)^q+\int_{(2 \theta n)^q}^{\infty} \mathbb{P}\left(X^q \geq t\right) d t    
\end{equation}
We note that $\operatorname{Var} X \leq \theta n$. By Bernstein's inequality, we have that for all $t \geq(2 \theta n)^q$,
\begin{equation*}
\mathbb{P}\left(X^q \geq t\right) =\mathbb{P}\left(X-\theta n \geq t^{1 / q}-\theta n\right) \leq \exp \Biggl(-\frac{1}{2} \frac{\bigl(t^{1 / q}-\theta n\bigr)^2}{\left(t^{1 / q}-\theta n\right)+\theta n}\Biggr)  \leq \exp \left(-\frac{1}{8} t^{1 / q}\right) 
\end{equation*}
Therefore, we may bound the second term of \eqref{lemma cq eqn 1} as
\begin{equation*}
\int_{(2 \theta n)^q}^{\infty} \mathbb{P}\left(X^q \geq t\right) d t \leq \int_{(2 \theta n)^q}^{\infty} e^{-\frac{t^{1 / q}}{8}} d t \leq \int_0^{\infty} q s^{q-1} e^{-\frac{s}{8}} d s < \infty.  
\end{equation*}
The lemma thus follows by writing $C_q := \int_0^{\infty} q s^{q-1} e^{-\frac{s}{8}} d s$.
\end{proof}

\begin{lemma}
\label{lem:edge_count_inequality}
Let $\bm{g} = (V, E)$ be an arbitrary graph without self-loop and let $A, B \subset V(\bm{g})$. Then, we have that
\[
\sum_{v \in A} |E(v, B)| \leq 2 |E(A, B)|.
\]
\end{lemma}
\begin{proof}
Define 
\[
I := A \cap B, \quad A_1 = A \cap B^c = A \cap I^c, \quad B_1 = B \cap A^c = B \cap I^c,
\]
and note that 
\[
E(A, B) = E(I, B) \sqcup E(A_1, B) = E(I, I) \sqcup E(I, B_1) \sqcup E(A_1, B).
\]

Then, we have that
\begin{align*}
\sum_{v \in A} |E(v, B)| &= \sum_{v \in I} |E(v, B)| + \sum_{v \in A_1} |E(v, B)| \\
&= \sum_{v \in I} |E(v, I)| + \sum_{v \in I} |E(v, B_1)| + \sum_{v \in A_1} |E(v, B)| \\
&= 2 |E(I, I)| + |E(I, B_1)| + |E(A_1, B)| \leq 2 |E(A, B)|.
\end{align*}
The lemma thus follows as desired.
\end{proof}

\begin{lemma}
\label{lemma integration}
For fixed integers $s \in \mathbb{N}$, 
\begin{equation*}
    \sum_{i=1}^{n-1} \frac{\left(\log i\right)^{s-1}}{2i-1} \asymp (\log n)^s.
\end{equation*}
\end{lemma}

\begin{proof}
It suffices to consider the sum from $i=2$ to $n-1$, the possible $i=1$ term is either zero or a fixed constant and is irrelevant for the asymptotic order.
We first show that $\sum_{i=2}^{n-1} \frac{\left(\log i\right)^{s-1}}{2i-1} \le \frac{(\log n)^s}{s}$, begin from the left hand side:
\begin{eqnarray*}
\sum_{i=2}^{n-1} \frac{\left(\log i\right)^{s-1}}{2i-1} &\le  &  \sum_{i=2}^{n-1} \left(\frac{1}{2i-1} \int^{i+1}_i \left(\log t\right)^{s-1}dt \right)\le \sum_{i=2}^{n-1} \int^{i+1}_i \frac{ \left(\log t\right)^{s-1}}{t}dt= \int^n_2 \frac{ \left(\log t\right)^{s-1}}{t}dt\\
&=& \int^{\log n}_{\log 2}  y^{s-1}d y=\frac{\left(\log n\right)^{s}}{s}-\frac{\left(\log 2\right)^{s}}{s} \leq \frac{(\log n)^s}{s}.
\end{eqnarray*}
In the first identity of the second line, we use a change of variables $y = \log t$. 

Then we show that $\sum_{i=2}^{n-1} \frac{\left(\log i\right)^{s-1}}{2i-1} \ge \frac{(\log n)^s}{4s}$, begin from the left hand side:
\begin{eqnarray*}
\sum_{i=2}^{n-1} \frac{\left(\log i\right)^{s-1}}{2i-1} &\ge & \sum_{i=2}^{n-1} \left(\frac{1}{2i-1} \int^{i}_{i-1} \left(\log t\right)^{s-1}dt \right)\ge \sum_{i=2}^{n-1} \int^{i}_{i-1} \frac{ \left(\log t\right)^{s-1}}{3t}dt = \int^{n-1}_1 \frac{ \left(\log t\right)^{s-1}}{3t}dt\\
&=& \int^{\log \left(n-1\right)}_{\log 1} \frac{1}{3} y^{s-1}d y=\frac{\left(\log \left(n-1\right)\right)^{s}}{3s}-\frac{\left(\log 1\right)^{s}}{3s} \geq \frac{(\log n)^s}{4s}.
\end{eqnarray*}
In the second identity of the second line, we use a change of variables $y = \log t$. The last inequality stands for all sufficiently large $n$.
The Lemma thus follows as desired. 
\end{proof}

\end{document}